\documentclass[11pt, letterpaper]{article}
\usepackage[utf8]{inputenc}
\usepackage{graphicx}
\usepackage{float}
\usepackage{caption}
\usepackage{subcaption}
\usepackage{multirow, multicol}

\usepackage{natbib}

\usepackage{amsmath, amssymb, amsthm, mathtools,stmaryrd}
\usepackage{bbm}
\usepackage{algorithm, algorithmic}

\usepackage{enumitem}
\usepackage{lmodern, float}
\usepackage{sidecap}
\usepackage[utf8]{inputenc} % allow utf-8 input
\usepackage[T1]{fontenc}    % use 8-bit T1 fonts
\usepackage{url}            % simple URL typesetting
\usepackage{booktabs}       % professional-quality tables
\usepackage{amsfonts}       % blackboard math symbols
\usepackage{nicefrac}       % compact symbols for 1/2, etc.
\usepackage{microtype}      % microtypography
\usepackage{braket}
\usepackage[colorlinks,allcolors=blue, pagebackref]{hyperref}

\usepackage{times}

\newtheorem{theorem}{Theorem}[section]
\newtheorem*{theorem*}{Theorem}

\newtheorem{proposition}[theorem]{Proposition}
\newtheorem*{proposition*}{Proposition}
\newtheorem{lemma}[theorem]{Lemma}
\newtheorem*{lemma*}{Lemma}
\newtheorem{corollary}[theorem]{Corollary}
\newtheorem*{conjecture*}{Conjecture}

\newtheorem*{fact*}{Fact}

\newtheorem*{hypothesis*}{Hypothesis}
\newtheorem{conjecture}[theorem]{Conjecture}
\newtheorem{itheorem}[theorem]{(Informal) Theorem}
\newtheorem{claim}[theorem]{Claim}
\newtheorem*{claim*}{Claim}

\theoremstyle{definition}
\newtheorem{definition}[theorem]{Definition}

\theoremstyle{remark}
\newtheorem{remark}[theorem]{Remark}
\newtheorem*{remark*}{Remark}

\newtheorem*{observation*}{Observation}
\newcommand{\R}{\mathbb{R}}

\newcommand{\calM}{\mathcal{M}}

\newcommand{\calQ}{\mathcal{Q}}

\newcommand{\poly}{\mathrm{poly}}
\newcommand{\diag}{\mathrm{diag}}

\newcommand{\norm}[1]{\lVert #1 \rVert}

\newcommand{\Bignorm}[1]{\Big\lVert#1\Big\rVert}
\newcommand{\spnorm}[1]{\norm{#1}}

\newcommand{\gnorm}[1]{{\vert\kern-0.25ex\vert\kern-0.25ex\vert #1 
    \vert\kern-0.25ex\vert\kern-0.25ex\vert}}

\newcommand{\unorm}[1]{\gnorm{#1}}
\newcommand{\Bigspnorm}[1]{\Bignorm{#1}}

\newcommand{\iprod}[1]{\langle#1\rangle}

\newcommand{\Bigiprod}[1]{\Big\langle#1\Big\rangle}

\newcommand{\Esymb}{\mathbb{E}}
\newcommand{\Psymb}{\mathbb{P}}

\DeclareMathOperator*{\E}{\Esymb}

\DeclareMathOperator*{\ProbOp}{\Psymb}

\DeclareMathOperator*{\argmin}{\text{argmin}}
\renewcommand{\Pr}{\ProbOp}

\newcommand{\tr}{\text{tr}}

\newcommand{\eps}{\varepsilon}
\renewcommand{\epsilon}{\varepsilon}

\newcommand{\dx}{dx}

\newcommand{\Bad}{\textsc{Bad Input}~}

\newcommand{\Mean}{\textsc{Mean}}
\newcommand{\RMean}{\textsc{RobustMean}}

\newcommand{\wt}{\widetilde}

\newcommand{\CPone}{\textbf{CP1}~}
\newcommand{\CPtwo}{\textbf{CP2}~}
\newcommand{\CPthree}{\textbf{CP3}~}
\newcommand{\CPfour}{\textbf{CP4}~}

\newcommand{\bfone}{\mathbf{1}}

\newenvironment{proofof}[1]{\bigskip \noindent {\it Proof of #1.}\quad }
{\qed\par\vskip 4mm\par}

\newenvironment{proofsketchof}[1]{\bigskip \noindent {\it Proof Sketch of #1.}\quad }
{\qed\par\vskip 4mm\par}

\newcommand{\sset}[1]{\{#1\}}
\newcommand{\intronorm}[1]{\norm{#1}_{q \to q^*}}

\newif\ifnotes\notesfalse

\ifnotes
\usepackage{color}
\definecolor{mygrey}{gray}{0.50}
\newcommand{\notename}[2]{{\textcolor{blue}{\footnotesize{\bf (#1:} {#2}{\bf ) }}}}

\else

\newcommand{\notename}[2]{{}}

\fi

\newcommand{\pnote}[1]{{\notename{Pranjal}{#1}}}
\newcommand{\anote}[1]{{\notename{Aravindan}{#1}}}
\newcommand{\xnote}[1]{{\notename{Xue}{#1}}}
\newcommand{\vnote}[1]{{\notename{Vaggos}{#1}}}

\usepackage{times}

\usepackage[T1]{fontenc}

\title{Adversarially Robust Low Dimensional Representations}
\author{Pranjal Awasthi\\ \small{Google Research and Rutgers University}\\ \small{pranjalawasthi@google.com} \and Vaggos Chatziafratis\\ \small{Google Research}\\ \small{vaggos@google.com} \and Xue Chen\thanks{Part of this work was done while the author was in Northwestern University.}\\ \small{George Mason  University}\\ \small{xuechen@gmu.edu} \and Aravindan Vijayaraghavan\thanks{The last author is supported by the National Science Foundation (NSF) under Grant No.~CCF-1652491 and CCF-1637585. }
\\ \small{Northwestern University} \\ \small{aravindv@northwestern.edu}}
\date{}

\begin{document}

\maketitle
Many machine learning systems are vulnerable to small perturbations made to inputs either at test time or at training time. This has received much recent interest on the empirical front due to applications where reliability and security are critical. However, theoretical understanding of algorithms that are robust to adversarial perturbations is limited.

In this work we focus on Principal Component Analysis (PCA), a ubiquitous algorithmic primitive in machine learning. 
We formulate a natural robust variant of PCA where the goal is to find a low dimensional subspace to represent the given data with minimum projection error, that is in addition robust to small perturbations measured in $\ell_q$ norm (say $q=\infty$). Unlike PCA which is solvable in polynomial time, our formulation is computationally intractable to optimize as it captures a variant of the well-studied sparse PCA objective as a special case. We show the following results:

\vspace{4pt}
\noindent \textbullet~Polynomial time algorithm that is constant factor competitive in the worst-case with respect to the best subspace, in terms of the projection error and the robustness criterion. 

\vspace{4pt}
\noindent \textbullet~We show that our algorithmic techniques can also be made robust to adversarial training-time perturbations, in addition to yielding representations that are robust to adversarial perturbations at test time. Specifically, we design algorithms for a strong notion of training-time perturbations, where every point is adversarially perturbed up to a specified amount. 

\vspace{4pt}
\noindent \textbullet~We illustrate the broad applicability of our algorithmic techniques in addressing robustness to adversarial perturbations, both at training time and test time. In particular, our adversarially robust PCA primitive leads to computationally efficient and robust algorithms for both unsupervised and supervised learning problems such as  clustering and learning adversarially robust classifiers.

\thispagestyle{empty}

\newpage

\tableofcontents

\thispagestyle{empty}

\anote{Some TODOs:
11/4:
1. Move Xue's counterexample to the main body, and call it a proposition.
2. Check about the Vu and Lei conjecture (should we be careful about wording) if conjecture was about the more general case?
3. Do we motivate or say anything about sparse mean estimation literature in stats?
4. Mention other forms of PCA and the randomized sketching approaches, column selection, Clarkson-Woodruff 
Older TODO:
1. Clarify that our projection matrices are all orthogonal projections.}

%2. $A^{\top}$ vs $A^T$ (consistency of notation)

%\anote{Pranjal: add something about the information theoretic statement for mean estimation and spectral clustering. }
\newpage

\setcounter{page}{1}

%\input{intro.tex}

%!TEX root=main.tex

%\section{Outline} \label{sec:outline}

\section{Introduction} \label{sec:intro}
%Modern machine learning systems are often vulnerable to adversarial corruptions. As AI systems are being deployed in almost every aspect of decision-making, it is vital for them to be reliable and secure, %This creates significant reliability and security concerns for current machine learning systems. This also 
%and emphasizes the need for classifiers and representations that are 
%and in particular robust to adversarial perturbations of various kinds.
Reliability and trustworthiness of machine learning systems are key requirements for their secure adoption in day to day life. Many algorithms in machine learning are brittle to small perturbations made to the data points either at
test time or at training time. 
While the design of robust machine learning algorithms has seen exciting recent developments in both statistics and computer science~\citep{huber2011robust, diakonikolas2019survey}, our theoretical understanding of robustness to adversarial perturbations is limited. 
% in handling corruptions made to a small fraction of the data
%work on robust high dimensional estimation, where a small fraction of the data is contaminated~\cite{Huber,DKKLMS,LRV16}. However there are several other kinds of errors and corruptions that are naturally motivated by practice, %particularly in emerging paradigms, 
%where our theoretical understanding of robustness and reliability is limited.   
%While robustness to errors of various forms like training data corruptions or data poisoning, drifting data distributions and partial observability are desirable, 
%However there are several other settings that are naturally motivated by practice, where our theoretical understanding of robustness and reliability is limited.
This lack of robustness to adversarial perturbations poses significant practical hurdles~\citep{szegedy2013intriguing, de2017understanding, de2018high}, and raises foundational questions of whether and how we can design basic machine learning primitives that are robust to adversarial perturbations.

% to the design of reliable learning systems~\cite{szegedy2013intriguing, de2017understanding, de2018high}. % and our theoretical understanding in this context is limited.
%  understanding whether and how we can design basic machine learning primitives that are robust to adversarial perturbations.

% However there are several practical settings where the lack of robustness to adversarial perturbations (both at test time, and training time) poses significant hurdles to the design of reliable learning systems~\cite{szegedy2013intriguing, de2017understanding, de2018high}.  
%

In this work we study the above question in the context of {\em principal component analysis (PCA)}, 
%or low-rank approximations is 
that is the predominant tool for obtaining succinct data representations, and used as a preprocessing primitive in many machine learning pipelines. 
%  The goal in PCA is to obtain a low dimensional representation that approximates a given dataset well. 
 Given data in a high-dimensional space $\R^n$ represented by the columns of a matrix $A$, the goal in PCA is to find a subspace of dimension at most $r \le n$ to represent the points, that minimizes the projection error (or reconstruction error) onto the subspace. This is formalized as follows where the matrix norm $\norm{\cdot}$ is either the Frobenius norm or the spectral norm:
\begin{equation} \label{eq:intro:PCA}
\min_{\Pi \in \mathcal{P}} \norm{\Pi^{\perp} A }^2= \min_{\Pi \in \mathcal{P}}\norm{A - \Pi A}^2, \text{ where } %\mathcal{P} \text{ is the set of orthogonal projection matrices 
%of rank} \le r .    
P=\{\text{orthogonal projections of rank} \le r \}.
\end{equation} 
%where the matrix norm $\norm{\cdot}$ is either the Frobenius norm or the spectral norm. 
The representation of each point $x \in \R^n$ corresponds to the projection $\Pi x$ onto the $r$-dimensional subspace given by $\Pi$ (one can also represent the point as an $r$-dimensional vector in terms of a basis for $\Pi$). 

% However, another key requirement for secure adoption of ML primitives is their reliability and trustworthiness~\cite{szegedy2013intriguing, de2017understanding, de2018high}. 
We propose a robust variant of PCA that corresponds to learning representations that are robust to adversarial perturbations to the data.
%made either at training time, or at test-time. %Motivated by this, 
% The primary goal of this work is to formulate and study a novel robust variant of Principal Component Analysis~(PCA), a primitive predominantly used in many machine learning pipelines. 
We model an adversarial perturbation $x'$ of a point $x$ as one for which the $\ell_q$ norm of the difference is small, i.e., $\|x-x'\|_q \leq \delta$, for some fixed $\delta > 0$ and $q > 2$. It is instructive to keep in mind the case of $q=\infty$, that is of particular interest in emerging paradigms such as adversarial machine learning~\citep{szegedy2013intriguing, madry2017towards}. A low dimensional subspace with an associated projection matrix $\Pi$ is robust if  $\norm{\Pi x -\Pi x'}_2$ is small for any adversarial perturbation $x'$ of $x$. 

{\em What data representations are adversarially robust?} Given an $r$-dimensional subspace of $\mathbb{R}^n$ with projection matrix $\Pi$, the adversarial robustness of $\Pi$ to $\delta$-perturbations in the $\ell_q$ norm is precisely captured by 
\begin{align}
\label{eq:intro-adversarial-robustness}
    \sup_{x,x': \|x-x'\|_q \leq \delta} \|\Pi(x-x')\|_2 = \delta \|\Pi\|_{q \to 2} .
\end{align}
The quantity $\kappa=\norm{\Pi}_{q \to 2}$ characterizes the robustness of the projection $\Pi$ to perturbations in $\ell_q$ norm around {\em every} point $x \in \R^n$ in the following sense. The distance between the projections of $x$ and a $\delta$-perturbation $x'$ of $x$ (in $\ell_q$ norm) is upper bounded by $\kappa \delta$.   
%An upper bound on $\norm{\Pi}_{q \to 2}$ gives an upper bound on the Euclideans $\norm{\Pi x - \Pi x'}_2$. 
On the other hand, around each point $x$ one can also realize a perturbation $x'=x+z$ with $\norm{z}_q \le \delta$ such that $\norm{\Pi x - \Pi x'}_2 = \kappa \delta$. We will call $\Pi$ a $(\kappa,q)$-robust rank-$r$ projection when $\Pi$ is an orthogonal projection matrix of rank at most $r$ with $\norm{\Pi}_{q \to 2} \le \kappa$; when the robustness parameter $\kappa$ and norm $q$ are understood, we will just call it a robust rank-$r$ projection. 
 %Now we can formulate our goal as that of given an $n \times m$ data matrix of $m$ points in $\mathbb{R}^n$ and $\kappa, r \geq 1$, to find a robust rank-$r$ projection matrix, i.e., with a small value of $\|\|_{q \to 2}$, and capturing most of the signal in the data. Formally, we can write this as
%  The range of values the robustness parameter can take is $1 \le \norm{\Pi}_{q \to 2} \le n^{1/2-1/q}$. 
%  %For $\ell_\infty$ norm, $\kappa \in [1,\sqrt{n}]$; 
%  The smaller the value of $\kappa$, the more robust the subspace $\Pi$ is (it will be instructive to think of $\kappa \approx n^{\eps}$, for some small constant $\eps=0.01$). 
 
This leads to the following natural formulation. 
%Our algorithmic goal 
Given a  
 data matrix $A \in \R^{n \times m}$ composed of $m$ points in $\R^n$, a robustness parameter $\kappa \ge 1$ and  the norm $q \in [2, \infty]$, find a robust rank-$r$ projection with low error: %formalized by  
\begin{align}\label{intro:obj}
    \min_{\Pi} &~~ \norm{\Pi^{\perp} A }^2 = \min_{\Pi} \norm{A - \Pi A}^2 \\
    \text{s.t.\ }& \Pi \text{ is an (orthogonal) projection matrix of rank at most } r, \text{ and } \norm{\Pi}_{q \to 2} \le \kappa \label{intro:operatornorm}.
\end{align}
One can also switch the objective and the constraint to consider the alternate formulation where we want to find a projection matrix with the minimum $\norm{\Pi}_{q \to 2}$ (i.e., the most robust projection), that achieves a prescribed projection error. %
%, captured as an upper bound constraint on $\norm{\Pi^{\perp} A}^2$.  
We will be interested in two versions of the problem, depending on whether we measure the projection error in {\em Frobenius norm} or {\em spectral norm}. Recall that the top-$r$ terms of the Singular Value Decomposition of $A$ simultaneously solve both of these problems in polynomial time, when there is no additional robustness constraint, or when $q=2$ (since $\norm{\Pi}_{2 \to 2} = 1$ for any non-trivial projection $\Pi$). We also remark that just as for the PCA objective~\eqref{eq:intro:PCA}, the above objective \eqref{intro:obj} can be equivalently rephrased as finding the best approximation among low-rank matrices, but among those with a ``robust column space'' (see Claim~\ref{lem:lowrankobj}).
 
%  Hence, our goal given a data matrix 
%   $A \in \R^{n \times m}$ as input and parameters $r, \kappa \ge 1$, is to find a robust rank-$r$ (orthogonal) projection matrix with low error:
% \begin{align}
%     \label{eq:intro-objective}
%     \min & \norm{A - \Pi A}^2 \,\, \text{ subject to}\\
%          & \|\Pi\|_{q \to 2} \leq \kappa \nonumber\\
%          & rank(\Pi) \leq r \nonumber.
% \end{align}
% The matrix norm above to measure the squared error in the projection will be either Frobenius norm or the spectral norm. Notice that without the constraint on the $\|\Pi\|_{q \to 2}$, the above formulation can be solved optimally in polynomial time via Singular Value Decomposition~(SVD). In fact, the top-$r$ subspace output by SVD is optimal for both the Frobenius norm and the spectral norm objective. As we will see later, this is not the case for the robust objective in (\ref{eq:intro-objective}). 

%Our optimization problem~\eqref{intro:obj} while motivated by adversarial robustness, 

\paragraph{Training-time robustness.} 
Our formulation in~\eqref{intro:obj} finds robust representations assuming access to the uncorrupted training dataset denoted by the matrix $A$. However in practice, large scale datasets often contain various kinds of measurement errors~\citep{sloutsky2013robust}, or even data that is poisoned by adversarial perturbations. Hence, it is important to design algorithms that are robust %desirable to make our algorithms robust 
to such training-time perturbations as well.

To capture training-time perturbations, we extend our formulation in \eqref{intro:obj} %to capture the above scenarios. In particular, we 
by assuming that we only have access to a corrupted dataset $\tilde{A}$, whose $i$th column $\tilde{A}_i$ is an adversarial perturbation of the corresponding column $A_i$ of the uncorrupted dataset $A$, i.e., $\norm{\tilde{A}_i - A_i}_q \le \delta$. %As discussed before we consider perturbations up to a specified amount measured in $\ell_q$ norm with $q \ge 2$. 
%An algorithm for solving~\eqref{intro:obj} would be considered 
%For robustness to training set perturbations if on input
Given as input $\tilde{A}$, our goal is to output a robust projection $\widehat{\Pi}$ that achieves {\em near optimal} error for the true dataset $A$, i.e., $\|{\widehat{\Pi}}^\perp A \|^2 \approx \min_\Pi \|\Pi^\perp A\|^2$.

% We study a strong model of training-time corruptions, where {\em every} point in the training data set can be adversarially perturbed up to a specified amount measured in $\ell_q$ norm with $q \ge 2$. 
We will show how to design algorithms for finding robust representations that are robust to adversarial perturbations at \emph{training-time}. In other words, we achieve robustness to both test-time and training-time perturbations simultaneously. 
%algorithmic techniques can be used to achieve such robustness to training-time perturbations while being simultaneously robust to test-time perturbations as well.
%\anote{Was:As we will see in Section~\ref{sec:results-training-errors}, the amount of adversarial perturbation $\delta$ to the training set that our algorithms is resilient to will crucially depend on the the $q \to 2$ operator norm of the projection matrix associated with the minimizer of \eqref{intro:obj}.}
As we will see in Section~\ref{sec:results-training-errors}, the resilience to adversarial perturbations in the training set will crucially depend on the the $q \to 2$ operator norm of the projection matrix associated with the minimizer of~\eqref{intro:obj}.
% that in turn dictates robustness to both types of error simultaneously.
%\anote{An interesting aspect is the $q \to 2$ norm. (linking sentence)}
%at both training time and test time simultaneously! 
%In fact our subsequent work~\cite{ACV20coltsub} shows that the $q \to 2$ norm precisely characterizes training-time robustness that is achievable in a natural average-case setting. 
%of the problem.  %for the classic problem of principal subspace recovery under a popular average case model namely, the spiked covariance model. %\xnote{July/13: Mention spiked-covariance model?}

%However, these results address training-time robustness, and do not capture the kind of adversarial perturbations that we are interested in  %as opposed to small perturbations to potentially every point 
%(see Section~\ref{sec:related-work}).  
% that can be processed by algorithms such as stochastic gradient descent. 
%As we will see in Section~\ref{sec:results-training-errors} and Section~\ref{sec:results:applications}, we will get provable algorithms guarantees against {\em both training-time and test-time} corruptions. 

\paragraph{Problem motivation.} 

Studying robust variants of PCA can lead to new robust primitives for problems in data analysis and machine learning. %algorithms for a variety of learning problems. 
(See Section~\ref{sec:results:applications} for specific examples.) 
% Our algorithms form a versatile new primitive for achieving robustness. %\footnote{Robust PCA handles randomly distributed sparse but unbounded corruptions to the data matrix $A$, whereas our goal is to handle bounded corruptions to potentially every entry of $A$. See Section~\ref{sec:results-training-errors} for a more detailed comparison.}  
% We demonstrate this by applying our robust PCA primitive to design new training-time robustness guarantees for unsupervised learning tasks such as clustering, and for designing linear classifiers robust to test-time perturbations. %\pnote{Mention that they also have implications for neural nets that are highly non-linear.}
%Apart from the natural motivation of studying a robust variant of PCA, 
%Apart from studying a robust formulation of PCA, a key primitive in machine learning, 
Our work is also motivated by emerging paradigms such as {\em adversarial machine learning} and {\em low precision machine learning}. The recent phenomenon of {\em adversarial robustness} identified by \citet{szegedy2013intriguing} shows that learning algorithms even when trained on high quality datasets are susceptible to small adversarial perturbations at test time. Even though empirical approaches have been proposed~\citep{madry2017towards, zhang2019theoretically}\pnote{add more citations}
 for designing
 algorithms that are robust to such perturbations, the current theoretical understanding is limited. Moreover in low-precision machine learning, one can achieve substantial performance improvements by quantizing the data to a few most significant bits (e.g., 8-bit arithmetic); this quantization noise is naturally captured as a small perturbation (in $\ell_\infty$ norm) to each training data point~\citep{de2017understanding, de2018high}. 
%  Furthermore, requiring a data representation to be adversarially robust is a natural notion of {\em individual fairness}~\cite{zemel2013learning, madras2018learning, oneto2019learning}. 
% % We believe that developing robust analogues of key primitives such as PCA will be important 
% in making even complex machine learning systems like neural networks more robust to adversarial perturbations. 

%in building more complex machine learning systems that are robust to adversarial perturbations. 
\vspace{5pt}
\noindent \emph{Practical implications.} Surprisingly, our techniques for learning robust linear representations also lead to algorithms for making deep neural networks that are highly non-linear in nature, more robust to test-time perturbations.
% An astute reader might notice that adversarial robustness is typically studied in the empirical machine learning literature in the context of deep neural networks that produce non-linear representations. While we study linear projections in this work, our results also have direct implications for robust training of deep neural networks. 
In a very recent work, \citet{awasthi2020adversarial} directly build on the theoretical insights developed in this work to design a practical algorithm for making deep neural networks more robust to adversarial perturbations as compared to the state-of-the-art.

\paragraph{Connection to Sparse PCA and generalizations.} While our formulation in \eqref{intro:obj} that is motivated by robustness is new to the best of our knowledge, it has rich connections to (and implications for) well studied problems like the sparse PCA problem~\citep{ZHT06,johnstone2009consistency}. Consider the setting when the perturbations are measured in $\ell_\infty$ norm and rank $r=1$. The robustness constraint on the projection $\Pi=vv^\top$ imposes an upper bound of $\kappa$ on the ``analytic sparsity'' of $v \in \R^n$ (measured as the ratio of $\ell_1$ and $\ell_2$ norms). In the special case of $r=1$ the formulation is  
\begin{equation}\label{eq:obj_rank1} \min  \|A - vv^\top A\|^2_F = \tr(AA^\top) - \max v^\top AA^\top v \,\, \text{ subject to }  \|v\|_1 \leq \kappa,~\text{ and } \|v\|_2=1.
\end{equation}
The complementary objective (i.e., $\max v^\top A A^\top v$) is the $\ell_1$ version of the maximization {\sc sparse PCA} objective;\footnote{It is within a factor $2$ of the $\ell_0$ version where the constraint $\|v\|_1 \leq \kappa$ is replaced by $\|v\|_0 \leq \kappa^2$ (see Section 10.3.3 of \citet{Vershynin}).} both the $\ell_0$ and the $\ell_1$ versions are notoriously hard in the worst-case~\citep{chan2016approximability} (see also Theorem~\ref{thm:hardness_SSE} in Appendix~\ref{sec:computational-lower-bound}). %\xnote{July/13: (1) 2nd review of FOCS complained ``complementary objective" means maximization. (2) For the worst-case hardness, shall we add a pointer to Appendix~\ref{sec:computational-lower-bound} and Theorem~\ref{thm:hardness_SSE}? }
For general $q \ge 2$, requiring robustness places a constraint on the dual $\ell_{q^*}$ norm of the direction $v$. %\anote{was: Moreover for rank $r \ge 1$, $\norm{\Pi}_{q \to 2}$ gives an upper bound on the $\ell_{q^*}$ norm of every unit vector in the subspace given by $\Pi$.}  
Moreover for projection matrices of higher rank $r \ge 1$, $\norm{\Pi}_{q \to 2}$ is a basis-independent quantity that captures the maximum  $\ell_{q^*}$ norm over all directions (unit vectors in $\ell_2$) in the subspace given by $\Pi$ (see  Lemma~\ref{lem:robustproperties}).  %for the statement, %, even though $\norm{\Pi}_{q \to 2}$ is a property of the subspace, and not a particular basis of it 
Hence robust projection matrices correspond to subspaces comprised of analytically sparse vectors measured in an appropriate norm e.g., $\ell_{1}$ norm, when $q=\infty$ (see also Claim~\ref{lem:basissparsity} for an approximate converse in terms of the sparsity of a basis for $\Pi$). \anote{Need to fill in this.} Appendix \ref{app:looklike} gives some examples of what robust projection matrices i.e., matrices with small $q \to 2$ operator norm, look like.

The range of values of the robustness parameter $\kappa$ is $1 \le \norm{\Pi}_{q \to 2} \le n^{1/2-1/q}$. For several real world datasets we expect $\kappa$ to be significantly smaller than the upper bound. As an example Figure~\ref{fig:cifar-10} shows that most of the signal in images from the CIFAR-10 dataset can be captured by a robust subspace with $\infty \to 2$ norm that is significantly smaller than $\sqrt{n}$. %The same phenomenon also holds for other domains like audio. %See also~\citet{awasthi2020adversarial} for more details on both image and audio data.  %For $\ell_\infty$ norm, $\kappa \in [1,\sqrt{n}]$; 
 The smaller the value of $\kappa$, the more robust the subspace is (it will be instructive to think of $\kappa \approx n^{\eps}$, for some small constant $\eps=0.01$). 
 
%\includegraphics{dctR-cifar.png} 

% \begin{figure}[htbp]
% \floatconts
% {fig:example}% label
% {\caption{An Example Figure}}% caption command
% {\includegraphics{dctG.png}}
% \end{figure} 

 \begin{figure}[htbp]
 %\vspace{-6pt}
    \centering
    \begin{minipage}{0.3\textwidth}
        \centering
        \includegraphics[width=1\textwidth]{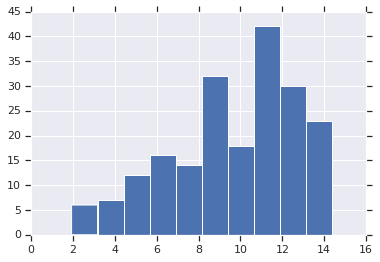} % first figure itself
        %\caption{first figure}
    \end{minipage}
    \begin{minipage}{0.3\textwidth}
        \centering
        \includegraphics[width=1\textwidth]{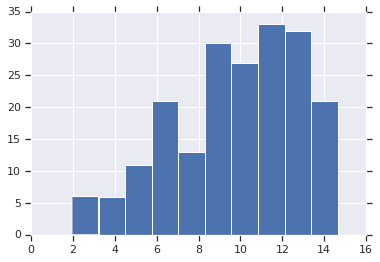} % second figure itself
        %\caption{second figure}
    \end{minipage}
        \begin{minipage}{0.3\textwidth}
        \centering
        \includegraphics[width=1\textwidth]{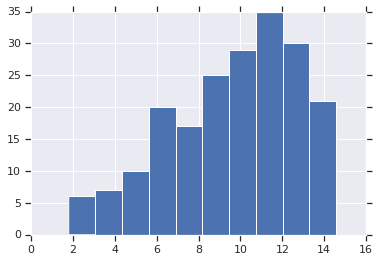} % second figure itself
        %\caption{second figure}
    \end{minipage}
     \caption{\small \label{fig:cifar-10} We study the CIFAR-10 image dataset~\citep{krizhevsky2009learning} transformed in the {\em discrete cosine}~(DCT) basis. We project each channel, of dimensionality $n=1024$, onto a robust $200$-dimensional subspace. The figure shows the histogram of the $\ell_1/\ell_2$ sparsity of the corresponding basis vectors. Each of the three projection matrices capture more than $99\%$ of the signal in the respective channel. Furthermore, for each projection matrix, the $\norm{\cdot}_{\infty \to 2}$ is in $[19,21]$, significantly smaller than $\sqrt{n} = 32$.}
    % demonstrating the existence of robust low-dimensional representations with very low projection error.}
%\vspace{-4pt}
\end{figure}

\section{Our Results} \label{sec:results}

While all the results that follow apply to $q \ge 2$, it will be useful to think of $q=\infty$ (perturbations measured in $\ell_\infty$ norm), and the robustness parameter $\kappa \ll n^{1/2}$, say $\kappa < n^{0.1}$.  %i.e., each point can be perturbed adversarially up to an amount $\delta$ in every coordinate. 
%However, our results will apply for all $q \ge 2$ generally. 
Intuitively, the larger the choice of $q$, the more unrestricted the adversarial perturbations can be (since $\norm{x}_p \le \norm{x}_q$ when $p \ge q$).\vnote{07/13: R2: complains about interpreting results when q close to 2: R2 says as projections all have identical 2->2 (spectral) norm, our new
model offers no robustness for q near 2. also, questions if robustness for low infinity->2 norm is
relevant, mentioning that there is disagreement about what
"robustness" means or how to model.} %Recall that our robustness parameter $\kappa \in [1, n^{1/2-1/q} ]$ for $q > 2$; for $q=\infty$ it will be instructive to think of $\kappa \ll n^{1/2}$ (e.g., $\kappa < n^{0.1}$ ), while interpreting the results.  \anote{added last line.}

\subsection{Algorithmic Guarantees for Robust Low-Rank Projections}
% \subsubsection{Worst-case Approximations for Adversarially Robust PCA}
\subsubsection{Approximation Algorithms for Adversarially Robust PCA}
\label{sec:results:worstcase}

We first consider the two variants of problem \eqref{intro:obj}, where the matrix norm $\gnorm{\cdot}$ represents either the Frobenius norm or spectral norm, in the worst-case setting. % where we are given a data matrix $A \in \R^{n \times m}$ composed of $m$ points in $\R^n$, a robustness parameter $\kappa \ge 1$ and  the norm $q \in [2, \infty]$. The goal is to find a robust low-rank projection with low error formalized by:  
% \begin{align}
%     \min_{\Pi} &\norm{\Pi^{\perp} A }= \norm{A - \Pi A}^2 \label{intro:obj}\\
%     \text{s.t.\ }& \Pi \text{ is a projection matrix of rank } r, \text{ and } \norm{\Pi}_{q \to 2} \le \kappa \label{intro:operatornorm}.
% \end{align}
% We will be interested in two versions of the problem, depending on whether we measure the error in {\em Frobenius norm} or {\em spectral norm}. Recall that the top-$r$ terms of the Singular Value Decomposition of $A$ solves both of these problems in polynomial time, when there is no additional robustness constraint. 
% \anote{Maybe move this into the introduction.}
%We give efficient algorithms that attain small constant factor bicriteria approximations. %to the projection error (approaching $2$ for the Frobenius norm error, and $3$ for the spectral norm error) when we are also allowed to relax the robustness of the projection by a constant factor. 
%\anote{Removed the informal wording.}
\begin{itheorem}\label{ithm:worstcase:frob}
There exist polynomial time algorithms that given $q \ge 2$, any $\gamma\in (0,1]$ and a data matrix $A$ 
%having a projection matrix $\Pi^*$ of rank at most $r$ with $\norm{\Pi^*}_{q \to 2} \le \kappa$, %and a robustness parameter $\kappa$, 
with a $(\kappa,q)$-robust projection matrix $\Pi^*$ of rank at most $r$ satisfying $\gnorm{A-\Pi^* A}^2 \le \eps \gnorm{A}^2$ for some $\eps \in [0,1]$,
find a projection matrix $\widehat{\Pi}$ of rank at most $r$ s.t. 
\begin{align} \label{eq:ithm:worstcase:frob}
\norm{\widehat{\Pi}}_{q \to 2} &\le O(1/\sqrt{\gamma}) \cdot \kappa , \text{ and } \gnorm{A-\widehat{\Pi} A}^2 \le \big(\alpha+\gamma \big) \cdot \eps \gnorm{A}^2, %\norm{A-\Pi^* A}^2, 
\end{align}
% where $OPT$ is the error obtained by the optimal $(\kappa,q)$-robust projection of rank at most $r$ 
where $\alpha=2$ for the Frobenius norm error objective and $\alpha=3$ for the spectral norm error objective. 
Moreover, for any $\gamma \in (0,1]$, there exist polynomial time algorithms that find an $r'\le r (1+O(\gamma^{-1}))$-dimensional projection $\widehat{\Pi}$ that gets a projection error of $(1+\gamma) \gnorm{A-\Pi^* A}^2$, and relaxes the robustness parameter by $O(1/\sqrt{\gamma})$ factor.
\end{itheorem}

% \begin{itheorem}\label{ithm:worstcase:frob}
% There exist polynomial time algorithms that given a data matrix $A$, $q \ge 2$ and a robustness parameter $\kappa$, find a projection matrix $\widehat{\Pi}$ of rank at most $r$ s.t.
% \begin{align} \label{eq:ithm:worstcase:frob}
% \forall \gamma\in (0,1], ~~~ \norm{\widehat{\Pi}}_{q \to 2} &\le O(1/\sqrt{\gamma}) \cdot \kappa , \text{ and } \norm{A-\widehat{\Pi} A}^2 \le \big(\alpha+\gamma \big) OPT, 
% \end{align}
% where $OPT$ is the error obtained by the optimal $(\kappa,q)$-robust projection of rank at most $r$ ($\alpha=2$ for the Frobenius norm error objective and $\alpha=3$ for the spectral norm error objective). 
% Moreover, for any $\gamma \in (0,1]$, there exist polynomial time algorithms that find an $r'\le r (1+O(\gamma^{-1}))$-dimensional projection $\widehat{\Pi}$ that gets an $(1+\gamma)$-approximation to the objective, and relaxes the robustness parameter by $O(1/\sqrt{\gamma})$ factor.
% \end{itheorem}
%Hence the above theorem gives a $(O(1),O(1))$-factor bicriteria approximation when we need an approximation of rank at most $r$. %; however we can also achieve an $((1+\gamma),O(1/\sqrt{\gamma}))$ by incurring an extra loss in the rank.
In other words, our algorithms attain small constant factor bicriteria approximation to the adversarially robust PCA problem.
 The algorithms for both objectives -- Frobenius norm and spectral norm, use convex relaxations and similar ideas, yet the algorithms (and relaxations) are different, unlike the case for standard PCA. Please see Theorem~\ref{thm:worstcase:frob} (Frobenius norm objective) and Theorem~\ref{thm:worstcase:spectral} (spectral norm objective) for the formal statements.
% As we mentioned in the introduction, the guarantees of the above theorem are most interesting when $q > 2$ since the case of $q=2$ corresponds to the vanilla PCA problem. %(without the robustness constraints). 
Our algorithms take in as input a guess for the robustness parameter $\kappa$. Recall that $\kappa \in [1, n^{1/2 - 1/q}]$.\footnote{When $q=2$, the robustness constraint becomes trivial, and problem reduces to the standard PCA problem as discussed earlier. } Alternately, one can also input a guess for the optimal projection error (or the desired projection error), and minimize the robustness parameter $\kappa$ approximately. 

Observe that the approximation guarantee in Theorem~\ref{ithm:worstcase:frob} is a constant independent of the desired rank $r$. %We remark that our guarantees are non-trivial 
Even if we do not restrict the rank $r$  % i.e., dimension of the representation to be small.
(set $r=n$) our algorithm finds among all subspaces that are %an orthogonal projection that is 
$O(\kappa)$-robust, the one with approximately optimal error. 
The constant factor loss in the robustness parameter depends on the value of $q \in [2,\infty)$. It is the largest for $q=\infty$ (where it is $\sqrt{\pi/2}$), and this is related to a variant of the Grothendieck problem~\citep{alon2004approximating,nesterov1998semidefinite}. This loss in the robustness parameter is unavoidable when $q>2$, due to the inapproximability  for certifying the $q \to 2$ norm, even for projection matrices~\citep{bhattiprolu2018inapproximability} (see Section~\ref{sec:operatornorms}).

Our result also has new implications for approximating the minimization objective for sparse PCA specified in \eqref{eq:obj_rank1}. Most existing theoretical guarantees for Sparse PCA have been established for average case models~\citep{BerthetRigollet}. There has also been work on studying the maximization version of the sparse PCA objective in worst case models~\citep{chan2016approximability}. To the best of our knowledge, we are not aware of any existing worst case guarantees for the minimization version as defined in \eqref{eq:obj_rank1}. Our results (applied with $r=1$) provides a small constant factor bicriteria approximation to problem \eqref{eq:obj_rank1}. %the minimization version of the sparse PCA problem under $\ell_1$ constraint on the sparsity. 
This is in stark contrast 
%We remark that these constant factor (bicriteria) approximations are in stark contrast 
to the approximability of the maximization version of the problem. Even when $r=1$ the best known polynomial time algorithm gives a $O(n^{1/3})$ factor approximation in the worst-case (for both the $\ell_1$ and $\ell_0$ versions); %\footnote{The $\ell_0$ and $\ell_1$ versions of the sparse PCA maximization problem are within a factor of $2$ from each other~\citep[see Section 10.3.3 of ][]{Vershynin}})
% This is true even when the sparsity can be relaxed by a $O(1)$ factor; 
moreover no constant factor approximation is possible assuming the SSE conjecture~\citep{chan2016approximability}. (see its implication to computational hardness of our minimization version \eqref{intro:obj} in Appendix~\ref{sec:computational-lower-bound}).
 %Based on this observation, see Section~\ref{sec:computational-lower-bound} for the computational intractability of the minimization version \eqref{eq:intro-objective} of the problem that we study. \pnote{Read previous line.}
 Furthermore, the minimization variant of the problem that we study (and our small approximation factors) 
 will be crucial in various downstream applications such as clustering. %Furthermore, our techniques also have implications for the average case formulation of the sparse PCA problem and its generalizations.
%\pnote{Aravindan, check paragraph above.}
\anote{Maybe we could rephrase the last line? e.g., Furthermore, the applications to downstream learning problems like clustering that follow also illustrate the merits of the minimization variant that we study. 
%(and our small approximation factors) will be crucial in various downstream applications such as clustering. 
 }

\subsubsection{Robustness to Adversarial Errors during Training}
\label{sec:results-training-errors}
%The notion of robust low-dimensional projections (and approximations) was motivated by the goal of robustness to test-time adversarial perturbations. 
%\anote{Rephrased following paragraph since definition is already in intro.}
We now discuss how %show how the robustness criterion in terms of the $q \to 2$ norm and our algorithmic techniques allow us 
%and the algorithms that we have developed (Theorem~\ref{ithm:worstcase:frob}) also allow us to 
to handle {\em data poisoning}, where points in the {\em training} data set $\wt{A}$ are adversarially perturbed. % -- this is often referred to as {\em data poisoning}. 
Recall that in the corruption model, {\em every} %training 
sample $A_i \in \R^{n}$ can potentially be adversarially perturbed up to a $\delta$ amount, as measured in $\ell_q$ norm for $q \ge 2$. 
So every column of $\wt{A}$ satisfies $\norm{\tilde{A}_i - A_i}_q \le \delta$; and we will refer to such an $\tilde{A}$ as a {\em $\delta$-corrupted instance} of $A$.  
%Our input instance is $\tilde{A} \in \R^{n \times m}$ where every column $i \in [m]$ satisfies $\norm{\tilde{A}_i - A_i}_q \le \delta$; we will refer to such an $\tilde{A}$ as a {\em $\delta$-corrupted instance} of $A$.  
While the input instance is $\wt{A}$, our goal now is to recover a robust low-rank projection for the uncorrupted matrix $A$. We will show that we can in fact output a robust low-rank projection $\widehat{\Pi}$ that is competitive with the best robust low-rank projection of $A$, even though $A$ is not known to us! %As we will see in the next subsection, these algorithms will also be useful in analogous guarantees for unsupervised learning tasks like mean estimation and clustering when each point in the training dataset is adversarially corrupted. 
We first state our result when the error is measured in Frobenius norm, and later describe the guarantees for the spectral norm variant.

\begin{itheorem}\label{ithm:robusttraining:frob}
Suppose $q \ge 2$ and $A \in \R^{n \times m}$ is the unknown uncorrupted data matrix, with a $(\kappa,q)$-robust projection matrix $\Pi^*$ of rank at most $r$ satisfying $\norm{A-\Pi^* A}_F^2 \le \eps \norm{A}_F^2$ for some $\eps \in [0,1]$. There exists a polynomial time algorithm that
given as input a $\delta$-corrupted instance $\tilde{A}$ of $A$ outputs a projection $\widehat{\Pi}$ of rank at most $r$ that is approximately optimal:
\begin{align} \label{eq:intro:robusttraining}
\forall \eta>0, ~\norm{\widehat{\Pi}}_{q \to 2} &\le O(\kappa) , \text{ and } \norm{A-\widehat{\Pi} A}_F^2  \le O(\eps+ \eta) \cdot \norm{A}_F^2 + O(\tfrac{1}{\eta}) \cdot \delta^2 \kappa^2 m.  
\end{align}
% \begin{align} \label{eq:intro:robusttraining}
% \norm{\widehat{\Pi}}_{q \to 2} &\le O(\kappa) , \text{ and } \norm{A-\widehat{\Pi} A}_F^2  \le O(\eps) \cdot \norm{A}_F^2 + O(\kappa \delta \sqrt{m}) \cdot \norm{A}_F.  
% \end{align}

In particular this gives an $O(1)$ approximation when $\delta^2 < (\eps^2/\kappa^2)\cdot \tfrac{1}{m}\norm{A}_F^2 $. 
% In particular this gives an $O(1)$ approximation when $\delta < (\eps/\kappa)\cdot \tfrac{1}{\sqrt{m}}\norm{A}_F $. 

\end{itheorem}
To interpret the results %\vnote{07/13: R2: wants to know what happens if q=2, instead of q=infinity, how do the results here compare to prior work?}, 
let $q=\infty$ and consider an uncorrupted dataset $A$ where every column (sample) is a unit vector in $\R^n$, and let $\kappa=n^{0.1}$. %When $\delta = o(n^{-1/2})$ i.e., every coordinate can be perturbed up to $\delta \ll 1/\sqrt{n}$, 
The total corruption to each point is at most $o(1)$ in Euclidean norm when $\delta=o(n^{-1/2})$; in this case one would expect that standard PCA applied to $\tilde{A}$ may recover a good solution. The above Theorem~\ref{ithm:robusttraining:frob} on the other hand guarantees to find a good (robust) low-rank approximation for the unknown matrix $A$ even when $\delta = o(1/\kappa)= o(n^{-0.1})$. Note that in this setting {\em every} point can be completely overwhelmed by the adversarial noise (in Euclidean norm).
The algorithm first denoises the input by solving a convex minimization problem before applying the algorithm from Theorem~\ref{ithm:worstcase:frob}.
Furthermore, %the above guarantees are optimal up to constant factors; in particular, 
Proposition~\ref{prop:additiveerror:tight} shows that the additive factor of $O(m \delta^2 \kappa^2)$ is unavoidable for every $\kappa, \delta=O(1/\kappa)$. %Additionally in Section~\ref{sec:} we provide some evidence that our notion of robust projections essentially characterizes robustness to perturbations during training as well. 
These results suggest that 
%our notion of robust projections
the robust projection structure (measured in $q \to 2$ operator norm) is {key} in understanding the resilience to small adversarial perturbations of every point during training, even without any test-time robustness considerations. Subsequent work by~\citet{ACV20coltsub} also characterize principal subspace recovery in an average-case setting in the presence of adversarial perturbations at training time using the $q \to 2$ norm robustness criterion (in an instance-optimal sense). %See also Section~\ref{sec:comparison} for a comparison to the structural properties necessitated by resilience to other notions of training-time corruptions. % like outlier noise or missing entries. 

%\anote{4/11: Added last line. Too much?}
%\anote{Too strong to say ``characterizes'' instead of ``key''? Check.} 
%\pnote{Perhaps we can make a formal claim here. Can we show that if a one-dimensional subspace is not robust, i.e, $\|v\|_1 > \kappa$, then there is a $\delta = 1/\kappa$ perturbation to the data such that there is another subspace that achieves smaller or the same error? A good place to start is the counter-example in Appendix~\ref{sec:app-lower-bound-additive-error}}

Our guarantees for spectral norm error in the presence of training-time adversarial perturbations are somewhat similar to Theorem~\ref{sec:wc:frob}. However, there is a qualitative difference: given as input an adversarial $\delta$-perturbation $\tilde{A}$ of an uncorrupted matrix $A$ that has a good solution (i.e., $\spnorm{A-\Pi^* A} \le \eps \spnorm{A}$ for some small $\eps \in (0,1)$), we will either find a robust low-dimensional projection of the unknown dataset $A$, or we will certify that the dataset has been poisoned substantially (i.e., $\spnorm{\tilde{A} - A} > \eps \spnorm{A}$). In particular, the algorithm will {never} output a low-dimensional representation that is bad for the unknown data matrix $A$.  
% In what follows $\spnorm{\cdot}$ will refer to the spectral norm. 
%  \begin{itheorem}\label{ithm:training:spectralnorm}
% % Suppose the (unknown) uncorrupted data matrix $A \in \R^{n \times m}$ has an $r$-dimensional projection $\Pi^{*}$ such that $\norm{\Pi^*}_{\infty \to 2} \le \kappa$ and the approximation error $OPT:=\spnorm{(I-\Pi^*) A} \le \eps \spnorm{A}$. 
% Fix $q \ge 2$ and $\eps \in (0,1)$. Let $A \in \R^{n \times m}$ be the unknown uncorrupted data matrix with a $(\kappa,q)$-robust projection $\Pi^*$ of rank at most $r$ with $\spnorm{A-\Pi^* A} \le \eps \spnorm{A}$. There exists a polynomial time algorithm that
% given as input a $\delta$-perturbed instance $\tilde{A}$, outputs either a robust projection $\widehat{\Pi}$ of rank at most $r$ with %an approximately optimal robust projection $\widehat{\Pi}$ for $A$ with 
% \begin{equation}
% \norm{\widehat{\Pi}}_{q \to 2} =O(\kappa), \text{ and } ~\norm{A-\widehat{\Pi} A}  \le O\Big(\sqrt{\eps} \spnorm{A}+\sqrt{m} \delta \kappa/\sqrt{\eps} \Big),
% \end{equation}
% or certifies that the data has been poisoned i.e., $\norm{\tilde{A}-A} > \eps \norm{A}$. 
%  \end{itheorem}
% %\vnote{say what $c_q$ is?}
Please see Theorem~\ref{thm:training:spectralnorm} for a formal statement. %Here again, the additive term of $\Omega(\sqrt{m} \delta \kappa)$ is unavoidable as shown in Proposition~\ref{prop:additiveerror:tight}. 
%To interpret the results for $q=\infty$, consider an uncorrupted dataset $A$ where every column (sample) is a unit vector in $\R^n$, and let $\kappa=n^{0.1}$. Suppose $A$ has a robust low-rank projection with error at most $\eps \spnorm{A}$, where $\eps$ be a small constant. In this case, when $\delta =o(1/\kappa)$ and given a perturbed data set $\tilde{A}$, we will either output a robust low-rank projection that achieves spectral norm error of $O(\eps \norm{A})$, or certify that the data was poisoned to such an extent that $\spnorm{A- \tilde{A}} = \Omega(\spnorm{A})$! These guarantees will hold even when $\delta = o(1/\kappa)= o(n^{-0.1})$ i.e., every point is perturbed up to $\delta$ amount, so that the total corruption to each point could be $n^{0.4}$ in Euclidean norm! 
We remark that information-theoretically we can design an estimator (that is computationally inefficient) that achieves the stronger qualitative guarantees as in Theorem~\ref{ithm:robusttraining:frob}. Designing a computationally efficient algorithm to do the same is an open question that we describe in more detail in Section~\ref{sec:techniques:training}. 

\subsection{Applications to Learning Problems.}
\label{sec:results:applications}
The algorithmic results that we have described so far, may be used as a robust primitive in lieu of standard PCA.
These lead to efficient, adversarially robust algorithms for learning problems of different flavors, further validating our formulation. In particular, we demonstrate the versatility of our robust primitive via the following three applications across both unsupervised and supervised learning: 

%\begin{enumerate}
 %\item 

 \vspace{5pt}
 \noindent \textbf{1. Clustering with training-time perturbations.} We study the classical unsupervised learning problem of $k$-means clustering. Let $A \in \R^{n \times m}$ denote $m$ data points with an unknown ground truth clustering into $k$ clusters. It is well known that if the ground truth clusters are well separated then the popular Lloyd's algorithm~\citep{lloyd1982least}, when properly initialized, recovers the ground truth clustering~\citep{kumar2010clustering, awasthi2012improved}. We extend this setting to consider a scenario where the input to the algorithm is the data matrix $\tilde{A}$ with each data point being adversarially corrupted up to a perturbation of a certain amount. Existing algorithms based on variants of the Lloyd's heuristic fail to handle large amounts of noise in this setting. This is due to the fact that these algorithms use PCA to initialize the cluster centers, and as we saw in previous sections, PCA is not robust to adversarial perturbations. 
 
 We instead design a robust variant of the Lloyd's heuristic that can handle a large amount of perturbation while successfully clustering the data according to the ground truth. In our algorithm, the adversarially robust PCA primitive plays a crucial role. We use the adversarially robust PCA primitive to obtain a good set of initial cluster centers. Additionally, during the iterative Lloyd's updates, we compute new cluster means via a new robust mean estimation procedure that we design in this work~(this is the special case of clustering with $k=1$). As a result we obtain a clustering algorithm that can handle adversarial perturbations to the training set, of magnitude up to $o(1/\kappa)$ where $\kappa$ is the robustness of the $k$-dimensional subspaces spanned by the cluster means. Hence our algorithm can handle significantly more noise when this subspace is robust, compared to standard approaches that break down unless the perturbation amount~(in $\ell_\infty$ norm, say) is of the order of $o(1/\sqrt{n})$. \anote{Added this line:} On the other hand, such a dependence on $\ell_1$ sparsity of the means is needed even in the case where $k=1$ (i.e., mean estimation in the presence of adversarial perturbations). 
 
 See Section~\ref{sec:applications} for details, including Theorem~\ref{thm:clustering-application-general} for the general case, and Theorem~\ref{thm:clustering-application-gaussians} for the specialization to clustering mixtures of Gaussians. 
    
  %  use our robust PCA primitive to design a robust variant of the Lloyd's heuristic that can handle a large amount of perturbation while successfully clustering the data according to the true ground truth. See Section~\ref{sec:applications} for details.
    
 \vspace{5pt}
 \noindent \textbf{2. Learning intersection of halfspaces under training-time and test-time perturbations.} We consider the problem of learning an intersection of $k$ halfspaces over the Gaussian distribution on $\R^n$ in the presence of adversarial perturbations to the samples, both at testing-time and training-time. We will represent an intersection of halfspaces by a function $h:\R^n \to \sset{0,1}$ denoted by $h(x)= \prod_{i=1}^k \bfone(w_i^\top x \ge \theta_i)$, where $\forall i \in [k], ~ \norm{w_i}_2 =1$ and $\theta_i \in \R$ and where $\bfone(\cdot)$ denotes the indicator function. Let $\mathcal{H}_k$ represent the hypothesis class of all intersections of at most $k$ halfspaces. 
 %We will also refer to `1' as the positive label, and `0' as the negative label. 
In the uncorrupted setting, the training points $x_1, \dots, x_m \in \R^n$ are drawn i.i.d. from a Gaussian distribution, and their corresponding labels $y_i=h^*(x_i)$ for some $h^* \in \mathcal{H}_k$. % (this corresponds to the realizable setting). 
%The special case of $k=1$ corresponds to standard linear classification. 
A series of well-known results~\citep{VempalaJACM,Vempala08,KOS08} shows that when we are given access to uncorrupted training samples in $\R^n$ drawn from a Gaussian distribution, one can learn an intersection of halfspaces in the PAC learning model,  in time $f(k)\cdot \poly(n)$. %, where $f(k)$ has a super-polynomial dependence on $k$. 
Crucially, these algorithms use PCA as a first step to reduce the learning problem to a low dimensional space. 
Our adversarially robust PCA primitive can be used to learn an intersection of $k=O(1)$ halfspaces even when there are adversarial perturbations {\em both} at {\em training-time} and {\em test-time}. 

What does a classifier $h$, say  %(even without any training-time robustness considerations). 
$h(x)= \bfone(w_1^\top x \ge 0) \cdot \bfone(w_2^\top x \ge 0)$, that is robust to adversarial $\delta$-perturbations at test-time look like? %Further let $\kappa=\max\sset{ \norm{w_1}_1, \norm{w_2}_1}$. 
%Let $\gamma:=\norm{w_1-w_2}_2$. 
%Consider an uncorrupted test point $x$ drawn from $N(0,\sigma^2 I)$, and suppose $h(x)=1$. We know that w.h.p., $|w_1^\top x|, |w_2^\top x| \le O(\sigma)$. 
First observe that $\max\sset{\norm{w_1}_1, \norm{w_2}_1} \le  O(1/\delta)$ is necessary, otherwise there exists a $\delta$-adversarial perturbation $\tilde{x}$ with $\norm{\tilde{x}-x}_\infty \le \delta$ that $h$ misclassifies w.h.p! Moreover the subspace $\Pi^*$ spanned by $w_1, w_2$ is robust as measured in $\kappa=\norm{\Pi^*}_{\infty \to 2}$ (see Claim~\ref{claim:robustintersection} for a formal claim). 
For general $k$, we will assume that there exists a robust classifier $h^*(x)=\prod_{i=1}^k \bfone(w_i^\top x \ge \theta_i)$ for the data such that 
%there exists a if the labels are generated by an intersection of $k$-halfspaces represented by $h^{*}(x):=\prod_{i=1}^k \bfone(w_i^\top x \ge \theta_i)$ with $\norm{w_i}_2=1~ \forall i \in [k]$, we assume that 
the projection matrix $\Pi^*$ onto the span of the normals $w_1, \dots, w_k$ satisfies $\norm{\Pi^*}_{q \to 2} \le \kappa$. 

We consider a natural model of training-time perturbations, where each training data-point is $\delta$-adversarially perturbed in $\ell_q$ norm ($q\ge 2$). Our robust algorithm follows the same general approach as in \citet{Vempala08}; however we use our primitive for adversarially robust PCA instead of standard PCA to bring down the dimension to $k$. %This gives an algorithm that in time $f(k) \cdot poly(n)$ outputs a robust classifier (intersection of $k$-halfspaces) that incurs an error $o(1)$, 
This allows us to handle adversarial perturbations of magnitude $\delta=o(1/\kappa)$ (as opposed to existing approaches that need $\delta=o(1/\sqrt{n})$ for $q=\infty$), and output a robust classifier (intersection of $k$-halfspaces) that incurs an error of $o(1)$. 
% Think of the setting when $k=2$, and $q=\infty$ as before. The above algorithm achieves an error $\eps=o(1)$ as long as $\delta = o(\sigma)/\kappa$, where $\kappa$ is the robustness parameter. 
Recall from the earlier discussion, that such a condition is necessary qualitatively: even a single half-space $\bfone(w_1^\top x \ge 0)$ is not robust when $\norm{w_1}_1 = \kappa$ and $\delta \gg 1/\kappa$. See Section~\ref{sec:intersection} for details.
  
%\anote{Say anything about $k=1$?}
%We also remark that the special case of the above result when $k=1$ 
%Another interesting aspect of this application is that there is no assumption on the data lying close to a robust subspace $\Pi^*$. Instead our algorithmic primitive exploits the robustness of the unknown classifier (the subspace containing the normals of the halfspaces).

 \vspace{5pt}
 \noindent \textbf{3. Trading off natural accuracy in classification for robustness to test-time perturbations.} Finally, in many scenarios it might be desirable to trade off natural accuracy for significant robustness to test-time perturbations. We demonstrate how our robust primitive can be used for this purpose. Specifically, we consider the {\em Gaussian data model}~\citep{Anderson_Multi_stat} that has been studied in recent works to understand adversarial robustness~\citep{tsipras2018robustness, schmidt2018adversarially}. In this model a labeled example $(x,y)$ is generated by first picking the label as $+1$ or $-1$ with equal probability. Then $x \in \R^n$ is drawn from either $\mathcal{N}(\mu_1, \Sigma)$ or $\mathcal{N}(\mu_2, \Sigma)$ depending on whether $y=-1$ or $y=+1$. We denote this model as $\mathcal{M}(\mu_1, \mu_2, \Sigma)$. 
    
     In the above model the (Bayes) optimal classifier is a linear classifier of the form $sgn(\iprod{w,x} + b)$ with weight vector $w = \Sigma^{-1}(\mu_1 - \mu_2)$. If the means are well separated then the above classifier has good accuracy but can be easily fooled during test-time via small perturbations. In other words, the robust accuracy of the above classifier is close to zero. In order to get a better trade-off of standard accuracy and robust accuracy, we could instead aim to look for robust subspaces where the variance is low and the means when projected are still separated by a non-trivial amount. We will show how our robust PCA primitive helps us achieve this and obtain a classifier with better robust accuracy. See  Section~\ref{sec:gaussian-model} for details.

\subsection{Proof Sketches and Technical Overview} \label{sec:techniques}

We %now 
give a flavor of the technical ideas involved in obtaining our main algorithmic results. 
% For the sake of exposition, we will restrict our attention to the case when the adversarial perturbations are measured in $\ell_\infty$ norm. 

\subsubsection{Constant Factor Approximation Algorithms}

%\anote{Mention the more efficient algorithm. And maybe relate to the recent empirical paper?}

Let us first consider the version of problem \eqref{intro:obj} of finding a robust rank-$r$ projection that has small error measured in Frobenius norm. A natural mathematical programming relaxation is the following:
% \begin{align}
%     &\min_X \norm{A}_F^2 - \iprod{AA^\top, X} \label{tech:sdp:obj}\\
%     \text{subject to  }~~& \tr(X) \le r, ~~ 0 \preceq X \preceq I 
%   ~\text{ and } \norm{X}_{\infty \to 2}\le \kappa \label{tech:sdp:norm}
% \end{align}
\begin{align}
    &\min_X \norm{A}_F^2 - \iprod{AA^\top, X} \label{tech:sdp:obj}\\
    \text{subject to  }~~& \tr(X) \le r, ~~ 0 \preceq X \preceq I 
  ~\text{ and } \norm{X}_{q \to 2}\le \kappa \label{tech:sdp:norm}
\end{align}

% This is a valid convex relaxation for the problem since the constraints are all satisfied by any rank-$r$ projection matrix that is robust i.e, $\norm{\Pi}_{\infty \to 2} \le \kappa$. 

This is a valid convex relaxation for the problem since the constraints are all satisfied by any rank-$r$ projection matrix that is robust i.e, $\norm{\Pi}_{q \to 2} \le \kappa$.

%Moreover this relaxation is convex; the last constraint in particular is an upper bound on a valid norm. Let $X^*$ be the optimal solution of this program.  
The first challenge however is that the operator norm constraint \eqref{tech:sdp:norm} is NP-hard to verify efficiently, even for the case of projection matrices. %\footnote{Specifically, the hard instance in Theorem 3.1 of~\cite{bhattiprolu2018inapproximability} that proves NP-hardness of approximating below a factor of $\sqrt{\tfrac{\pi}{2}}$ is a projection matrix.} %, {even when the matrix is a projection (and this is tight).} matrices~\cite{bhattiprolu2018inapproximability}. 
However, these operator norm $\norm{\cdot}_{q \to p}$ computation problems %are APX-hard for most values of $p,q$. % except when $q=1,\infty$ or $p=q=2$, and 
%They 
form a rich class of problems related to the Grothendieck problem~\citep{alon2004approximating, nesterov1998semidefinite}, and polynomial time $O(1)$ factor approximations are known for general $q \to 2$ norms with $q \ge 2$ (see Section~\ref{sec:operatornorms}). %, polynomial time constant factor approximation algorithms based on semi-definite programming (SDP) are known~\citep{nesterov1998semidefinite, steinberg2005computation}.

% The bigger challenge is in producing a projection matrix from $X^*$ that simultaneously (a) achieves a good objective value, (b) has rank at most $r$, and (c) is $O(\kappa)$-robust i.e., has bounded $\infty \to 2$ norm. A natural approach for producing a good low-rank solution is to output a rank-$r$ projection matrix $\Pi_r$ that corresponds to the large singular values of $X^*$. However we have no control on the robustness of the subspace $\norm{\Pi_r}_{\infty \to 2}$. In fact, the algorithmic problem~\eqref{intro:obj} is challenging even when there is no rank constraint ($r=n$). The main issue is to relate the $\infty \to 2$ operator norm of the projection matrix we output to that of the relaxation solution $\norm{X^*}_{\infty \to 2}$ which is upper bounded by $\kappa$. 

The bigger challenge is in producing a projection matrix from $X^*$ that simultaneously (a) achieves a good objective value, (b) has rank at most $r$, and (c) is $O(\kappa)$-robust i.e., has bounded $q \to 2$ norm. A natural approach for producing a good low-rank solution is to output a rank-$r$ projection matrix $\Pi_r$ that corresponds to the large singular values of $X^*$. However we have no control on the robustness of the subspace $\norm{\Pi_r}_{q \to 2}$. In fact, the algorithmic problem~\eqref{intro:obj} is challenging even when there is no rank constraint ($r=n$). The main issue is to relate the $q \to 2$ operator norm of the projection matrix we output to that of the relaxation solution $\norm{X^*}_{q \to 2}$ which is upper bounded by $\kappa$. 

Our crucial insight is that we can indeed design a rounding scheme that achieves all three goals if the norm in the constraint \eqref{tech:sdp:norm} is a {\em monotone norm}! %A matrix norm 
% $\gnorm{\cdot}$ is monotone iff $\gnorm{A+B} \ge \gnorm{A}$, { for any pair of positive semidefinite (PSD) matrices } $A,B$. 
$$\text{A matrix norm }\gnorm{\cdot} \text{ is monotone iff }~~ \forall A,B \succeq 0,~\text{ we have }~ \gnorm{A+B} \ge \gnorm{A}.$$

(See Definition~\ref{def:monotonicity} for details.) This monotonicity property allows us to truncate terms in the eigendecomposition of $X^*$ without any loss in robustness $\kappa$, and get fine control on the robustness $\kappa$ when we rescale different rank-$1$ terms appropriately.   
% Unfortunately however, the $\infty \to 2$ operator norm is not monotone (see e.g., Claim~\ref{clm:counterexample_infinity_2},  Claim~\ref{claim:counterexample:monotone}). 
Unfortunately however, the $q \to 2$ operator norm is not monotone in general (see e.g., Claim~\ref{clm:counterexample_infinity_2},  Claim~\ref{claim:counterexample:monotone}).
% for various counterexamples among such non-unitarily-invariant norms).%\footnote{A large class of matrix norms that are known to be monotone are unitarily invariant norms.})! 

% Our next important observation is that 
% we can replace the constraint \eqref{tech:sdp:norm} by a similar constraint in terms of the $\infty \to 1$ norm (more generally $q \to 2$ norm in terms of the $q \to q^*$ norm). This is because for any matrix $B$, we have that $\norm{B}_{q \to 2}^2 = \norm{B^\top B}_{q \to q^*}$ where $\ell_{q^*}$ is the dual norm for $\ell_q$ and satisfies $\tfrac{1}{q^*}+\tfrac{1}{q}=1$. The main advantage of this reformulation is that the $q \to q^*$ operator norms {\em are} indeed {\em monotone}. %(see Claim~\ref{lem:monotone} for a simple proof).

Our next important observation is that 
we can replace the constraint \eqref{tech:sdp:norm} by a similar constraint in terms of the $q \to q^*$ norm. This is because for any matrix $B$, we have that $\norm{B}_{q \to 2}^2 = \norm{B^\top B}_{q \to q^*}$ where $\ell_{q^*}$ is the dual norm for $\ell_q$ and satisfies $\tfrac{1}{q^*}+\tfrac{1}{q}=1$. The main advantage of this reformulation is that the $q \to q^*$ operator norms {\em are} indeed {\em monotone}. %(see 

\begin{claim}[Same as Claim~\ref{lem:monotone}]
%[Monotonicity of $q \to q^*$ operator norm]\label{lem:monotone}
For any $q \ge 1$, the operator norm  $\norm{\cdot}_{q \to q^*}$ is monotone.% i.e., $\norm{A+B}_{q \to q^*} \ge \norm{A}_{q \to q^*}$ for any $A, B \succeq 0$.
\end{claim}

Moreover polynomial time $O(1)$-approximate separation oracles based on semidefinite programs exist for these norms when $q \ge 2$. This motivates convex programming relaxation {\CPone}and its equivalent but more elegant convex relaxation {\CPtwo}shown in Figure~\ref{intro: fig:relaxations:fr}. 
%can be efficiently separated with an $O(1)$ slack factor when $q \ge 2$. % since $\norm{M}_{\infty \to 1}=\max_{x,y \in \sset{\pm 1}^n} \iprod{M,xy^\top}$. %In fact, the $\infty \to 1$ norm computation problem corresponds to the famous Grothendieck problem that has been widely studied in functional analysis and theoretical CS~\citep{Gro56,alon2004approximating, nesterov1998semidefinite,braverman2013grothendieck}. 
%i.e.,  $\norm{Ax}_{\infty \to 2}^2=\max_{x \in \sset{\pm 1}^n} x^\top AA^\top x = \max_{x \in \sset{\pm 1}^n} \iprod{AA^\top,xx^\top}$ for a given matrix $A$. 

\begin{figure}[h]
\fbox{\begin{minipage}{.48\linewidth}
\CPone:
\begin{align*}
    \min_X &\norm{A}_F^2 - \iprod{AA^\top, X} \\
    \text{s.t. }& \tr(X) \le r, ~~ 0 \preceq X \preceq I \\
    & \max_{\substack{Y \in \calQ}}\iprod{X,Y}\le C_G \kappa^2 , \nonumber \text{ where}\\
    \calQ=& \sset{ Y \in \R^{n \times n}:~ Y \succeq 0, \sum_{i = 1}^n Y_{ii}^{q/2} \le 1 }  
\end{align*}
\end{minipage}}%
\fbox{\begin{minipage}{0.48 \linewidth}
\CPtwo:
\begin{align*}
    \min_{X \in \R^{n \times n}, d \in \R_{\ge 0}^n} &\norm{A}_F^2 - \iprod{AA^\top, X} \\
    \text{s.t. }& \tr(X) \le r, ~~ 0 \preceq X \preceq I \\
                & X \preceq \diag(d) \\
 \norm{d}_{q/(q-2)} := &\Big(\sum_{i=1}^n d_i^{q/(q-2)}\Big)^{(q-2)/q} \le C_G \kappa^{2}. 
\end{align*}
\end{minipage}}
\caption{Two equivalent tractable convex relaxations \CPone  and \CPtwo  for problem \eqref{int:obj}. % with Frobenius norm objective. 
See Lemma~\ref{lem:convexdual} for proof of equivalence using convex duality.} \label{intro: fig:relaxations:fr}
\end{figure}

\vspace{-10pt}

% The final algorithm in fact uses convex duality to design an elegant convex relaxation where the constraint $\norm{X}_{\infty \to 2} \le \kappa$ in \eqref{tech:sdp:norm} is replaced by 
% $$ X \preceq \text{diag}(d), ~\text{ and }~ \sum_{i=1}^n d_i \le  O(\kappa^2), $$
% where $d \in \R^n$ represents an additional set of variables to minimize over (see \CPtwo in Section~\ref{sec:worstcase} for details about general $q \ge 2$). %This relaxation can be solved in polynomial time (see \CPtwo in Section~\ref{sec:worstcase}).
%
%Our final algorithm approximately solves the mathematical program \eqref{tech:sdp:obj} with constraint \eqref{tech:sdp:norm} replaced by $\norm{X}_{\infty \to 1} \le \kappa^2$. 
Let $X^*$ be the optimal solution to the convex program. We obtain the required robust low-rank projection matrix from $X^*$ by a simple rounding procedure that focuses on the large singular values of $X^*$. %We produce a low-rank projection matrix by first truncating the solution to the program $\widehat{X}$ to the large singular values, and then picking the $r$ terms that contribute the most towards the objective. Our 
%With the monotonicity property in hand, a neat analysis shows that the resulting low-rank projection is $O(\kappa)$-robust, while also achieving an $O(1)$ approximation to the objective. %See Theorem~\ref{prop:gen:fr} for an analysis with a general monotone norm constraint, and Theorem~\ref{thm:worstcase:frob} for our desired problem.
The monotonicity property of the norm leads to an elegant analysis to guarantee that the resulting low-rank projection is $O(\kappa)$-robust, while also achieving small error. 
%
%\anote{Not sure if we need to include Lemma statement. Commented out. }
%
We now sketch the proof of Theorem~\ref{ithm:worstcase:frob} with the Frobnenius norm objective. 
%\noindent {\bf Spectral norm objective.} 
Similar ideas also work when the projection error is measured in terms of the spectral norm. However, the objective function is instead rephrased as $\min \spnorm{(A^\top (I-X) A}$ where $\spnorm{\cdot}$ is the spectral norm; the algorithm and analyis, while slightly different again leverage the monotonicity property of the $q \to q^*$ norms. %One can verify that this also gives a convex relaxation; the rounding and analysis while being a little different, again leverage the monotonicity property of $q \to q^*$ operator norms. 

% \begin{lemma}\label{lem:lowrank}
% Let $\eps, \delta>0$, and $M \succeq 0$. Suppose $X$ satisfies the SDP constraints \eqref{gen:sdp:spectralnorm} and $\iprod{M,X} \ge (1-\eps) \tr(M)$. Suppose $P^{1-\delta}_X$ is the projection operator onto the subspace spanned by eigenvectors of $X$ with eigenvalues at least $(1-\delta)$. Then we have
% \begin{equation}
%     \iprod{I-P^{1-\delta}_X, M} \le \frac{\eps}{\delta} \cdot \tr(M).
% \end{equation}
% \end{lemma}

\begin{proofsketchof}{Theorem~\ref{ithm:worstcase:frob}}
Assume $\norm{A}_F=1$ without loss of generality, and let $OPT=\eps \in [0,1]$. It is easy to see any feasible projection matrix $\Pi$ of rank $r$ satisfying $\norm{\Pi}_{q \to 2} \le \kappa$ forms feasible solutions to {\CPone}and {\CPtwo} (for an appropriate feasible $d$) with the correct objective value. Moreover the relaxations \CPone and \CPtwo can be solved in polynomial time to arbitrary accuracy using the Ellipsoid algorithm. See Claim~\ref{claim:feasibility} for details.

Set $\delta:=1/(1+\gamma)$. Let $\widehat{X}=\sum_i \lambda_i v_i v_i^\top$ and $S=\sset{i : \lambda_i \ge 1-\delta}$. Define for each $i \in [S]$, $\alpha_i:= \iprod{v_i v_i^\top, AA^\top}$. We form $T$ from $S$ by picking the $\min\sset{r,|S|}$ ones with the largest $\sset{\alpha_i}$ values. %We sort the elements of $S$ based on $\sset{\alpha_i}$, and pick greedily the first $\min\sset{r, |S|}$ of them to form $T \subseteq S$.
Let $\Pi_S=\sum_{i \in S} v_i v_i^\top$. Our projection matrix will be $\Pi_T =\sum_{i \in T} v_i v_i^\top$. %\pnote{rest of sentence not needed.} for some appropriately chosen subset $T \subseteq S$. 

We use monotonicity of $\intronorm{\cdot}$ to show the operator norm constraint is satisfied:
%To show that the operator norm constraint is approximately satisfied. By the monotonicity of the $\intronorm{\cdot}$ we have
\begin{align*}
    \intronorm{\Pi_T} &= \intronorm{\sum_{i \in T} v_i v_i^\top} \le   
     \frac{1}{1-\delta}\intronorm{\sum_{\substack{i: \lambda_i > 1-\delta}} \lambda_i v_i v_i^\top } \le  \frac{\intronorm{\widehat{X}}}{1-\delta}     \le \frac{\alpha \kappa}{1-\delta} = \frac{\alpha(1+\gamma) \kappa}{\gamma}  
\end{align*}

\noindent Also, since $\Pi_S$ (and hence $\Pi_T$) projects onto the large eigenspace of $\widehat{X}$, we can prove that by truncating onto the large eigenvalues (see Lemma~\ref{lem:lowrank})
\begin{align}
\iprod{I-\Pi_S, AA^\top}&=\sum_{i \notin S} \iprod{v_i v_i^\top, AA^\top} \le \frac{\eps}{\delta}, \text{ and } \nonumber\\
\sum_{i \in S} \iprod{(1-\lambda_i) v_i v_i^\top, AA^\top} &\le \sum_i \iprod{(1-\lambda_i) v_i v_i^\top, AA^\top} \le 1 - \iprod{X,AA^\top} \le \eps. \nonumber
\end{align}
$$\text{Hence, }  \sum_{i \in S} \lambda_i \alpha_i =\sum_{i \in S} \lambda_i \iprod{v_i v_i^\top, AA^\top}  \ge 1- \eps\Big(1+\frac{1}{\delta}\Big)= 1-(2+\gamma)\eps,$$
for our choice of $\delta=1/(1+\gamma)$.
 By our greedy choice of $T$, we have $\sum_{i \in T} \alpha_i \ge \sum_{i \in S} \lambda_i \alpha_i$, as $\sum_{i \in S} \lambda_i \le  \min\sset{\tr(X),|S|} = |T|$, with each $\lambda_i \in [0,1]$. Thus $\norm{\Pi_T^{\perp} A}_F^2 \le (2+\gamma)\eps$. This completes the proof. For the bicriteria guarantee with rank $r/(1-\delta)$ we output $\Pi_S$. The objective and $\intronorm{\cdot}$ bounds follow using similar arguments.  

% % {\bf Assume (*):}{\em We have a randomized (dependent) rounding scheme such that we pick $T \subseteq S$ with $|T| \le r$  such that $\Pr[i \in T] \ge \lambda_i$.}\\
% % If we have (*), then 
% % $$\E_T \Big[ \iprod{P_T, M} \Big] =  \sum_{i \in S} \Pr[i \in T] \iprod{v_i v_i^\top , M}  \ge \sum_{i \in S} \lambda_i \iprod{v_i v_i^\top , M} \ge 1-\eps\Big(1+\frac{1}{\delta}\Big),$$
% % as required.

% Observe that $|S| \le r / (1-\delta)$ from the trace constraint. The bicriteria guarantee with rank $r/(1-\delta)$ is obtained by returning the projection $\Pi_S=\sum_{i\in S} v_i v_i^\top$.  The operator norm bounds follows using the same argument as \eqref{eq:worstcase:bicriteria} with $T=S$. Moreover, the objective value follows directly from Lemma~\ref{lem:lowrank}. 
\end{proofsketchof}

\vspace{-20pt}

% \noindent {\bf Spectral norm objective.} Similar ideas can be used when the projection error is measured in terms of the spectral norm. The objective function in the relaxations \CPone and \CPtwo is instead rephrased as $\min \spnorm{(A^\top (I-X) A}$ where $\spnorm{\cdot}$ is the spectral norm. One can verify that this also gives a convex relaxation; the rounding and analysis while being a little different, again leverage the monotonicity property  %\vnote{do we mean $q\to q^*$ here?} 
% of $q \to q^*$ operator norms. 
%Finally, we remark that a subsequent paper~\citep{awasthi2020adversarial} that builds on ideas in this work, gives a practical heuristic for finding robust projections. 
%Finally, while the above algorithm uses the Ellipsoid algorithm with an SDP-based approximate separation oracle, we believe that these insights could lead to practical algorithms for the problem (see \cite{awasthi2020adversarial} for a fast heuristic that builds on these ideas). 

\subsubsection{Technical Overview for Training-Time Adversarial Perturbations}\label{sec:techniques:training}

%\anote{Motivate why it is not a triangle inequality. }

%We now describe how our algorithm for finding a robust low-dimensional representation is itself robust to adversarial perturbations at training time. 
Let $q=\infty$. Recall that our input instance $\tilde{A}$ is a $\delta$-corrupted instance obtained from $A$ by potentially corrupting every entry of it by a $\delta$ amount. % i.e., each co-ordinate of {\em every} data point up to an amount of $\delta>0$. 
Our goal is to output a robust low-rank projection matrix $\Pi$ of rank at most $r$ for the uncorrupted matrix $A$, that is not known to us. This question is interesting even from a purely statistical standpoint; but additionally, we would also like our algorithm to run in polynomial time.

Why should this be possible? 
Suppose the uncorrupted matrix $A$ has a robust low-rank projection $\Pi^*$ of small error i.e., $\norm{A-\Pi^* A}< \eps \norm{A}$ (where $\norm{A}$ is either the Frobenius norm or spectral norm). Also assume for just this discussion that the average column (Euclidean) length of $A$ is 1, $\kappa=n^{0.1}$ say and $\delta = o(n^{-0.1})$. For any $\kappa$-robust projection $\Pi$, $\Pi A_j \approx \Pi \tilde{A}_j$ for each data point $j \in [m]$ .  So one could apply the worst-case algorithm on the corrupted input $\tilde{A}$, and hope to  also get a robust projection of low-error for the unknown matrix $A$. 

However, there are two major challenges in implementing this strategy. {\em (1) Solution value of $\tilde{A}$:} the robust projection $\Pi^*$ may not achieve low error on $\tilde{A}$; in fact, $\tilde{A}$ may not have any good robust low-rank approximation --  in this case the algorithm output may be useless. This is because the entry-wise perturbations could make $A$ and $\tilde{A}$ far away in aggregate e.g., $\norm{A-\tilde{A}}_F$ could be $\delta \sqrt{n m} \gg \sqrt{m} \approx \norm{A}_F$. \\%(whereas, even $\norm{A}_F \approx \sqrt{m})$. 
\noindent {\em (2) Identifiability issue:} perhaps more importantly, even if the perturbation $\tilde{A}$ has a robust low-rank projection of small error, we need to argue that this subspace indeed attains small error on $A$! The second issue is crucial in resolving the purely information-theoretic aspect of the question; it involves ruling out the scenario where $\tilde{A}$ has good robust low-rank approximation that is very different from any robust low-rank approximation for $A$.

%\noindent {\em Issue (2): Identifiability.} 

To address the second issue (identifiability), we prove that if the projection $\widehat{\Pi}$ gives a small error on $\tilde{A}$, it necessarily gives a low-error on $A$. Roughly speaking, if there are two data-matrices $A$ and $B$ with $\norm{A-B}_\infty \le \delta$, then for $\gamma \in (0,1)$
$$ \norm{A-\Pi_1 A} , \norm{B-\Pi_2 B} < \gamma \norm{A}~~~\implies ~~    \norm{A - \Pi_2 A} \le \gamma_1 \norm{A}+ \tfrac{1}{\gamma_2}\sqrt{m} \delta \kappa, \text{ (and similarly for } B),$$
where $\gamma_1=\gamma_1(\gamma), \gamma_2=\gamma_2(\gamma) \in (0,1)$.  % (See Lemma~\ref{lem:training:spectralnorm} for a formal statement). 
One can show that $\norm{\Pi_1 A- \Pi_1 B}$ and $\norm{\Pi_2 A - \Pi_2 B}$ are small since $\Pi_1, \Pi_2$ are robust (see Lemma~\ref{lem:robustproperties}); however this does not give a handle on $\norm{A-\Pi_2 A}$.  
Note that the above statement does not follow from an application of the triangle inequality since we do not have any prior control on $\Pi_1 - \Pi_2$. This statement is particularly tricky to show for the spectral norm.  
A natural approach is to argue that $\Pi_1 A$ and $\Pi_2 B$ are close by arguing about their actions on any unit vector. %The difficulty arises when some directions in $\Pi_2$ are close to the subspace of $\Pi_1$, yet their orthogonal component is small, but spread out (even when $r=1$). %For e.g., even when $r=1$, where $\Pi_1, \Pi_2$ are identified with analytically sparse directions $v_1, v_2$, it could be the case that $v_2$ has a small component orthogonal to $v_1$ that is very spread out (hence non-sparse).
%in effect, the above statement also shows that $\norm{A-B}$ has to be small. We could try to construct a common robust subspace that gives low error for both $A$ and $B$ (thus proving that $A$ and $B$ are close). A natural choice would be the span of the union of the subspaces given by $\Pi_1, \Pi_2$. Unfortunately this subspace may not even be robust! For example even for $r=1$ where $\Pi_1, \Pi_2$ are identified with analytically sparse directions $v_1, v_2$, it could be the case that $v_2$ has a small component orthogonal to $v_1$ that is very spread out (hence non-sparse). When $r=1$, we can still salvage the claim by doing an win-win analysis: either $v_1, v_2$ are already close, or the union is also robust (see Lemma~\ref{lem:spectral-norm-closeness-of-subspaces}). But when $r \gg 1$, this approach breaks down since different directions in subspace corresponding to $\Pi_2$ may have very different projection lengths onto the subspace of $\Pi_1$. 
We use a somewhat indirect proof; we show that for every direction $v \in \mathbb{S}^{n-1}$, (1) the lengths $\norm{Av}_2$ and $\norm{Bv}_2$ are similar and (2) the difference in the lengths $|\norm{Av}_2 - \norm{Bv}_2|$ is (approximately) lower bounded by $\norm{(A - \Pi_2 A)v}_2$.  This will allow us to conclude that $\norm{A - \Pi_2 A}$ is small. % in spectral norm. 

%\noindent{\em Issue (1): Solution value of $\tilde{A}$.} 

To tackle the first issue (solution value), we first  preprocess (denoise) to find an alternate matrix $A'$ with a good solution value. %know that there is a data matrix (in particular the uncorrupted matrix $A$) that can be obtained from $\tilde{A}$ by perturbing each entry by at most $\delta$ (i.e., it is in the $\ell_\infty$ neighborhood of $\tilde{A}$) that has a robust low-rank approximation of low value $\eps \norm{A}^2$. Let us 
Suppose we have an algorithm to find
%can solve the following optimization problem: 
\begin{equation} \label{eq:tech:training}
A' = \argmin_{\substack{B: \norm{B - \tilde{A}}_\infty \le \delta}} ~~~\min_{\substack{\Pi: \text{rank}(\Pi)=r,~ \norm{\Pi}_{\infty\to 2} \le \kappa}} \norm{B - \Pi B}^2.
\end{equation}
We know that the uncorrupted matrix $A$ is a feasible solution with good value. Hence the optimal solution $A'$ of \eqref{eq:tech:training} has an even better solution. %a robust low-rank approximation of even smaller solution value. 
Moreover $\norm{A-A'}_\infty \le 2\delta$. This reduces the first issue to a computational question of solving \eqref{eq:tech:training}. 
For Frobenius norm error, we can obtain a good $A'$ by instead solving a simple convex optimization problem. % of finding $A'$ that $\text{argmin}_B \norm{B}_F^2$ over all matrices $B$ that are valid $\delta$-perturbations of $\tilde{A}$ i.e., $\norm{\tilde{A}-B}_q \le \delta$. %The minimizer here just reduces the magnitude of each entry by $\delta$ or until it is $0$. For more general $q \ne \infty$, 
%For all $q$,  this corresponds to a simple convex minimization problem.  

For the spectral norm problem, we do not know of an efficient algorithm for \eqref{eq:tech:training}. However by running our worst-case algorithm (for spectral norm error) on $\tilde{A}$, we will either find a good solution that also works for $A$, or we will certify that %$\tilde{A}$ has no good robust low-rank approximation; this certifies that 
$\spnorm{A - \tilde{A}}$ is too large i.e., the data was poisoned significantly. 
Finally we remark that we get the stronger computationally efficient guarantee for the spectral norm error (as for the Frobenius norm error) if we can resolve the spectral norm variant of \eqref{eq:tech:training}, which is an open question.

\subsection{Related and Concurrent Work on Training-time corruptions.}\label{sec:comparison}

Subsequent work by \citet{ACV20coltsub}, studies training time robustness in an average-case setting namely, the spiked covariance model where the goal is to recover the top principle subspace of the data distribution. They extend the algorithms developed in this work~(Section~\ref{sec:worstcase}) to the average case setting, and in fact show that the $q\to 2$ operator norm of the principal subspace almost characterizes its robustness to adversarial perturbations at training time in the spiked covariance model. Very recently, \citet{kothari2020pca} studies the problem of recovering an $\ell_0$-sparse\footnote{Note that any $\ell_0$ sparse unit vector is also $\ell_1$ sparse; see Claim~\ref{claim:norms-holder}.} principal component where there are adversarial perturbations in $\ell_\infty$ norm to the training data points, again focusing on the spiked covariance model. In contrast, our work studies the worst case formulation of the problem.

\vspace{2pt}
\noindent\emph{Comparison to the Huber contamination model and the robust PCA problem.} There is a vast amount of literature in designing robust algorithms in a different model, the Huber's contamination model, where, unlike our setting, a small fraction of the data can be arbitrarily corrupted~\citep{huber2011robust, diakonikolas2018learning, lai2016agnostic, diakonikolas2019survey}.
% There is a vast literature in robust statistics and TCS on designing robust algorithms for high-dimensional estimation under the Huber contamination model, where a small fraction of the (training) points or samples are corrupted arbitrarily~\citep{huber2011robust, diakonikolas2018learning, lai2016agnostic, diakonikolas2019survey}. 
Our notion of training-time adversarial perturbations is very different in flavor -- it involves bounded adversarial perturbations to potentially every training point. 
%The problem of PCA under training-time corruptions has received significant attention in recent years~\cite{de2003framework, candes2011robust, chandrasekaran2011rank, huber2011robust, diakonikolas2019survey}. 
%The Huber contamination model considers the setting where a small fraction of the data points are corrupted arbitrarily~\cite{huber2011robust}. Recent works give robust high-dimensional estimators . 
Another popular model is the robust PCA problem proposed in~\citet{candes2011robust}. It assumes that a given corrupted matrix $\tilde{A}$ is a sum of two matrices, the true matrix $A$ that is low-rank and a sparse corruption matrix $S$ with sparsity pattern being essentially random. The corruptions although sparse can be unbounded in magnitude. This setting is again fundamentally different from ours. Recovery in this model necessitates incoherence type structural assumptions that the principal components of $A$ are spread out, whereas in our setting {\em sparsity or localization} of the signal dictates the recovery error. 

Please see Section~\ref{sec:related-work} for more details, and comparison to other related work.

\section*{Acknowledge}
The authors would like to thank the anonymous reviewers for their valuable feedback. AV was supported by the National Science Foundation (NSF) under Grant No. CCF-1652491, CCF-1637585 and CCF 1934931.

\bibliographystyle{apalike} 
\bibliography{main}

\newpage 
\appendix

% !TEX root=main.tex

\section{Notation and Preliminaries} \label{sec:prelims}
\anote{Make some of these lemmas into claims, since they're too simple for lemmas?}

\paragraph{Norms.}
For every $q \ge 1$ and $x \in \R^n$, we will use $\norm{x}_q$ to denote the $\ell_q$ norm of the vector $x$ i.e., $\norm{x}_q^q = \sum_{i \in [n]} |x_i|^q$. The dual norm of $\ell_q$ is $\ell_{q^*}$ where $1/q^* + 1/q =1$. We will heavily use Holder's inequality which states that
\begin{equation} \label{eq:holder}
    \text{(H\"older's inequality)}~~ |\iprod{u,v}| \le \norm{u}_{q} \cdot \norm{v}_{q^*} ~~ \forall u, v \in \R^n.
\end{equation}

When not specified, $\norm{x}$ will denote the Euclidean norm of $x$. 
Further $\mathbb{S}^{n-1}$ will represent the unit sphere for the Euclidean norm. For convenience, we will use $\norm{x}_0$ to denote the sparsity i.e., the size of the support of $x$ (note that $\ell_0$ is not a valid norm on vectors). 

\paragraph{Operator Norms of Matrices.} We will use the following matrix norms. For any $q, p \ge 1$ and any matrix $M \in \R^{n \times m}$, we will denote by  $\norm{M}_{q \to p} = \max_{y \in \R^m, \norm{y}_q \le 1} \norm{M y}_p$. By duality of vector norms, we have
$$\norm{M}_{q \to p}  = \max_{y \in \R^m, \norm{y}_q \le 1} \max_{z \in \R^n, \norm{z}_{p^*} \le 1} z^\top M y = \max_{z \in \R^n, \norm{z}_{p^*} \le 1} \max_{y \in \R^m, \norm{y}_q \le 1} y^\top M^\top z = \norm{M^\top}_{p^* \to q^*}.$$

When $p=q=2$, this corresponds to the spectral norm of the matrix $M$ i.e., the maximum singular value of $M$. When unspecified, we will use $\spnorm{M}$ to denote the spectral norm of $M$. 
 (Note that the above equalities from duality also show that the $\spnorm{A}=\spnorm{A^\top}$ i.e., the maximum right singular value is the same as the maximum left singular value).
% \pnote{Should we make all matrices be $n \times m$ to be consistent with the rest of the paper?}\anote{11/3:fixed.} 
We will also make use of the following claim relating the $q\to 2$ and $q\to q^*$ norms of a matrix.
 
 \begin{claim}\label{lem:extra}
For any projection matrix $\Pi$, and $q \geq 2$,  $\norm{\Pi}_{q \to q^*}=\norm{\Pi}_{q \to 2}^2$ (this is crucially an equality, and not just an inequality). More generally, for any matrix $B$, we have $\norm{B^\top B}_{q \to q^*}=\norm{B}^2_{q \to 2}$.
\end{claim}
\begin{proof}
Note that by duality of norms, we have for any matrix $B$ we have:
\begin{align*}
    \norm{B^\top B}_{q \to q^*}&=\max_{\norm{y}_q \le 1} \norm{B^\top B y}_{q^*}= \max_{\norm{y}_q \le 1, \norm{z}_q \le 1} z^\top B^\top B y = \max_{\norm{y}_q \le 1, \norm{z}_q \le 1} \iprod{B z , B y} \\
    &= \max_{\norm{y}_q \le 1} \iprod{B y , B y}
    = \max_{\norm{y}_q \le 1} \norm{B y}_2^2 = \norm{B}_{q \to 2}^2.
\end{align*}
For a projection matrix $\Pi$ we also have $\Pi=\Pi^\top$ and $\Pi^2=\Pi$. Hence the lemma follows. 
% \begin{align*}
%     \norm{\Pi}_{q \to q^*}&=\max_{\norm{y}_q \le 1} \norm{\Pi y}_{q^*}= \max_{\norm{y}_q \le 1, \norm{z}_q \le 1} z^\top \Pi y = \max_{\norm{y}_q \le 1, \norm{z}_q \le 1} \iprod{\Pi z , \Pi y} = \max_{\norm{y}_q \le 1} \iprod{\Pi y , \Pi y}\\
%     &= \max_{\norm{y}_q \le 1} \norm{\Pi y}_2^2 = \norm{\Pi}_{q \to 2}^2
% \end{align*}
% where we used that $\Pi^2=\Pi$ (in the third equality). 
%\vspace{-0.5cm}
\end{proof}

\paragraph{Entry-wise Norms of Matrices.}

We will also consider various matrix norms obtained by considering a matrix $M \in \R^{m \times n}$ as a vector of size $m n$. In particular, for any $q \ge 1$ we will use $\norm{M}_q$ to denote the $\ell_q$ norm of the ``flattened'' vector corresponding to $M$ i.e., $\norm{M}_q^q = \sum_{i=1,j=1}^{m,n} |M(i,j)|^q$. The Frobenius norm  $\norm{M}_F = \norm{M}_2$. 
Moreover for matrices $A, B$, we use $\iprod{A,B} := \tr(A^\top B)$ to represent the trace inner product.

\paragraph{Monotonicity of Matrix Norms.}
The following property of certain matrix norms will be crucial in designing constant factor approximation algorithms for the low-rank approximations. 

\begin{definition}[Monotone matrix norm]\label{def:monotonicity}
A matrix norm $\gnorm{\cdot}$ is said to be monotone if and only if 
\begin{equation}\label{eq:monotonenorm}
\forall A, B \succeq 0, ~~    \gnorm{A+B} \ge \gnorm{A}.
\end{equation}
\end{definition}
Observe that it suffices to check the above condition for all rank-$1$ PSD matrices $B$ i.e., $B=vv^\top$ for $v \in \R^n$. 
It is well known that all unitarily invariant matrix norms\footnote{A matrix norm $\gnorm{\cdot}$ is unitarily invariant iff $\gnorm{A}=\gnorm{UAV}$ for all matrices $A$ and all unitary matrices $U,V$.} are monotone (this is because unitarily invariant norms are just norms on the singular values). On the other hand, many other matrix norms including other entry-wise norms $\norm{X}_q$ or general operator norms $\norm{X}_{q \to p}$ are not necessarily monotone (see Claim~\ref{claim:counterexample:monotone} and Claim~\ref{clm:counterexample_infinity_2} for some counterexamples).  Perhaps surprisingly, the $q \to q^{*}$ matrix operator norms are monotone (see Claim~\ref{lem:monotone} for a simple proof of this fact)!  

\paragraph{High probability bounds.} We will say that an event holds {\em with high probability (w.h.p.)} if the probability of failure on a given instance is less than {\em any} polynomial of the input parameters e.g., the dimension $n$, and the number of data points $m$. We remark that in all our settings, one can amplify the success probability to $1-\eta$ for any small $\eta>0$ by repeating the algorithm $\log(1/\eta)$ times. %(hence these guarantees will in fact hold with exponentially small failure probability). 

\subsection{Properties of Robust Projections.} \label{sec:properties}

Throughout the paper we will use the term projections and projection matrices to always refer to orthogonal projection matrices on to linear subspaces of $\R^n$. Next we list and prove some simple properties of subspaces with {\em robust} projection matrices i.e., subspaces with $\norm{\Pi}_{\infty \to 2}$ (or more generally $q \to 2$ norm for some $q \ge 2$) that is upper bounded.

For any $q^* \in [1,2]$, the ratio of the $\ell_{q^*}$ vs $\ell_2$ corresponds to an analytic notion of sparsity. The following claim gives an upper bound on the $\ell_{q^*}$ norm in terms of the sparsity. 

\begin{claim}[Analytic Sparsity] \label{claim:norms-holder}
Consider any vector $v \in \R^n$ of support size $k$. For any $q^* \in [1,2]$, we have
$$ \norm{v}_{q^*} \le k^{\tfrac{1}{q^*}-\tfrac{1}{2}} \norm{v}_2 .$$
In particular, $\norm{v}_1 \le \sqrt{k} \norm{v}_2$ for vectors with support size at most $k$.
\end{claim}
On the other hand, it is easy to see that the bound given here is tight for any vector that is equally spread out among its support of size $k$.  
\begin{proof}
Without loss of generality suppose $\norm{v}_2 =1$ (if $v=0$ it holds trivially).
Let $v$ have support $S$ of size $k$. Set $p:=2/q^*$, and let $u$ be the vector such that $u_i=|v_i|^{q*}$ for each $i \in [n]$. By Holder's inequality
$$    \norm{v}_{q^*}^{q^*} = \sum_{i \in S} 1 \cdot u_i \le  \norm{\mathbf{1}_S}_{p^*} \norm{u}_p \le k^{1/p^*} \Big(\sum_i |v_i|^{p q^*} )^{1/p} \le k^{1- q^*/2} \norm{v}_2^{2/p} = k^{1-q^*/2},$$
hence establishing the lemma.
\end{proof}

Recall that $\ell_{q^*}$ corresponds to the dual norm for $\ell_q$, and $q^* \in [1, 2]$ when $q \ge 2$. 
The following simple lemma proves two useful properties of robust subspaces i.e., subspaces having projection matrices with bounded $\infty \to 2$ norm (or more generally $q \to 2$ norm for $q >2$). The first property shows that any two vectors that are close in $\ell_\infty$ norm will have nearby projections onto any subspace that is robust.
The second property shows that a subspace is robust (i.e., has a robust projection matrix) exactly when every vector in the subspace is {\em analytically} sparse. 
\begin{lemma}\label{lem:robustproperties}[Properties of Robust Subspaces and Projections]
Consider any subspace of $\mathcal{V} \subseteq \R^n$ with projection matrix $\Pi \in \R^{n \times n}$ satisfying $\norm{\Pi}_{q \to 2} \le \kappa$. We have the following two properties:

\begin{enumerate}
    \item[I.] {\bf Closeness of projections of pertubations:} For any vector $v$ and its perturbation $\tilde{v}$
$$ \norm{v - \tilde{v}}_{q} \le \delta ~~\implies ~~ \norm{\Pi \tilde{v} - \Pi v}_2 \le \kappa \delta .$$
    
    \item[II.] {\bf Analytic sparsity:} For any $v \in \mathcal{V}$, we have $\norm{v}_{q^*} \le \kappa \norm{v}_2$, where $q^*=q/(q-1)$. Moreover, if every vector in $\mathcal{V}$ has $\norm{v}_{q^*}  \le \kappa \norm{v}_2$, then $\norm{\Pi}_{q \to 2} \le \kappa$. In particular  $\norm{\Pi}_{\infty \to 2} \le \kappa$ if and only if $\norm{v}_1 \le \kappa$ for all unit vectors $v \in \mathbb{S}^{n-1} \cap \mathcal{V}$. 
\end{enumerate}
\end{lemma}

\begin{proof}
We first show property (I). Let $u:=v - \tilde{v}$. Then
$$ \norm{\Pi \tilde{v} - \Pi v}_2 = \norm{\Pi u} \le \norm{\Pi}_{q \to 2} \norm{u}_q \le \kappa \delta.$$ 
To show property (II), note that by duality of matrix operator norms we have $\norm{\Pi}_{q \to 2} = \norm{\Pi^\top}_{2 \to q^* } = \norm{\Pi}_{2 \to q^*}$. %Let $\mathcal{V}_{\Pi} \subseteq \R^n$ be the subspace corresponding to $\Pi$.
 %Hence
 $$\text{Hence } ~~~\forall v \in \mathbb{S}^{n-1} \cap \mathcal{V}, ~~~ \norm{v}_{q^*} = \norm{\Pi v}_{q^*} \le \norm{\Pi}_{2 \to q^*} \norm{v}_2 \le \kappa.$$
 For the converse, if there exists $v \in \mathbb{S}^{n-1} \cap \mathcal{V}$ s.t. $\norm{v}_{q^*} > \kappa$, then by duality $\norm{\Pi}_{2 \to q^*}=\norm{\Pi}_{q \to 2}> \kappa$.

\end{proof}

Observe that the robustness condition on the subspace as captured by the $q \to 2$ operator norm bound of its projection matrix $\Pi$ is basis independent.The following simple claim gives a simple sufficient condition on the basis of the subspace that implies robustness of the subspace spanned by them. This relates our robustness of the subspace to alternate notions of sparsity of subspaces that have been studied in the literature on sparse PCA~\citep{VuLei12, VuLei13}. 

\begin{claim}\label{lem:basissparsity}
Given any orthonormal basis $v_1, v_2,\dots, v_r$ for a subspace $\mathcal{V}$ such that $\norm{v_i}_{q^*} \le \kappa$ for each $i \in [r]$, we have $\norm{\Pi}_{q \to 2} \le \sqrt{r} \kappa$. 
\end{claim}
\begin{proof}
Firstly, $\Pi= \sum_{i=1}^r v_i v_i^\top$, and $\norm{\Pi}_{q \to 2} = \norm{\Pi}_{2 \to q^*}$. We have
$$\norm{\Pi}_{q \to 2} = \max_{u: \norm{u}_q \le 1} \Bignorm{\sum_{i=1}^r \iprod{u,v_i} v_i}_2 = \max_u \sqrt{ \sum_{i=1}^r \iprod{u,v_i}^2} \le \sqrt{r} \cdot \max_{u: \norm{u}_q \le 1} \max_{v: \norm{v}_{q^*} \le \kappa} |\iprod{u,v}| \le \sqrt{r} \kappa.$$
\end{proof}

For a given matrix $B \in \R^{n \times m}$, let us denote by $\Pi(B)$ to the projection matrix onto the column space of $B$. The following lemma shows that the best low-rank $(\kappa,q)$-robust projection objective \eqref{intro:obj} also finds the low-rank approximation that has smallest error among ones with a $(\kappa,q)$ robust column space. 
%\xnote{Nov 2nd: SODA review 2 mentioned why not use $\|A - \Pi(B) A\| \le \|A-B\|$ where $\Pi(B)$ denotes the projection of $B$?}
\begin{claim}\label{lem:lowrankobj}
Let $\mathcal{P}_r$ be the set of all rank-$r$ projection matrices.
Given a data matrix $A \in \R^{n \times m}$ and a given parameter $\kappa \ge 1$, $q>0$, we have 
$$ \min_{\substack{\Pi \in \mathcal{P}_r \\ \norm{\Pi}_{q \to 2} \le \kappa}}  \norm{A-\Pi A} = \min_{\substack{B: \text{rank}(B) \le r, \\ \norm{\Pi(B)}_{q \to 2} \le \kappa}} \norm{A- B},$$
where $\norm{M}$ here stands for the spectral norm. The above statement is also true for the Frobenius norm. 
\end{claim}
\begin{proof}   
Let $B^*$ be the minimizer for the right minimization problem and let $\Pi_2=\Pi(B^*)$ be its projection matrix, and let $\Pi_1$ be the minimizer for the left optimization problem. 
It is easy to see that $\norm{A-\Pi_1 A} \ge \norm{A-B^*}$, since $\Pi_1 A$ is also a feasible choice for $B$ in the right minimization problem. 
\xnote{11/5: New argument.}
The other direction follows from the fact that $\|A - \Pi(B) A\| \le \|A - B\|$ holds for both Frobenius norm and the spectral norm (specifically, $\|A v - \Pi(B) A v\|_2 \le \|A v - Bv\|_2$ for any $v \in \mathbb{R}^{m}$).

\iffalse
To prove the other direction, if $A_i, B^*_i$ are the $i$th columns of $A,B^*$ respectively then,
\begin{align*}
\norm{A - B^*} &= \norm{A - \Pi_2 B^*} = \norm{\Pi_2(A - B^*) + \Pi_2^{\perp} A} \\
&= \max_{u \in \mathbb{S}^{m-1}}  \norm{\sum_{i \in [m]} u_i \Pi_2(A_i - B^*_i) + \sum_{i \in [m]} u_i \Pi_2^{\perp} A_i}_2\\ 
%= \max_{u \in \mathbb{S}^{m-1}}  \sqrt{\norm{\sum_{i \in [m]} u_i \Pi_2(A_i - B_i)}_2^2 + \norm{\sum_{i \in [m]} u_i \Pi_2^{\perp} A_i}_2^2} 
& \ge \max_{u \in \mathbb{S}^{m-1}}  \norm{\sum_{i \in [m]} u_i \Pi_2^{\perp}A_i}_2
= \norm{\Pi_2^{\perp} A } \ge \norm{\Pi_1^{\perp}A} 
\end{align*} 
as required. The first inequality above follows since the column space of $\Pi_2^{\perp} A$ and the column space of $\Pi_2(A-B^*)$ are orthogonal. An identical proof also follows for the Frobenius norm.  
\fi
\end{proof}

\anote{Should we include monotonicity here? Or should it go into the relevant section?}

\subsection{Approximation Algorithms for Operator Norms.}\label{sec:operatornorms}
%\vnote{Feel free to ping me for changes.}
%\vnote{OK thanks. Fixed.}\anote{Looks great! Thanks! Made some minor edits. Also is the hardness assuming P vs NP or Unique Games conjecture?}
\anote{I also changed the references to the approximations from Sec 4 to refer to Corollary~\ref{cor:norms-related} now. Please check in Sec 4 if it reads fine now. }
Here we briefly describe some known positive and negative results for approximating the $q\to p$ operator norm of a matrix (sometimes referred to as the $(\ell_q,\ell_p)$-Grothendieck problem). We will say that a randomized algorithm gives an {\em $\alpha$-factor approximation} for the $q \to p$ operator norm (for some $\alpha \ge 1$) iff for any input matrix $M$ the algorithm outputs with probability at least $(1-n^{-\omega(1)})$ a vector $\widehat{x} \ne 0$ such that $\norm{M\widehat{x}}_p/\norm{\widehat{x}}_q  \ge \tfrac{1}{\alpha} \cdot \norm{M}_{q \to p}$. 
%There are three simple cases ($q = 1, 1 \le p\le \infty$, and $p = \infty, 1 \le q \le \infty$, and $q = p = 2)$, where the norm can be computed in polynomial time. 
%Some notable cases that are known to be intractable are the $\infty \to 1$ and $\infty\to2$ norm (hence also $2\to 1$ by duality). 
The $\infty \to 1$ norm is the well-known Grothendieck's problem~\citep{Gro56} (that is related to the cut-norm of a matrix~\citep{alon2004approximating}  and has a rich history.

% More generally, if $1\le p < q \le \infty$ (non-hypercontractive norms), computing the $q\to p$ norm is NP-hard~\cite{steinberg2005computation} and in fact the problem exhibits a dichotomy in terms of approximation. Specifically, whenever $2\in [p,q]$, the problem admits constant factor approximation algorithms, whereas whenever $2 \notin [p, q]$, the problem is hard to approximate within almost polynomial factors~\cite{bhaskara2011approximating}. 
There is a lot of work on approximation algorithms and inapproximability results for computing these $q \to p$ norms~\citep{nesterov1998semidefinite, alon2004approximating, bhaskara2011approximating, barak2012hypercontractivity, bhattiprolu2018approximating, bhattiprolu2018inapproximability}. 
Regarding approximation algorithms, the works of ~\citep{nesterov1998semidefinite, wolkowicz2012handbook, steinberg2005computation} %~\cite{steinberg2005computation} builds upon Nesterov's theorems and extensions~\cite{nesterov1998semidefinite,wolkowicz2012handbook} and 
provides a $1/(\tfrac{2\sqrt3}{\pi}-\tfrac23)\approx 2.29$ approximation for when $1\le p\le 2 \le q\le \infty$, and for the special case $p=2$ or $q=2$ the factor becomes $\sqrt{\pi/2}\approx 1.25$. %These are achieved by semidefinite programming (SDP) techniques. 
Recently, improved upper and (almost matching) lower bounds were proved for many settings of $q,p$ in~\citep{bhattiprolu2018approximating,bhattiprolu2018inapproximability}. Formally, we have the following guarantee where $\gamma^{q^*}_{q^*}$ is the ${q^*}$th moment of a standard normal random variable.

 \begin{theorem}[\citep{bhattiprolu2018approximating,bhattiprolu2018inapproximability, nesterov1998semidefinite,steinberg2005computation}]\label{cor:norms-related}
 For computing the $\infty \to 2$ norm, there is a randomized polynomial time algorithm that gives a $\sqrt{\pi/2}\approx 1.25$-approximation, and for the $q\to 2$ norm there is a randomized polynomial time algorithm that gives a $1/\gamma_{q^*}$-factor approximation. Furthermore, when the input matrices are positive semidefinite, the integrality gap of the aforementioned SDP is $\pi/2$ for the $\infty \to 1$ norm, and $1/\gamma_{q^*}^2$ for the $q \to q^*$ operator norm respectively. Using a generalization of random hyperplane rounding, this SDP yields approximation algorithms that succeed with high probability for any given instance.%there exists randomized polynomial time algorithms that give a $\pi/2$ approximation for the $\infty \to 1$ norm, and a $1/\gamma_{q^*}^2$ factor approximation for the $q \to q^*$ operator norm respectively. These algorithms are SDP-based randomized algorithms that succeed with high probability for any given instance. 
 \end{theorem}

\section{Worst-case Approximation Guarantees}\label{sec:worstcase}
In this section, we show the approximation algorithm bounding the \emph{Frobenius} norm error in Section~\ref{sec:wc:frob} and the approximation algorithm bounding the \emph{spectral} norm error in Section~\ref{sec:wc:spectral} separately.

\subsection{Approximations in Frobenius Norm Error} \label{sec:wc:frob}

We will aim to obtain a bicriteria approximation for the robust low-rank approximation problem given in \eqref{intro:obj} for the Frobenius norm error. In the rest of the section, we will focus on the the formulation where the objective is to minimize projection error, subject to a specified robustness requirement. It is easy to see that one can switch the role of the objective and constraint here and obtain similar guarantees for minimizing the robustness parameter $\kappa$, subject to a bound on the projection error. In what follows $q \in [2, \infty]$.  
\begin{align}
    \min_{\Pi} &\norm{\Pi^{\perp} A }_F^2= \min_{\Pi} \norm{A}_F^2 - \iprod{AA^\top,\Pi} \label{int:obj}\\
    \text{s.t.\ }& \Pi \text{ is a projection matrix of rank } \le r, \text{ and } \norm{\Pi}_{q \to 2} \le \kappa \label{int:operatornorm}.
\end{align}

We prove the following theorem.
%\anote{Should we write this in terms of the Frobenius norm squared loss, or Frobenius norm loss?}

%\anote{changed $\gamma \in (0,1]$ to $\gamma >0$. Check if ok?}

\begin{theorem}\label{thm:worstcase:frob}
Suppose the data matrix $A \in \R^{n \times m}$ has an (orthogonal) projection $\Pi^{*}$ of rank at most $r$ such that $\norm{\Pi^*}_{q \to 2} \le \kappa$ and the approximation error $OPT:=\norm{(I-\Pi^*) A}_F^2$.
There exists a polynomial time algorithm \anote{11/2: removed w.h.p} such that given any $\gamma>0$, it finds a projection matrix $\widehat{\Pi}$ of rank at most $r$ satisfying
\begin{align} \label{eq:worstcase:frob}
\norm{\widehat{\Pi}}_{q \to 2} &\le {\sqrt{C_G(q)(1+1/\gamma)}} \cdot \kappa , \text{ and } \norm{(I-\widehat{\Pi}) A}_F^2 \le (2+\gamma) OPT,  
\end{align}
where $C_G(q)>0$ is a constant that only depends on $q$ as given in Theorem~\ref{cor:norms-related} (for $q=\infty$ this value is at most $\pi/2$). 
Moreover, for any $\gamma>0$, there exists an algorithm that runs in polynomial time  and  finds an $r'\le r (1+1/\gamma)$-dimensional orthogonal projection $\widehat{\Pi}$ such that
\begin{align} \label{eq:worstcase:bicriteria}
\norm{\widehat{\Pi}}_{q \to 2} &\le {\sqrt{C_G(q)(1+1/\gamma)}}\cdot \kappa, \text{ and } \norm{(I-\widehat{\Pi}) A}_F^2 \le (1+\gamma) OPT.  
\end{align}
% \pnote{be consistent in (17) and (18). One of them use cdot. The other does not.}
% \pnote{$\sqrt{1+\gamma}$ missing in (17)?}\anote{fixed both i think}
\end{theorem}

The theorem above will be established by proving a statement about the %In fact, our algorithm and guarantee will also work for the 
more general problem of finding a low-rank projection under {\em any} monotone norm constraint that can be approximately certified. While the $q \to 2$ norm is not monotone as discussed in Section~\ref{sec:prelims}, we will show that applying the more general guarantee on an appropriate monotone norm helps prove Theorem~\ref{thm:worstcase:frob} above. 
Let $\gnorm{\cdot}$ be a monotone matrix norm. Consider the following generalization of problem \eqref{eq:intro:PCA} that given a data matrix $A \in \R^{n \times m}$, and a parameter $\kappa \ge 1$, finds a projection
\begin{equation}\label{prob:gen:fr}
    \min_{\Pi}  \norm{A}_F^2 - \iprod{AA^\top,\Pi} 
    ~\text{s.t.}~ \Pi \text{ is a projection matrix of rank}\le r, \text{ and } \gnorm{\Pi} \le \kappa.
\end{equation}

\anote{11/2: removing randomization. Also changed name to certifiable instead of separable.}
\begin{definition}\label{assump:norm:optimization}[$\alpha$-approximately certifiable matrix norm]
A matrix norm $\gnorm{\cdot}$ over $\R^{n \times m}$ matrices is $\alpha$-approximately certifiable for some $\alpha \ge 1$ iff there exists an algorithm that runs in time $\poly(n,m)$, and when given a PSD matrix $B \in \R^{n \times n}$ and a parameter $\kappa$ as input 
%runs in time $\poly(n,m)$ and 
will either certify that $\gnorm{B} \le \alpha \kappa$, or 
%with high probability 
finds a $Z \in \R^{n \times n}$ such that (1) $\iprod{B,Z} > \kappa$, and (2) $\iprod{M,Z} \le \kappa$ for all $M$ s.t. $\gnorm{M} \le \kappa$.      
\end{definition}
%\xnote{11/4: Is it $\mathbb{R}^{n \times m}$ in the 1st line of the definition? We never use $m$ later.}
%Note that in the above definition, the only potential randomness is from the potential random choices of the algorithm. 
As we will see later the operator norms that we will consider (e.g., $q \to 2$ norm and the $q \to q^*$ norm for $q \ge 2$) will be $O(1)$-approximately certifiable. %This will allow us to construct an appropriate separation oracle for using the Ellipsoid algorithm.   

The following general theorem gives an $O(1)$ bicriteria approximation for the problem assuming the monotone matrix norm $\gnorm{\cdot}$ is approximately certifiable. %We will denote by $\gnorm{\cdot}_*$ the dual norm of $\gnorm{\cdot}$. 

\begin{theorem}\label{prop:gen:fr}
Let $\gnorm{\cdot}$ be any matrix norm that is {\em monotone} and $\alpha$-approximately certifiable for some $\alpha \ge 1$.  Suppose we are given as input a data matrix $A \in \R^{n \times m}$ that has an (orthogonal) projection $\Pi^{*}$ of rank at most $r$ such that $\gnorm{\Pi^*} \le \kappa$ and the approximation error $OPT:=\norm{(I-\Pi^*) A}_F^2$.
There exists a polynomial time algorithm such that given every $\gamma \in (0,1)$, it finds an orthogonal projection matrix $\widehat{\Pi}$ of rank at most $r$ satisfying
\begin{align} \label{eq:gen:frob}
\gnorm{\widehat{\Pi}} &\le \alpha \big(1+\tfrac{1}{\gamma}\big) \cdot \kappa , \text{ and } \norm{(I-\widehat{\Pi}) A}_F^2 \le \Big(2+\gamma \Big) OPT.  
\end{align}
\pnote{This is confusing as earlier theorem is stated as $\alpha \sqrt{1+1/\gamma}$. Ideally, the earlier theorem should follows as a special case. Otherwise discuss what's happening: $q \to 2$ vs. $q \to q^*$?}\anote{Addressed this in the proof of Theorem~\ref{thm:worstcase:frob}.}

Moreover, for any $\gamma \in (0,1)$, there exists a polynomial time algorithm that finds an $r'\le (1+1/\gamma)r$-dimensional orthogonal  projection $\widehat{\Pi}$ such that
\begin{align} \label{eq:gen:bicriteria}
\gnorm{\widehat{\Pi}} &\le \alpha\big(1+\tfrac{1}{\gamma}\big)\cdot \kappa, \text{ and } \norm{(I-\widehat{\Pi}) A}_F^2 \le (1+\gamma) \cdot OPT.  
\end{align}
\end{theorem}

% We see that for any $\gamma>0$, \eqref{eq:gen:frob} gives a $(2+\gamma)$-factor approximation while violating the operator norm constraint by $\alpha(1+\gamma)/\gamma$ factor. %(by setting $1/\delta=1+\gamma$). 
% \vnote{above, $C_G$ was removed, so $\alpha$ is better?}
%On the other hand, \eqref{eq:gen:bicriteria} gives a $1+\gamma$ factor approximation with the same operator norm guarantee as before, but violates the rank constraint by a $1+1/\gamma$ factor. 
%We remark that the only randomization in the algorithm is in the construction of the separation oracle for the matrix norm constraint $\gnorm{\Pi} \le \kappa$. In particular the algorithm from Theorem~\ref{prop:gen:fr} has a {\em Las Vegas guarantee} i.e., the algorithm is always correct, and the running time is polynomial with high probability (in fact with exponentially small failure probability), and hence in expectation.  
%\anote{11/27: added above line. Should I change the ordering of approximation quality and running time so that it becomes clearer that w.h.p. is only for polynomial time?}

%\anote{Need to rewrite proof so that we prove the above theorem, and write down our required result as an implication.}

We consider the following mathematical programming relaxation for the problem. In the alternate formulation where we minimize the robustness parameter subject to an upper bound on projection error, the roles of \eqref{gen:sdp:obj} and \eqref{gen:sdp:norm} below is switched. 
\begin{align}
    \min_X &\norm{A}_F^2 - \iprod{AA^\top, X} \label{gen:sdp:obj}\\
    \text{s.t.}& \tr(X) \le r \label{gen:sdp:tr}\\
                & 0 \preceq X \preceq I \label{gen:sdp:spectralnorm}\\
    & \gnorm{X}\le \kappa \label{gen:sdp:norm}
\end{align}

First we observe that this is a valid convex relaxation to the problem. In fact any feasible projection matrix $\Pi$ of rank at most $r$ for \eqref{prob:gen:fr} is a feasible solution to the above program \eqref{gen:sdp:obj}-\eqref{gen:sdp:norm} with the same value. The intended solution here is just $X=\Pi$. All the eigenvalues of $\Pi$ are $0$ or $1$, since $\Pi$ is a projection matrix; hence \eqref{gen:sdp:tr}, \eqref{gen:sdp:spectralnorm} are satisfied. Moreover \eqref{gen:sdp:norm} is satisfied just because of the same constraint as in \eqref{prob:gen:fr}. Finally, the objective value is preserved since 
$$\norm{A}_F^2 - \iprod{AA^\top, \Pi} = \norm{A}_F^2 - \tr(A A^\top \Pi) = \norm{A}_F^2 - \tr(\Pi A ( \Pi A)^\top) = \norm{A}_F^2 - \norm{\Pi A}_F^2.$$
%({\color{red}Vaggos: should it be $\tr(\Pi A ( \Pi A)^\top)$ above?})
In the above program, the objective \eqref{gen:sdp:obj} and constraints \eqref{gen:sdp:tr}-\eqref{gen:sdp:spectralnorm} define a semi-definite program (SDP). Moreover, for \eqref{gen:sdp:norm}, we see that for any $\lambda \in [0,1]$, by triangle inequality $\gnorm{\lambda X_1 + (1-\lambda) X_2 } \le \lambda \gnorm{X_1} + (1-\lambda) \gnorm{X_2}$. Hence the set of all $X$ that satisfies constraints \eqref{gen:sdp:tr} - \eqref{gen:sdp:norm} is convex. In general, constraint \eqref{gen:sdp:norm} may be NP-hard 
%({\color{red}vaggos: what does ``may not'' mean here?}) 
to verify for a given PSD matrix $X$. % (this is true even when $\gnorm{\cdot}$ is an operator norm of the form $\norm{X}_{q \to p}$ for most values of $q,p$). 
However, we can use the fact that $\gnorm{\cdot}$ is approximately certifiable to get a approximately feasible solution to the program in polynomial time, using the Ellipsoid method. 

The following lemma shows that by truncating a solution of the program \eqref{gen:sdp:obj}-\eqref{gen:sdp:norm} to just the terms corresponding to the large eigenvalues, we retain much of the objective. 

\begin{lemma}\label{lem:lowrank}
Let $\eps, \delta>0$, and $M \succeq 0$. Suppose $X$ satisfies the SDP constraints \eqref{gen:sdp:spectralnorm} and $\iprod{M,X} \ge (1-\eps) \tr(M)$. Suppose $P^{1-\delta}_X$ is the projection operator onto the subspace spanned by eigenvectors of $X$ with eigenvalues at least $(1-\delta)$. Then we have
\begin{equation}
    \iprod{I-P^{1-\delta}_X, M} \le \frac{\eps}{\delta} \cdot \tr(M).
\end{equation}
\end{lemma}
\begin{proof}
\anote{11/2: shortened/ simplified the proof. }
We can assume without loss of generality that $\tr(M)=1$, since $M$ can be scaled accordingly. 
Let $X=\sum_{i=1}^n \lambda_i v_i v_i^\top$ be the eigendecomposition of $X$ (note that $\lambda_i \ge 0$ since $X$ is p.s.d.), and let $S=\sset{i: \lambda_i \ge 1-\delta}$. We have 
\begin{align*}
     (1-\eps) \tr(M) &\le \iprod{M,X}= \sum_i \lambda_i v_i^\top M v_i \\
      \tr(M) &= \sum_i v_i^\top M v_i, \text{ since } \sset{v_i : i \in [n]} \text{ is an orthonormal basis}.
    %   \text{ Hence } \sum_i (\lambda_i - 1 + \eps) v_i^\top M v_i &\ge 0~~ \implies \sum_{ i \in S} \eps v_i^\top M v_i +  \sum_{i \notin S}( \tau +\eps-1) v_i ^\top M v_i \ge 0\\ 
    %   \sum_{ i \in S} \eps v_i^\top M v_i &\ge  \sum_{i \notin S}( \delta -\eps) v_i^\top M v_i \\
    %   \text{ Hence, } \iprod{I-{P^{1-\delta}_X},M} &\le   \frac{\eps}{\delta-\eps} \cdot \iprod{P^{1-\delta}_X,M}. 
\end{align*}
By subtracting the two inequalities, we get
\begin{align*}
    \sum_{i=1}^n (1-\lambda_i) v_i^\top M v_i &\le \eps \cdot \tr(M) = \eps.\\
    \sum_{i\notin S} \delta v_i^\top M v_i & \le \sum_{i \notin S} (1-\lambda_i) v_i^\top M v_i + \sum_{i \in S} (1-\lambda_i) v_i^\top M v_i \le \eps 
    \end{align*}
from definition of  $S, M \succeq 0$, and  \eqref{gen:sdp:spectralnorm}. Hence 
\begin{align*}
    \iprod{I-P^{1-\delta}_X, M} &= \sum_{i\notin S} v_i^\top M v_i \le \frac{\eps}{\delta} = \frac{\eps}{\delta} \cdot  \tr(M),  \text{ as required.}
\end{align*}
% Substituting the above inequality in $\iprod{P^{1-\delta}_X,M}+\iprod{I-P^{1-\delta}_X,M}=\tr(M) = 1$, we get
% $$ \iprod{I-{P^{1-\delta}_X},M} \le \frac{1}{1+\tfrac{\delta - \eps}{\eps}} = \frac{\eps}{\delta}.$$
\end{proof}

%\anote{11/27: Fixed the proof with respect to high probability.}
\begin{proofof}{Theorem~\ref{prop:gen:fr}}
We can scale the matrix $A$ appropriately so that we can assume $\norm{A}_F=1$ without loss of generality. 
Let $OPT=\eps $ for some $\eps \in [0,1]$. %Without loss of generality, in the rest of the proof we will assume that $\|A\|_F=1$. %({\color{red}vaggos: should we state case for when $OPT=0$?}).
We will use the Ellipsoid algorithm to approximately solve the relaxation in \eqref{gen:sdp:obj}-\eqref{gen:sdp:norm}. As we have explained before, the feasible set is convex. We now show how to design an approximate hyperplane separation oracle for \eqref{gen:sdp:norm}; the rest of the constraints just correspond to a simple SDP. Since $\gnorm{\cdot}$ is $\alpha$-approximately certifiable, we have a polynomial time algorithm that given a matrix $\widehat{X} \succeq 0$ ,
either certifies that $\gnorm{\widehat{X}} \le \alpha \kappa$ (e.g., when the SDP value is at most $\alpha \kappa$), and otherwise produces a separating hyperplane of the form $\iprod{Z,X} \le \kappa$ that is not satisfied by $\widehat{X}$. 
%w.h.p. accepts when $\gnorm{\widehat{X}} \le \kappa$, but when $\gnorm{\widehat{X}} > \alpha \kappa$ it produces a separating hyperplane of the form $\iprod{Z,X} \le \kappa$ that is not satisfied by $\widehat{X}$. 
%Let $T=\poly(n,m)$ be the number of iterations taken by the Ellipsoid algorithm to produce a solution of value at most $OPT$ (up to exponentially small additive error), assuming access to a separation oracle. Hence from a union bound over all the $T$ iterations, we can use the above randomized \xnote{Nov 2nd: Do we need "randomized" here?} polynomial time algorithm for separating \eqref{gen:sdp:norm}, and 
We run the Ellipsoid algorithm to find a solution $\widehat{X}$ that satisfies $\gnorm{\widehat{X}} \le \alpha \kappa$, and has objective value that is arbitrarily close to $OPT$. %This algorithm runs in polynomial time with high probability \xnote{Nov 2nd: Do we need "w.h.p.", the remark in the bracket  and the last sentence?} (with an exponentially small probability it may not terminate when the separation oracle does not terminate in polynomial time). For the rest of the analysis, we condition on the event that the algorithm terminates in polynomial time. %\anote{rand: added w.h.p.} 

Set $\delta:=1/(1+\gamma)$. Let $\widehat{X}=\sum_i \lambda_i v_i v_i^\top$ and let $S=\sset{i : \lambda_i \ge 1-\delta}$. Define for each $i \in [S]$, $\alpha_i:= \iprod{v_i v_i^\top, M}$, where $M=AA^\top$. We sort the elements of $S$ based on $\sset{\alpha_i}$, and pick greedily the first $\min\sset{r, |S|}$ of them to form $T \subseteq S$. Our projection matrix will be $\Pi_T =\sum_{i \in T} v_i v_i^\top$. %\pnote{rest of sentence not needed.} for some appropriately chosen subset $T \subseteq S$. 

We first argue that the operator norm constraint is approximately satisfied. By the monotonicity of the $\gnorm{\cdot}$ we have
\begin{align}
    \gnorm{\Pi_T} &= \gnorm{\sum_{i \in T} v_i v_i^\top} \le   
     \frac{1}{1-\delta}\gnorm{\sum_{\substack{i: \lambda_i > 1-\delta}} \lambda_i v_i v_i^\top } \le  \frac{\gnorm{\widehat{X}}}{1-\delta}     \le \frac{\alpha \kappa}{1-\delta} = \frac{\alpha(1+\gamma) }{\gamma} \cdot \kappa \label{eq:operatornorm:argument}
\end{align}

\noindent Also, from Lemma~\ref{lem:lowrank}, we have
\begin{align}
\sum_{i \notin S} \iprod{v_i v_i^\top, M}&=\iprod{I-P_{\widehat{X}}^{1-\delta}, M} \le \frac{\eps}{\delta}, \text{ and } \nonumber\\
\sum_{i \in S} \iprod{(1-\lambda_i) v_i v_i^\top, M} &\le \sum_i \iprod{(1-\lambda_i) v_i v_i^\top, M} \le 1 - \iprod{X,M} \le \eps. \nonumber\\ 
\text{Hence, }  \sum_{i \in S} \lambda_i \alpha_i &=\sum_{i \in S} \lambda_i \iprod{v_i v_i^\top, M}  \ge 1- \eps\Big(1+\frac{1}{\delta}\Big)= 1-(2+\gamma)\eps, \nonumber
\end{align}
for our choice of $\delta=1/(1+\gamma)$.
By our greedy choice of $T$, we have $\sum_{i \in T} \alpha_i \ge \sum_{i \in S} \lambda_i \alpha_i$, as $\sum_{i \in S} \lambda_i \le  \min\sset{\tr(X) ,|S|}=|T|$, with each $\lambda_i \in [0,1]$. Thus $\norm{\Pi_T^{\perp} A}_F^2 \le (2+\gamma)\eps$. 

% {\bf Assume (*):}{\em We have a randomized (dependent) rounding scheme such that we pick $T \subseteq S$ with $|T| \le r$  such that $\Pr[i \in T] \ge \lambda_i$.}\\
% If we have (*), then 
% $$\E_T \Big[ \iprod{P_T, M} \Big] =  \sum_{i \in S} \Pr[i \in T] \iprod{v_i v_i^\top , M}  \ge \sum_{i \in S} \lambda_i \iprod{v_i v_i^\top , M} \ge 1-\eps\Big(1+\frac{1}{\delta}\Big),$$
% as required.

The guarantee in \eqref{eq:worstcase:bicriteria} is obtained by returning the projection $\Pi_S=\sum_{i\in S} v_i v_i^\top$. Observe that $|S| \le r / (1-\delta)$ from \eqref{gen:sdp:tr}. The operator norm bounds follows using the same argument as \eqref{eq:operatornorm:argument} with $T=S$. Moreover, the objective value follows directly from Lemma~\ref{lem:lowrank}. 
\end{proofof}

\paragraph{Guarantees for the $q \to 2$ norm (Proof of Theorem~\ref{thm:worstcase:frob}).}

%Theorem~\ref{prop:gen:fr} implies the main theorem of the section (Theorem~\ref{thm:worstcase:frob}), which shows that we can find an approximately optimal low-rank robust projection as measured in $\norm{\cdot}_{q \to 2}$. \pnote{The above line does not make sense.}
% We see that for any $\gamma>0$, \eqref{eq:worstcase} gives a $(2+\gamma)$-factor approximation while violating the operator norm constraint by $\sqrt{C_G(1+\gamma)/\gamma}$ (by setting By setting $1/\delta=1+\gamma$). On the other hand, \eqref{eq:worstcase:bicriteria} gives a $1+\gamma$ factor approximation with the same operator norm guarantee as before, but violates the rank constraint by a $1+1/\gamma$ factor. 
Our goal will be to apply Theorem~\ref{prop:gen:fr} to obtain our required guarantee. However the $q \to 2$ operator norm is not monotone when $q>2$; see Claim~\ref{clm:counterexample_infinity_2} for a counter-example. 
Our crucial insight is that we can instead use the $\norm{\cdot}_{q \to q^*}$ norm which we show %is indeed monotone! 
%The following lemma shows that the norm in constraint \eqref{sdp:norm} 
indeed satisfies the monotonicity property (Definition~\ref{def:monotonicity}), so that we can apply Theorem~\ref{prop:gen:fr}.

\begin{claim}[Monotonicity of $q \to q^*$ operator norm]\label{lem:monotone}
For any $q \ge 1$, the operator norm  $\norm{\cdot}_{q \to q^*}$ is monotone.% i.e., $\norm{A+B}_{q \to q^*} \ge \norm{A}_{q \to q^*}$ for any $A, B \succeq 0$.
\end{claim}
\begin{proof}
Let $B \in \R^{n \times n}$. It suffices to prove for any $B \succeq 0, v \in \R^n$, $\norm{B+vv^\top}_{q \to q^*} \ge \norm{B}_{q \to q^*}$. %Let $B:=A+vv^\top$.
\begin{align*}
    \norm{B}_{q \to q^*}&= \max_{\norm{y}_q \le 1, \norm{z}_q \le 1} z^\top B y = \max_{\norm{y}_q \le 1, \norm{z}_q \le 1} \iprod{B^{1/2} z , B^{1/2} y} \\
    &= \max_{\norm{y}_q \le 1} \iprod{B^{1/2}y, B^{1/2} y}= \max_{\norm{y}_q \le 1} y^\top B y.
\end{align*}
In other words, the quadratic form is maximized by $y=z$. Moreover for every $y$,  
$$y^{T} (B+vv^\top) y = y^{T} B y + \iprod{y,v}^2 \ge y^\top B y.$$
Hence, $\norm{B+vv^\top}_{q \to q^*} \ge \norm{B}_{q \to q^*}$. 
    
\end{proof}

This gives the following mathematical programming relaxation for the problem, where the robustness constraint is captured by the $q \to q^*$ operator norm. 

\begin{align}
    \min_X &\norm{A}_F^2 - \iprod{AA^\top, X} \label{sdp:obj}\\
    \text{s.t. }& \tr(X) \le r %\label{sdp:tr} 
        %        & 
         \text{ and }       0 \preceq X \preceq I \label{sdp:spectralnorm}\\
    & \norm{X}_{q \to q^*} \le \kappa^2 \label{sdp:norm}
\end{align}

In the above program, the objective \eqref{sdp:obj} and constraints \eqref{sdp:spectralnorm} define a semi-definite program (SDP). However, the $q \to q^*$ operator norm is NP-hard to compute. Instead we will use the standard SDP relaxation~\citep{nesterov1998semidefinite, steinberg2005computation} for certifying the $q \to q^*$ norm of any PSD matrix $X$, where $C_G=C_G(q)>0$ is the constant given by Theorem~\ref{cor:norms-related}: 
\begin{equation} \label{eq:relax:constraint:approx}
\norm{X}_{q \to q^*} \le     \max_{\substack{Y \in \R^{n \times n}: Y \succeq 0\\ \sum_i Y_{ii}^{q/2} \le 1}} \iprod{X,Y} \le  C_G \norm{X}_{q \to q^*},
\end{equation}
This shows that $q \to q^*$ norm is $\alpha=C_G$ approximately certifiable. At this point, we have all the ingredients to prove Theorem~\ref{thm:worstcase:frob} by using Theorem~\ref{prop:gen:fr} as a black-box. However, let us first consider the convex relaxation(s) suggested by this certificate, with a view towards designing a more efficient algorithm based on convex relaxations.

\paragraph{Convex relaxations for the Frobenius norm objective.}
The program \eqref{sdp:obj} along with the approximate certificate for $q \to q^*$ norm from \eqref{eq:relax:constraint:approx} leads to the following two tractable relaxations for our problem \eqref{int:obj} that are equivalent.  
%%%%%%%%%%%% Needs to be rearranged %%%%%%%%%%%%%%%%%%%%%%%%
\CPone involves a universal quantifier over $Y \in \calQ$. Fortunately, there exists an efficient hyperplane separation oracle for the constraint \eqref{eq:convert:constr} (which by itself is an SDP relaxation). This can be used with the Ellipsoid algorithm to solve the above relaxation in polynomial time. We now instead give an equivalent convex relaxation for \eqref{eq:relax:redef} that is much more efficient to solve. The main idea is to use Lagrangian duality to convert the universal quantifier in \eqref{eq:convert:constr} into an existential quantifier. In what follows for a vector $d \in \R^n$, we will use $\diag(d)$ to denote the diagonal matrix in $\R^{n \times n}$ defined by $d$. \anote{Move the last line to prelims?}

\begin{figure}[h]
\fbox{\begin{minipage}{.48\linewidth}
\CPone:
\begin{align}
    \min_X &\norm{A}_F^2 - \iprod{AA^\top, X} \label{eq:relax:redef}\\
    \text{s.t. }& \tr(X) \le r, ~~ 0 \preceq X \preceq I \label{eq:relax:stdconstr}\\
    & \max_{\substack{Y \in \calQ}}\iprod{X,Y}\le C_G \kappa^2 , \nonumber \text{ where}\\
    \calQ=& \sset{ Y \in \R^{n \times n}:~ Y \succeq 0, \sum_{i = 1}^n Y_{ii}^{q/2} \le 1 } \label{eq:convert:constr}%\label{eq:def:Z}  
\end{align}
\end{minipage}}%
\fbox{\begin{minipage}{0.48 \linewidth}
\CPtwo:
\begin{align}
    \min_{X \in \R^{n \times n}, d \in \R_{\ge 0}^n} &\norm{A}_F^2 - \iprod{AA^\top, X} \label{eq:newsdp:fr}\\
    \text{s.t. }& \tr(X) \le r, ~~ 0 \preceq X \preceq I \label{eq:newsdp:stdconstr}\\
                & X \preceq \diag(d) \\
 \norm{d}_{q/(q-2)} := &\Big(\sum_{i=1}^n d_i^{q/(q-2)}\Big)^{(q-2)/q} \le C_G \kappa^{2}. \label{eq:newsdp:constr}
\end{align}
\end{minipage}}
\caption{Two tractable mathematical relaxations \CPone  and \CPtwo  for problem \eqref{int:obj} with Frobenius norm objective.}\label{fig:relaxations:fr}
\end{figure}

The equivalence of \CPone and \CPtwo follows immediately from the following lemma that uses Lagrangian duality. 
\begin{lemma}\label{lem:convexdual}
Consider the following two programs defined for $q \ge 2$ and input $C \in \R^{n \times n}$ with $C_{ii} \ge 0~ \forall i \in [n]$. 
\begin{align}
\textbf{primal:}~  \text{val}_P := \max_{Y \in \R^{n \times n}} \iprod{C,Y} \text{ s.t. } \sum_{i=1}^n Y_{ii}^{q/2} \le 1, ~~ Y \succeq 0.  \label{eq:convexprimal}\\ 
\textbf{dual:}~  \text{val}_D  :=\min_{d \in \R^{n}} \norm{d}_{q/(q-2)}= \Big(\sum_{i=1}^n d_i^{q/(q-2)} \Big)^{(q-2)/q} \text{ s.t. }  ~~ \diag(d) \succeq C, ~~ d \ge 0.  \label{eq:convexdual}
\end{align}
For any feasible solution $Y$ of the primal, and any feasible solution $d$ of the dual, we have $\iprod{C,Y} \le \norm{d}_{q/(q-2)}$ i.e., weak duality holds. Moreover, the optimum values of the primal and the dual relaxations are equal $\text{val}_P=\text{val}_D$ (strong duality). 
\end{lemma}

We remark that in the special case when $q=\infty$, the last constraint \eqref{eq:newsdp:constr} becomes the simple linear constraint $\sum_i d_i \le c_G \kappa^2$. 
The following claim shows that the above convex programs \CPone and \CPtwo are valid relaxations for \eqref{int:obj} and can be solved in polynomial time. 

\begin{claim}\label{claim:feasibility}
Any feasible projection matrix $\Pi$ of rank $r$ satisfying \eqref{int:operatornorm} forms feasible solutions to \CPone, \CPtwo (for an appropriate feasible $d$) and \eqref{sdp:obj} with the same objective value as \eqref{int:obj}. Moreover the relaxations \CPone and \CPtwo can be solved in polynomial time to arbitrary accuracy.  
\end{claim}
\begin{proof}
First we argue about the feasibility of \CPone and \eqref{sdp:obj}. The intended SDP solution here is just $X=\Pi$. All the eigenvalues of $\Pi$ are $0$ or $1$, since $\Pi$ is a projection matrix; hence \eqref{sdp:spectralnorm} and the corresponding constraints in the programs \eqref{eq:relax:stdconstr} and \eqref{eq:newsdp:stdconstr} are satisfied. To verify \eqref{sdp:norm}, note that from Claim~\ref{lem:extra}, $\norm{\Pi}_{q \to q^*}=\norm{\Pi}_{q \to 2}^2$. Moreover from \eqref{eq:relax:constraint:approx}, we also have that \eqref{eq:convert:constr} is satisfied by $X$.  Hence \CPone (and \eqref{sdp:obj}) are feasible. Moreover $X$ is also feasible for \CPtwo for an appropriate choice of $d$ since the constraints \eqref{eq:convert:constr} and \eqref{eq:newsdp:constr} are equivalent from Lemma~\ref{lem:convexdual}.

Finally, the objective value is preserved since 
$$\norm{A}_F^2 - \iprod{AA^\top, \Pi} = \norm{A}_F^2 - \tr(A A^\top \Pi) = \norm{A}_F^2 - \tr(\Pi A (A \Pi)^\top) = \norm{A}_F^2 - \norm{\Pi A}_F^2.$$

It is easy to check that both relaxations \CPone and \CPtwo are convex. We now argue about the computational tractability of \CPtwo.  
The relaxation is a semi-definite program (SDP) with an extra constraint \eqref{eq:newsdp:constr}. The constraint \eqref{eq:newsdp:constr} is convex when $q \ge 2$. Let $q'=q/(q-2)$; the dual norm for $\ell_{q'}$ is $\ell_{q/2}$. 
Moreover there is a simple separation oracle for this constraint since by duality
\begin{align*}
\norm{d}_{q/(q-2)} = \max_{y \in \R^n: \norm{y}_{q/2} \le 1} \iprod{y,d} = \Big\langle \frac{d^*}{\norm{d^*}_{q/2}},d \Big \rangle, \text{ where } d^*_i= \text{sign}(d_i) |d(i)|^{2/(q-2)}~~ \forall i \in [n].
\end{align*}
Hence by using the Ellipsoid algorithm, this problem can be solved in polynomial time. %\footnote{In fact Projected Gradient Descent Algorithm can also be used here; see \cite{sra2012fast}.} 
Finally \CPone can also be solved in polynomial time since \eqref{eq:convert:constr} is itself a semi-definite program: this can be solved in polynomial time (by Ellipsoid method for example) and can in turn be used as a hyperplane separation oracle for constraint \eqref{eq:convert:constr} and the outer relaxation. We remark that solving \CPone is less computationally efficient compared to \CPtwo due to the repeated use of the Ellipsoid method (to certify \eqref{eq:convert:constr}).    
\end{proof}

%%%%%%%%%%%%%%%%%%%%%%%%%%%%%%%%%%%%%%%%%%%%%%%%%%%%%%%%%%%%%%%%%%%

%The following lemmas shows that the above programs are valid relaxation to the problem. They also show that 

% \begin{lemma} \label{lem:frob:normconstr}
% Any feasible projection matrix $\Pi$ of rank $r$ satisfying \eqref{prob:gen:fr} is a feasible SDP solution with the same value. 
% \end{lemma}
% \begin{proof}
% The intended SDP solution here is just $X=\Pi$. All the eigenvalues of $\Pi$ are $0$ or $1$, since $\Pi$ is a projection matrix; hence \eqref{sdp:tr}, \eqref{sdp:spectralnorm} are satisfied. To verify \eqref{sdp:norm}, note that from Lemma~\ref{lem:extra}, $\norm{\Pi}_{q \to q^*}=\norm{\Pi}_{q \to 2}^2$ as required.

% Finally, the objective value is preserved since 
% $$\norm{A}_F^2 - \iprod{AA^\top, \Pi} = \norm{A}_F^2 - \tr(A A^\top \Pi) = \norm{A}_F^2 - \tr(\Pi A (A \Pi)^\top) = \norm{A}_F^2 - \norm{\Pi A}_F^2.$$
% \end{proof}

\anote{Not sure how modular this is. Can you check to ensure this is OK?}

\begin{proofof}{Theorem~\ref{thm:worstcase:frob}}
We will apply Theorem~\ref{prop:gen:fr} with $\gnorm{\cdot}:= \norm{\cdot}_{q \to q^*}$.  
From Lemma~\ref{lem:monotone}, we have that the $\norm{\cdot}_{q \to q^*}$ is monotone. 
We now show that $\norm{\cdot}_{q \to q^*}$ is $O(1)$-approximately certifiable. For this we use the constraint \eqref{eq:convert:constr} of \CPone or equivalently \eqref{eq:newsdp:constr} of \CPtwo. As shown in Claim~\ref{claim:feasibility} both of these constraints can be certified efficiently. By \eqref{eq:relax:constraint:approx}, this in turn shows that $\norm{\cdot}_{q \to q^*}$ is $\alpha=C_G(q) =O_q(1)$-approximately certifiable (in particular, $\alpha=\pi/2$ for $\infty \to 1$). Hence, applying Theorem~\ref{prop:gen:fr} we see that the algorithm outputs a projection matrix $\widehat{\Pi}$ with $\norm{\widehat{\Pi}}_{q \to 2}^2=\norm{\widehat{\Pi}}_{q \to q^*} \le C_G(q) (1+\tfrac{1}{\gamma})\kappa^2$ which obtains an objective value of $(2+\gamma) OPT$. This completes the proof of Theorem~\ref{thm:worstcase:frob}.   

%Claim~\ref{claim:feasibility} shows the feasibility and polynomial time solvability of the convex programs \CPone and \CPtwo.
We remark that applying Theorem~\ref{prop:gen:fr} as a black-box involves solving an SDP to $\alpha$-approximately certify the $\norm{\cdot}_{q \to q^*}$ -- this corresponds to solving the convex relaxation \CPone. However, using the equivalence between \CPone and \CPtwo (from Lemma~\ref{lem:convexdual}), we can also apply the same rounding algorithm from Theorem~\ref{prop:gen:fr} to the solution $X$ obtained from \CPtwo to obtain the same guarantees.    
\end{proofof}

We now complete the proof of Lemma~\ref{lem:convexdual} which establishes the equivalence of \CPone and \CPtwo. 

\begin{proofof}{Lemma~\ref{lem:convexdual}}
First we observe that the primal \eqref{eq:convexprimal} is a conic program (PSD cone) with an additional convex constraint. Consider the following Lagrangian dual of the conic program
\begin{align}\label{eq:lagrangian}
    \min_{\substack{\lambda \ge 0,U \succeq 0}} L(\lambda, U), &\text{ where } L(\lambda,U):= \lambda + \max_{Y \in \R^{n \times n}} \iprod{C,Y} - \lambda \sum_{i=1}^n Y_{ii}^{q/2} + \iprod{U,Y}.\\ %=\lambda + \max_{Y \in \R^{n \times n}} \iprod{C+U,Y} - \lambda \sum_{i=1}^n Y_{ii}^{q/2}.
    \text{Moreover by duality},&~~     \text{val}_P=\min_{\substack{\lambda \ge 0,U \succeq 0}} L(\lambda, U).  \label{eq:dualityclaim}
\end{align}
The second line follows from strong duality for the Lagrangian dual --  it is easy to see that Slater's condition holds for the primal (there exists a primal feasible solution $Y$ where all the constraints hold strictly).  
Note that $U \succeq 0$ (since the PSD cone is its own self dual). Note that for the optimal $U$,  we have $C_{ij}+U_{ij}=0$ for all $i \ne j \in [n]$ (otherwise $Y$ can be chosen to make the objective go to $\infty$). By introducing the variable $d \in \R^n$ with $d_i=C_{ii}+U_{ii}$ for all $i \in [n]$, we can rewrite \eqref{eq:lagrangian} as
\begin{align}\label{eq:lagrangian2}
    \min_{\substack{\lambda \ge 0, d\in \R^n \\ \diag(d) \succeq C}} L(\lambda, d), \text{ where } L(\lambda,d)= \lambda + \max_{Y \in \R^{n \times n}} \sum_{i=1}^n d_i Y_{ii} - \lambda Y_{ii}^{q/2}. %=\lambda + \max_{Y \in \R^{n \times n}} \iprod{C+U,Y} - \lambda \sum_{i=1}^n Y_{ii}^{q/2}. 
\end{align}
Note that the inner maximization objective is concave in the variables $\sset{Y_{ii}}$; hence its maximum is obtained where the gradient vanishes i.e., 
\begin{align*}
    \forall i \in [n],~ d_i &= \lambda \Big(\frac{q}{2} \Big) Y_{ii}^{q/2-1} ~~\implies~ Y_{ii}= \Big(\frac{2 d_i}{q}\Big)^{q/(q-2)}.  
\end{align*}
Further $d_i \ge 0$ for all $i \in [n]$. 
Substituting in \eqref{eq:lagrangian2} and simplifying we get that \eqref{eq:lagrangian2} is equivalent to
\begin{align}\label{eq:lagrangian3}
    \min_{\substack{\lambda \ge 0, d\in \R_{\ge 0}^n \\ \diag(d) \succeq C}} L(\lambda, d), \text{ where } L(\lambda,d)= \lambda + \Big(\frac{q-2}{q}\Big) \Big(\frac{2}{q\lambda } \Big)^{2/(q-2)} \sum_{i=1}^n d_i^{q/(q-2)}. 
\end{align}
Again $L(\lambda, d)$ is convex is $\lambda$; hence its minimum is attained at the critical point (that is strictly positive) $$\lambda = \frac{2}{q} \Big( \sum_{i=1}^n d_i^{q/(q-2)} \Big)^{(q-2)/q}.$$
Substituting this value of $\lambda$ in \eqref{eq:lagrangian3} proves that \eqref{eq:lagrangian} is equivalent to the claimed dual formulation \eqref{eq:convexdual}. Finally from Lagrangian duality and \eqref{eq:dualityclaim} the lemma follows. 
\end{proofof}

\subsection{Approximations in the Spectral Norm} \label{sec:wc:spectral}
%\anote{Should we again make this about monotone norms? And then specialize it? I did not do it this way to save space.}

We now show how techniques similar to those in Section~\ref{sec:wc:spectral} can also be extended to robust low-rank approximations, when the error is measured in spectral norm as opposed to the Frobenius norm. In this section we will use $\norm{A}$ to denote the spectral norm of matrix $A$.   
For convenience of exposition, we will measure the projection error in \eqref{intro:obj} using the spectral norm $\norm{A-\Pi A}$ as opposed to the squared spectral norm. 

\begin{theorem}\label{thm:worstcase:spectral}
Suppose the data matrix $A \in \R^{n \times m}$ has an (orthogonal) projection $\Pi^{*}$ of rank at most $r$ such that $\norm{\Pi^*}_{q \to 2} \le \kappa$ and the approximation error $OPT:=\spnorm{(I-\Pi^*) A}$.
There exists a polynomial time algorithm such that given any $\gamma \in (0,1)$, it finds an (orthogonal) projection matrix $\widehat{\Pi}$ of rank at most $r$ satisfying
\begin{align} \label{eq:spectral:worstcase}
\norm{\widehat{\Pi}}_{q \to 2} &\le \sqrt{C_G(q)(1+2/\gamma)} \cdot \kappa , \text{ and } \spnorm{(I-\widehat{\Pi}) A} \le \sqrt{(3+\gamma)} \cdot OPT,  
\end{align}
where $C_G(q)>0$ is a constant that only depends on $q$ as given in Theorem~\ref{cor:norms-related}. For $q=\infty$ this value is known to be at most $\pi/2$. 

Moreover, for any $\gamma \in (0,1)$, there exists a polynomial time algorithm that finds an %$r'\le r/(1-\delta)$-dimensional
$r'\le (1+\tfrac{2}{\gamma})r$ dimensional
orthogonal projection $\widehat{\Pi}$ such that
\begin{align} \label{eq:spectral:bicriteria}
\norm{\widehat{\Pi}}_{q \to 2} &\le \sqrt{C_{G}(q)(1+2/\gamma)}\cdot \kappa, \text{ and } \spnorm{(I-\widehat{\Pi}) A} \le \sqrt{(1+\gamma)} \cdot OPT.  
\end{align}
\end{theorem}
%We will set $\delta:= 2/(2+\gamma)$.

% We remark that as in Theorem~\ref{thm:worstcase:frob}, the only randomization in the algorithm from Theorem~\ref{thm:worstcase:spectral} is in the construction of the separation oracle for the matrix norm constraint. In particular the algorithm from Theorem~\ref{thm:worstcase:spectral} has a {\em Las Vegas guarantee} i.e., the algorithm is always correct, and the running time is polynomial with high probability (and hence in expectation).  
% \anote{11/27: added above line. }
% \pnote{Perhaps change ``and hence in expectation'' to ``and also in expectation''?}

\paragraph{Proof of Theorem~\ref{thm:worstcase:spectral}.}
\anote{11/3/2020: Needs to be modified to incorporate the new relaxation.}
We will use the following mathematical relaxation for the problem. 
 \begin{align}
 \min ~~&\lambda \label{eq:spectralsdp}\\
 \text{s.t. } A^\top (I-X) A &\preceq \lambda I \label{eq:spectralsdp:obj2}\\     
 \tr(X) \le r, ~~&\text{ and } 0 \preceq X \preceq I   \label{eq:spectralsdp:constr}\\
 \norm{X}_{q \to q^*} &\le \kappa^2 \label{eq:spectralsdp:norm}
 \end{align}
The last constraint \eqref{eq:spectralsdp:norm} is NP-hard to certify. So as in the previous section we will relax it and consider the following two convex relaxations \CPthree and \CPfour. 

\begin{figure}[h]
\fbox{\begin{minipage}{.48\linewidth}
\CPthree:
\begin{align}
\min ~&\lambda \label{cp3:obj1}\\
\text{s.t. } &A^\top (I-X) A \preceq \lambda I \label{cp3:obj2}\\     
&\tr(X) \le r, ~~\text{ and } 0 \preceq X \preceq I   \label{cp3:stdconstr}\\
     &\max_{\substack{Y \in \calQ}}\iprod{X,Y}  \le C_G \kappa^2 ,\nonumber \text{ where}\\
    \calQ &= \sset{ Y \in \R^{n \times n}:~ Y \succeq 0, \sum_{i} Y_{ii}^{q/2} \le 1 } \label{cp3:convert:constr} %\label{cp3:def:Z}  
\end{align}
\end{minipage}}%
\fbox{\begin{minipage}{0.48 \linewidth}
\CPfour:
\begin{align}
\min ~&\lambda \label{cp4:obj1}\\
\text{s.t. } & A^\top (I-X) A \preceq \lambda I \label{cp4:obj2}\\     
& \tr(X) \le r, ~~ 0 \preceq X \preceq I   \label{cp4:stdconstr}\\
                & X \preceq \diag(d) \\
 \norm{d}_{q/(q-2)} &:=\Big(\sum_{i=1}^n d_i^{q/(q-2)}\Big)^{(q-2)/q} \le C_G \kappa^{2} \label{cp4:newsdp:constr}
\end{align}
\end{minipage}}
\caption{Two tractable relaxations \CPthree  and \CPfour for spectral norm objective.}\label{fig:relaxations:fr}
\end{figure}

The following claim (which is analogous to Claim~\ref{claim:feasibility}) shows that the above convex programs are valid relaxations and can be solved in polynomial time. 

\begin{claim}\label{claim:spectral:feasibility}
Any feasible projection matrix $\Pi$ of rank $r$ satisfying $\norm{\Pi}_{q \to 2} \le \kappa$ forms feasible solutions to \CPthree, \CPfour (for an appropriate feasible $d$) and \eqref{eq:spectralsdp} with objective value $\spnorm{A - \Pi A}^2$. Moreover the relaxations \CPthree and \CPfour can be solved in polynomial time to arbitrary accuracy.  
\end{claim}
\begin{proof}
The proof follows the same argument as Claim~\ref{claim:feasibility}, with a small modification to account for the different objective. 
We first argue that \CPthree and \CPfour are valid relaxations. Consider any projection matrix $\Pi$ of rank $r$ satisfying $\norm{\Pi}_{q \to 2} \le \kappa$. From Claim~\ref{claim:feasibility}, we have that constraints \eqref{cp3:stdconstr} and \eqref{cp3:convert:constr} of \CPthree and \eqref{cp4:stdconstr} and \eqref{cp4:newsdp:constr} of \CPfour are satisfied. To see that the objective value is preserved, note that for any projection matrix $\Pi$,
$$ \spnorm{A^\top (I-\Pi) A}= \spnorm{A^\top \Pi^{\perp} \Pi^{\perp} A} = \spnorm{\Pi^{\perp} A}^2 = \spnorm{A - \Pi A}^2,$$ 
 as required. This establishes the first part of the claim. 

We now show that \CPthree and \CPfour are polynomial time solvable up to accuracy $\eta>0$ in time polynomial in the input size and $\log(1/\eta)$. We use the Ellipsoid algorithm to solve both the relaxations. It is easy to verify that the feasible sets are convex. % is convex (including \eqref{eq:spectralsdp:norm}, which is just a upper bound constraint on a matrix norm).
The argument in Claim~\ref{claim:feasibility} proves that the constraints \eqref{cp3:stdconstr}, \eqref{cp3:convert:constr}, \eqref{cp4:stdconstr} and \eqref{cp4:newsdp:constr} are all efficiently separation. We now argue about the objective i.e., the constraint \eqref{cp3:obj2} and \eqref{cp4:obj2}. 
Finally, given $\lambda, X$,  \eqref{cp3:obj2} and \eqref{cp4:obj2} can also be efficiently separated by computing the maximum eigenvalue of $A^\top (I-X) A$. Let $v \in \mathbb{S}^{n-1}$ be the corresponding eigenvector. If the constraint is violated, the hyperplane separator is of the form
$$ \iprod{vv^\top , A^\top A} - \iprod{vv^\top, A^\top X A} - \lambda \le 0, \text{ i.e., } \iprod{vv^\top , A^\top A} - \iprod{A vv^\top A^\top, X } - \lambda \le 0, $$
since $\tr(vv^\top A^\top X A)= \tr(A vv^\top A^\top X)$. This completes the proof. 
\end{proof}

The proof of Theorem~\ref{thm:worstcase:spectral} also crucially uses the monotonicity of $q \to q^*$ matrix operator norm, and also follows the same outline as the Frobenius norm objective. The primary difference arises in analyzing the objective.  

\begin{proofof}{Theorem~\ref{thm:worstcase:spectral}}
Let $OPT^2:=\eps^2 \spnorm{A}^2$ for some $\eps \in (0,1]$.
Set $\delta:= 2/(2+\gamma)$. Claim~\ref{claim:spectral:feasibility} shows that in polynomial time
%with high probability 
we obtain a solution $X \succeq 0$ satisfying \eqref{eq:spectralsdp:constr}, \eqref{eq:spectralsdp:obj2} with $\lambda \le OPT$, and $\max_{Y \in \calQ} \iprod{Y,Z} \le C_G(q) \kappa^2$ . From \eqref{eq:relax:constraint:approx}, this implies $\norm{X}_{q \to q^*} \le C_G(q) \kappa^2$.
%(with a very small probability it returns FAIL). %
%The rest of the algorithm (and analysis) conditions on the success of this event.  

Let $X=\sum_i \lambda_i v_i v_i^\top$ and let $S=\sset{i : \lambda_i \ge 1-\delta}$. 
For the rest of the analysis we will assume without loss of generality that $\spnorm{A}=1$. 
We first show the guarantee in \eqref{eq:spectral:bicriteria}. The projection output is just $\Pi_S=\sum_{i\in S} v_i v_i^\top$. Observe that $|S| \le r / (1-\delta)$ from \eqref{eq:spectralsdp:constr}. Since the projector we output is just $\Pi_S$, each of its associated eigenvalues are at least $1-\delta$. Hence, the operator norm bounds follows using the monotonicity of the norm since $\Pi_S \preceq \tfrac{1}{1-\delta} X$. To verify the objective value we see that
\begin{align*}
    \Bigspnorm{\sum_{i \in [n]} (1-\lambda_i) A^\top v_i v_i^\top A } &=\Bigspnorm{A^\top (I-X) A} \le \eps^2 \spnorm{A}^2 = \eps^2 \\  
        \Bigspnorm{\sum_{i \notin S} \delta A^\top v_i v_i^\top A } &\le         \Bigspnorm{\sum_{i \notin S} (1-\lambda_i) A^\top v_i v_i^\top A } \le \Bigspnorm{\sum_{i \in [n]} (1-\lambda_i) A^\top v_i v_i^\top A } \le \eps^2 \\
        \text{Hence }  \Bigspnorm{A^\top \Pi_S^{\perp} A }  &\le \frac{\eps^2}{\delta} = \Big(1+\frac{\gamma}{2} \Big)\eps^2 ,
\end{align*}
as required. 
We now show the guarantee in \eqref{eq:spectral:worstcase} where we output a projection of rank at most $r$ (with no slack). Let $M':= \sum_{i \in S} A^\top v_i v_i^\top A$. Let $\Pi'$ be the projection matrix for the subspace corresponding to the best rank $r$ projection of $M'$. The algorithm outputs $\Pi'$. %\xnote{May wanna say $\Pi'$ is also robust because any vector in $\Pi'$ is also in $M'$}

Note that $\Pi' \preceq \Pi_S$, hence by monotonicity, the $q \to q^*$ operator norm constraint is satisfied up to a $\alpha:=C_G$ factor. Also note that if $\Pi^*$ is the projection that gives the optimal solution to the problem,
$$\spnorm{M' - A^\top \Pi^* A} \le \spnorm{M' - A^\top A} + \spnorm{A^\top A - A^\top \Pi^* A} \le \eps^2 + \frac{\eps^2}{\delta}= \eps^2 (1+\tfrac{1}{\delta})= (2+\tfrac{\gamma}{2})\eps^2.$$
But $A^\top \Pi^* A$ is a valid approximation of $M'$ of rank at most $r$. Hence, we have that \begin{align*}
    \spnorm{M' - \Pi' M' \Pi'} &\le \eps^2 (1+\tfrac{1}{\delta}) \\
    \spnorm{A^\top A - \Pi' A^\top A \Pi'} &\le \spnorm{A^\top A - \Pi' \Pi_S A^\top A \Pi_S \Pi'} \le \spnorm{A^\top A - M'} + \spnorm{M'-\Pi' M' \Pi'}\\
    &\le \frac{2\eps^2}{\delta} + \eps^2 =  (1+\tfrac{2}{\delta}) \eps^2 = (3+\gamma) \eps^2,
\end{align*}
as required. 

\end{proofof}

As before the same ideas also give the following more general theorem for any monotone matrix norm $\gnorm{\cdot}$ that is approximately certifiable. 
\begin{theorem}\label{prop:gen:spectral}
Let $\gnorm{\cdot}$ be any matrix norm that is {\em monotone} and $\alpha$-approximately certifiable for some $\alpha \ge 1$.  
Suppose the data matrix $A \in \R^{n \times m}$ has a projection $\Pi^{*}$ of rank at most $r$ such that $\gnorm{\Pi^*} \le \kappa$ and the approximation error $OPT:=\spnorm{(I-\Pi^*) A}$.
There exists a polynomial time algorithm such that given any $\gamma \in (0,1)$, it finds an orthogonal projection $\widehat{\Pi}$ of dimension at most $r$ satisfying
\begin{align} 
\gnorm{\widehat{\Pi}} &\le \sqrt{\alpha (1+2/\gamma)} \cdot \kappa , \text{ and } \spnorm{(I-\widehat{\Pi}) A} \le \sqrt{(3+\gamma)} \cdot OPT.  
\end{align}

Moreover, for any $\gamma \in (0,1)$, there exists a polynomial time algorithm that finds an %$r'\le r/(1-\delta)$-dimensional
orthogonal projection $\widehat{\Pi}$ of rank $r'\le (1+\tfrac{2}{\gamma})r$ such that
\begin{align} 
\gnorm{\widehat{\Pi}} &\le \sqrt{\alpha(1+2/\gamma)}\cdot \kappa, \text{ and } \spnorm{(I-\widehat{\Pi}) A} \le \sqrt{(1+\gamma)} \cdot OPT.  
\end{align}
\end{theorem}
We omit the proof, since the ideas are identical to Theorem~\ref{thm:worstcase:spectral}.

\subsection{Recovering the Optimal Projection Matrix} \label{sec:recovery}
\anote{We can potentially remove this section. It may not be used anywhere?}
We now show that if the optimal robust low-rank projection has very small error compared to the $r$th smallest singular value of $A$, then we can in fact approximately recover the subspace itself up to small error measured in terms of the principal angles. For two subspaces with projection matrices $\Pi_1, \Pi_2$, the $\text{Sin}$ of the canonical angles matrix is given by $\Pi_1^{\perp} \Pi_2$.
These techniques will also be helpful for recovery in the Spiked Covariance model. 
The following simple corollary will work for both Frobenius norm error and spectral norm error. For this purpose, we will just use $\unorm{A}$ to denote the norm of $A$, where the unspecified matrix norm $\unorm{\cdot}$ is norm in which we are measuring the error --  either Frobenius norm or spectral norm. 

\begin{corollary}\label{corr:recovery:fr}
Suppose the data matrix $A \in \R^{n \times m}$ has an $r$-dimensional projection $\Pi^{*}$ such that $\norm{\Pi^*}_{q \to 2} \le \kappa$, the approximation error $OPT:=\unorm{(I-\Pi^*) A}^2 < \eps^2 \unorm{A}^2$ and $\sigma_r (\Pi^* A) \ge \theta$.
There exists a polynomial time algorithm that finds a projection $\widehat{\Pi}$ of rank at most $r$ such that 
\begin{equation}\label{eq:strongerrecovery}
    \unorm{\Pi^{\perp} \widehat{\Pi}} \le O(1+\alpha)\cdot \frac{\eps \unorm{A}}{\theta}.
\end{equation}
where the subspace corresponding to $\Pi$ is a subset of the subspace given by $\Pi^*$ and $\alpha$ is the approximation factor attained by the algorithm in Theorem~\ref{thm:worstcase:frob} (or Theorem~\ref{thm:worstcase:spectral}).
\end{corollary}

% \begin{corollary}\label{corr:recovery:fr}
% Suppose the data matrix $A \in \R^{n \times m}$ has an $r$-dimensional projection $\Pi^{*}$ such that $\norm{\Pi^*}_{q \to 2} \le \kappa$ and the approximation error $OPT:=\unorm{(I-\Pi^*) A}^2 < \eps^2 \unorm{A}^2$.
% There exists a deterministic polynomial time algorithm that finds an $r$-dimensional projection $\widehat{\Pi}$ such that 
% \begin{equation} \unorm{\Pi^{\perp} \widehat{\Pi}} \le O(1+\alpha) r \sqrt{\eps},
% \end{equation}
% where the subspace corresponding to $\Pi$ is a subset of the subspace given by $\Pi^*$ and $\alpha$ is the approximation factor attained by the algorithm in Theorem~\ref{thm:worstcase:frob} (or Theorem~\ref{thm:worstcase:spectral}).\\
% Moreover, in addition to the above conditions,  
% \begin{equation}\label{eq:strongerrecovery}
% \text{if }~~ \sigma_r (\Pi^* A) \ge \theta, ~~\text{ then we have }~~    \unorm{\Pi^{\perp} \widehat{\Pi}} \le O(1+\alpha)\cdot \frac{\eps \unorm{A}}{\theta}.
% \end{equation}
% \end{corollary}
% \anote{The first part of the corollary seems incorrect.}
Note that the above bound holds for the spectral norm error and the Frobenius norm error. \anote{Similar bound for spectral norm.}
\begin{proof}
The algorithm is exactly the same algorithm used in Theorem~\ref{thm:worstcase:frob}. Let $\Pi$ denote the best robust low-rank subspace for $A$. We will then use the Davis-Kahan $\sin \Theta$ theorem about perturbations of singular vectors to show that the subspaces given by $\Pi_1$ and $\Pi_2$ are close. Note that the Davis-Kahan theorem states that if $\Pi_i$ is the projection matrix onto eigenspaces of $A_i A_i^\top$ respectively ($i \in \sset{1,2}$) with the least singular values of $\Pi_1 A_1$ being at least $\delta>0$ more than the singular values of $\Pi_2^{\perp} A_2$, then for any unitarily invariant norm $\unorm{\cdot}$, 
$$ \unorm{\Pi_2^{\perp} \Pi_1} \le \frac{\unorm{A_1-A_2}}{ \delta }.$$

We would like to apply it with $A_2= \widehat{\Pi} A, A_1 = \Pi^* A$ and $\Pi_2=\widehat{\Pi}, \Pi_1 = \Pi^*$. We know that by the triangle inequality, for some constant $\alpha$ given by the approximation ratio in Theorem~\ref{thm:worstcase:frob} (or Theorem~\ref{thm:worstcase:spectral}), 
\begin{align}\label{eq:recovery:1}
    \unorm{\Pi^* A - \widehat{\Pi} A} &\le \unorm{\Pi^* A - A} + \unorm{A - \widehat{\Pi} A} \le \eps \unorm{A} + \alpha \eps \unorm{A} \le (\alpha+1) \eps \unorm{A}, 
\end{align}
where we used the fact that $\widehat{\Pi}$ gives an $\alpha$-factor approximation to the objective. Moreover, in our case $A_1=\Pi^* A$ is itself of rank-$r$ and $\Pi_2^{\perp} A_2 =0$. Under the stronger assumption in \eqref{eq:strongerrecovery}, we have $\sigma_r(\Pi^* A) \ge \theta$. Hence we see that \eqref{eq:strongerrecovery} holds since
\begin{align*}
    \unorm{\widehat{\Pi}^{\perp} \Pi^{*}} &\le \frac{\unorm{\Pi^* A - \widehat{\Pi} A}}{\theta} \le \frac{(1+\alpha) \eps}{\theta}.
\end{align*}

\end{proof}

% !TEX root = main.tex

\section{Data Poisoning and Robustness to Adversarial Perturbations at Training Time}
In this section, we consider training-time robustness, where the input matrix $\tilde{A}$ is an adversarial perturbation of $A$ and the goal is to recover the robust projection of $A$ rather than $\tilde{A}$. In Section~\ref{sec:training:frob}, We will study the approximation algorithms in Frobenius norm error. In Section~\ref{sec:training:spectral}, we will study its counterpart in spectral norm error. Finally we show a lower bound in Section~\ref{sec:app-lower-bound-additive-error}. \xnote{July/11: Add this short overview.}

\subsection{Training-Time Robustness: Approximations in Frobenius Norm Error}\label{sec:training:frob}

\anote{Need to read through this carefully, and fill in other explanations. Vaggos/Xue, can you look at this?}
%We will focus on $\ell_\infty$ perturbations for this portion. 

\begin{theorem}\label{thm:robusttraining:frob}
Suppose $q \ge 2$ and $A \in \R^{n \times m}$ is the (unknown) uncorrupted data matrix, with an (orthogonal) projection matrix $\Pi^*$ of rank at most $r$ that is robust i.e., $\norm{\Pi^*}_{q \to 2}\le \kappa$ satisfying $\norm{A-\Pi^* A}_F^2 \le \eps \norm{A}_F^2$ for some $\eps \in [0,1]$. There exists a polynomial time algorithm that
given as input any (adversarially perturbed) data matrix $\tilde{A}$ s.t. for each column $j\in [m]$, $\norm{\tilde{A}_j - A_j}_q \le \delta$, outputs an orthogonal projection $\widehat{\Pi}$ of rank at most $r$ such that for any $\eta>0$
\begin{align} \label{eq:robusttraining}
\norm{\widehat{\Pi}}_{q \to 2} &\le O(\kappa) , \text{ and } \norm{A-\widehat{\Pi} A}_F^2  \le O(\eps+ \eta) \cdot \norm{A}_F^2 + O(\tfrac{1}{\eta}) \cdot \delta^2 \kappa^2 m.  
\end{align}
\end{theorem}
To get a multiplicative approximation we will set $\eta = O(\eps)$, and get an extra additive term of $\delta^2 \kappa^2 m / \eps$. 
Here think of $\delta^2 \kappa^2 \ll \tfrac{1}{m} \cdot \eps \norm{A}_F^2$. 
Further we remark that the above guarantees are optimal up to constant factors; in particular, the additive factor of $O(m \delta^2 \kappa^2)$ is unavoidable (see Proposition~\ref{prop:additiveerror:tight}).

The main challenge here is that while $A$ has a good low-rank projection (in fact a robust one), $\tilde{A}$ may be very far from a rank-$r$ matrix (let alone having a robust rank-$r$ approximation). Further, the best robust low-rank approximation of $\tilde{A}$ could be very different from the best robust low-rank projection of $A$. This is because the entry-wise perturbations of $\delta$ could be too large in aggregate; for instance, it could be the case that $\norm{\tilde{A}}_F^2 \gg \norm{A}_F^2$. 
Suppose $\Pi^*$ is the best robust low-rank projection of $A$. 
We will run the algorithm in the previous section not on the given matrix $\tilde{A}$, but on a suitably modified matrix $A'$. 

\begin{figure}[htbp]
\begin{center}
\fbox{\parbox{0.98\textwidth}{
%{\bf Algorithm for Robust Adversarial Training in Spectral norm}

{\bf Input:} $\tilde{A}$, the corrupted $n \times m$ data matrix, rank $r$, robustness parameter $\kappa \ge 1$ and norm $q \ge 2$. 
\begin{enumerate}
\item Compute $A'$ (using Lemma~\ref{lem:computenearbyA}) such that $$A'=\argmin_{\substack{B \in \R^{n \times m} \text{ s.t. }\\ \norm{B_j-\tilde{A}_j}_q \le \delta, \forall j \in [m]}} \norm{B}_F.$$ 
\item Run the algorithm from Theorem~\ref{thm:worstcase:frob} on $A'$, to obtain a rank-$r$ projection matrix $\widehat{\Pi}$.
\item Output $\widehat{\Pi}$.
\end{enumerate}
}}
\end{center}
\caption{\label{ALG:training:frob} Robust rank-$r$ approximations in Frobenius norm under adversarial perturbations during training.}
\end{figure}

\begin{lemma}\label{lem:computenearbyA}
There is a polynomial time algorithm that given any matrix $M \in \R^{n \times m}$, can find  
$$\Gamma_q(M)=\min_{\substack{B \in \R^{n \times m} \text{ s.t. }\\ \norm{B_j-M_j}_q \le \delta, \forall j \in [m]}} \norm{B}^2_F,$$
up to arbitrary accuracy. 
\end{lemma}
\begin{proof}
First we note that since $\norm{B}_F^2 = \sum_j \norm{B_j}_2^2$, the optimization problem is separable across each of the $m$ samples i.e., 
$$ \min_{\substack{B \in \R^{n \times m} \text{ s.t. } \\ \norm{B_j-M_j}_q \le \delta, \forall j \in [m]}} \norm{B}^2_F  = \sum_{j \in [m]} \min_{\substack{B_j \in \R^{n} \text{ s.t. } \\\norm{B_j-M_j}_q \le \delta}} \norm{B_j}_2^2 
%= \sum_{j \in [m]} \min_{\substack{z \in \R^{n} \\ \norm{z}_q \le \delta}} \norm{M_j + z}_2^2.
$$
We now describe how to solve each of the $m$ subinstances corresponding to the column $j \in [m]$, which for a given $b \in \R^n$ is of the form $$\min_{\substack{z \in \R^{n} \\ \norm{z}_q \le \delta}} \norm{b - z}_2^2.$$
Note that the least-squares objective $\norm{b-z}_2^2$ is convex. Moreover the constraint $\norm{z}_q \le \delta$ is also convex; further there is a simple separation oracle for this constraint since by duality
$$\norm{z}_q = \max_{y \in \R^n: \norm{y}_{q^*} \le 1} \iprod{y,z} = \Big\langle \frac{z^*}{\norm{z^*}_{q^*}},z \Big \rangle, \text{ where } z^*_i= \text{sign}(z_i) |z(i)|^{q-1}~~ \forall i \in [n].$$
Hence by using the Ellipsoid algorithm, this problem can be solved in polynomial time. \footnote{In fact Projected Gradient Descent Algorithm can also be used here; see \citet{sra2012fast}.} 
\end{proof}

Note that when $q=\infty$,  it is easy to find the matrix $A'$, by just setting  $$A'_{ij}= \text{sign}(M_{ij}) \cdot \max\sset{ 0, |M_{ij}|-\delta }, ~~\forall i,j \in [n].$$

We will argue that $\Pi^*$ also gives a good low-rank approximation to $A'$. This crucially uses the fact that $\Pi^{*}$ has bounded $\ell_q \to 2$ norm, which implies the following useful lemma. 

\begin{lemma}\label{lem:closeprojections}
Suppose $A, B \in \R^{n \times m}$ are two matrices such that for each column $j \in [m]$, $\norm{A_j-B_j}_q \le \delta$, and let $\Pi$ be any rank-$r$ projection matrix such that $\norm{\Pi}_{q \to 2} \le \kappa$. Then for any $\eta \in (0,1)$,   
$$(1-\eta) \norm{\Pi A}_F^2 -  (\tfrac{1}{\eta}-1) \delta^2 \kappa^2 m \le \norm{\Pi B}_F^2 \le (1+\eta) \norm{\Pi A}_F^2 +  (\tfrac{1}{\eta}+1) \delta^2 \kappa^2 m. $$
\end{lemma}
\begin{proof}
For each $j\in [m]$, let $A_j,B_j$ be the $j$th columns of $A$ and $B$ respectively. Then $\norm{\Pi (A_j - B_j)}_2 \le \norm{\Pi}_{q \to 2} \norm{A_j - B_j}_q \le \delta \kappa$. Using this along with the triangle inequality we get,  
\begin{align*}
    \norm{\Pi B}_F^2 & = \sum_{j=1}^m \norm{\Pi (B_j-A_j) + \Pi A_j}_2^2 \ge \sum_{j=1}^m ( \norm{\Pi A_j}_2 - \delta \kappa )^2 \\
    &\ge \sum_{j=1}^m   \norm{\Pi A_j}_2^2  - 2 \Big(\frac{\delta \kappa}{\sqrt{\eta}} \Big) (\sqrt{\eta} \norm{\Pi A_j}_2) +(\delta \kappa)^2  \nonumber\\
    &\ge (1-\eta) \norm{\Pi A}_F^2 - (\tfrac{1}{\eta}-1) \delta^2 \kappa^2 m, \text{ for any } \eta \in (0,1). 
\end{align*}
This proves the first inequality. A similar argument also shows the other inequality. 
\end{proof}

We now prove that Algorithm~\ref{ALG:training:frob} finds an approximately optimal robust low-rank projection for unknown, uncorrupted data matrix $A$. 

\begin{proofof}{Theorem~\ref{thm:robusttraining:frob}}
The first step of the algorithm finds the matrix $A'$ given by
 $$ A' = \argmin_{\substack{B \in \R^{n \times m} \text{ s.t. }\\ \norm{B_j-\tilde{A}_j}_q \le \delta, \forall j \in [m]}} \norm{B}^2_F. $$
% Note that it is easy to find the matrix $A'$, by just setting  $$A'_{ij}= \text{sign}(\tilde{A}_{ij}) \cdot \max\sset{ 0, |\tilde{A}_{ij}|-\delta }, ~~\forall i,j \in [n].$$
Note that $\norm{A'}_F \le \norm{A}_F$ since $A$ is also a feasible solution for the above minimization. Moreover since $\norm{A_j-A'_j}_q \le 2\delta$ for each $j \in [m]$, we get from Lemma~\ref{lem:closeprojections}, 
\begin{align}
    \norm{\Pi^* A'}_F^2 %&= \sum_{j=1}^m \norm{\Pi^* (A'_j-A_j) + \Pi^* A_j}_2^2 \ge \sum_{j=1}^m ( \norm{\Pi^* A_j}_2 - 2\delta \kappa )^2 \ge \sum_{j=1}^m   \norm{\Pi^* A_j}_2)^2  - 2 \Big(\frac{2\delta \kappa}{\eta} \Big) (\eta \norm{\Pi^* A_j}_2) +(2\delta \kappa)^2  \nonumber\\
    &\ge (1-\eta) \norm{\Pi^* A}_F^2 - 4(\tfrac{1}{\eta}-1) \delta^2 \kappa^2 m, ~\text{ for any } \eta \in (0,1). \label{eq:train:approx}
\end{align}

 Now we run the algorithm from the previous section (Theorem~\ref{thm:worstcase:frob}) on $A'$. From Theorem~\ref{thm:worstcase:frob} (with $\delta=1/2$ say), we find a rank-$r$ projection matrix $\Pi$ with $\norm{\Pi}_{\infty \to 2} \le O(\kappa)$ such that
 \begin{align*}
     \norm{A' - \Pi A'}_F^2 &\le 3 \Big(\norm{A'}_F^2 - (1-\eta) \norm{\Pi^* A}_F^2 + 4(\tfrac{1}{\eta}-1) \delta^2 \kappa^2 m \Big)\\
     &\le 3 \Big( \norm{A- \Pi^* A}_F^2 \Big) + 3\eta \norm{\Pi^* A}_F^2+ 12 (\tfrac{1}{\eta}-1) \delta^2 \kappa^2 m \\
     &\le 3(\eps+\eta) \norm{A}_F^2+ 12 (\tfrac{1}{\eta}-1) \delta^2 \kappa^2 m.
 \end{align*} 
 However we know that $\norm{A'}_F^2 \ge \norm{\Pi^* A'}_F^2$. Hence
 \begin{align*}
     \norm{\Pi A'}_F^2 &\ge \norm{A'}_F^2 - \norm{A'-\Pi A'}_F^2 \ge \norm{\Pi^* A'}_F^2 - \norm{A'-\Pi A'}_F^2 \\
     &\ge (1-\eta) \norm{\Pi^* A}_F^2 - 3 (\eps + \eta) \norm{A}_F^2 - 16(\tfrac{1}{\eta} -1 ) \delta^2 \kappa^2 m   \\
    \text{ Hence, }\norm{A-\Pi A}_F^2 &= \norm{A}_F^2 - \norm{\Pi A}_F^2 \le_{\substack{\text{Lemma} ~\ref{lem:closeprojections}}} \norm{A}_F^2 - (1-\eta) \norm{\Pi A'}_F^2 +  (\tfrac{1}{\eta} + 1) \delta^2 \kappa^2 m\\
   & \le \norm{A}_F^2 - (1-\eta)^2 \norm{\Pi^* A}_F^2 + 3 (\eps+\eta) (1-\eta) \norm{A}_F^2 \\
   &~~~+ m\delta^2 \kappa^2 \Big(1+\tfrac{1}{\eta} + 16 (1-\eta) (\tfrac{1}{\eta} -1) \Big)\\
   &\le \norm{A - \Pi^* A}_F^2 + (3\eps+5\eta) \norm{A}_F^2  + \Big(1+\tfrac{17}{\eta}\Big) \delta^2 \kappa^2 m \\
   &\le  O(\eta) \norm{A}_F^2 + O(\tfrac{1}{\eta}) \delta^2 \kappa^2 m, 
 \end{align*}
 for any $\eta \ge 4\eps$.
 \end{proofof}
 
 \subsection{Training-Time Robustness: Approximations in Spectral Norm Error} \label{sec:training:spectral}
  
%  The algorithm in Section~\ref{sec:training:frob} that finds a low-rank robust projection with small frobenius norm error also in fact gives low spectral norm error. 
  
% \begin{corollary}\label{thm:robusttraining:frob}
% Suppose the uncorrupted data matrix $A \in \R^{n \times m}$ has an $r$-dimensional projection $\Pi^{*}$ such that $\norm{\Pi^*}_{\infty \to 2} \le \kappa$ and the approximation error $OPT:=\spnorm{(I-\Pi^*) A}^2 \le \eps \norm{A}_2^2$.

% Given as input any (adversarially perturbed) data matrix $\tilde{A}$ such that $\norm{\tilde{A} - A}_\infty \le \delta$, there exists a deterministic polynomial time algorithm that finds an $r$-dimensional projection $\widehat{\Pi}$ such that for any $\eta>0$
% \begin{align} 
% \norm{\widehat{\Pi}}_{\infty \to 2} &\le O(\kappa) , \text{ and } \spnorm{(I-\widehat{\Pi}) A}^2  \le O(\eps+ \eta) \cdot r \spnorm{A}^2 + O(\tfrac{1}{\eta}) \cdot \delta^2 \kappa^2 m.  
% \end{align}
% \end{corollary}
%   \begin{proof}
%   Consider the projection output by algorithm in Section~\ref{sec:training:frob}. Let $\Pi_1$ be the optimal projection for the Frobenius norm error. By Theorem~\ref{thm:robusttraining:frob}, we have that $\norm{\widehat{\Pi}}_{\infty \to 2} \le O(\kappa)$ and at the same time,
%  \begin{align*}
%      \norm{(I-\widehat{\Pi}) A}_F^2  &\le  O(\norm{A - \Pi_1 A}_F^2) +\eta \norm{A}_F^2 + O(\tfrac{1}{\eta}) \cdot \delta^2 \kappa^2 m \\
%         & \le \norm{A - \Pi^* A}^2 + O(\tfrac{1}{\eta}) \cdot \delta^2 \kappa^2 m
%  \end{align*}
 
%   \end{proof}

 We now show guarantees for low-rank approximations in spectral norm error that are similar to Theorem~\ref{thm:robusttraining:frob}. However, there is a qualitative difference: we will either find a robust low-dimensional projection of the unknown dataset $A$, or we will certify that the dataset has been poisoned substantially. In particular, the algorithm will {\em never} output a low-dimensional representation that is bad for the unknown data matrix $A$. We will later see how these guarantees also imply training-time robustness for downstream unsupervised learning applications like spectral clustering, robust mean estimation and learning mixture models. In what follows $\spnorm{\cdot}$ will refer to the spectral norm. 

 \begin{theorem}\label{thm:training:spectralnorm}
% Suppose the (unknown) uncorrupted data matrix $A \in \R^{n \times m}$ has an $r$-dimensional projection $\Pi^{*}$ such that $\norm{\Pi^*}_{\infty \to 2} \le \kappa$ and the approximation error $OPT:=\spnorm{(I-\Pi^*) A} \le \eps \spnorm{A}$. 
Suppose $q \ge 2$ and $A \in \R^{n \times m}$ is the (unknown) uncorrupted data matrix, and let $\Pi^*$ have the smallest spectral norm error $\spnorm{A-\Pi A}$ among (orthogonal) projections of rank at most $r$ that are robust i.e., $\norm{\Pi}_{q \to 2}\le \kappa$. There exists a polynomial time algorithm (Alg.~\ref{ALG:training:spectral}) that given as input any (adversarially perturbed) data matrix $\tilde{A}$ s.t. for each column $j\in [m]$, $\norm{\tilde{A}_j - A_j}_q \le \delta$ and a parameter $\tau>0$, outputs either a projection matrix $\widehat{\Pi}$ of rank at most $r$ or outputs \Bad s.t. 
\begin{enumerate}
    \item[(I)] if the algorithm outputs a projection $\widehat{\Pi}$ of rank at most $r$, then it is a near-optimal robust low-rank approximation for the unknown matrix $A$ i.e., for some small universal constant $c\ge 1$, 
    \begin{equation} \label{eq:training:spectralnorm}
\forall \eta>0, ~~\norm{\widehat{\Pi}}_{q \to 2} \le c_q \kappa , \text{ and } \norm{(I-\widehat{\Pi}) A}  \le O\Big(1+\tfrac{1}{\eta}\Big) \Big( \tau+\spnorm{A-\Pi^* A}+\sqrt{m} \delta \kappa \Big) + \sqrt{2\eta}\norm{A}.  \end{equation}
% \begin{equation} \label{eq:training:spectralnorm}
% \forall \eta>0, ~~\norm{\widehat{\Pi}}_{q \to 2} \le c_q \kappa , \text{ and } \norm{(I-\widehat{\Pi}) A}  \le c\Big(\tau+ \spnorm{A-\Pi^* A}+ \eta \cdot \norm{A} + \tfrac{1}{\eta} \cdot \delta \kappa \sqrt{m} \Big).  \end{equation}
\item[(II)] if the algorithm outputs \Bad, then either the data was poisoned i.e., $\norm{A-\tilde{A}} > \tau$, or there is no good robust spectral norm approximation for $A$ i.e., $\norm{A - \Pi A} > \tau$ for all rank-$r$ projection matrices $\Pi$ s.t. $\norm{\Pi}_{q \to 2} \le \kappa$.   
%Output \textsc{Bad Input}. Certify that the data has been poisoned i.e., $\norm{A-\tilde{A}} \ge \eps \norm{A}$.
\end{enumerate}
In particular, if we are promised that $A$ has a good robust projection $\Pi^*$ of value $\spnorm{A-\Pi^* A} \le \eps \spnorm{A}$, then the algorithm either finds an approximately optimal robust projection $\widehat{\Pi}$ of rank at most $r$ for $A$ with 
\begin{equation}\label{eq:spectral:betterguarantee}
\norm{\widehat{\Pi}}_{q \to 2} \le c_q\kappa, \text{ and } \forall \eta>0,~~\norm{(I-\widehat{\Pi}) A}  \le O\Big(1+\tfrac{1}{\eta}\Big) \Big( \spnorm{A-\Pi^* A}+\sqrt{m} \delta \kappa) + \sqrt{2\eta}\norm{A},
\end{equation}
or certifies that the data has been poisoned i.e., $\norm{\tilde{A}-A} > \eps \norm{A}$.
 \end{theorem}
\anote{Can $\eta$ be set to $\tau/\norm{A}$ by default?}
Our algorithm just runs the worst-case approximation algorithm from Theorem~\ref{thm:worstcase:spectral} on $\tilde{A}$ to find a projection $\widehat{\Pi}$. If the error is less than $\tau$, it outputs $\widehat{\Pi}$; else it certifies that the data is corrupt.

The main feature of the above algorithm is that it is {\em always correct}. The algorithm certifies that the input is {\sc Bad} only when the data has been poisoned i.e., $\tilde{A}$ is substantially far from $A$, or $A$ did not have a good robust low-rank approximation to begin with. More crucially, when it does output a projection matrix $\widehat{\Pi}$, it is guaranteed to be a valid robust projection\footnote{ In particular it rules out the scenario where the algorithm finds a solution that it thinks is good (on $\tilde{A}$), but is in fact bad for the unknown, uncorrupted matrix $A$.} for the unknown matrix $A$. %(with an exponentially small probability, the algorithm may not terminate in polynomial time; this happens exactly when the algorithm from Theorem~\ref{thm:worstcase:spectral} fails to terminate in polynomial time\pnote{11/27 Aravindan: check this.}).
% (with some small probability, the algorithm may also return FAIL; this happens exactly when the algorithm from Theorem~\ref{thm:worstcase:spectral} outputs FAIL \pnote{change this. the algorithm never fails.}). %We note that this is stronger guarantee than other scenarios where the algorithm just claims the dataset is poisoned if it did not work. In such cases, the algorithm could have found a solution on $\widehat{A}$
We remark that the additive error term of $\Omega(\delta \kappa \sqrt{m})$ is unavoidable here information-theoretically; see Proposition~\ref{prop:additiveerror:tight} for an example. 
%  \begin{theorem}\label{thm:training:spectralnorm}
%  Suppose the (unknown) uncorrupted data matrix $A \in \R^{n \times m}$ has an $r$-dimensional projection $\Pi^{*}$ such that $\norm{\Pi^*}_{\infty \to 2} \le \kappa$ and the approximation error $OPT:=\spnorm{(I-\Pi^*) A} \le \eps \spnorm{A}$. 
% Given as input any (adversarially perturbed) data matrix $\tilde{A}$ such that $\norm{\tilde{A} - A}_\infty \le \delta$, there exists a deterministic polynomial time algorithm that outputs either a projection matrix $\widehat{\Pi}$ or outputs \textsc{Bad Input}, such that
% \begin{enumerate}
%     \item[(I)] If the algorithm outputs an $r$-dimensional projection $\widehat{\Pi}$, then it is robust for the unknown matrix $A$ i.e., 
% \begin{equation} \label{eq:training:spectralnorm}
% \forall \eta>0, ~~\norm{\widehat{\Pi}}_{\infty \to 2} \le O(\kappa) , \text{ and } \norm{(I-\widehat{\Pi}) A}  \le O(\eps+ \eta) \cdot \norm{A} + O(\tfrac{1}{\eta}) \cdot \delta \kappa \sqrt{m}.  \end{equation}
% \item[(II)] If the algorithm outputs \Bad, then either the data was poisoned i.e., $\norm{A-\tilde{A}} > \tau$, or there was no good robust spectral norm approximation for $A$ i.e., $\norm{A - \Pi^* A} > \tau$.   
% %Output \textsc{Bad Input}. Certify that the data has been poisoned i.e., $\norm{A-\tilde{A}} \ge \eps \norm{A}$.
% \end{enumerate}
%  \end{theorem}

\anote{should the threshold be based on $\norm{\widehat{A}}$ and not $\norm{A}$?. Mention what to set $\tau$ for the moreover part?} 
\anote{Maybe ditch figure because it's too simple?}
\begin{figure}[htbp]
\begin{center}
\fbox{\parbox{0.98\textwidth}{
%{\bf Algorithm for Robust Adversarial Training in Spectral norm}

{\bf Input:} $\tilde{A}$, the corrupted $n \times m$ data matrix, tolerance parameter $\tau>0$, rank $r$, robustness parameter $\kappa \ge 1$ and norm $q \ge 2$. 
\begin{enumerate}
\item Run the algorithm from Theorem~\ref{thm:worstcase:spectral} on $\tilde{A}$, to obtain a rank-$r$ projection matrix $\widehat{\Pi}$.
\item If the robust low-rank approximation error on $\tilde{A}$, $\spnorm{\tilde{A} - \widehat{\Pi} \tilde{A}} \le \tau$, output $\widehat{\Pi}$.
\item Otherwise output \Bad.
\end{enumerate}
}}
\end{center}
\caption{\label{ALG:training:spectral} Robust rank-$r$ approximations in Spectral norm error under adversarial perturbations in training.}
\end{figure}

 The following is the key lemma that argues that if the projection $\widehat{\Pi}$ gives a small error on $\tilde{A}$, it necessarily gives a low-error on $A$. 
 \begin{lemma}\label{lem:training:spectralnorm}
 Let $\delta \in \R_+$ and $A,B \in \R^{n \times m}$ such that $\norm{A-B}_q \le \delta$. Let $\Pi_1, \Pi_2$ be projection matrices such that $\norm{\Pi_1}_{q \to 2}, \norm{\Pi_2}_{q \to 2} \le \kappa$, and = $\norm{A-\Pi_1 A} \le \eps_1$ and $\norm{B - \Pi_2 B} \le \eps_2$. Then we have that 
for any $\eta \in (0,1)$,
 \begin{align}
     \norm{A - \Pi_2 A} &\le O\Big(1+\tfrac{1}{\eta}\Big) \Big( \eps_1+\eps_2+\sqrt{m} \delta \kappa \Big) + \sqrt{2\eta}\norm{A}, \\
    \text{and } \norm{B - \Pi_1 B} &\le O\Big(1+\tfrac{1}{\eta}\Big) \Big( \eps_1+\eps_2+\sqrt{m} \delta \kappa \Big) + \sqrt{2\eta}\norm{B}.
\end{align}
 \end{lemma}
\begin{proof}
The projection matrices $\Pi_1, \Pi_2$ are both robust. For $\ell \in \sset{1,2}$
\begin{align}
    \norm{\Pi_\ell A - \Pi_\ell B}^2 &\le \norm{\Pi_\ell (A-B)}_F^2 = \sum_{j \in [m]} \norm{\Pi_\ell (A_j - B_j)}_2^2 \le m \kappa^2 \delta^2 \nonumber \\
   \text{Hence  } \Big| \norm{\Pi_\ell A} - \norm{\Pi_\ell B} \Big| &\le \sqrt{m} \kappa \delta. \label{eq:closeproj} 
\end{align}

Let $\gamma := \sqrt{m} \delta \kappa$.
We also know that $\norm{A - \Pi_1 A} \le \eps_1$. 
\begin{align}
&    \norm{ A - \Pi_1 B} \le \norm{A - \Pi_1 A} + \norm{\Pi_1 A - \Pi_1 B} \le \eps_1+ \gamma  \nonumber\\
    \text{Hence } \forall v \in \mathbb{S}^{n-1}, ~& \norm{A v - \Pi_1 B v}_2 \le \eps_1 + \gamma, 
    ~~\text{and similarly }  \norm{B v - \Pi_2 A v}_2 \le \eps_2 + \gamma. \label{eq:spectralrobust:1}
\end{align}

But $Bv= \Pi_1 Bv + \Pi_1^{\perp} Bv$. We have for any $\eta \in (0,1)$
\begin{align*}
    \norm{Bv}_2^2 &= \norm{\Pi_1 Bv }_2^2+ \norm{\Pi_1^{\perp} Bv}_2^2 \ge (\norm{Av}_2 - \eps_1 -\gamma)^2 + \norm{\Pi_1^{\perp} Bv}_2^2 \\
    &\ge (1-\eta)\norm{Av}_2^2 - (\eps_1 +\gamma)^2(1+\tfrac{1}{\eta}) + \norm{\Pi_1^{\perp} Bv}_2^2\\
    \text{Similarly, } \norm{Av}_2^2 &\ge (1-\eta)\norm{Bv}_2^2 - (\eps_2 +\gamma)^2(1+\tfrac{1}{\eta}) + \norm{\Pi_2^{\perp} Av}_2^2
\end{align*} 
Combining the two, we get that
\begin{align*}
\norm{Bv}_2^2 &\ge (1-\eta)^2 \norm{Bv}_2^2 - (1+\tfrac{1}{\eta})\Big((\eps_1 +\gamma)^2 + (\eps_1+\gamma)^2\Big) \\
&\qquad + (1-\eta) \norm{\Pi_2^{\perp} Av}_2^2 + \norm{\Pi_1^{\perp} Bv}_2^2 \\
(1-\eta) \norm{\Pi_2^{\perp} Av}_2^2 + \norm{\Pi_1^{\perp} Bv}_2^2 &\le (2\eta-\eta^2) \norm{Bv}_2^2 + (1+\tfrac{1}{\eta})\Big((\eps_1 +\gamma)^2 + (\eps_1+\gamma)^2\Big) \\
\forall v \in \mathbb{S}^{n-1}, ~~ \norm{\Pi_1^{\perp} Bv}_2^2 &\le 2\eta \norm{Bv}_2^2 + (1+\tfrac{1}{\eta})\Big((\eps_1 +\gamma)^2 + (\eps_1+\gamma)^2\Big).\\
\text{Hence, } \norm{B - \Pi_1 B}^2 &\le 2\eta \norm{B}^2 + (1+\tfrac{1}{\eta})(\eps_1 +2\gamma+\eps_2)^2,
\end{align*}
as required. A similar statement also follows for $A$ using a symmetric proof. 
\end{proof}

\begin{proof}[Proof of Theorem~\ref{thm:training:spectralnorm}]
Firstly the algorithm from Theorem~\ref{thm:worstcase:spectral} runs on $\tilde{A}$ and produced a robust projection matrix $\widehat{\Pi}$.% (with some exponentially small probability, the algorithm fails to terminate in polynomial time).
\pnote{11/27 Aravindan: check this.}
%(with some tiny probability it may output FAIL, in which case this algorithm also outputs FAIL). 
%We now condition on the event that algorithm from Theorem~\ref{thm:worstcase:spectral} terminates in polynomial time.
The proof consists of two parts. We first argue that if the algorithm outputs any robust rank-$r$ projection matrix, then it has to be robust for $A$. 
Any such $\widehat{\Pi}$ satisfies $\spnorm{\tilde{A}-\widehat{\Pi}\tilde{A}} \le \tau$. Applying Lemma~\ref{lem:training:spectralnorm} with $\eps_2=\tau$ ($B=\tilde{A}$) and $\eps_1=\spnorm{A-\Pi^* A}$, we have
$$ \spnorm{A - \widehat{\Pi} A} \le O\Big(1+\tfrac{1}{\eta}\Big) \Big( \tau+\spnorm{A-\Pi^* A}+\sqrt{m} \delta \kappa \Big) + \sqrt{2\eta}\norm{A}.$$

On the other hand, if the input $\tilde{A}$ is not ``{\sc Bad}'' i.e., (a) for the unknown matrix $A$,  $\spnorm{A-\Pi^* A} \le \tau$, and (b) $\spnorm{A-\tilde{A}} \le \tau$, we now show that the algorithm outputs a good solution for $A$. In this case we have that $\spnorm{\tilde{A} - \Pi^* A} \le 2\tau$; hence, 
$$\spnorm{\tilde{A}- \Pi^* \tilde{A}} \le \spnorm{\tilde{A} - \Pi^*A} + \spnorm{\Pi^*\tilde{A}- \Pi^* A} \le 2 \tau + \sqrt{\sum_{j \in [m]} \norm{\Pi^* A_j - \Pi^* \tilde{A}_j}_2^2} \le 2 \tau + \sqrt{m} \kappa \delta.$$
Hence, by Lemma~\ref{lem:training:spectralnorm} applied with $\eps_1=\tau$ and $\eps_2=$ ($B=\tilde{A}$), we have that   
$$ \spnorm{A - \widehat{\Pi} A} \le O\Big(1+\tfrac{1}{\eta}\Big) \Big( \tau+\sqrt{m} \delta \kappa \Big) + \sqrt{2\eta}\norm{A}.$$
This proves the theorem. The moreover part follows by setting $\tau:=\eps \spnorm{A}$.

\end{proof}

In fact, Lemma~\ref{lem:training:spectralnorm} implies a stronger information-theoretic statement about finding a robust low-rank approximation of the unknown, uncorrupted matrix $A$ with low spectral norm (just like Theorem~\ref{thm:robusttraining:frob} for Frobenius norm error). In fact we get a polynomial time algorithm assuming access to a polynomial time algorithm approximation algorithm for solving the following problem: given a matrix $\tilde{A} \in \R^{n \times m}$, find\footnote{This problem is reminiscent of the concept of $\epsilon$-rank~\citep{AlonVempala}, that corresponds to the smallest rank attainable by changing every entry of the given matrix by at most $\delta$.}
\begin{equation}
    \min_{\substack{B: \norm{B_j - \tilde{A_j}}_q \le \delta} ~ \forall j \in [m]} ~~~\min_{\substack{\Pi: \text{rank}(\Pi)=r,~ \norm{\Pi}_{q \to 2} \le \kappa}} \norm{B - \Pi B}^2,\label{eq:problem:spectral}
\end{equation}
where $\norm{\cdot}$ stands for the spectral norm.
\begin{proposition}\label{prop:training:spectral:stat}
Suppose $q \ge 2$ and $A \in \R^{n \times m}$ is the (unknown) uncorrupted data matrix, and let $\Pi^*$ have the smallest spectral norm error $\spnorm{A-\Pi A}$ among rank-$r$ projections that are robust i.e., $\norm{\Pi}_{q \to 2}\le \kappa$. Suppose further that there is an efficient algorithm for finding an $\alpha$-factor approximation algorithm for \eqref{eq:problem:spectral}.
Then there exists an algorithm that runs in polynomial time, and
given as input any (adversarially perturbed) data matrix $\tilde{A}$ s.t. for each column $j\in [m]$, $\norm{\tilde{A}_j - A_j}_q \le \delta$ and a parameter $\tau>0$, outputs a robust projection matrix $\widehat{\Pi}$ of rank at most $r$ that is near optimal in approximation error for the unknown matrix $A$ i.e., for some small universal constant $c\ge 1$, \begin{equation} \label{eq:it:training:sp}
\forall \eta>0, ~~\norm{(I-\widehat{\Pi}) A}  \le O\Big(1+\tfrac{1}{\eta}\Big) \Big( \alpha \spnorm{A-\Pi^* A}+\sqrt{m} \delta \kappa \big) + \sqrt{2\eta}\norm{A}.  \end{equation}
Moreover, the above bound is achieved information-theoretically by an algorithm (that potentially does not have polynomial running time), by using an inefficient algorithm for problem~\eqref{eq:problem:spectral}.
\end{proposition}
We remark that the main difference between the above proposition and Theorem~\ref{thm:training:spectralnorm} is that Proposition~\ref{prop:training:spectral:stat} will always output a good robust projection for $A$ (just like Theorem~\ref{thm:robusttraining:frob} for Frobenius norm error), but the algorithm is not computationally efficient unless \eqref{eq:problem:spectral} can be solved efficiently.
\begin{proof}
Given $\tilde{A}$, the algorithm first runs the $\alpha$-factor approximation algorithm for solving \eqref{eq:problem:spectral} on $\tilde{A}$. The uncorrupted matrix $A$ is itself a feasible solution; hence the solution output by the algorithm $A'$ has a robust low-rank approximation of error $O(\alpha)\norm{A - \Pi^* A}$. Such a robust low-rank projection $\widehat{\Pi}$ for $A'$ i.e., a projection for rank at most $r$ with $\norm{\widehat{\Pi}}_{q \to 2} \le O(\kappa)$ and $\norm{A' - \widehat{\Pi} A'} \le O(\alpha) \norm{A- \Pi A}$ can be found by running Theorem~\ref{thm:worstcase:spectral} on $A'$. Moreover $A'$ and $A$ are valid $2\delta$ adversarial perturbations of each other. Now applying Lemma~\ref{lem:training:spectralnorm} with $A, \Pi^*$ and $A', \widehat{\Pi}$ completes the proof. 
\end{proof}

\subsection{Lower Bound for the Additive Error in Training with Adversarial Perturbations}
\label{sec:app-lower-bound-additive-error}
We now show that the additive terms of $\Omega(m\delta^2 \kappa^2)$ in Theorem~\ref{thm:robusttraining:frob} is unavoidable.

\begin{proposition}\label{prop:additiveerror:tight}
For any data matrix $A$ with the following two properties:
\begin{enumerate}
    \item Each column $\|A_j\|_2 \in [1/10,10]$,
    \item There exists $\Pi^*$ of rank $1$ and $\|\Pi^*\|_{\infty \to 2} \ge \kappa$ (which is at least 2) satisfying $\Pi^* A=A$,
\end{enumerate}
there exists $\delta_0$ (depending on $A$) such that for any $\delta \le \delta_0$ \anote{1/2/2020: unclear what this means. do we mean $\delta<c_0$ for some small enough absolute constant $c_0>0$? }, there exist $A'$ as a $\delta$-perturbation of $A$ (i.e., $\|A-A'\|_{\infty}\le \delta$) and a projection matrix $\Pi'$ of rank 1 satisfying
\begin{enumerate}
    \item $\Pi'$ is robust  $\|\Pi'\|_{\infty \to 2} \le \|\Pi^*\|_{\infty \to 2}$. \anote{1/2/2020: Edited it from robust than $\Pi$. Why is it relevant?} 
    \item We still have $\Pi' A'=A'$ but $\|A - \Pi' A\|_F = \Omega(\delta \kappa \sqrt{m})$. Since $A - A'$ is of rank $2$, this also implies a similar lower bound for the spectral norm.
\end{enumerate}
\end{proposition}
When $A$ is a $k$-sparse flat matrix where every entry is either $0$ or $\Theta(1/\sqrt{k})$,  $\delta_0$ is as large as $\Theta(1/\sqrt{k})$.

\begin{proof}
Let $v$ be the unit eigenvector of $\Pi^*$ such that $\|v\|_1 \ge \kappa$. Without loss of generality, we assume $|v_1| \ge |v_2| \cdots \ge |v_n|$ and $\ell=\lfloor \textrm{supp}(v)/2 \rfloor$. Notice that $\textrm{supp}(v) \ge 2\ell$ by this definition such that $v_{2\ell} \neq 0$. At the same time, because of the Cauchy-Schwartz inequality, we have $\|v\|_1^2 \le \textrm{supp}(v) \cdot \|v\|_2^2$ so that $\ell \ge \kappa^2/2-1$. 

Then we set $\delta_0=|v_{2\ell}|$ and consider any $\delta$ less than it. We perturb $v$ to another ``sparser'' vector $u$ whose coordinate-wise absolute values are given by
\[
\bigg( |v_1| + \delta, \ldots, |v_{\ell}| + \delta , |v_{\ell+1}| - \delta,\ldots, |v_{2\ell}| - \delta,  |v_{2\ell+1}|,\ldots,|v_n| \bigg).
\] However, since $v_i$ may be negative or positive, we define $u$ according to the $\textrm{sign}$ function:
\[u=\bigg( v_1 + \textrm{sign}(v_1) \cdot \delta, \ldots,v_{\ell} + \textrm{sign}(v_{\ell}) \cdot \delta , v_{\ell+1} - \textrm{sign}(v_{\ell+1}) \cdot \delta,\ldots, v_{2\ell} - \textrm{sign}(v_{2\ell}) \cdot \delta,  v_{2\ell+1},\ldots,v_n \bigg).
\]
We have $
\|u\|_1 = \sum_{i=1}^{\ell} |v_i| + \delta + \sum_{i=\ell+1}^{2\ell} |v_i| - \delta + \sum_{i > 2 \ell} v_i=\|v\|_1$ and
\begin{align*}
\|u\|^2_2 & =(|v_1|+\delta)^2 + \cdots + (|v_{\ell}|+\delta)^2 + (|v_{\ell+1}|-\delta)^2 + \cdots + (|v_{2\ell}| - \delta)^2 + v_{2 \ell+1}^2 + \cdots + v_n^2\\
& =\sum_i v_i^2 + 2(\sum_{i=1}^{\ell} |v_i| - \sum_{i=\ell+1}^{2\ell} |v_i|)\delta + 2\ell \delta^2.    
\end{align*}
Since $|v_1| \ge |v_2| \ge \cdots \ge |v_n|$, this is at least $\sum_{i} v_i^2 + 2 \ell \cdot \delta^2 > \|v\|_2^2$.
So let $\overline{u}=u/\|u\|_2$ with unit $\ell_2$ norm and $\Pi'=\overline{u} \cdot  \overline{u}^{\top}$. So $\|\Pi'\|_{\infty \to 2}=\|\overline{u}\|_1 < \|v\|_1 = \|\Pi^*\|_{\infty \to 2}$.

Next we consider $A$. Since $\Pi^* A = A$, we assume $A=[c_1 \cdot v , c_2 \cdot v, \ldots, c_m \cdot v]$ with coefficient $|c_i| \le [1/10,10]$. We set $A'=[c_1 \cdot u,\ldots,c_m \cdot u]$ such that $\|A-A'\|_{\infty} \le 10 \delta$ and $\Pi' A'=A'$.

Finally we lower bound $\|A-\Pi'A\|_F^2$. Notice that 
\[
\langle u, v \rangle = \sum_{i=1}^{\ell} (v_i^2 + |v_i| \delta ) + \sum_{i=\ell+1}^{2\ell} (v_i^2 - |v_i| \delta) + \sum_{i>2\ell} v_i^2 = 1 + (\sum_{i=1}^{\ell} |v_i| - \sum_{i=\ell+1}^{2\ell} |v_i|)\delta.
\]
Thus $\Pi' v = \langle u, v \rangle u / \|u\|_2^2 = \alpha u$ for $\alpha=\frac{1 + (\sum_{i=1}^{\ell} |v_i| - \sum_{i=\ell+1}^{2\ell} |v_i|)\delta}{1 + (\sum_{i=1}^{\ell} |v_i| - \sum_{i=\ell+1}^{2\ell} |v_i|)\delta + 2\ell \delta^2}<1$. So we lower bound the distance between $v - \Pi' v$ by counting the $\ell$ entries from $v_{\ell+1}$ to $v_{2\ell}$:
\[
\sum_{i=\ell+1}^{2 \ell} [v_i - \alpha(v_i - \textrm{sign}(v_i) \delta) ]^2 = \sum_{i=\ell+1}^{2 \ell} [(1-\alpha) |v_i| + \alpha \delta]^2.
\]
Since $\delta \le |v_{2 \ell}|$, each term in the summation is at least $\delta^2$. So $\|A-\Pi'A\|_F =\sqrt{c_1^2 + \cdots + c_m^2} \cdot \delta \cdot \sqrt{\ell} = \Omega(\delta \kappa \sqrt{m})$.
\end{proof}

%section 6 commented out for COLT submission
%\input{spike_SDP.tex}
%\input{spikedcovariance.tex}
%\input{information_LB.tex}

% !TEX root = main.tex
\section{Robustness to Adversarial Perturbations in Clustering}\label{sec:applications}

\anote{Commented out a whole iffalse portion till line 310. Was very confusing to tell where the section begins :). }

As a concrete application of our guarantees for resilience to adversarial perturbations, we study the problem of clustering under adversarial perturbations. Here our goal is to approximately recover the clusters of a well defined ground truth clustering, as well as good approximations to the cluster centers.
Our main result is to apply the guarantee from Theorem~\ref{thm:training:spectralnorm} to show how to perform clustering of a well-clustered instance when every data point in the instance could be corrupted. Our guarantees will apply to clustering a mixture of well separated Gaussians and more general data distributions. In particular, we will show that a robust modification of the popular Lloyd's algorithm~\citep{lloyd1982least} (also known as the $k$-means algorithm) can be used to perform clustering in our model, thereby providing further evidence towards the widespread applicability of the algorithm. 
 Existing guarantees for using Lloyd's algorithm~\citep{kumar2010clustering, awasthi2012improved} for clustering a mixture of Gaussians and general datasets assume that every pair of means $\mu_i, \mu_j$ are separated by $\sim \sigma \sqrt{k}$, where $\sigma$ is the maximum variance of the dataset around the mean and $k$ is the number of clusters~(see (\ref{eq:spectral-stability}) for the formal condition). %Under this assumption, one can show that appropriately initialized Lloyd's updates recover the optimal ground truth clustering. 
 In the presence of adversarial perturbations of magnitude $\delta$, even if the optimal clustering (according to unperturbed ground truth) of the perturbed data is provided to us, the best we can hope for is to estimate the cluster means up to an error that goes to zero with $\delta$.

\subsection{Overview of clustering results}

%%%%%%%%%%%%%%%%%%%%%%%%%%%%%%%%%%%%%%%%%%%
%%%%%%%%%%%% Start of the portion from FOCS intro
%%%%%%%%%%%%%%%%%%%%%%%%%%%%%%%%%%%%%%%%%

Let $A \in \R^{n \times m}$ be clustered into $k$ clusters of equal sizes with 
%$C^*_1, C^*_2, \dots, C^*_k$ 
means $\mu_1, \mu_2, \dots, \mu_k$. %and define $\mu_{\max} = \max_r \|\mu_r\|$. 
Furthermore, let $C \in \R^{n \times m}$ be the matrix of corresponding centers for each column of $A$ and let $\sigma$ be such that $\|A-C\| \leq \sigma \sqrt{m}$. 
%Defining $\alpha = (1+\frac{\mu_{\max}}{\sigma})^{2/3}$, 
Then $A$ satisfies $c$-spectral stability if for each pair of optimal clusters, say, cluster $r$ and $s$ with means $\mu_r$ and $\mu_s$, any point in cluster $r$, when projected onto the line joining $\mu_r$ and $\mu_s$ is closer to $\mu_r$ than $\mu_s$ by an additive amount of $\Delta_{r,s} := c\alpha k \cdot \sigma$. 
% \begin{align}
% \label{eq:ithm-spectral-stability}
%     \Delta_{r,s} := c \alpha \sqrt{k}  \Big(\frac{\sqrt{m}}{\sqrt{|C^*_r|}} + \frac{\sqrt{m}}{\sqrt{|C^*_s|}}\Big) \sigma ,
% \end{align}
Here $\alpha$ is a quantity that captures the signal-to-noise ratio and the relative perturbation magnitude.\footnote{We show that unlike standard clustering, the dependence on $\alpha$ is unavoidable with corruptions even for $k=1$ (mean estimation).}\xnote{4/15: add a period here. Feel strange to put the footnote before the period.} %i.e., robust mean estimation.
% We would like to point out that when clustering a stable instances as defined above under no corruptions, the dependence on $\alpha$ is not needed~\cite{kumar2010clustering}. However, in our model of corruptions, the lower bound in Theorem~\ref{thm:robust-mean-lower-bound} shows that some dependence on $\alpha$ is unavoidable for the clustering task to be meaningful.}. 
When $A$ is a set of $m = \poly(n,k)$ points\anote{4/11: change $1/w_{min}$ to $k$} drawn i.i.d. from a mixture of Gaussians with the variance of each Gaussian being bounded by $\sigma^2$, and with uniform mixture weight $1/k$ each, the separation condition becomes
%\begin{align}\label{eq:ithm-spectral-stability-gaussians}
  $  \Delta_{r,s} = c \alpha k \cdot \text{polylog}(nk) \cdot \sigma$.
%\end{align}
Below we denote $\kappa$ to be the robustness, as measured in $\|\cdot \|_{q \to 2}$, of the subspace spanned by the true means $\{\mu_1, \mu_2, \dots, \mu_k\}$. 
%We prove the following guarantee on our modified Lloyd's updates. 
\begin{itheorem}
\label{ithm:clustering-application-general}[Robust Clustering]
Fix $q \geq 2$, and let $c_q$ be a constant that depends on $q$. Let $A \in \R^{n \times m}$ satisfy $c$-spectral stability, for $c > 200c_q$. Then given as input a $\delta$-corrupted instance $\tilde{A}$ of $A$, 
%such that $\kappa \delta = O(\sigma)$, 
there is a Lloyd's style algorithm that either certifies that the dataset is poisoned, i.e, $\|A - \tilde{A}\| = \Omega(\sigma \sqrt{m})$, or recovers each mean $\mu_r$ up to error $O(\alpha \sqrt{k} \sigma )$. 
Using the computed centers to cluster $\tilde{A}$, we obtain a clustering of $\tilde{A}$ such that the corresponding induced clustering on $A$ that misclassifies $O(1/k)$- fraction of the points.
%close to the optimal clustering of $A$, i.e., the one using the true means $\mu_1, \dots, \mu_k$. 
    
In the special case of a mixture of Gaussians with equal mixing weights we recover the means upto error $\tilde{O}(\alpha \sigma)$,
% that
% %\begin{align*}
%     $\|\mu_r - \hat{\mu}_{\pi(r)}\| \leq \tilde{O}(\alpha) \sigma$,
% %\end{align*}
where we hide a $\text{polylog}(m,n)$ factor in the $\tilde{O}$ notation. This implies $O(1/k^2)$-fraction clustering error.
% Using the cocenters $\hat{\mu}_1, \dots, \hat{\mu}_k$ to cluster $\tilde{A}$, will induce a clustering on $A$ that is $O(1/k^2)$-close to the optimal clustering of $A$.
\end{itheorem}

See Theorem~\ref{thm:clustering-application-general} and Theorem~\ref{thm:clustering-application-gaussians} for formal statements that also handle more general cluster sizes and mixing weights. 
Finally, as in Section~\ref{sec:results-training-errors} we can also prove that there is an algorithm (though computationally inefficient) that 
% performs robust mean estimation up to the (near optimal) error of $O((1+\frac{\kappa \delta}{\sigma})\max(\sigma, \sqrt{\sigma \|\mu\|}))$ on all instances (it will estimate the uncorrupted mean even for instances that are certified to have been poisoned). See Theorem~\ref{thm:robust-mean-stat} for the formal statement. As a result, we 
can cluster well-clustered instances up to the claimed error above, without the need for certification. Whether this can be achieved in polynomial time is an open question.

Our analysis proceeds in three stages: a) an initialization stage, b) a center improvement stage, and c) analyzing the robust Lloyd's updates. Each stage poses unique challenges arising from working with $\tilde{A}$ where each data point is potentially corrupted. The standard way to initialize Lloyd's algorithm via PCA\footnote{This initialization is needed for theoretical bounds. In practice, the initial centers are chosen as random data points.} can be arbitrarily bad when every data point is corrupted. Using our algorithm for spectral norm error from Section~\ref{sec:results-training-errors} we instead project the data onto a robust $k$-dimensional subspace $\Pi$ for $\tilde{A}$ with small error, or certify that the dataset has been poisoned. This then lets us compute initial centers that are $O(\alpha k \sigma)$-close to the true means.

In the second stage we improve the initial center estimates by a $\sqrt{k}$ factor. Our main technical contribution here is to establish a stronger version of the statements that appear in~\citet{kumar2010clustering, awasthi2012improved}~(see Lemma~\ref{lem:cluster-purity}). This lemma simultaneously helps us argue about the clustering error, and also the variance of each current cluster around its mean, a quantity crucial to bound in order to analyze the iterative updates later. The first two stages together help us establish the guarantee for general well-clustered instances.

To establish the stronger guarantee for mixtures of Gaussians we first analyze the ``ideal'' iterative updates, as if we had access to the uncorrupted data. This largely follows the analysis in~\citet{kumar2010clustering} and helps us argue that if the current center estimates are $\beta \alpha \sqrt{k} \sigma$ close to the corresponding means (where $\beta < 1$), then in the next step the ideal updates give estimates that are $\frac{\beta}{4} \alpha \sqrt{k} \sigma$ close. A key technical step is to show that when performing ideal updates, the variance of the formed clusters around their means is bounded even though the clusters themselves are impure! Using the bounded variance property, we next analyze the actual updates and use a specialized robust mean estimation procedure~(Lemma~\ref{thm:robust-mean}) to 
%argue that when given perturbed set as input, the procedure will either certify that the set is poisoned, or will output 
get an estimate that is
% within $\tilde{O}(\alpha \sigma)  + \frac{\beta}{4} \alpha \sqrt{k} \sigma$ of the mean of the ideal updates. This in turn means that the new estimate output by the robust mean estimation procedure will be 
within $\tilde{O}(\alpha \sigma) + \frac{\beta}{2} \alpha \sqrt{k} \sigma$ of the true mean $\mu_r$. Hence, the updates will keep improving until the unavoidable error of $\tilde{O}(\alpha \sigma)$.

%%%%%%%%%%%%%%%%%%%%%%%%%%
%%% End of portion copied from FOCS subm
%%%%%%%%%%%%%%%%%%%%%%%%%
%%%%%%%%%%%%%%%%%%%%

\paragraph{Guarantees for mean estimation.} For purpose of recovering the means we will crucially rely on a subroutine to robustly estimate the mean of a cluster. We sketch below this procedure and state the associated guarantee.

\begin{figure}[htbp]
\begin{center}
\fbox{\parbox{0.98\textwidth}{
{\bf \RMean($\tilde{A}, \kappa, \sigma^2$)}

{\bf Input:} The corrupted data matrix $\tilde{A} \in \R^{n \times m}$  with columns $\tilde{A}_i$ for $i \in [m]$, and the upper bound $\kappa$ on the robustness of the subspace $\|\Pi^*\|_{q \to 2}$.
\begin{enumerate}
%\item Let $\tilde{A}$ be the corrupted $n \times m$ data matrix with columns $A_i$ for $i \in [m]$. Let $\kappa$ be the given input value representing the bound on $\|\Pi^*\|_{q \to 2}$.
\item Run the algorithm from Figure~\ref{ALG:training:spectral} with $\tau = 2\sigma \sqrt{m}, r=1$, and $\kappa$.%with $\kappa$ equal to the given input value. 
\item If the algorithm outputs \Bad then terminate and certify that the data has been poisoned. Otherwise let $\Pi$ be the $1$-dimensional subspace output by the algorithm. Return $\hat{\mu} = \Mean(\Pi \tilde{A})$.
\end{enumerate}
}}
\end{center}
\caption{\label{ALG:Robust-Mean} Robust Mean Estimation.}
\end{figure}
\begin{lemma}
\label{thm:robust-mean}
Let $A$ be an $n \times m$ matrix representing $m$ data points in $n$ dimensions and let $\mu$ be a vector such that $\norm{\Mean(A)-\mu}_2 \leq \eta$. Let $C$ be the $n \times m$ matrix with each column being $\mu$. Let $\Pi^* = \mu \mu^\top/\|\mu\|^2$ be the one dimensional subspace denoting the projection onto $\mu$ and assume that $\|\Pi^*\|_{q \to 2} \leq \kappa$, for some $q \geq 2$. Let $\tilde{A}$ be the given input such that for every column $j \in [m]$ we have $\|A_j-\tilde{A}_j\|_q \leq \delta$. Furthermore, let $\sigma^2 >0$ be a given upper bound on the variance of the data around $\mu$, i.e., $\|A-C\| \leq \sigma \sqrt{m}$. Then
%if $\kappa \delta = O(\sigma)$, 
the algorithm from Figure~\ref{ALG:Robust-Mean} when run on $\tilde{A}$, runs in polynomial time, and either certifies that the data has been poisoned, i.e., $\|\tilde{A} - A\| = \Omega(\sigma \sqrt{m})$, or outputs an estimate $\hat{\mu}$ of the true mean $\mu$ such that
% \begin{align*}
%     \|\hat{\mu} - {\mu}\|_2 &\leq O(c_q ) \max \Big(\sigma, \sqrt{\sigma \|\mu\|} \Big),
% \end{align*}
\begin{align*}
    \|\hat{\mu} - {\mu}\|_2 &\leq \eta + O(c_q)(1+\frac{\kappa \delta}{\sigma}) \max \Big(\sigma, \sqrt{\sigma \|\mu\|} \Big),
\end{align*}
where $c_q$ is a constant that depends on $q$. In particular, the above implies a relative error guarantee of
\begin{align*}
    \frac{\|\hat{\mu} - {\mu}\|_2}{\|\mu\|} &\leq \eta + O(c_q)(1+\frac{\kappa \delta}{\sigma}) \max \Big(\frac{\sigma}{\|\mu\|}, \sqrt{\frac{\sigma}{\|\mu\|}} \Big).
\end{align*}
% In general, for any $\kappa, \delta$ we get a guarantee of
% \begin{align*}
%     \frac{\|\hat{\mu} - {\mu}\|_2}{\|\mu\|} &\leq O(c_q) \max \left(\frac{\sigma + \kappa \delta}{\|\mu\|}, \sqrt{\frac{\sigma + \kappa \delta}{\|\mu\|}} \right).
% \end{align*}
\end{lemma}
We provide the proof of the above lemma in Appendix~\ref{app:robust-mean}. In the appendix we also provide a matching lower bound stating that in general the above bound on estimation error cannot be improved. See also subsequent work \citep{ACV20coltsub} for other algorithms for robust mean estimation, without the need for certification. However, the associated guarantees in \citet{ACV20coltsub} are incomparable and typically have a multiplicative dependence on $\eta$ which is not desirable for the clustering application. In the above lemma we make use of the fact that the data has a small projection onto $\Pi$ to get a stronger additive guarantee, or certify that the data has been poisoned.

\paragraph{Guarantees for $k$-means clustering.}
 
From the above discussion, in the context of clustering, even if one is given the original optimal clustering of the given perturbed dataset, we must incur a loss of $\Omega(\sigma \cdot \max(1, \sqrt{\|\mu\|/\sigma}))$ in simply estimating the true mean $\mu$ of a cluster. This suggests a separation condition of the type $\sim \alpha \sigma \sqrt{k}$, where $\alpha$ depends on $(1+{\frac{\mu_{\max}}{\sigma}})$ and the guarantee to aim for is to estimate means upto error $O(\alpha \sigma)$ error. Here $\mu_{\max}$ is the maximum $\ell_2$ norm of the any of the $k$ mean vectors. In this section we will show that a modified Lloyd's combined with our certification procedure can indeed achieve this guarantee or certify that the dataset has been poisoned.

More formally, we will assume that there is a set of $m$ points in $\mathbb{R}^n$ with ground truth clustering $C^*_1, C^*_2, \dots, C^*_k$, and means $\mu_r = \Mean(C^*_r)$ for $r \in [k]$ and $\mu_{\max} = \max_r \|\mu_r\|$. Let $A$ be the $n \times m$ data matrix and $C$ be the matrix of corresponding centers. We will assume that we have an upper bound $\sigma^2$ on the maximum variance of the data points around their mean, i.e. $\|A-C\|^2 \leq \sigma^2 m$ and define $\alpha = (1+\frac{\kappa \delta}{\sigma})(1+{\frac{\mu_{\max}}{\sigma}})^{2/3}$. We will enforce the spectral stability condition studied in~\citet{kumar2010clustering} on our clustering instance. This condition implies that for each pair of clusters $C^*_r,C^*_s$ with means $\mu_r, \mu_s$ and each point $x \in C^*_r$, $\bar{x}$ is closer to $\mu_r$ than to $\mu_s$ by a margin $\Delta_{r,s}$. Here $\bar{x}$ is the projection of $x$ onto the line joining $\mu_r$ and $\mu_s$. For a constant $c > 0$, the $c$-spectral stability condition requires that for each $r \neq s$,
\begin{align}
    \label{eq:spectral-stability}
    \Delta_{r,s} \geq c \alpha \sigma \sqrt{k} { }\Big(\frac{\sqrt{m}}{\sqrt{|C^*_r|}}+\frac{\sqrt{m}}{\sqrt{|C^*_s|}}\Big)
\end{align}
Notice that the above also implies that every pair of means are separated i.e.,
$$
\|\mu_r - \mu_s\| \geq c \alpha \sigma \sqrt{k} { }\Big(\frac{\sqrt{m}}{\sqrt{|C^*_r|}}+\frac{\sqrt{m}}{\sqrt{|C^*_s|}}\Big).
$$
It is worth mentioning that in the typical analysis of Lloyd's algorithm~\citep{kumar2010clustering, awasthi2012improved} the dependence on $\alpha$ in the separation condition is not needed. However, as discussed before, in our noise model, some dependence on $\alpha$ is unavoidable to get a meaningful clustering guarantee. 

\noindent \textbf{Assumptions I:} Fix $q \geq 2$. We will assume that we are given access to $\tilde{A}$ such that for every $j \in [m]$, $\|A_j-\tilde{A}_j\|_q \leq \delta$. Furthermore, define $\kappa$ to be the robustness, of the subspace spanned by the means $\mu_1, \dots, \mu_k$. Formally, let $\Pi_C$ be the projection matrix for the orthogonal projection onto the space of the means. Then $\kappa$ is such that $\|\Pi_C\|_{q \to 2} \leq \kappa$.
%, and as in Section~\ref{sec:application-mean-estimation} we will assume that $\kappa \delta = O(\sigma)$. 
Under Assumptions I, we prove the following theorem that applies to any stable dataset as defined in (\ref{eq:spectral-stability}).
\begin{theorem}
\label{thm:clustering-application-general}
Fix $q \geq 2$, and let $c_q$ be a constant that depends on $q$. Let $A$ be a dataset that satisfies $c$-spectral stability for $c > 200c_q$. Under Assumptions I, there is a Lloyd's style algorithm that takes $\tilde{A}$ as input,  runs in polynomial time, and either certifies that the dataset is poisoned, i.e, $\|A - \tilde{A}\| = \Omega(\sigma \sqrt{m})$, or outputs a clustering $\hat{C}_1, \hat{C_2}, \dots, \hat{C_k}$ and means $\hat{\mu_1}, \hat{\mu_2}, \dots, \hat{\mu_k}$ such that
\begin{align*}
    \sum_{r=1}^k |C^*_r \triangle \hat{C}_{\pi(r)}| \leq O \Big(\frac{c^2_q m}{k \alpha^2 c^2} \Big)\\
    \|\mu_r - \hat{\mu}_{\pi(r)}\| \leq c_q \alpha \frac{\sigma \sqrt{m}}{\sqrt{|C^*_r|}}.
\end{align*}
for an appropriately chosen bijection $\pi: [k] \to [k]$. 
\end{theorem}
While the above theorem works for any data set that satisfies spectral stability, notice that it leads to a sub optimal mean estimation error of $O(\alpha \sigma \sqrt{m}/\sqrt{|C^*_r|})$ for each cluster $r$. For example, when the clusters are balanced, this will lead to a guarantee of $O(\alpha \sigma) \sqrt{k}$. Next, we show that for data sets that additionally satisfy Gaussian type concentration, we can indeed get $O(\alpha \sigma)$ estimation error even when each data point is corrupted. 

\noindent \textbf{Assumptions II:} Let $A$ be a given dataset with optimal clustering $C^*_1, C^*_2, \dots, C^*_k$. We will assume that we are given $\tilde{A}$ that satisfies Assumptions I.
% but with separation that is a $\sqrt{k}$ factor more than what is required in Theorem~\ref{thm:clustering-application-general}, i.e., $c > 200c_q \sqrt{k}$ . 
Furthermore, we will assume that $\|C^*_r\| \geq n^3$ for each $r \in [k]$ and that for any subset $S \subset C^*_r$ of points such that $|S| > n \log n$, we have that $\|A_S - C_S\| \leq \sigma \sqrt{|S|}\cdot \poly\log(m,n)$. Here $A_S, C_S$ are the matrices $A$ and $C$ restricted to the columns of the points in $S$. Additionally, we require a pointwise guarantee that for each $r \in [k]$, and $A_i \in C^*_r$, $\|A_i - \mu_r\|^2 \leq 2\sigma^2 n \cdot \poly\log(m,n)$.
It is easy to see that $m \geq \poly(m, 1/w_{\min})$ samples generated from a mixture of Gaussians with maximum variance $\sigma^2$ and minimum mixture weight $w_{\min}$ will, with high probability, satisfy the above assumptions.   
Under Assumptions II, we prove the following theorem that applies to any stable dataset as defined in (\ref{eq:spectral-stability}).
\begin{theorem}
\label{thm:clustering-application-special}
Fix $q \geq 2$, and let $c_q$ be a constant that depends on $q$. Let $A$ be a dataset that satisfies $c$-spectral stability for $c > 200c_q$. Under Assumptions II, there is a Lloyd's style algorithm that takes $\tilde{A}$ as input, runs in polynomial time, and either certifies that the dataset is poisoned, i.e, $\|A - \tilde{A}\| = \Omega(\sigma \sqrt{m})$, or outputs a clustering $\hat{C}_1, \hat{C_2}, \dots, \hat{C_k}$ and means $\hat{\mu_1}, \hat{\mu_2}, \dots, \hat{\mu_k}$ such that
\begin{align*}
    \sum_{r=1}^k |C^*_r \triangle \hat{C}_{\pi(r)}| \leq O \Big(\frac{c^2_q m}{k^2 \alpha^2  c^2}\Big)\\
    \|\mu_r - \hat{\mu}_{\pi(r)}\| \leq \tilde{O}(\alpha \sigma).
\end{align*}
for an appropriately chosen bijection $\pi: [k] \to [k]$, where we hide a polylogarithmic~(in $n,m$)  factor in the $\tilde{O}$ notation. 
\end{theorem}
As a corollary we get the following statement about robustly clustering a mixture of Gaussians.
\pnote{change $\|\mu_{\max}\|$ to $\mu_{\max}$ everywhere.}
\pnote{Done.}
\begin{theorem}
\label{thm:clustering-application-gaussians}
Fix $q \geq 2$, and let $c_q$ be a constant that depends on $q$. Define $\mathcal{M}$ to be a distribution that is a mixture of $k$ Gaussians, i.e., $\mathcal{M} := \sum_{r=1}^k w_r \mathcal{N}(\mu_r, \Sigma_r)$. Furthermore, let $\Sigma_r \preceq \sigma^2 I$ and define $w_{\min} = \min_r w_r$ and $\mu_{\max} = \max_r \|\mu_r\|$, $\alpha = (1+\frac{\kappa \delta}{\sigma})(1+\frac{\mu_{\max}}{\sigma})^{2/3}$. Let $A$ be a set $\poly(n, 1/w_{\min})$ samples generated i.i.d. from the mixture. If the mixture if well separated, i.e, $\|\mu_r - \mu_s\| \geq c \alpha \sigma \sqrt{k} \cdot \poly\log(n/w_{min})/\sqrt{w_{\min}}$ for $c > 200 c_q$, and the means span a $\kappa$ robust subspace, then given access to $\tilde{A}$ such that $\|\tilde{A}_j - A_j\|_q \leq \delta$,
%and $\kappa \delta = O(\sigma)$, 
there is a Lloyd's style algorithm that, runs in polynomial time, and either certifies that the data is poisoned, i.e., $\|A - \tilde{A}\| \geq \Omega(\sigma \sqrt{m})$, or outputs a clustering $\hat{C}_1, \hat{C_2}, \dots, \hat{C_k}$ and means $\hat{\mu_1}, \hat{\mu_2}, \dots, \hat{\mu_k}$ such that
\begin{align*}
    \sum_{r=1}^k |C^*_r \triangle \hat{C}_{\pi(r)}| \leq O \Big(\frac{c^2_q m}{k^2 \alpha^2 c^2}\Big)\\
    \|\mu_r - \hat{\mu}_{\pi(r)}\| \leq \tilde{O}(\alpha \sigma).
\end{align*}
for an appropriately chosen bijection $\pi: [k] \to [k]$. 
\end{theorem}
% Let $A$ be a set of dataset that satisfies $c$-spectral stability for $c > 200c_q$. Under Assumptions I, there is a Lloyd's style algorithm that takes $\tilde{A}$ as input and either certifies that the dataset is poisoned, i.e, $\|A - \tilde{A}\| = \Omega(\sigma \sqrt{m})$, or outputs a clustering $\hat{C}_1, \hat{C_2}, \dots, \hat{C_k}$ and means $\hat{\mu_1}, \hat{\mu_2}, \dots, \hat{\mu_k}$ such that
% \begin{align*}
%     \sum_{r=1}^k |C_r \triangle \hat{C}_{\pi(r)}| \leq \frac{m}{k^2 c^2}\\
%     \|\mu_r - \hat{\mu}_{\pi(r)}\| \leq O(\tau \sigma).
% \end{align*}
% for an appropriately chosen bijection $\pi: [k] \to [k]$. 

% Our goal will be to recover a clustering that is good w.r.t. the ground truth clustering of $A$. We will show that using our robust projection guarantee, along with an off the shelf approximation algorithm for $k$-means clustering will output a clustering that is $O(\frac 1 {c^2})$-close to $C$. In other words, our algorithm will output a clustering $\tilde{C}$ of $\tilde{A}$ such that the corresponding induced clustering $\hat{C}$ on $A$ satisfies
% \begin{align*}
%     \sum_{i=1}^k |C_i \triangle \hat{C}_{\sigma(i)}| \leq \frac{m}{c^2}
% \end{align*}
\noindent \textbf{Computing Good Initial Centers.} The first step in establishing the above theorems is to compute centers/means that are close to the true ones. A common approach for this step is to use PCA to project the data onto the top-$k$ subspace of the input data matrix, and run any constant factor approximation algorithm for $k$-means clustering~\citep{kumar2010clustering}. However this can be arbitrarily bad if the data is corrupted as in our model. We instead show that by projecting the data onto a robust subspace as output by our guarantee from Theorem~\ref{thm:training:spectralnorm} and then using a $k$-means approximation algorithm, we do indeed recover good centers. Our initialization algorithm is shown in Figure~\ref{ALG:Initialization}. We next provide proofs for our claims. 
%For simplicity of presentation, we will assume that $\kappa \delta \leq \sigma$. All the arguments can be changed to account for $\kappa \delta = O(\sigma)$ at the expense of constant factors.
\begin{figure}[htbp]
\begin{center}
\fbox{\parbox{0.98\textwidth}{
{\bf Input:} The corrupted data matrix $\tilde{A} \in \R^{n \times m}$  with columns $\tilde{A}_i$ for $i \in [m]$, upper bound $\kappa$ on the robustness of the subspace $\|\Pi_C\|_{q \to 2}$, upper bound $\sigma^2$ on the data variance $\|A-C\|^2/m$.
\begin{enumerate}
\item Run the algorithm from Figure~\ref{ALG:training:spectral} with $\tau = 2\sigma \sqrt{m}, r=k$, and $\kappa$.
\item If the algorithm outputs \Bad then certify that the data has been poisoned and terminate.
\item If the algorithm outputs a projection matrix $\Pi$ then project $\tilde{A}$ onto $\Pi$. Denote $\hat{A} = \Pi \tilde{A}$ as the projected matrix.
\item Run a $10$-approximation algorithm for $k$-means clustering on $\hat{A}$, and obtain $k$ centers $\nu_1, \nu_2, \ldots, \nu_k$.
\item Output $\nu_1, \nu_2, \dots, \nu_k$.
\end{enumerate}
}}
\end{center}
\caption{\label{ALG:Initialization} Computing initial center estimates.}
\end{figure}

\anote{Should we write down the robust mean estimation part separately?}
\pnote{Done.}
\begin{theorem}
\label{thm:lloyds-initialization-guarantee}
Assume that the clustering instance $A$ is $c$-stable for $c > 200 c_q$. If Assumptions I hold, then the algorithm in Figure~\ref{ALG:Initialization} runs in polynomial time, and either certifies that the data has been poisoned, i.e., $\|A-\tilde{A}\| > 2 \sigma \sqrt{m}$, or the algorithm outputs centers $\nu_1, \nu_2, \dots, \nu_k$ such that for all $r \in [k]$,
\begin{align*}
    \|\mu_r - \nu_{\pi(r)}\| \leq 30 c_q \alpha \sqrt{k}\frac{\sigma \sqrt{m}}{\sqrt{|C^*_r|}}.
\end{align*}
for an appropriately chosen bijection $\pi$.
\end{theorem}
%\xnote{$C^*_r$ shall be $C^*_i$ in the equation of the statement?}
%\pnote{Done.}
\begin{proof}
The proof will follow the general outline of Lemma 5.1 of~\citet{kumar2010clustering}, except that we need to argue following two stronger conditions. Firstly, we need to establish that $\tilde{A}$ projected on to $\Pi$ has cost comparable to that of $k \sigma^2 m$. This will ensure that the approximation algorithm will output a clustering of low cost. Secondly, we also simultaneously need to establish that $\tilde{A}$ when projected on to $\Pi$ has low cost clustering when true means $C$ are used to cluster it. 
Together with the fact that $\tilde{A}$ and $A$ are pointwise close in the projected space, we can then claim that missing out on a good approximation for even a single cluster center of $C^*$ will incur a cost of $\Omega(k \sigma^2 m)$, thereby contradicting the approximation guarantee of the $k$-means algorithm used in step 2. 

Establishing that $\Pi \tilde{A}$ has low cost with respect to $C$ boils down to showing that $\Pi$ is good for $A$ given that it is good for $\tilde{A}$, a perturbation of $A$. This statement, established in Theorem~\ref{thm:training:spectralnorm} is the key in analyzing the initialization phase, and is a generalization of Lemma~\ref{lem:spectral-norm-closeness-of-subspaces} to higher dimensional subspaces. 
%The robust projection guarantee of Theorem~\ref{thm:training:spectralnorm} is key to establishing these statements.  
Let's first establish that $\tilde{A}$ projected on to $\Pi$ has a low cost clustering. 
We have
\begin{align}
    \|\Pi \tilde{A} - \Pi C\|_F &\leq \sqrt{3k}\|\Pi \tilde{A} - \Pi C\| \,\, (\text{since both $\Pi \tilde{A}$ and $\Pi C$ have rank at most $k$}) \nonumber \\
    &\leq \sqrt{3k}\Big(\|\Pi (\tilde{A}-A)\| + \|\Pi (A-C)\| \Big)\nonumber \\
    &\leq \sqrt{3k}(c_q \delta \kappa \sqrt{m} + \|A-C\|) \nonumber \\
    &\leq 2\sqrt{3}c_q \sqrt{k}\sigma \sqrt{m} \leq c_q(1+\frac{\kappa \delta}{\sigma}) \sqrt{12k} \sigma \sqrt{m}. \label{eq:initialization-proof-part-1}
\end{align}
Here the third inequality follows from the fact that for any $n \times m$ matrix $M$, $\|M\| \leq \ell_{\max} \sqrt{m}$, where $\ell_{\max}$ is the maximum $\ell_2$ norm of a column of $M$. Furthermore, from the robustness of $\Pi$ we know that for any $j \in [m]$,
$$
\|\Pi(\tilde{A}_j-A_j)\| \leq c_q \kappa \delta. 
$$
Next, let's establish that $\|A-\Pi A\|$ is small. By triangle inequality we know that
\begin{align*}
    \|A- \Pi_C A\| &\leq \|A-C\| + \|C-\Pi_C A\|\\
    &= \|A-C\| + \|\Pi_C(C-A)\|\\
    &\leq 2\|A-C\| \leq 2\sigma \sqrt{m}.
\end{align*}
Furthermore, from the guarantee of Theorem~\ref{thm:training:spectralnorm} we have that for any $\eta \in (0,1)$,
\begin{align*}
    \|A-\Pi A\| &\leq O(1+\frac{1}{\eta}) \Big(2\sigma \sqrt{m} + \|A-\Pi_C A\| + \kappa \delta \sqrt{m} \Big) + \sqrt{2\eta}\|A\|\\
    &\leq O(1+\frac{1}{\eta}) \Big(4\sigma + \kappa \delta \Big)\sqrt{m} + \sqrt{2\eta}\|A\|.
\end{align*}
Setting $\eta = (5\sigma \sqrt{m}/\|A\|)^{2/3}$ we get that
\begin{align}
    \|A - \Pi A\| &\leq 4 c_q(1+\frac{\kappa \delta}{\sigma}) \sigma\Big(1+ \Big(\frac{\|A\|}{\sigma \sqrt{m}} \Big)^{2/3} \Big)\sqrt{m} \nonumber \\
    &\leq {4 c_q \alpha \sigma \sqrt{m}} \label{eq:initialization-proof-part-2}.
\end{align}
The last inequality above follows from the fact that 
\begin{align*}
    \|A\| &\leq \|A-C\| + \|C\|\\
    &\leq \sigma \sqrt{m} + \mu_{\max} \sqrt{m}.
\end{align*}
% Hence, from the guarantee of Theorem~\ref{thm:training:spectralnorm} we have
% \begin{align}
%     \|A - A\Pi\| & \leq O(\|A-C\|) + O(\delta \kappa \sqrt{n}) \nonumber \\
%     &\leq 4 \|A-C\|
%     \label{eq:initialization-proof-part-2}.
% \end{align}
Next notice that
\begin{align}
    \|\Pi \tilde{A}-C\|_F &\leq \sqrt{3k}\|\Pi \tilde{A}-C\| \,\, (\text{since both $\Pi \tilde{A}$ and $C$ have rank at most $k$}) \nonumber\\
    &\leq \sqrt{3k}\Big(\|\Pi(\tilde{A}-A)\| + \|\Pi A-A\| + \|A-C\|)\nonumber \\
    &\leq \sqrt{3k}(c_q \delta \kappa \sqrt{m} + 5{c_q \alpha \sigma \sqrt{m}} + \sigma \sqrt{m})\nonumber \\
    &\leq 6\sqrt{3k}c_q \alpha \sigma \sqrt{m} \label{eq:bound-tildeA-pi-from-C}.
\end{align}
Now we are ready to establish the claim of the theorem. From~(\ref{eq:initialization-proof-part-1}) we get that the centers $\nu_1, \nu_2, \dots, \nu_k$ will have $k$-means cost at most $120k (1+\frac{\kappa \delta}{\sigma})^2 c^2_q \sigma^2 m$ on $\tilde{A}$. Furthermore, suppose that there exists a center $\mu_r$ such that every $\nu_s$ is far from it. For any point $A_i$, let $\nu_{c(i)}$ be the center in the set $\{\nu_1, \nu_2, \dots, \nu_k\}$ that is closest to the projection of $\tilde{A}_i$ on to $\Pi$. Then we have that
\begin{align}
120kc^2_q (1+\frac{\kappa \delta}{\sigma})^2 \sigma^2 m \geq   \sum_{A_i \in C_r} \|\Pi \tilde{A}_i - \nu_{c(i)}\|^2 &= \sum_{A_i \in C_r} \|\Pi \tilde{A}_i-\mu_r +\mu_r - \nu_{c(i)}\|^2 \nonumber\\
    &\geq \frac 1 2 |C_r|\|\mu_r - \nu_{c(i)}\|^2 - \sum_{A_i \in C_r}\|\Pi \tilde{A}_i-\mu_r\|^2 \nonumber\\
    &\geq \frac 1 2 |C_r|\|\mu_r - \nu_{c(i)}\|^2 - \|\Pi \tilde{A}-C\|^2_F \nonumber\\
    &\geq 450k \alpha^2 c^2_q \sigma^2 m - \|\Pi \tilde{A}-C\|^2_F \nonumber\\
    &> 120 k c^2_q \alpha^2 \sigma^2 m.
    \label{eq:initialization-part-3}
\end{align}
%\xnote{Not sure where is the $\alpha^2$ comes from? The term in THM 6.6 does not have any dependency on $\alpha$.}
%\pnote{Done.}
%Plugging into~(\ref{eq:initialization-part-3}) we get that
% \begin{align*}
%     120kc^2_q \sigma^2 m \geq   \sum_{A_i \in C_r} \|\Pi \tilde{A}_i - \nu_{c(i)}\|^2 &\geq 900k c^2_q \alpha^2 \sigma^2 m - \|\Pi \tilde{A}-C\|^2_F\\
%     &\geq 888kc^2_q \alpha^2 \sigma^2 m.
% \end{align*}
Noticing that $\alpha \geq  (1+\frac{\kappa \delta}{\sigma})$, we get a contradiction to the fact that $\mu_r$ is far from every $\nu_s$. This combined with the fact that the clustering instance is $c$-stable for $c>200c_q$ implies that one can find a bijection $\pi:[k] \mapsto [k]$ between $\{\mu_1, \dots, \mu_k\}$ and $\{\nu_1, \dots, \nu_k\}$ such that each $\mu_i$ is close to a unique $\nu_{\pi(i)}$.
\end{proof}
\subsection{Analyzing Lloyd's Updates}
Next we will use the obtained initial centers and run the robust Lloyd's algorithm starting with these centers as shown in Figure~\ref{ALG:Lloyds}. Our goal in this section is to analyze the updates and establish Theorem~\ref{thm:clustering-application-general} and Theorem~\ref{thm:clustering-application-special}. 
% As in~\cite{awasthi2012improved} a minor modification we make is introducing step $3$ in the algorithm where we first refine the initial center estimates by only updating the means on points that are much closer to the current means. We only do this once and from then on run the regular Lloyd's updates in the projected space.
% In fact we will be able to show that steps $1$-$3$ of the algorithm are enough to establish Theorem~\ref{thm:clustering-application-general}.

% In this section we will show that by carefully analyzing the popular Lloyd's algorithm updates, when initialized with the centers obtained from the algorithm in Figure~\ref{ALG:Initialization}, we can in fact get center estimates that are $O(\frac{\|A-C\|}{\sqrt{n}})$-close to the true means. This will also let us argue that even in the presence of corruptions, we can get a clustering that is $o(1)$-close to the optimal clustering of the uncorrupted data set $A$! In particular, this will imply that for a mixture of well separated Gaussians with roughly equal weights and maximum variance $\sigma^2$ in any direction, we can either certify that the data is corrupted, or estimate the true means up to $O(\sigma)$ error in the $\ell_2$ norm~(See Corollary~\ref{}). 
\begin{figure}[htbp]
\begin{center}
\fbox{\parbox{0.98\textwidth}{
{\bf Input:} The corrupted data matrix $\tilde{A} \in \R^{n \times m}$  with columns $\tilde{A}_i$ for $i \in [m]$, upper bound $\kappa$ on the robustness of the subspace $\|\Pi_C\|_{q \to 2}$, upper bound $\sigma^2$ on the data variance $\|A-C\|^2/m$.
\begin{enumerate}
\item Let $\Pi$ be the robust $k$ dimensional projection matrix as computed by the initialization algorithm in Figure~\ref{ALG:Initialization}. 
\item 
%Let $\hat{A} = \Pi \tilde{A}$ and let $\hat{A}_i$ to be the $i$th column of $\hat{A}$. 
Define initial center estimates $\nu^{(0)}_1, \nu^{(0)}_2, \dots, \nu^{(0)}_k$ to be the centers output by the algorithm in Figure~\ref{ALG:Initialization}.
\item For each $r \in [k]$, define $S_r = \{\Pi\tilde{A}_i: \|\Pi \tilde{A}_i - \nu_r\| \leq \frac{\|\Pi \tilde{A}_i - \nu_s\|}{3}, \forall s \neq r\}$. Update $\nu^{(1)}_r = \Mean(S_r)$.
\item For $t=2, 3, \dots$ do:
\begin{enumerate}
\item For each $r \in [k]$, compute $S_r = \{\tilde{A}_i: \|\Pi \tilde{A}_i - \nu^{(t-1)}_r\| \leq \|\Pi \tilde{A}_i - \nu^{(t-1)}_s\|, \forall s \neq r \}$.
\item For each $r \in [k]$, update $\nu^{(t)}_r = \RMean(S_r, \kappa, 4\sigma)$ \small{//If {\RMean} outputs \Bad then certify that data is poisoned.} 
\end{enumerate}
\end{enumerate}
}}
\end{center}
\caption{\label{ALG:Lloyds} Iterative Updates of the Lloyd's Algorithm.}
\end{figure}

\noindent \textbf{Overview of Analysis and Challenges.} 
Our analysis of the modified Lloyd's updates proceeds in two stages: a) a center improvement step, and b) analyzing robust Lloyd's updates. In (a), we first improve the initial center estimates obtained form the initialization phase to get estimates $\nu^{(1)}_1, \dots, \nu^{(1)}_k$ such that each $\nu^{(1)}_r$ is $\sim \Delta_r/2$-close to the corresponding $\mu_r$, where $\Delta_r = 40c_q \alpha \sigma \sqrt{m}/\sqrt{|C^*_r|}$. In other words, we get a factor $\sqrt{k}$ improvement over the initial estimates. This is sketched in step $3$ of the algorithm in Figure~\ref{ALG:Lloyds}. 
% While such a step is not needed in standard analysis of Lloyd's, it is crucial for our robust version. See Section~\ref{sec:application-clustering} for a discussion on this.
First, we motivate the need for this intermediate step, since it is not necessary in the analysis of Lloyd's algorithm for uncorrupted data. 

Just as in standard analysis of Lloyd's updates, we would like to argue that if we have non-trivial estimates of the centers, as obtained from the initialization stage, forming clusters using these points and moving to the means of these clusters will improve the center estimates. To argue this we will crucially rely on the fact that when projected onto $\Pi$, $A$ and $\tilde{A}$ are close pointwise. Hence, we can come up with a charging argument to assign mistakes made by the current centers on the uncorrupted points to the mistakes made by the centers on the corrupted points. We can then bound the number of such mistakes by observing that on $\Pi \tilde{A}$, the true means have a small $k$-means cost. This forces us to work in the projected space $\Pi$, but as a result inherently limits the accuracy to which we can obtain center estimates. Notice that if the initialization algorithm outputs $\Pi$, then $\Pi$ is guaranteed to be good overall for $\tilde{A}$, in the sense that $\|\tilde{A} - \Pi \tilde{A}\| = O(\sigma \sqrt{m})$. However, $\Pi$ has no per cluster guarantee, and in general $\|\tilde{A}_r - \Pi \tilde{A}_r\|$ when restricted to a cluster $C^*_r$ could be as large as $\sigma \sqrt{m}/\sqrt{|C^*_r|}$. Hence, to achieve our goal of estimating the centers upto $\tilde{O}(\alpha \sigma)$ accuracy, we also need to work outside of the projection $\Pi$ at the same time. Due to these conflicting demands, notice that the Lloyd's updates we analyze in step $4$ of the algorithm in Figure~\ref{ALG:Lloyds} perform clustering using current centers in the projected space, but perform robust mean estimation on the original input data. 

Furthermore, from our guarantee on robust mean estimation in Theorem~\ref{thm:robust-mean}, we know that in the {\RMean} step of the algorithm the centers will be accurate upto $\sim \alpha \sigma_S$, where $\sigma_S$ is the standard deviation of the uncorrupted data points in $S_r$ around the uncorrupted mean of $S_r$. As a result we need a stronger argument that not only shows that we have low clustering error given the current estimates, but also helps us argue about the variance of the formed clusters $S_r$ at each step. Such an argument~(Lemma~\ref{lem:cluster-purity}) is a main technical contribution in the analysis. 

Unfortunately, the argument~(Lemma~\ref{lem:cluster-purity}) only kicks in when we have much better center estimates than the one provided by the initialization stage, thereby requiring an additional center improvement stage.
To argue about the center improvement stage, we use a trick from~\citet{awasthi2012improved} and form sets $S_r$ that correspond to points in $\Pi \tilde{A}$ that are significantly close to one of the centers $\nu^{(0)}_r$ than any other center $\nu^{(0)}_s$. Notice that these sets do not form a partitioning of the data. We then argue that any mistake made by this assignment must have also been made if one had used the true centers $\mu_1, \dots, \mu_k$, to cluster $\Pi \tilde{A}$. Using the fact that the true means have small $k$-means cost on $\Pi \tilde{A}$ we can bound the number of such mistakes and hence get sets $S_r$ that have low error, thereby helping us show that the means of these sets will be much closer to the true centers. This is established in Theorem~\ref{thm:center-closeness}. The above arguments help us establish Theorem~\ref{thm:clustering-application-general}. 
We next state the key technical lemma for our analysis.
\begin{lemma}
\label{lem:cluster-purity}
Let $\Pi$ be the robust subspace computed in step (1) of the algorithm in Figure~\ref{ALG:Lloyds}. For each cluster $C^*_r$ in the optimal clustering of $A$, define $\Delta_r = 40c_q \alpha \sigma \sqrt{m}/\sqrt{|C^*_r|}$. Suppose we have center estimates $\{\nu_1, \nu_2, \dots, \nu_k\}$ such that $\|\nu_r - \mu_r\| \leq \beta \Delta_r$, $\forall r \in [k]$, and some $\beta < 1$. 
%and $\|\nu_r\|_{q^*} \leq c_q \kappa$. 
When using $\nu_i$s to cluster $\Pi \tilde{A}$, define $T_{r,s}$ to be the set of points that are misclassified, w.r.t. the induced clustering on $A$, i.e., $T_{s \to r} = \{i: A_i \in C^*_r \text{ and } \|\Pi \tilde{A}_i - \nu_s\| \leq \|\Pi \tilde{A}_i - \nu_r\| \}$. There exists a constant $c_1 > 0$ depending on $q$ such that if the clustering instance is $c$-stable for $c > 200c_q$ then we have that $|T_{s \to r}| \leq \frac{c_1 \beta^2 \sigma^2 m}{k c^2 \|\mu_r - \mu_s\|^2}$. 
\end{lemma}
% Furthermore, if the centers $\nu_1, \dots, \nu_k$ lie in $\Pi$ then we have that $|T_{s \to r}| \leq \frac{k c_1 \beta^2 \sigma^2 m}{c^2 \|\mu_r - \mu_s\|^2} \min(|C_r|, |C_s|)$.
\begin{proof}
Fix $s \neq r$ and let $W$ be the subspace spanned by $\{\mu_r, \mu_s, \nu_r, \nu_s\}$ with $\Pi_W$ being the projection matrix for the orthogonal projection on to the subspace. Define $\bar{A}_i$ to be the projection of $A_i$ onto the line joining $\mu_r$ and $\mu_s$. Since $W$ contains $\mu_r, \mu_s$, this is also the same as the projection of $\Pi_W A_i$ on to the line joining $\mu_r$ and $\mu_s$. Similarly, define $\bar{\tilde{A}}_i$ to be the projection of $\tilde{A}_i$ on to the line joining $\mu_r$ and $\mu_s$, and again this is the same as the projection of $\Pi_W \tilde{A}_i$ on to the line joining $\mu_r$ and $\mu_s$. We will crucially make use of the fact that
\begin{align}
\label{eq:proof-cluster-purity-fact1}
    \|\bar{\tilde{A}}_i - \mu_s\|-\|\bar{\tilde{A}}_i - \mu_r\| &\geq \Delta_{r,s} - O(\kappa \delta) \geq \Delta_{r,s}/2.
\end{align}
The above holds since from $c$-stability we know that $\|{\bar{A}}_i - \mu_s\|-\|{\bar{A}}_i - \mu_r\| \geq \Delta_{r,s}$. Furthermore, since $\|\tilde{A}_i - {A}_i\|_q \leq \delta$ and each of $\mu_r, \mu_s$ is $\kappa$-sparse in $\ell_{q^*}$ norm, we have that $\|\bar{\tilde{A}}_i - \bar{A}_i\| \leq O(\kappa \delta)$. Here $q^*$ is such that $1/q + 1/q^*=1$.
Next, let $v = \Pi_W \tilde{A}_i $. Then we have that
\begin{align}
    \|v - \mu_s\|^2 - \|v-\mu_r\|^2 &= \|\bar{\tilde{A}}_i - \mu_s\|^2 - \|\bar{\tilde{A}}_i-\mu_r\|^2 \nonumber \\
    &\geq \frac{\Delta_{r,s} \|\mu_r - \mu_s\|}{4} \,\, (\text{using the fact that $\bar{\tilde{A}}_i$ lies on the line joining $\mu_r$ and $\mu_s$}) \label{eq:proof-cluster-purity-fact2}.
\end{align}
By triangle inequality we also have that,
\begin{align}
    \|v - \mu_s\|^2 - \|v-\mu_r\|^2 
    &\leq (\|v - \nu_s\| + \beta \Delta_s)^2 - (\|v - \nu_r\| - \beta \Delta_r)^2 \nonumber \\
    &\leq (\|v - \nu_r\| + \beta \Delta_s)^2 - (\|v - \nu_r\| - \beta \Delta_r)^2 \nonumber \\
    &\leq \beta(\Delta_s + \Delta_r)\|v-\nu_r\| \label{eq:proof-cluster-purity-fact3}.
\end{align}
Here the first inequality uses the fact that $\nu_r,\nu_s$ are close to $\mu_r, \mu_s$ respectively and the second inequality uses the fact that $\tilde{A}_i$ is closer to $\nu_s$ than to $\nu_r$, the same holds true for $\tilde{A}_i$ projected on to any subspace that contains $\nu_r$ and $\nu_s$. From (\ref{eq:proof-cluster-purity-fact2}) and (\ref{eq:proof-cluster-purity-fact3}), and substituting the bound for $\Delta_{r,s}$ we get that $\|v - \nu_r\| \geq \frac{c\sqrt{k} \|\mu_r - \mu_s\|}{10c_q \beta}$, which in turn implies that $\|v - \mu_r\| \geq c\sqrt{k} \frac{\|\mu_r-\mu_s\|}{8c_q \beta}$. Hence we get that
\begin{align}
\label{eq:proof-cluster-purity-fact4}
\sum_{i \in T_{s \to r}} \| \Pi_W \tilde{A}_i - \mu_r\|^2 \geq |T_{s \to r}| \frac{c^2 k \|\mu_r - \mu_s\|^2}{64 c^2_q \beta^2}.
\end{align}
Combining with the fact that $\|\mu_r - \nu_r\| \leq \beta \Delta_r$ we get that
\begin{align}
\label{eq:proof-cluster-purity-fact5}
\sum_{i \in T_{s \to r}} \| \Pi_W \tilde{A}_i - \nu_r\|^2 \geq |T_{s \to r}| \frac{c^2 k \|\mu_r - \mu_s\|^2}{128 c^2_q \beta^2}.
\end{align}
On the other hand we also have that 
\begin{align}
\label{eq:proof-cluster-purity-fact6}
\sum_{i \in T_{s \to r}} \| \Pi_W \tilde{A}_i - \nu_r\|^2    &\leq \sum_{i \in T_{s \to r}} 2\|\Pi_W \tilde{A}_i - \mu_r\|^2 + 2|T_{s \to r}| \|\mu_r - \nu_r\|^2 \,\, (\text{by triangle inequality})\nonumber \\
&=\sum_{i \in T_{s \to r}} 2\|\Pi_W \tilde{A}_i - \Pi_W \mu_r\|^2 + 2|T_{s \to r}| \|\mu_r - \nu_r\|^2 \,\, (\text{since $\mu_r$ lies in $\Pi_W$})\nonumber \\
&\leq \sum_{i \in T_{s \to r}} 2\|\Pi_W \tilde{A}_i - \Pi_W \mu_r\|^2 + 2|T_{s \to r}| \beta^2 \Delta^2_r \nonumber \\
&\leq \sum_{i \in T_{s \to r}} 4\|\Pi_W \tilde{A}_i - \Pi_W \Pi \mu_r\|^2 + \sum_{i \in T_{s \to r}} 4\|\Pi_W(\Pi \mu_r - \mu_r)\|^2 + 2|T_{s \to r}| \beta^2 \Delta^2_r \,\, \nonumber\\%(\text{by triangle inequality})\nonumber\\
&\leq \sum_{i \in T_{s \to r}} 4\|\Pi_W \tilde{A}_i - \Pi_W \Pi \mu_r\|^2 + 4|T_{s \to r}|\|\Pi \mu_r - \mu_r\|^2 + 2|T_{s \to r}| \beta^2 \Delta^2_r, 
\end{align}
where the last but one line also uses triangle inequality. 
Next notice that
\begin{align*}
\|\Pi \mu_r - \mu_r\| &= \frac{1}{|C^*_r|} \|\sum_{A_i \in C_r} (\Pi A_i - A_i)\| = \frac{1}{|C^*_r|} \|\mathbf{1}^{\top}(\Pi A - A)\| \nonumber \\
&\leq \frac{1}{\sqrt{|C^*_r|}}\|\Pi A - A\| \leq \frac{4c_q \alpha \sigma \sqrt{m}}{\sqrt{|C^*_r|}} = \frac{2}{5}\Delta_r.
\end{align*}
Substituting into (\ref{eq:proof-cluster-purity-fact6}) we get that
\begin{align}
\sum_{i \in T_{s \to r}} \| \Pi_W \tilde{A}_i - \nu_r\|^2 &\leq \sum_{i \in T_{s \to r}} 4\|\Pi_W \tilde{A}_i - \Pi_W \Pi \mu_r\|^2 + |T_{s \to r}|(\frac{16}{5}+2 \beta^2) \Delta^2_r \nonumber\\
&= \sum_{i \in T_{s \to r}} 4\|\Pi_W \Pi^\top \Pi \tilde{A}_i- \Pi_W \Pi^\top \Pi^2 \mu_r\|^2 + |T_{s \to r}|(\frac{16}{5}+2 \beta^2) \Delta^2_r \nonumber \\
&= \sum_{i \in T_{s \to r}} 4\|\Pi_W \Pi^\top(\Pi \tilde{A}_i - \Pi \mu_r)\|^2 + |T_{s \to r}|(\frac{16}{5}+2 \beta^2) \Delta^2_r \nonumber \\
&\leq 4\|\Pi_W \Pi^{\top}(\Pi \tilde{A}-\Pi C)\|^2_F + |T_{s \to r}|(\frac{16}{5}+2 \beta^2) \Delta^2_r \nonumber \\
&\leq 16\|\Pi\tilde{A}- \Pi C\|^2 + |T_{s \to r}|(\frac{16}{5}+2 \beta^2) \Delta^2_r \nonumber
\end{align}
since $\Pi_W \Pi^{\top}(\Pi \tilde{A}- \Pi C)$ has rank at most $4$. Hence
\begin{align}
\sum_{i \in T_{s \to r}} \| \Pi_W \tilde{A}_i - \nu_r\|^2 &\leq 32 c^2_q (1+\frac{\kappa \delta}{\sigma})^2 \sigma^2 m + |T_{s \to r}|(\frac{16}{5}+2 \beta^2) \Delta^2_r.
\label{eq:proof-cluster-purity-fact7}
\end{align}
The last inequality uses the fact that
\begin{align*}
    \|\Pi \tilde{A} - \Pi C\| &\leq \|\Pi(\tilde{A}-A)\| + \|\Pi A - \Pi C\|
    \leq \kappa \delta \sqrt{m} +  \|A-C\|
    \leq \kappa \delta \sqrt{m} +  \sigma \sqrt{m}.
\end{align*}
Combining, (\ref{eq:proof-cluster-purity-fact5}) and (\ref{eq:proof-cluster-purity-fact7}) we get the desired claim. 
% In the case when $\nu_1, \nu_2, \dots, \nu_k$ lie in $\Pi$, we can get a better bound on $|T_{s \to r}|$ by bounding $\sum_{i \in T_{s \to r}} \| \Pi_W \tilde{A}_i - \nu_r\|^2$ as
% \begin{align*}
% \sum_{i \in T_{s \to r}} \| \Pi_W \tilde{A}_i - \nu_r\|^2    &= \sum_{i \in T_{s \to r}} \|\Pi_W \Pi^\top \Pi \tilde{A}_i- \Pi_W \Pi^\top \Pi \nu_r\|^2 \nonumber \\
% &= \sum_{i \in T_{s \to r}} \|\Pi_W \Pi^\top(\Pi \tilde{A}_i - \nu_r)\|^2 \nonumber \\
% &\leq 2\sum_{i \in T_{s \to r}} \|\Pi_W \Pi^\top(\Pi \tilde{A}_i - \mu_r)\|^2 + 2|T_{s \to r}| \|\mu_r - \nu_r\|^2 \nonumber \\
% &\leq 8\|\tilde{A}\Pi-C\|^2 + 2|T_{s \to r}|\beta^2 \Delta^2_r \nonumber \\
% &\leq 4 c^2_q \sigma^2 m + 2|T_{s \to r}| \beta^2 \Delta^2_r.
% \end{align*}
\end{proof}

In order to apply Lemma~\ref{lem:cluster-purity} iteratively we need initial centers such that $\|\mu_r - \nu_r\| \leq \beta \Delta_r$, for $\beta \leq \frac{1}{4}$. However, notice that the the initialization procedure of Figure~\ref{ALG:Initialization} only guarantees $\beta \leq 30c_q \sqrt{k}$. We next argue that step (3) of the algorithm in Figure~\ref{ALG:Lloyds} provides center estimates that are much closer to the true means, thereby allowing us to analyze the iterative Lloyd's updates in step (4) of the algorithm. 
\begin{theorem}
\label{thm:center-closeness}
If the clustering instance $A$ is $c$-stable as defined in Theorem~\ref{thm:lloyds-initialization-guarantee}, then given $\tilde{A}$ as input, steps $1$-$3$ of the Algorithm in Figure~\ref{ALG:Lloyds}, run in polynomial time, and output centers $\nu^{(1)}_1, \dots, \nu^{(1)}_k$ such that %for all $r \in [k]$,
\begin{align*}
\forall r \in [k],~~    \|\mu_r - \nu^{(1)}_{\sigma(r)}\| \leq \beta \Delta_r,
\end{align*}
for an appropriately chosen bijection $\sigma$. Here $\Delta_r = 40c_q \alpha \sigma \sqrt{m}/\sqrt{|C^*_r|}$ and $\beta < 1$. 
%\xnote{I am not sure about the last sentence.}
%\pnote{Done.}
\end{theorem}
\begin{proof}
The proof strategy closely follows the one in~\citet{awasthi2012improved} and consists of three main steps. We first define clusters $T_r$ for $r \in [k]$ such that $T_r$ consists of points $\Pi \tilde{A}_j$ for $A_j \in C^*_r$. In other words, $\{T_1, T_2, \dots, T_k\}$ is the clustering induced on the data set $\Pi \tilde{A}$ by the optimal clustering $\{C^*_1, C^*_2, \dots, C^*_k\}$. We first argue that $S_r$ is pure w.r.t. $T_r$ i.e., at most $O(\frac{1}{c^2}|C^*_r|)$ points of $T_r$ do not belong to $S_r$ and in total at most $O(\frac{1}{k}|C^*_r|)$ points from $T_s$, for $s \neq r$, end up belonging to $S_r$. 
%Notice that by applying Lemma~\ref{lem:cluster-purity} on centers $\nu^{(0)}_1, \dots, \nu^{(0)}_k$ we get a weaker bound by a factor of $k$ since in that case $\alpha \sim \sqrt{k}$. 
Next use the fact that any points that belongs to $|S_r \cap T_s|$ for $s \neq r$, will also be misclassified when using centers $\Pi \tilde{\mu}_1, \dots, \Pi \tilde{\mu}_k$ instead of centers $\nu^{(0)}_1, \dots, \nu^{(0)}_k$. Here $\tilde{\mu}_r = \Mean(T_r)$.
Now each projected center $\Pi \tilde{\mu}_r$ is much closer to the corresponding true center $\mu_r$. To see this notice that
\begin{align*}
    \|\Pi \tilde{\mu}_r - \mu_r\| &= \frac{1}{|C^*_r|} \|\sum_{\Pi \tilde{A}_i \in T_r} (\Pi \tilde{A}_i - A_i)\|\\
    &\leq \frac{1}{|C^*_r|} \|\sum_{\Pi \tilde{A}_i \in T_r} (\Pi (\tilde{A}_i - A_i))\| + \frac{1}{|C^*_r|} \|\sum_{\Pi \tilde{A}_i \in T_r} (\Pi {A}_i - A_i)\|\\
    &\leq c_q \kappa \delta + \frac{1}{|C^*_r|}\|\mathbf{1}^{\top}(\Pi A-A)\| 
    \leq c_q \kappa \delta + \frac{1}{\sqrt{|C^*_r|}}\|\Pi A - A\|\\
    &\leq c_q \kappa \delta + \frac{4c_q \alpha \sigma\sqrt{m}}{\sqrt{|C^*_r|}}
    \leq \frac{\Delta_r}{9}.
\end{align*}
With the above idea, arguing that $|T_s \cap S_r|$ is small and $T_r$ has large overlap with $S_r$ follows verbatim from Lemmas 4.2 and 4.3 of~\citet{awasthi2012improved} by substituting $\Pi \tilde{A}_i$ instead of $\hat{A}_i$ in the proofs. In the final step we use the following standard fact stated in Lemma~\ref{lem:mean-closeness-AS} below and adapted from its original version in~\citet{awasthi2012improved,kumar2010clustering}. From the guarantees on $|T_s \cap S_r|$ and $|T_r \cap S_r|$ we can set $\rho_{out} = \frac{1}{8}$ and $\rho_{in} = c/10k$ to get that
\begin{align*}
    \|\Mean(S_r) - \Mean(\hat{C}_r)\| &\leq 2\frac{\sigma \sqrt{m}}{\sqrt{|C^*_r|}}.
\end{align*}
Furthermore we also have that
\begin{align*}
    \|\Mean(\hat{C}_r) - \mu_r\| &= \frac{1}{|C^*_r|} \|\sum_{\Pi {A}_i \in C^*_r} (\Pi {A}_i - A_i)\|\\
    &= \frac{1}{|C^*_r|} \|\mathbf{1}^{\top}(\Pi A-A)\|
    \leq \frac{1}{\sqrt{|C^*_r|}}\|\Pi A - A\|\\
    &\leq \frac{4c_q \alpha \sigma\sqrt{m}}{\sqrt{|C^*_r|}}
    \leq \frac{\Delta_r}{10}.
\end{align*}
Combining the above two we get that
\begin{align*}
    \|\mu_r - \nu^{(1)}_{r}\| \leq O(\beta \Delta_r),
\end{align*}
for $\beta < 1$.
\end{proof}
\begin{lemma}[Fact 1.3 from~\citet{awasthi2012improved}]
\label{lem:mean-closeness-AS}
Fix a target cluster $C^*_r$ and let $\hat{C}_r$ be the projection of points in $C^*_r$ onto $\Pi$. Let $S_r$ be a set of points created by removing $\rho_{out} |C^*_r|$ points
from $\hat{C}_r$ and adding $\rho_{in} |C^*_s|$ points from each cluster $\hat{C}_s$ for $s \neq r$, s.t. every added point $x$ satisfies $\|x - \Pi \mu_s\| \geq \frac{2}{3}\|x-\Pi \mu_r\|$. If $\rho_{out} < 1/4$ and $\rho_{in} := \sum_{s \neq r} \rho_{in} <1/4$ then we have that
\begin{align}
    \label{eq:as12-mean-closeness}
    \|\Mean(S_r) - \Mean(\hat{C}_r)\| \leq 2\Big(\sqrt{\frac{\rho_{out}}{|C^*_r|}} + \frac{3\sqrt{k}}{2}\sqrt{\frac{\rho_{in}}{|C^*_r|}} \Big)\sigma \sqrt{m}
\end{align}
\end{lemma}
\begin{proof}[Proof of Theorem~\ref{thm:clustering-application-general}]
The theorem follows from using steps $1$-$3$ of the algorithm in Figure~\ref{ALG:Lloyds} and from the guarantees in Lemma~\ref{lem:cluster-purity} and Lemma~\ref{thm:center-closeness}. 
\end{proof}
\noindent \textbf{Achieving $\tilde{O}(\alpha \sigma)$ Guarantee for Mean Estimation.}
\begin{proof}[Proof of Theorem~\ref{thm:clustering-application-special}]
Notice that Theorem~\ref{thm:center-closeness} gives us centers $\nu_1, \dots, \nu_k$ that are $\beta \Delta_r$ close to the corresponding true centers $\mu_1, \dots, \mu_k$. We start with these centers and perform Lloyd's updates as shown in step $4$ of the algorithm in Figure~\ref{ALG:Lloyds}. Next suppose that at iteration $t$ we have centers $\nu^{(t)}_1, \dots, \nu^{(t)}_k$ such that $\|\nu^{(t)}_r - \mu_r\| \leq \beta \Delta_r$ for $r \in [k]$. We will argue that using $\nu^{(t)}_1, \dots, \nu^{(t)}_k$ to form clusters $S_1, S_2, \dots, S_r$ and computing new means by calling the {\RMean} procedure on the sets $S_r$, either leads to a certification that the dataset is poisoned or leads to new centers estimates $\nu^{(t+1)}_1, \dots, \nu^{(t+1)}_k$ that satisfy $\|\nu^{(t+1)}_r - \mu_r\| \leq \Big(\frac{\beta}{2} \Delta_r + \tilde{O}(\alpha \sigma)\Big)$. Hence the estimates will improve until the unavoidable error of $\tilde{O}(\alpha \sigma)$. We will prove the claim in two steps. First we analyze the ``ideal'' updates. For each $S_r$ define $S^*_r$ to the set $S_r$ with corrupted points replace by the original points, i.e., $S^*_r = \{A_i: \tilde{A}_i \in S_r\}$. We next show that the mean of $S^*_r$ is close to $\mu_r$ upto $\frac{\beta}{2} \Delta_r$ error. As in Lemma~\ref{lem:cluster-purity} define $T_{r \to r} = S^*_r \cap C^*_r$ and for $s \neq r$, define $T_{r \to s} = S^*_r \cap C^*_s$. Then we have by triangle inequality that,
\begin{align}
\label{eq:clustering-gaussian-mean-estimation-1}
    \|\Mean(S^*_r) - \mu_r\| &= \Bignorm{\frac{|T_{r \to r}|}{|S^*_r|} (\Mean(T_{r \to r})-\mu_r)  + \sum_{s \neq r} \frac{|T_{r \to s}|}{|S^*_r|} (\Mean(T_{r \to s})-\mu_r )}\nonumber \\
    &\leq  \frac{|T_{r \to r}|}{|S^*_r|}\Bignorm{\Mean(T_{r \to r})-\mu_r} + \sum_{s \neq r} \frac{|T_{r \to s}|}{|S^*_r|}\Bignorm{\Mean(T_{r \to s})-\mu_r}\nonumber \\
    &\leq  \frac{|T_{r \to r}|}{|S^*_r|}\Bignorm{\Mean(T_{r \to r})-\mu_r} + \sum_{s \neq r} \frac{|T_{r \to s}|}{|S^*_r|}\Bignorm{\Mean(T_{r \to s})-\mu_s} \\
    & ~~~+ \sum_{s \neq r} \frac{|T_{r \to s}|}{|S^*_r|}\Bignorm{\mu_r-\mu_s}
\end{align}
Next we notice that
\begin{align*}
    \Bignorm{\Mean(T_{r \to r})-\mu_r} &= \frac{|C^*_r \setminus T_{r \to r}|}{|T_{r \to r}|} \Bignorm{\sum_{A_i \in C^*_r \setminus T_{r \to r}} (A_i - \mu_r)}\\
    &\leq \frac{\sqrt{|C^*_r \setminus T_{r \to r}|}}{{|T_{r \to r}|}} \sigma \sqrt{m}\\
    &= \frac{\sqrt{|C^*_r|- |T_{r \to r}|}}{{|T_{r \to r}|}}  \sigma \sqrt{m}.\\
\end{align*}
The first inequality above follows from the fact that
\begin{align*}
\Bignorm{\sum_{A_i \in C^*_r \setminus T_{r \to r}} (A_i - \mu_r)} &= \|\mathbf{1}^{\top}_S (A-C)\| \,\, (\text{$\mathbf{1}_S$ is the indicator vector for points in $C^*_r \setminus T_{r \to r}$})\\
&\leq \frac{\sigma \sqrt{m}}{\sqrt{|C^*_r \setminus T_{r \to r}|}}.
\end{align*}
Next, by Assumptions II regarding large subsets of optimal clusters we have that
for sets $T_{r \to s}$ either $|T_{r \to s}| \leq n \log n$ or 
\begin{align*}
        \Bignorm{\Mean(T_{r \to s})-\mu_s} &= \frac{1}{|T_{r \to s}|} \Bignorm{\sum_{A_i \in T_{r \to s}} (A_i - \mu_s)}\\
    &\leq \frac{1}{\sqrt{|T_{r \to s}|}} \sigma \poly \log(m,n).
\end{align*}
Furthermore, we also have the pointwise guarantee that for every $A_i \in T_{r \to s}$, $\|A_i - \mu_s\| \leq 2\sigma \sqrt{n} \cdot \poly \log(m,n)$.
Hence we get that
\begin{align*}
    \frac{|T_{r \to s}|}{|S^*_r|}\Bignorm{\Mean(T_{r \to s})-\mu_r} &\leq \max \Big(\frac{\sqrt{|T_{r \to s}|}}{|S^*_r|}\sigma \poly \log(m,n) , \frac{\sigma \poly \log(m,n)}{n} \Big).
\end{align*}
Substituting back into (\ref{eq:clustering-gaussian-mean-estimation-1}) we get that
\begin{align*}
\|\Mean(S^*_r) - \mu_r\| &\leq \frac{\sqrt{|C^*_r|-|T_{r,r}|}}{|S^*_r|} \sigma \sqrt{m} + \sum_{s \neq r}\frac{\sqrt{|T_{r \to s}|}}{|S^*_r|}\sigma \text{polylog}(m,n) + \sum_{s \neq r} \frac{|T_{r \to s}|}{|S^*_r|}\Bignorm{\mu_r-\mu_s} + \sigma\\
&\leq \frac{\sqrt{|S_r|-|T_{r,r}|}}{|S^*_r|} \sigma \sqrt{m} + \frac{\sqrt{|S_r \triangle C^*_r|}}{|S^*_r|} \sigma \sqrt{m}  \\
& \qquad \quad +  \sum_{s \neq r}\frac{\sqrt{|T_{r \to s}|}}{|S^*_r|}\sigma \text{polylog}(m,n)+ \sum_{s \neq r} \frac{|T_{r \to s}|}{|S^*_r|}\Bignorm{\mu_r-\mu_s} + \sigma\\
&= \frac{\sigma \sqrt{m}\sqrt{\sum_{s \neq r}|T_{r \to s}|}}{|S^*_r|}  + \frac{\sigma \sqrt{m}\sqrt{\sum_{s \neq r} |T_{r \to s}|  + |T_{s \to r}|}}{|S^*_r|}  \\
&\qquad\qquad +  \sum_{s \neq r}\frac{\sigma \text{polylog}(m,n)\sqrt{|T_{r \to s}|}}{|S^*_r|}+ \sum_{s \neq r} \frac{|T_{r \to s}|}{|S^*_r|}\Bignorm{\mu_r-\mu_s} + \sigma\\
&\leq 4 \sigma \sqrt{m} \sum_{s \neq r} \frac{\sqrt{|T_{r \to s}|}}{|S^*_r|} + 4\sigma \sqrt{m} \sum_{s \neq r} \frac{\sqrt{|T_{s \to r}|}}{|S^*_r|} \\
&\qquad \quad \quad +\sum_{s \neq r} \frac{|T_{r \to s}|}{|S^*_r|}\Bignorm{\mu_r-\mu_s} + \sigma
\end{align*}
Noticing that $|S^*_r| > |C^*_r|/2$ we get that
\begin{align*}
    \|\Mean(S^*_r) - \mu_r\| &\leq 4 \sigma \sqrt{m} \sum_{s \neq r} \frac{\sqrt{|T_{r \to s}|}}{|C^*_r|} + 4\sigma \sqrt{m} \sum_{s \neq r} \frac{\sqrt{|T_{s \to r}|}}{|C^*_r|} + \sum_{s \neq r} \frac{2|T_{r \to s}|}{|C^*_r|}\Bignorm{\mu_r-\mu_s} + \sigma
\end{align*}
Substituting the bound on $T_{r \to s}$ from Lemma~\ref{lem:cluster-purity} we get that
\begin{align*}
    \|\Mean(S^*_r) - \mu_r\| &\leq \frac{8 c_1 \sigma \sqrt{m}}{|C^*_r|} \sum_{s \neq r} \frac{\beta \sigma \sqrt{m}}{c \sqrt{k} \|\mu_r - \mu_s\|} + \sum_{s \neq r} \frac{2|T_{r \to s}|}{|C^*_r|}\Bignorm{\mu_r-\mu_s} + \sigma
\end{align*}
where $c_1$ is an absolute constant depending on $q$. Substituting the lower bound on $\|\mu_r - \mu_s\|$ and using the definition of $\Delta_r$ we get that
\begin{align*}
    \|\Mean(S^*_r) - \mu_r\| &\leq \frac{\beta \Delta_r}{4}\sum_{s \neq r} \frac{1}{c^2 k}+ \sum_{s \neq r} \frac{2|T_{r \to s}|}{|C^*_r|}\Bignorm{\mu_r-\mu_s} + \sigma\\
    &\leq \frac{\beta \Delta_r}{4} + \sum_{s \neq r} \frac{2|T_{r \to s}|}{|C^*_r|}\Bignorm{\mu_r-\mu_s} + \sigma.
\end{align*}
To bound the second term, we again substitute the guarantee on $|T_{r \to s}|$ from Lemma~\ref{lem:cluster-purity} and get that
\begin{align*}
    \sum_{s \neq r} \frac{2|T_{r \to s}|}{|C^*_r|}\Bignorm{\mu_r-\mu_s} &\leq \sum_{s \neq r} \frac{2c_1  \beta^2 \sigma^2 m}{c k |C^*_r| \|\mu_r - \mu_s\|}\\
    &\leq \sum_{s \neq r} \frac{2 c_1  \beta^2 \sigma \sqrt{m} \min(\sqrt{|C^*_r|}, \sqrt{|C^*_s|})}{\alpha c^2 k\sqrt{k}  |C^*_r|}\\
    &\leq \frac{\beta^2 \Delta_r}{\alpha} \sum_{s \neq r} \frac{1}{c^2 k \sqrt{k}}
    \leq \frac{\beta^2 \Delta_r}{\alpha c^2 \sqrt{k}}.
\end{align*}
Combining the above we get that
\begin{align}
\label{eq:clustering-gaussian-mean-estimation-6}
    \|\Mean(S^*_r) - \mu_r\| &\leq \frac{\beta \Delta_r}{3}.
\end{align}
Next we analyze the true updates that correspond to running the {\RMean} procedure on the set $S_r$. Notice from the guarantee of Theorem~\ref{thm:robust-mean}, when run on $S_r$, that either the algorithm will certify that the dataset if poisoned or will output an approximation to \Mean($S^*_r$) upto a factor of $O(\alpha \sigma_{S^*_r})$ where $\sigma_{S^*_r}$ is the variance of the set $S^*_r$ around $\mu_r$. We next bound this value.
\begin{align}
\label{eq:clustering-gaussian-mean-estimation-2}
    \sigma^2_{S^*_r} &= \max_{v: \|v\|=1} \frac{1}{|S^*_r|} \sum_{A_i \in S^*_r} \Big( (A_i - \mu_r)\cdot v \Big)^2 \nonumber \\
    &\leq \max_{v: \|v\|=1} \frac{1}{|S^*_r|} \sum_{A_i \in T_{r \to r}} \Big( (A_i - \mu_r)\cdot v \Big)^2 + \sum_{s \neq r} \max_{v: \|v\|=1} \frac{1}{|S^*_r|} \sum_{A_i \in T_{r \to s}} \Big( (A_i - \mu_r)\cdot v \Big)^2
\end{align}
Since $|T_{r \to r}| \geq \frac{7}{8}|C^*_r| \geq n^2$, from Assumptions II regarding large subsets of clusters we can bound the first term by 
\begin{align}
\label{eq:clustering-gaussian-mean-estimation-3}
    \max_{v: \|v\|=1} \frac{1}{|S^*_r|} \sum_{A_i \in S^*_r} \Big( (A_i - \mu_r)\cdot v \Big)^2 &\leq O(\sigma^2 \poly \log(m,n)).
\end{align}
To bound the second term we have by triangle inequality that
\begin{align*}
    \max_{v: \|v\|=1} \frac{1}{|S^*_r|} \sum_{A_i \in T_{r \to s}} \Big( (A_i - \mu_r)\cdot v \Big)^2 &\leq \max_{v: \|v\|=1} \frac{1}{|S^*_r|} \sum_{A_i \in T_{r \to s}} \Big( (A_i - \mu_s)\cdot v \Big)^2 + \frac{|T_{r \to s}|}{|S^*_r|}\|\mu_r - \mu_s\|^2.
\end{align*}
Here again the first term is either small due to $|T_{r \to s}|$ being small or is bounded due to Assumptions II about variance of large subsets. In particular, we have that
\begin{align}
\label{eq:clustering-gaussian-mean-estimation-4}
    \max_{v: \|v\|=1} \frac{1}{|S^*_r|} \sum_{A_i \in T_{r \to s}} \Big( (A_i - \mu_s)\cdot v \Big)^2 &\leq \max \Big(\frac{2 \sigma^2 n \log n}{|S^*_r|}, 2\sigma^2 \poly \log(m,n) \frac{|T_{r \to s}|}{|S^*_r|} \Big). 
\end{align}
Finally, using the bound on $|T_{r \to s}|$ from Lemma~\ref{lem:cluster-purity} we have that
\begin{align}
    \label{eq:clustering-gaussian-mean-estimation-5}
    \frac{|T_{r \to s}|}{|S^*_r|}\|\mu_r - \mu_s\|^2 &\leq \frac{2 \beta^2 \Delta^2_r}{\alpha^2 c  k}.
\end{align}
Combining (\ref{eq:clustering-gaussian-mean-estimation-3}), (\ref{eq:clustering-gaussian-mean-estimation-4}), and (\ref{eq:clustering-gaussian-mean-estimation-5}) we get that
\begin{align*}
        \sigma^2_{S^*_r} &= \max_{v: \|v\|=1} \frac{1}{|S^*_r|} \sum_{A_i \in S^*_r} \Big( (A_i - \mu_r)\cdot v \Big)^2\\
        &\leq O(\sigma^2 \poly \log(m,n)) + 2\frac{\beta^2 \Delta^2_r}{\alpha^2 c}.
\end{align*}
Hence, at each step the {\RMean} procedure will either certify that the dataset is poisoned or will find estimates $\nu^{(t+1)}_1, \dots, \nu^{(t+1)}_k$ such that 
\begin{align*}
    \|\nu^{(t+1)}_r - \Mean(S^*_r)\| &\leq \tilde{O}(\alpha \sigma) + \frac{\beta \Delta_r}{4}.
\end{align*}
Combining with (\ref{eq:clustering-gaussian-mean-estimation-6}) we get that at iteration $t+1$
\begin{align*}
    \|\nu^{(t+1)}_r - \mu_r\| &\leq \tilde{O}(\alpha \sigma) + \frac{\beta \Delta_r}{2}.
\end{align*}
Hence, the updates will keep improving until the unavoidable error of $\tilde{O}(\alpha \sigma)$.
\end{proof}

\paragraph{Information Theoretic Upper Bounds (Computationally Inefficient Algorithms).}
Finally, we would like to mention that using Proposition~\ref{prop:training:spectral:stat}, via an (inefficient) algorithm we can get the same guarantees as in this section on clustering without the need for certification. In other words, if exponential time is allowed, then there exist algorithms for robust mean estimation and robust clustering that, given any $\delta$-corrupted instance of the problem, will {\em always} output solutions achieving the error guarantees in Theorem~\ref{thm:robust-mean}, Theorem~\ref{thm:clustering-application-general}, Theorem~\ref{thm:clustering-application-special} and Theorem~\ref{thm:clustering-application-gaussians} from this section. In order to achieve this, we simply use the (inefficient) robust mean estimation procedure from the guarantee of Theorem~\ref{thm:robust-mean-stat} when performing the modified Lloyd's updates and we use the guarantee of Proposition~\ref{prop:training:spectral:stat} to always compute good initial centers without the need for certification.

\section{Learning Intersection of Halfspaces} \label{sec:intersection}

%\anote{Work under progress...}

We next demonstrate the applicability of our primitives in supervised learning as well. We will consider the problem of learning an intersection of $k$ halfspaces over the Gaussian distribution on $\R^n$ in the presence of adversarial perturbations to the samples, both at testing-time and training-time. We will represent an intersection of halfspaces by a Boolean function $h:\R^n \to \sset{0,1}$ denoted by $h(x)= \prod_{i=1}^k \bfone(w_i^\top x \ge \theta_i)$, where $\forall i \in [k], ~ \norm{w_i}_2 =1$ and $\theta_i \in \R$ and where $\bfone(\cdot)$ denotes the indicator function. Let $\mathcal{H}_k$ represent the hypothesis class of all intersections of at most $k$ halfspaces. We will also refer to `1' as the positive label, and `0' as the negative label. 

In the uncorrupted setting, the training points $x_1, \dots, x_m \in \R^n$ are drawn i.i.d. from a Gaussian distribution, and their corresponding labels $y_i=h^*(x_i)$ for some $h^* \in \mathcal{H}_k$ (this corresponds to the realizable setting). The special case of $k=1$ corresponds to standard linear classification. A series of well-known results~\citep{VempalaJACM,Vempala08,KOS08} starting with \citet{VempalaJACM} shows that when we are given access to uncorrupted training samples in $\R^n$ drawn from a Gaussian distribution \pnote{grammatical error in this line}, one can PAC-learn\pnote{Perhaps say learn in the probably approximately correct (PAC) model} an intersection of half-spaces in time $f(k)\cdot \poly(n)$, where $f(k)$ has a super-polynomial dependence on $k$. Our algorithmic techniques will be used to learn an intersection of $k=O(1)$ half-spaces even when there are adversarial perturbations {\em both} at {\em training-time} and {\em test-time}. For simplicity we will focus on the case when the uncorrupted points are drawn from a spherical Gaussian $N(0, \sigma^2 I)$. We believe that the same ideas should also extend to general Gaussians, and other convex geometrical concepts as in \citet{Vempala08}.    

Consider a classifier $h \in \mathcal{H}_2$ that is adversarially robust i.e.,  suppose $h(x)= \bfone(w_1^\top x \ge 0) \cdot \bfone(w_2^\top x \ge 0)$ is robust to adversarial $\delta$-perturbations at test-time measured in $\ell_q$ norm, and let $b:= \norm{w_1 - w_2}_2 \in (0,2)$. %Further let $\kappa=\max\sset{ \norm{w_1}_1, \norm{w_2}_1}$. 
%Let $\gamma:=\norm{w_1-w_2}_2$. 
%For a fresh test sample $x \sim N(0,\sigma^2 I)$ with $h(x)=1$, it is easy to see that w.h.p., $|w_1^\top x|, |w_2^\top x| \le O(\sigma)$. It has to be the case 
It is easy to show that $\max\sset{\norm{w_1}_{q^*}, \norm{w_2}_{q^*}} \le  O(\sigma)/\delta$, otherwise for most positive examples there exists $\delta$-adversarial perturbation that $h$ misclassifies w.h.p! Moreover for such an adversarially robust classifier, we can assume that the subspace $\Pi^*$ spanned by $w_1, w_2$ satisfies $\kappa:=\norm{\Pi^*}_{q \to 2} \le O(\sigma/(\delta b))$ (see Claim~\ref{claim:robustintersection}). For general $k$, if the labels are generated by an intersection of $k$-halfspaces represented by $h^{*}(x):=\prod_{i=1}^k \bfone(w_i^\top x \ge \theta_i)$ with $\norm{w_i}_2=1~ \forall i \in [k]$, we assume that the projection matrix $\Pi^*$ onto the span of the normals $w_1, \dots, w_k$ satisfies $\norm{\Pi^*}_{q \to 2} \le \kappa$.

%, and intersection of $k$ half-spaces as studied in \cite{Vempala97,Vempala08, KOS08}.   

We consider the following natural model, where each of the samples can be corrupted adversarially up to $\delta$ measured in $\ell_q$ norm for $q \ge 2$:
\begin{itemize}
    \item Samples $x_1, x_2, \dots, x_m \in \R^n$ are drawn i.i.d from $N(0, \sigma^2 I)$. The labels $y_1=h^*(x_1), \dots, y_m=h^*(x_m)$. 
    \item For each $j \in [m]$, an adversary corrupts (corruptions could be dependent) the points to produce $\tilde{x}_1, \dots, \tilde{x}_m \in \R^n$ such that $\forall j \in [m],~ \norm{\tilde{x}_j - x_j}_q \le \delta$. 
    \item The input consists of $\sset{(\tilde{x}_1, y_1), (\tilde{x}_2, y_2), \dots, (\tilde{x}_m, y_m)}$.
\end{itemize}

The goal is to find an intersection of $k$ half-spaces that achieves low-error and is adversarially robust to $\delta$-perturbations at test-time (this is sometimes referred to as robust accuracy). %We show the following theorem. 
Now we state our main result in this section.
\newcommand{\bone}{\mathbbm{1}}

\begin{theorem}\label{thm:intersect2}
Suppose $\kappa>0, q\ge 2, \delta>0$, and $k \le n^{1/2}$. Let $h^*= \prod_{i=1}^k \bfone(w_i^\top x \ge \theta_i)$ with the normal vectors $w_1,\ldots,w_k$ spanning a $(\kappa,q)$-robust subspace. For convenience, let $\eps=O\big(k^{4/3} \cdot (\kappa \delta \sqrt{k}/ \sigma  + \kappa^2 \delta^2/\sigma^2)^{1/3} \big)$ denote the desired learning error rate. Suppose we are given $m=\poly(n,1/\eps)$ samples $\sset{(\wt{x}_i,y_i): i \in [m]}$  where $\wt{x}_i$ is a $\delta$-perturbation (under $\ell_q$ norm) of $x_i \sim N(0,\sigma^2 I)$ and $y_i=h^*(x_i)$. There exists an algorithm that runs in time $\poly(n) \cdot (\frac{k}{\eps})^{{O}(k^2)}$ to output $\wt{h}=\prod_{i=1}^k \bfone\big( (w'_i)^{\top} x \ge \theta' \big)$ such that with probability 0.9, $$\Pr_{x \sim N(0,\sigma^2 I)}[\wt{h}(x) = h^*(x)] \ge 1 - \eps, \text{ and } \Pr_{x \sim N(0,\sigma^2 I)}[\forall z \text{ s.t. } \norm{z}_q \le \delta,~ \wt{h}(x+z) = h^*(x)] \ge %1 - \eps - O\Big(\frac{k \kappa \delta}{\sigma}\Big) \ge
1-2\eps.$$
%This implies that the robust error is at most $\eps+O(k \kappa \delta/\sigma) \le 2\eps$.
%Moreover, the subspace spanned by the $k$ halfspaces $w_1',\ldots,w_k'$ in $\wt{h}$ has a projection with $\|\cdot\|_{\infty \to 2}$ operator norm at most $O(\kappa)$. This guarantees that for a random $x \sim N(0,\sigma^2 I)$ with arbitrary $\delta$-perturbation $\wt{x}$, $\Pr[\wt{h}(x)=\wt{h}(\wt{x})] \ge 1 - O(k \kappa \delta/\sigma)$.
\anote{Reworded last part of theorem.}
\end{theorem}
The above algorithm runs in polynomial time and returns an intersection of $k$ half-spaces that achieves error $\eps=o(1)$ as long as $\kappa \delta = o(\sigma)$. For example, when $\kappa \approx n^{0.1}$ this allows us to tolerate $\delta=1/\kappa = o(n^{-0.1})$ as opposed to a tolerance of $\delta=o(n^{-1/2})$ for the naive approach. Recall from the earlier discussion, that such a condition is necessary qualitatively: even a single half-space $\bfone(w_1^\top x \ge 0)$ is not robust when $\norm{w_1}_{q^*} = \kappa$ and $\kappa \delta \gg \sigma$.

\paragraph{Notation.} We will use the following notation specific to this problem. Let $X \in \R^{n \times m}$ be the uncorrupted points, and $\tilde{X} \in \R^{n \times m}$ be the points obtained after adversarial perturbations. In particular, let $m_{+}$ denote the number of positive labels and $X_+, \tilde{X}_+ \in \R^{n \times m_+}$ correspond to the positive examples. In what follows $\bone=(1,1,\dots)$ will represent the all-ones vector of appropriate dimension. %and $X_-, \tilde{X}_- \in \R^{n \times m_{-}}$ correspond to the negative examples. %Note that $m=m_+ + m_-$.
Let $B=\tilde{X}_+ - \tfrac{1}{m_+}\tilde{X}_+\bone \bone^\top$ be the centered input matrix corresponding to the (corrupted) positive examples. Hence, we can construct the covariance matrices, {\em uncorrupted} and {\em corrupted} by $M_+= \E\Big[ (x-\mu_+) (x-\mu_+)^\top ~|~ h^*(x)=+1 \Big]$, and $\tilde{M}_+= \tfrac{1}{m_+}BB^\top$. 

We will assume without loss of generality that $m_+ \ge (\kappa \delta/\sigma) \cdot m$. Otherwise, we can output the trivial hypothesis $x_1>0 \wedge (-x_1> 0)$ that achieves an accuracy of $1-O(\kappa \delta/\sigma)$ with high probability. 

Finally, we will say that an intersection of halfspaces $h$ is in a subspace $S \subset \R^n$ iff $h$ can be represented as $h(x)=\prod_{i=1}^k \bfone(w_i^\top x \ge \theta_i)$, where $w_1, \dots, w_k \in S$. 

%where $\eta_0>0$ comes from Theorem~\ref{thm:intersect2}. 

\paragraph{Algorithm description and overview.} The algorithm (Algorithm~\ref{ALG:intersection}) follows the same general approach as~\citet{Vempala08}. The main idea in~\citet{Vempala08} is to consider the co-variance matrix of just the positive examples $X_+$. With infinite samples, the (population) variance of $X_+$ in all the directions orthogonal to the span of $w_1,\ldots,w_k$ in $h^*$ is $\sigma^2$. On the other hand, the variance along directions in $\textrm{span}\{w_1,\ldots,w_k\}$ is less than $\sigma^2$ because any thresholding (or any convex restriction) can only make the variance smaller; quantitative bounds on the gap are given in Lemma~\ref{lem:Vempala}! Suppose the data is uncorrupted i.e., we are given $X$, we can just find the eigenspace corresponding to the $k$ smallest eigenvalues of the covariance matrix corresponding to $X_+$, and learn the hypothesis in the $k$-dimensional subspace.

%%%%%%%%%%%%%%%%%%%%%%%%%%%%%%%%%%%%%%%%%%%%%%%%%%%%%%%%%%%%%%%%%%%%%%%%%%%
\begin{lemma}[Lemma 4.8 in \citet{Vempala08}]\label{lem:Vempala}
Let $g$ be the standard Gaussian density function in $\R^n$ and $f: \R^n \to \R_+$ be any logconcave function. Define the function $h$ to the density $h(x) = f(x) g(x)/\beta$ where $\beta=\int_{\R^n} f(x) g(x) \dx$. Then for any unit vector $u \in \R^n$,
$$\mathrm{Var}_h(u^\top x) \le 1 - \frac{e^{-b^2}}{2\pi} ,$$
where the support of $f$ along $u$ is $[a_0, a_1]$ and $b=\min\sset{|a_0|, |a_1|}$. In the particular the above statement also holds when $f$ corresponds to the indicator function over any convex set.
\end{lemma}

%%%%%%%%%%%%%%%%%%%%%%%%%%%%%%%%%%%%%%%%%%%%%%%%%%%%%%%%%%%%%%%%%%%%%%%%%%%
\begin{figure}[htb]
\begin{center}
\fbox{\parbox{0.98\textwidth}{
%{\bf Algorithm for Robust Adversarial Training in Spectral norm}

{\bf Input:} Samples $\tilde{X} \in \R^{n \times m}$ with labels $y_1, \dots, y_m \in \sset{0, 1}$, $\sigma$, robustness parameter $\kappa \ge 1$, and the perturbation parameters $\delta$ and $q$. Set $\tau=k/(k+1)$. 
\begin{enumerate}
\item Split the samples into two parts of $T_1, T_2$ where $|T_2|=\poly(k,\sigma/\kappa \delta)$. 
\item Let $\tilde{X}_+$ be the positive examples in $T_1$ and let $m_+$ be the number of positive examples. Set $B=\tilde{X}_+ - \tfrac{1}{m_+} \tilde{X}_+ \bone \bone^\top$. 
\item If $m_+ < (\kappa \delta/\sigma) m$, output the trivial hypothesis $h(x)=(x_1 >0) \wedge (-x_1 >0)$.  
\item Else solve the convex program \eqref{least:sdp:obj} on input $B$, with parameters $\kappa, q$ to get a PSD matrix $Y$. 
\item Let $t$ be the number of eigenvalues of $Y$ that are at least $\tau$. Let $\widehat{\Pi}$ be the orthogonal projection given by the top $\min\sset{t,k}$ and $S_1'$ be this subspace.
\item We run a net argument on $S_1'$ to find a hypothesis $h:$
\begin{enumerate}
    \item Project the samples in $T_2$ onto the subspace $S_1'$ to get samples $\sset{(x'_j,y_j):j \in T_2}$ where $x'_j = \widehat{\Pi} \cdot \tilde{x}_j$. 
    \item Set $\eps=0.01 \cdot \max\big\{(\kappa \delta/\sigma)^3, (k+1) \eta \big\}$ where $\eta$ is the same parameter in Lemma~\ref{lem:intersect:projection}.
    \item Let $\mathcal{W}$ be an $\eps$-net of unit vectors in $S_1': \forall v \in S_1', \exists u \in W$ s.t. $\|v-u\|_2 \le \eps$. Let $T$ be an $\eps \sigma$-net of thresholds in $[-5\sigma \cdot \log 1/\eps, 5 \sigma \cdot \log 1/\eps]$.  
    \item Output any $h=\prod_{i=1}^k \bfone(w_i^{\top} x \ge \theta_i)$ with each $w_i \in \mathcal{W}$ and $\theta_i \in T$ satisfying $$\sum_{j \in T_2} \bfone\bigg(h(x'_j)=y_j\bigg) \ge |T_2| \cdot (1 - C \cdot k^{4/3} \eta^{1/3}).$$
\end{enumerate}
\end{enumerate}
}}
\end{center}
\caption{\label{ALG:intersection} Learning a robust intersection of halfspaces with training corruptions.}
\end{figure}

We are given a matrix $\tilde{X}$ that corresponds to a $\delta$-perturbation of $X$ (a training time perturbation). %we will use a similar convex-programming approach as in Section~\ref{sec:worstcase}. However, 
We will use a convex-programming approach as in Section~\ref{sec:worstcase}, but to find the robust analog of a {\em least} singular subspace for the covariance matrix corresponding to $\tilde{X}_+$ i.e., our goal is to find an (orthogonal) projection matrix $\Pi$ of rank $r$ that is $(\kappa,q)$ robust i.e., $\norm{\Pi}_{q \to 2} \le \kappa$ and that minimizes $\norm{\Pi B}_F^2$.  
%In what follows $q \in [2, \infty]$. Formally the optimization problem is  
% \begin{align}
%     \min_{\Pi} &\norm{\Pi B }_F^2 \label{least:int:obj}\\
%     \text{s.t.\ }& \Pi \text{ is a projection matrix of rank } \ge r, \text{ and } \norm{\Pi}_{q \to 2} \le \kappa \label{least:int:operatornorm}.
% \end{align}
%
%Section~\ref{sec:worstcase} that gives a $O(1)$ approximation on a suitably {\em denoised} matrix. But there are two challenges. The subspace of interest $\Pi^*$ corresponds to the bottom-$r$ singular space of $X_+X_+^\top$. Moreover as in  Section~\ref{sec:training:frob} $\tilde{X}_+ \tilde{X}_+^\top$ could be very far from $XX^\top$; so we will need to denoise the matrix appropriately. Our hope is to find a matrix $\widehat{A}$ that is $\delta$-close $A$, where $A$ satisfies $AA^\top$ ``='' $m_+I - X_+X_+^\top$. However such an $A$ may not even exist since $m_+ I - X_+ X_+^\top$ may not even be positive semi-definite! 

We consider the following mathematical programming relaxation for the problem.
\begin{align}
    \min_Y & \iprod{BB^\top, Y} \label{least:sdp:obj}\\
    \text{s.t. \ }& \tr(Y) = r \label{least:sdp:tr}\\
                & 0 \preceq Y \preceq I \label{least:sdp:spectralnorm}\\
    & \norm{Y}_{q \to q^*}\le \kappa^2 \label{least:sdp:norm}
\end{align}

Note that the above program is a relaxation where for any $(\kappa,q)$-robust (orthogonal) projection matrix $\Pi$  of rank $r$, $Y=\Pi$ is a feasible solution. Moreover as in Theorem~\ref{thm:worstcase:frob}, we can use the Ellipsoid algorithm along with a $O(1)$ factor separation oracle for the constraint \eqref{least:sdp:norm}, to find in polynomial time $\widehat{Y}$ such that $\norm{\widehat{Y}}_{q \to q^*} \le c_q \kappa^2$ (for some absolute constant $c_q>0$ and the objective value attained by $\widehat{Y}$ is at most the optimum solution value of the above program (up to arbitrarily small accuracy).

The following lemma shows that we will recover a $(O(\kappa), q)$-robust projection matrix that captures all the directions where $M_+$ takes value significantly smaller than $\sigma^2$. 

\begin{lemma}\label{lem:intersect:projection}
There exists a constant $c_1>0$ such that the projection matrix $\widehat{\Pi}$ output after step 5 of the algorithm satisfies $\norm{\widehat{\Pi}}_{q \to 2} \le c_1 \kappa$ and for any $\lambda \in (0,1)$
\begin{align*}
\forall v \in \mathbf{S}^{n-1} &\text{ s.t. } \Pi^*v =v \text{ and } v^\top M_+ v < \sigma^2(1- \lambda), ~\text{ we have } \norm{\widehat{\Pi} v}_2^2 \ge 1 - \frac{2(r+1)\eta}{\lambda}\\ &~\text{ where } \eta = 2 \sqrt{k} \cdot \frac{\kappa \delta }{\sigma}  + \frac{\kappa^2 \delta^2}{\sigma^2} + O\Big( \frac{n \log n}{\sqrt{m_+}} +  \frac{\kappa \delta}{\sigma} \cdot \frac{ \sqrt{n}\log n}{\sqrt{m_+}} \Big).    \end{align*}
\end{lemma}

The error $\eta$ in Theorem~\ref{thm:intersect2} inherits the same $\eta$ from the above lemma where we simplify the last two terms to $\kappa^2\delta^2/\sigma^2$ by assuming $m_{+} \ge (\kappa \delta/\sigma) m$ and $m=\poly(n,k,\sigma/\kappa\delta)$ is sufficiently large. We defer the proof of Lemma~\ref{lem:intersect:projection} to Section~\ref{sec:proof_intersect:proj} and finish the proof of Theorem~\ref{thm:intersect2} in the rest of this section. 
\xnote{I change the organization a little bit.}

%%%%%%%%%%%%%%%%%%%%%%%%%%%%%%%%%%%%%%%%%%%%%%%%%%%%%%%%%%%%%%%%%%%%%%%%%%%

%For convenience, given a subspace $V$, we say $h$ is an intersection of halfspaces in $V$ if $\exists w_1,\cdots, w_\ell \in \mathbb{R}^n$ and $\theta_1,\ldots,\theta_\ell$ such that each $w_i$ is in $V$ and $h(x)=\prod_{i=1}^{\ell} \bfone(w_i x \ge \theta_i)$. 
\anote{Moved definition about halfspace being in subspace to Notation portion.}
Lemma~\ref{lem:Vempala} implies the following claim (see the proof of Theorem 1.3 in \citet{Vempala08}).

\begin{claim}\label{clm:vempala_gap_eig}
Let $\sset{\lambda_i: i \in [n] }$ and $\sset{v_i : i \in [n]}$ be the eigenvalues and eigenvectors of $M_+$. Given any $\gamma \in (0,1)$, we set a subspace $S_1:=\textrm{span}\{v_i| \lambda_i < (1 - \gamma) \sigma^2 \}$. Then there exists $h_1 \in \mathcal{H}_k$ in the subspace $S_1$, that agrees with $h^*$ with probability at least $1-O(\gamma \cdot k)$ over the uncorrupted samples:
\[
\Pr_{x \sim  N(0,\sigma^2 I)}[h^*(x)=h_1(x)] > 1-O(\gamma \cdot k). 
\]
\end{claim}

\noindent We use the following claim along with Lemma~\ref{lem:intersect:projection} to show that the algorithm (up to step 5) recovers a subspace $S_1'$ very close to $S_1$.

\begin{claim}\label{claim:close_projection_intersection}
Given $\gamma<1/2$ and a subspace $S_1$ of dimension $k$ with projection $\Pi^*$, let $S'_1$ be a subspace whose projection matrix $\wt{\Pi}$ satisfies the following property: for every unit  vector $v$ in $S_1$, $\| \wt{\Pi} v \|_2 \ge 1 - \gamma$ . Then for any $h \in \mathcal{H}_k$  in subspace $S_1$, there exists another $h' \in \mathcal{H}_k$ in subspace $S_1'$ (given by a natural projection onto $S'_1$) that agrees with $h$ with probability at least $1-O(k \cdot \sqrt{\gamma \log \frac{1}{\gamma}})$ on an uncorrupted sample.
\end{claim}

\begin{proof}
For an intersection $h(x)=\prod_{i=1}^{k} \bfone(w_i^\top x \ge \theta_i)$ in $S_1$, we have $\Pi^* w_i=w_i$. Without loss of generality, we assume $\|w_i\|_2=1$. Let $w'_i=\wt{\Pi} w_i$ for each $i$ and $h'(x)=\prod_{i=1}^{k} \bfone( (w'_i)^\top x \ge \theta_i )$ be the projection of $h(x)$ into $S_1'$. So
\[
\Pr[h(x) \neq h'(x)] \le \sum_{i=1}^k \Pr\bigg[ \bfone\big( w_i^{\top} x \ge \theta_i \big) \neq \bfone\big( (w'_i)^{\top} x \ge \theta_i \big) \bigg].
\]
Next we bound each probability. Notice that $(w'_i)^\top x$ is a random variable drawn from $\sigma \cdot N(0,  \|w'_i\|_2^2)$ and $w_i^\top x$ is drawn from $(w'_i)^\top x + \sigma \cdot N(0,\|w_i\|_2^2 - \|w'_i\|_2^2)$. Note that $\norm{w'_i}_2^2 \ge (1-\gamma)^2$ and $\norm{w_i}_2^2=1$.
So for $c:= 6\sigma \sqrt{\gamma \log 1/\gamma}$
\begin{align*}
\Pr\bigg[ \bfone\big( w_i^{\top} x \ge \theta_i \big) \neq \bfone\big( (w'_i)^{\top} x \ge \theta_i \big)  \bigg] &\le \Pr_{g \sim \sigma \cdot  N(0,\|w'_i\|_2^2)}\bigg[ |g-\theta_i| \le c \bigg] + \Pr_{\substack{g' \sim \\\sigma \cdot N(0, \|w_i\|_2^2-\|w'_i\|_2^2)}}\bigg[ |g'| \ge c \bigg]\\
&\le O(\sqrt{\gamma \log 1/\gamma}).
\end{align*}
We get the last inequality by using Gaussian anticoncentration for the first term, and standard Gaussian tail bounds for the second term: 
since $\|w_i\|_2^2 - \|w'_i\|_2^2 \le 2\gamma$ from the guarantee of $S_1'$, the second term is bounded by $O(\sqrt{\gamma \log (1/\gamma)})$. Finally note that one can always rescale $\sset{w'_i}$ (along with the thresholds) to be unit vectors without changing $h'$.
\end{proof}

\noindent Next we use the VC dimension to bound the empirical risk error.
\begin{claim}\label{clm:VC_dimension_k_halfspaces}
Let $S_1'$ be a subspace of dimension $\ell$ with projection matrix $\wt{\Pi}$ and consider a fixed $h^* \in \mathcal{H}_k$ in subspace $S_1'$. Then $m=O(k \cdot l \log k/\eps)^2$ random Gaussian points $x_1,\ldots,x_m$ satisfy that with probability 0.99, any $h$, an intersection of $k$ halfspaces in $S_1'$, will have
\[
\Pr_{x \sim N(0, \sigma^2 I)}\big[ h(\wt{\Pi} x) = h^*(\wt{\Pi} x) \big] = \frac{1}{m} \sum_{i=1}^{m} \bfone\big( h(\wt{\Pi} x) = h^*(\wt{\Pi} x) \big) \pm \eps. \]
\end{claim}
\begin{proof}
We first bound the VC dimension of intersections of at most $k$ halfspaces by $2k(\ell+1) \log 5k$. The VC dimension of halfspaces in $S_1'$ is $\ell+1$. Then the intersections have VC dimension at most $2k(\ell+1) \cdot \log (5k)$. So by the learnability of VC dimension~\citep{VC_dimension}, the empirical error is $\eps$ for any $h$ with probability at least $0.99$.
\end{proof}

\noindent Finally we finish the proof of Theorem~\ref{thm:intersect2} assuming Lemma~\ref{lem:intersect:projection} and the above claims.  

\begin{proof}{of Theorem~\ref{thm:intersect2}}
%Recall that our goal is to learn $h^*$. 
  Set $\eta:=O\big( \sqrt{k} \cdot \frac{\kappa \delta }{\sigma} + \frac{\kappa^2 \delta^2}{\sigma^2} \big)$ and $\gamma:= ((k+1) \eta)^{1/3}$. Let $S^*$ be the subspace spanned by $w_1, \dots, w_k$; by assumption $\Pi^*$ is its projector. 
 Let $S_1$ be the subspace spanned by the eigenvectors of $M_{+}$ whose corresponding eigenvalues are at most $\sigma^2 (1-\gamma)$; note that the dimension of $S_1$ is at most $k$ and $S_1 \subset S^*$. Let $\bar{h}_1 \in \mathcal{H}_k$ be the classifier in the subspace $S_1$ given by Claim~\ref{clm:vempala_gap_eig} that approximates $h^*$ up to error $O(k \cdot \gamma)$. 
 
 Then we apply Lemma~\ref{lem:intersect:projection} to obtain $S_1'$ and $\wt{\Pi}$ such that $\|\wt{\Pi} v\|_2^2 \ge 1 - 2(k+1)\eta/\gamma$ for any unit vector $v$ in $S_1$, since for our choice of $\eta$, it holds that $\eta \ge 2 \kappa \delta \sqrt{k}/\sigma + \kappa^2 \delta^2/\sigma^2 +  O(\frac{n \log n}{\sqrt{m_{+}}}+\kappa \delta/\sigma \cdot \frac{\ \sqrt{n} \log n}{\sqrt{m_{+}}})$ and $m_+ \ge (\kappa \delta/\sigma) m$.
 %$\eta=2 \kappa \delta \sqrt{k}/\sigma + \kappa^2 \delta^2/\sigma^2 +  O(\frac{n \log n}{\sqrt{m_{+}}}+\kappa \delta/\sigma \cdot \frac{\ \sqrt{n} \log n}{\sqrt{m_{+}}})=O(\kappa \delta \sqrt{k}/\sigma + \kappa^2 \delta^2/\sigma^2)$ given $m_+ \ge (\kappa \delta/\sigma) m$. Also set $\gamma:= ((k+1) \eta)^{1/3}$.
 For convenience, let $h^*_1$ be the projection of $\bar{h}_1$ in $S_1'$ from Claim~\ref{claim:close_projection_intersection}. From the guarantees in Claim~\ref{clm:vempala_gap_eig} and~\ref{claim:close_projection_intersection}, we have 
\begin{equation}\label{eq:inter_h_1_*}
\Pr_{x \sim N(0,\sigma^2 I)}[h^*_1(x)=h^*(x)] \ge 1 - O(\gamma k) - O\Big(k \cdot \sqrt{ \tfrac{(k+1)\eta}{\gamma} \cdot \log \tfrac{\gamma}{(k+1)\eta } } \Big).
\end{equation} 

In the rest of this proof, we will focus on learning $h^*_1$, or a good proxy for it, using the second part of samples in $T_2$. For convenience, we use $m_2=|T_2|$ and $\big( \wt{x}_i,h(x_i) \big)_{i \in [m_2]}$ to denote the input. Since $h^*$ and $h_1^*$ are fixed and $m_2> (k^2 \log k/\eps)^2$, we have with probability at least $0.99$,
\begin{equation}\label{eq:inter_emp_h_*}
\frac{1}{m_2} \sum_{i=1}^{m_2} \bfone\big( h_1^*(x_i) \neq h^*(x_i) \big) \le O\Big(\gamma k +k \cdot \sqrt{ \tfrac{(k+1)\eta}{\gamma} \cdot \log (\tfrac{\gamma}{(k+1)\eta } )} \Big).
\end{equation}

Recall that $\mathcal{W}$ is an $\eps$-net of unit vectors in $S_1': \forall v \in S_1', \exists u \in W$ s.t. $\|v-u\|_2 \le \eps$. Let $T$ be an $(\eps\sigma)$-net of thresholds in $[-5\sigma\log 1/\eps, 5\sigma\log 1/\eps]$. So $\mathcal{W} \times T$ is a net for halfspaces in $S_1'$: for any halfspace $\bfone(w^\top x \ge \theta)$ with $w \in S_1'$, exist $w' \in \mathcal{W}$ and $\theta' \in T$ such that $\Pr[\bfone\big( w^\top x \ge \theta \big) = \bfone\big( (w')^\top x \ge \theta'\big)] \ge 1-O(\sqrt{\eps \log 1/\eps)})$. Similarly, $(\mathcal{W} \times T)^{\otimes k}$ gives a net for intersections of $k$ halfspaces in subspace $S_1'$: for any intersection $h$ of $k$ halfspaces, exist $w_1,\ldots,w_k \in \mathcal{W}$ and $\theta_1,\ldots,\theta_k \in T$ such that 
$$\Pr_{x \sim N(0,\sigma^2 I)}[h(x)=\prod_{j=1}^k \bfone(w_j^\top x \ge \theta_j)] \ge 1 - O(k \sqrt{\eps \log 1/\eps}).$$

Next we consider the empirical estimation over $m_2$ Gaussian points for $m_2>(k^2 \log k/\gamma)^2$.  By Claim~\ref{clm:VC_dimension_k_halfspaces}, we have 
\begin{equation}\label{eq:emp_before_perturbation}
\Pr_{x \sim N(0,\sigma^2 I)} \big[ h_1^*(x) = h(x) \big] = \frac{1}{m_2} \sum_{i=1}^{m_2} \bfone\bigg( h_1^*(x_i)=h(x_i) \bigg) \pm \gamma. 
\end{equation}

Now we consider the empirical estimation after \emph{perturbations}. For adversarial perturbations $\wt{x}_1,\ldots,\wt{x}_m$, we use $\|\wt{\Pi}\|_{q \to 2} = O(\kappa)$ to bound $\|\wt{\Pi}(x_i - \wt{x}_i) \|_2 \le O(\kappa \delta)$. Let us fix  $h \in \mathcal{H}_k$ in subspace $S_1'$; this could be $h_1^*$ or any classifier in the net $(\mathcal{W} \times T)^{\otimes k}$. For a random $x_i \sim N(0,I)$, its adversarial perturbation $\wt{x}_i$ changes its label in $h$ only if $w_i^{\top} \cdot (\wt{\Pi} x_i) \ge \theta_i$ and $w_i^{\top} \cdot (\wt{\Pi} \wt{x}_i) < \theta_i$ for some $i \in [k]$. Since their difference $w^{\top} \cdot (\wt{\Pi} X_i - \wt{\Pi} \wt{X}_i)$ is always bounded by $\|w\|_2 \cdot \|\wt{\Pi} X_i - \wt{\Pi} \wt{X}_i\|_2 \le \kappa \delta$ in absolute value. So $\wt{x}_i$ changes its label in $h$ with probability at most 
$$
k \cdot \Pr_{x_i} \bigg[ \big| w^\top (\wt{\Pi} x_i) \big| \le \kappa \delta \bigg] \le k \cdot \frac{2 \kappa \delta }{\sigma \cdot \sqrt{2 \pi}}.
$$ 
For $m_2$ random points $x_1,\ldots,x_{m_2}$, we have by standard concentration bounds,
$$
\Pr\left [\sum_{i=1}
^{m_2} \bfone\big( h(\wt{\Pi} \cdot x_i) \neq h(\wt{\Pi} \cdot \wt{x}_i) \big) \ge k \cdot \frac{2 \kappa \delta }{\sigma \sqrt{2 \pi}} \cdot {m_2} + 5 \sqrt{{m_2} \cdot k \log |\mathcal{W} \times T|} \right] \le \frac{1}{200 \cdot |\mathcal{W} \times T|^k }.
$$
For convenience, let $err$ denote the normalized error $k \cdot \frac{2 \kappa \delta }{\sigma \sqrt{2 \pi}} + 5 \sqrt{\frac{k \log |\mathcal{W} \times T|}{m_2}}$, which is $O(k \cdot \kappa \delta/\sigma + k^2 \log \frac{1}{\eps\sigma}/\sqrt{m_2})$. We apply a union bound over the net of classifiers to claim that with probability at least $0.99$,
\begin{equation}\label{eq:after_perturbation}
\frac{1}{m_2} \sum_{i=1}^{m_2} \bfone\big( h(\wt{\Pi} \cdot x_i) \neq h(\wt{\Pi} \cdot \wt{x}_i) \big) \le err \text{ for any } h \in \bigg\{ (\mathcal{W} \times T)^{\otimes k}, h_1^* \bigg\}.
\end{equation} 
Plug this into Equation~\eqref{eq:emp_before_perturbation}, we have
\begin{equation}\label{eq:intersection_emp_after_pert}
\Pr_{x}\big[ h_1^*(x)=h(x)\big] = \frac{1}{m_2} \sum_{i=1}^{m_2} \bfone\big( h_1^*(\wt{\Pi} \cdot \wt{x}_i)=h(\wt{\Pi} \cdot \wt{x}_i)\big) \pm \gamma \pm 2err.
\end{equation}

We are ready to show the correctness of Algorithm~\ref{ALG:intersection}. It will output a solution because there exists $h \in (\mathcal{W} \times T)^{\otimes k}$ very close to $h_1^*$: \begin{align*}
\frac{1}{m_2} \sum_{i=1}^{m_2} \bfone\big( h_1^*(\wt{\Pi} \wt{x}_i)=h(\wt{\Pi} \wt{x}_i)\big) & \ge \Pr_x[h_1^*(x)=h(x)] - \gamma - 2 err \tag{Equation~\eqref{eq:intersection_emp_after_pert}}\\
& \ge 1 - O(k \sqrt{\eps \log 1/\eps}) - \gamma - 2 err. \tag{from the property of the net}
\end{align*} At the same time, by equation~\eqref{eq:inter_emp_h_*}, we also have
\[
\frac{1}{m_2} \sum_{i=1}^{m_2} \bfone\big( h_1^*(\wt{\Pi} \wt{x}_i)=h^*(x_i)\big) \ge \frac{1}{m_2} \sum_{i=1}^{m_2} \bfone\big( h_1^*(\wt{\Pi} x_i)=h^*(x_i)\big) - err \ge 1 - O(\gamma k + k \sqrt{(k+1)\eta/\gamma}) - err.
\]
Thus with probability 0.9, we have
\[
\frac{1}{m_2} \sum_{i=1}^{m_2} \bfone\big( h(\wt{\Pi} \wt{x}_i)=h^*(x_i) \big)\ge 1 - O(\gamma k + k \sqrt{(k+1)\eta/\gamma} + k \sqrt{\eps \log 1/\eps} + err).
\]
We set the parameters and simplify the error: Let $\gamma=(k+1)^{1/3}\eta^{1/3}$ such that $\gamma = \sqrt{(k+1)\eta/\gamma}$ and $\eps=0.01 \cdot \max( \gamma^3,\kappa^3 \delta^3/\sigma^3)$ s.t. $\sqrt{\eps \log 1/\eps} \le \max\{\gamma,\kappa\delta\}$. Since $k<n^{1/2}$, the error becomes \xnote{Is this a over-simplification?}
$$
1 - O(k \gamma + err) = 1 - O\Big(k^{4/3} \cdot \eta^{1/3} + k \cdot \kappa \delta/\sigma + k^2 \log (\tfrac{1}{\eps})/\sqrt{m_2} \Big) = 1 - O(k^{4/3} \cdot \eta^{1/3}).
$$

We finish the proof by showing any $h$ in the net satisfying $\frac{1}{m_2} \sum_{i=1}^{m_2} \bfone\big( h(\wt{\Pi} \wt{x}_i)=h^*(x_i) \big) \ge 1 - c \cdot k^{4/3} \cdot \eta^{1/3}$ is close to $h^*$: $\Pr_{X \sim N(0,I)}[h(X)=h^*(X)] \ge 1 - O(k^{4/3} \cdot \eta^{1/3})$. By equation~\eqref{eq:after_perturbation}, we rewrite the guarantee of $h$ into
$$
\frac{1}{m_2} \sum_{i=1}^{m_2} \bfone\big( h(\wt{\Pi} x_i)=h^*(x_i) \big) \ge 1 - c k^{4/3} \cdot \eta^{1/3} - err.
$$
Furthermore, we use \eqref{eq:inter_emp_h_*} to rewrite it as
$$
\frac{1}{m_2} \sum_{i=1}^{m_2} \bfone\big( h(\wt{\Pi} x_i)=h_1^*(x_i) \big) \ge 1 - c k^{4/3} \cdot \eta^{1/3} - err -O(\gamma k + k \cdot \sqrt{(k+1)\eta/\gamma \cdot \log \gamma/(k+1)\eta }).
$$

By Equation~\eqref{eq:emp_before_perturbation}, we have
$$
\Pr[h(X)=h_1^*(X)] \ge 1 - c \cdot k^{4/3} \cdot \eta^{1/3} - err - O(\gamma k + k \cdot \sqrt{(k+1)\eta/\gamma \cdot \log \gamma/(k+1)\eta }) - \gamma.
$$
We combine it with \eqref{eq:inter_h_1_*} using triangle inequality to obtain
$$
\Pr[h(X)=h^*(X)] \ge 1 - c \cdot k^{4/3} \cdot \eta^{1/3} - err - \gamma - O(\gamma k + k \cdot \sqrt{(k+1)\eta/\gamma \cdot \log \gamma/(k+1)\eta })=1-O(k^{4/3} \cdot \eta^{1/3}).
$$

This establishes the required (natural) accuracy bound on the classifier $\wt{h}$. We now show the classifier $\wt{h}$ that we output is adversarially robust to test-time perturbations. Since $\|\wt{\Pi}\|_{\infty \to 2}= O(\kappa)$, we know the subspace spanned by the normals of halfspaces in $\wt{h}$ also has $\|\cdot\|_{\infty \to 2} \le \|\wt{\Pi}\|_{\infty \to 2}=  O(\kappa)$. Since each $w_i$ is in the subspace with $\|w_i\|_{q^*}=O(\kappa)$, we know $w_i \cdot (x-\wt{x})=O(\kappa \delta)$ for any $\delta$ perturbation in $\ell_q$ norm. Using the Gaussian anti-concentration again, we bound the probability of a flip in the label from $x$ to $\wt{x}$ for all valid $\delta$-perturbations $\wt{x}$ by $\Pr[\wt{h}(x) \neq \wt{h}(\wt{x})] \le O(k \kappa\delta/\sigma)$. By triangle inequality, this implies that the required robust error is at most $\eps+ O(k \kappa \delta/\sigma) \le 2\eps$.
\end{proof}

%%%%%%%%%%%%%%%%%%%%%%%%%%%%%%%%%%%%%%%%%%%%%%%%%%%%%%%%%%%%%%%%%%%%%%%%%%%
\subsection{Proof of Lemma~\ref{lem:intersect:projection}}\label{sec:proof_intersect:proj}
We first prove the following lemma which shows that that the robustness constraint ensures that the value on $M_+$ and $\tilde{M}_+$ are close.
\begin{lemma}\label{lem:error:intersection}
%Let $\tilde{X}$ be any $\delta$-perturbation of the uncorrupted samples $X \in \R^{m \times n}$ i.e., $\norm{\tilde{X}_j - X_j}_q \le \delta$ for all $j \in [k]$ generated according to the above process. 
%Let $M_+= I- \E\Big[ (x-\mu_+) (x-\mu_+)^\top ~|~ h^*(x)=+ \Big]$, and let $\tilde{M}_+= \tfrac{1}{m_+}(\tilde{X}_+ - \tilde{X}_+ \bone \bone^\top ) (\tilde{X}_+ - \tilde{X}_+ \bone \bone^\top )^\top$. 
For any projection matrix $Y$ of rank $r$ with $\norm{Y}_{q \to q^*} \le \kappa^2$, we have with probability $0.99$
\begin{equation}\label{eq:error:intersection}
\Big| \iprod{Y, M_+ - \tilde{M}_+} \Big| \le  2 \kappa \delta \sqrt{r} \sigma  + \kappa^2 \delta^2 + \gamma(m_+), \text{ where } \gamma(m_+) = O\Big( \frac{\sigma^2 n \log n}{\sqrt{m_+}} +   \frac{\sigma \kappa \delta \sqrt{n}\log n}{\sqrt{m_+}} \Big)
\end{equation}
goes to $0$ as the number of positive samples $m_+ \to \infty$. 
\end{lemma}

\begin{proof}
For convenience, let $\mu_+ = \E[x ~|~ h^*(x)=+]$, $\tilde{C}_+ = \tilde{X}_+ \bone \bone^\top$, $C_+= X \bone \bone^\top$. %, and let $C^{\Delta}_+ = \tilde{C}_+-C_+$. 
Let $Z_+^\top = \tilde{X}_+ -X_+ - (\tilde{C}_+ - C_+) $ (and similarly for $Z$). 
Consider the following shifted covariance matrices for the positive examples for the corrupted % and uncorrupted population average as given by
\begin{align}
    \tilde{M}_+&= \frac{1}{m_+} BB^\top= \frac{1}{m_+} \Big( (\tilde{X}_+-\tilde{C}_+) (\tilde{X}_+ - \tilde{C}_+)^\top \Big)\\
    &= \frac{1}{m_+} \Big( (X_+ - C_+) (X_+ - C_+)^\top + Z_+ Z_+^\top+  (X_+ - C_+) Z_+^\top + Z_+ (X_+ - C_+)^\top\Big) \nonumber,\\
    &= M_+ + E + \frac{1}{m_+} \Big( Z_+ Z_+^\top+  (X_+ - C_+) Z_+^\top + Z_+ (X_+ - C_+)^\top\Big), \\
%    M_+ &=I- \E\Big[ (x-\mu_+) (x-\mu_+)^\top ~|~ h^*(x)=+ \Big]\\
%\text{Finally, }~~   
\text{where } E &= \frac{1}{m_+} (X_+ - C_+) (X_+ - C_+)^\top- \E\Big[ (x-\mu_+) (x-\mu_+)^\top ~\big|~ h^*(x)=+ \Big]
\end{align}
represents the sampling error in the covariance matrix among positive examples, even in the absence of any corruptions. Moreover the samples (columns) of $X_+ - C_+$ are distributed according to a restriction of a spherical Gaussian onto a convex set. From Lemma~\ref{lem:Vempala}, it follows that the (population) variance of $X_+$ in every direction is at most $\sigma^2$. Hence, the operator norm
\begin{align*}
    \norm{\tfrac{1}{m_+}(X_+ - C_+) (X_+ - C_+)^\top} &\le \norm{M_+}+\Bignorm{\tfrac{1}{m_+}(X_+ - C_+) (X_+ - C_+)^\top - M_+ }  \\
    &\le \sigma^2 + O\Big(  \frac{\sigma^2 \sqrt{n} \log n}{\sqrt{m_+}} \Big).
\end{align*}
Recall that $\norm{Y^{1/2}}_{q \to 2} = \sqrt{\norm{Y}_{q \to q^*}}\le \kappa$. 
Hence, for any PSD matrix $Y$ satisfying $\tr(Y)=r, 0 \preceq Y \preceq I$ and $\norm{Y}_{q \to q^*}\le \kappa^2$,
\begin{align*}
    \Big| \iprod{Y, \tilde{M}_+- M_+} \Big| &\le \Bigiprod{Y, \tfrac{1}{m_+}Z_+ Z_+^\top} + |\iprod{Y,E}| + 2\Bigiprod{Y, \tfrac{1}{m_+}(X_+ - C_+) Z_+^\top} \Big|\\
        &\le \tr(Y) \norm{E}+ \frac{1}{m_+}\norm{Y^{1/2} Z_+}_F^2+   \frac{2}{m_+} \iprod{Y^{1/2} Z_+, Y^{1/2} (X_+ - C_+)} \\
        &\le r \norm{E} + \frac{1}{m_+} \underbrace{\norm{Y^{1/2} Z_+}_F^2}_{\text{ using } \norm{Y^{1/2}}_{q \to 2} \le \kappa} + 2 \frac{1}{\sqrt{m_+}} \norm{Y^{1/2} Z_+}_F \cdot \frac{1}{\sqrt{m_+}}\underbrace{\norm{Y^{1/2} (X_+ - C_+)}_F}_{\le \norm{Y^{1/2}}_F \cdot \norm{(X_+ - C_+)}}\\ %~\text{(by Cauchy-Schwartz inequality)}\\
        &\le O\Big(  \frac{r \sigma^2 n \log n}{\sqrt{m_+}} +  \frac{\sigma \kappa \delta \sqrt{n}\log n}{\sqrt{m_+}} \Big)+ 2 \kappa \delta \cdot \sqrt{r} \sigma + \kappa^2 \delta^2 ,% \text{ where } \gamma(m_+) = \dots \le \dots
\end{align*}
with probability $0.99$. %, where $\gamma(m_+)$ is the sampling error in Frobenius norm of estimating the covariance matrix with $m_+$ samples. By standard concentration bound, we have
%$\gamma(m_+)\le O(\sqrt{n^2 \log n/m_+})$ suffices. 
\end{proof}

We now prove Lemma~\ref{lem:intersect:projection}.
\begin{proof}[Proof of Lemma~\ref{lem:intersect:projection}]
Let $\eta:=2 \kappa \delta \sqrt{r} \sigma  + \kappa^2 \delta^2 + \gamma(m_+)$ be the error in Lemma~\ref{lem:error:intersection}. Also let $\Pi^*M_+ \Pi^* = \Sigma'$; by Lemma~\ref{lem:Vempala} we have $\Sigma' \preceq \sigma^2 I$. The optimal objective value of the relaxation \eqref{least:sdp:obj}
\begin{align*}
    \iprod{\Pi^*, \tilde{M}_+} & \ge \iprod{Y, \tilde{M}_+}   \\
    \iprod{\Pi^*, \Sigma'} = \iprod{\Pi^*, M_+} & \ge \iprod{Y, M_+} - 2\eta = \iprod{Y, \sigma^2 I - \sigma^2 \Pi^*+ \Sigma'} - 2\eta \\
    &= r -\iprod{Y, \Pi^* (\sigma^2 I-\Sigma') \Pi^*} - 2\eta,~~\text{ using Lemma~\ref{lem:error:intersection}}.\\
    \text{Hence } ~\iprod{Y, \Pi^* (\sigma^2 I-\Sigma') \Pi^*} & \ge  \iprod{\Pi^*, \Pi^*(\sigma^2 I- \Sigma') \Pi^*} - 2\eta \\
     \text{i.e.,  }~\iprod{\Pi^*(I-Y) \Pi^*, \Pi^* (\sigma^2 I-\Sigma') \Pi^*} & \le 2 \eta. 
\end{align*}
We have a unit vector $v$ in the subspace given by $\Pi^*$ with $v^\top (\sigma^2 I-\Sigma') v \ge \sigma^2 \lambda$. %Consider an orthonormal basis $v_1=v, v_2, \dots, v_r$ for $\Pi^*$.  
Moreover $0 \preceq Y \preceq I, \Sigma' \preceq \sigma^2 I$.
\begin{align}
    (v^\top (I-Y) v ) (v^\top (\sigma^2 I-\Sigma') v) &\le \iprod{\Pi^*(I-Y) \Pi^*, \Pi^* (\sigma^2 I-\Sigma') \Pi^*} \le 2\eta \nonumber\\
    (1 - v^\top Y v) \cdot \sigma^2 \lambda &\le 2\eta. ~~ \text{Hence } ~ v^\top Y v \ge 1- \frac{2\eta}{\lambda\sigma^2 }. \label{eq:Yv}
\end{align}

Let $Y=\sum_{i=1}^n \lambda_i u_i u_i^\top$ be the eigendecomposition of $Y$ with eigenvalues $\lambda_1 \ge \lambda_2 \ge \dots \ge 0$, and $\lambda_1 \le 1$. As in the algorithm let $t$ be the number of eigenvalues that are larger than $r/(r+1)$. Since $\tr(Y)=r$, we have $t \le r$. Note that $\widehat{\Pi}=\sum_{i=1}^t u_i u_i^\top$. By monotonicity of matrix norm $\norm{\widehat{\Pi^*}}_{q \to q^*} \le c_q (1+1/r) \kappa^2$. Finally,
\begin{align*}
1-\frac{2\eta}{\sigma^2 \lambda}  \le v^\top Y v & = \sum_{i=1}^n \lambda_i \iprod{u_i, v}^2 \le \sum_{i=1}^t \iprod{u_i, v}^2 + \lambda_{t+1} \sum_{i=t+1}^n \iprod{u_i, v}^2 \le 1-\norm{\Pi^{\perp} v}_2^2 + \frac{r}{r+1} \norm{\Pi^{\perp} v}_2^2 \\
\norm{\Pi^{\perp} v}_2^2 & \le \frac{2 (r+1) \eta}{\sigma^2 \lambda}, ~\text{ hence proving the lemma}.
\end{align*}

\end{proof}

\subsection{Properties of test-time robust classifiers}

The following claim shows that for any test-time robust classifier given by the intersection of $k=2$ halfspaces, the normals of the halfspaces are sparse. Moreover the subspace spanned by the normals is robust. We remark that a simple statement holds for general $k$ with a dependence on the least (non-trivial) singular value of the matrix given by normals.    
\begin{claim}\label{claim:robustintersection}
Let $h:\R^n \to \{0,1\}$ represent a classifier given by $h(x)= \bfone(w_1^\top x \ge 0) \cdot 1 (w_2^\top x\ge 0)$, where $\norm{w_1}_2 = \norm{w_2}_2 =1$ and $\norm{w_1 -w_2}_2 = \gamma \in (0,2)$. Suppose for $x \sim N(0,\sigma^2 I)$ with $h(x)=1$ \big(note that this happens with probability at least $\Omega(1-\gamma/2)$\big), we have with probability at least $2/3$ that
$$ \forall x'\in \R^n \text{ s.t. } \norm{x - x'}_q \le \delta, ~~ h(\tilde{x})= h(x).$$
Then there exists some universal constant $c>0$ such that we have $\max\sset{\norm{w_1}_{q^*} , \norm{w_2}_{q^*}} \le c \sigma/ \delta$. Moreover if $\Pi^*$ is the projection matrix onto the span of $w_1, w_2$, we have that $\norm{\Pi^*}_{q \to 2} \le (c/\gamma) \cdot  \sigma/ \delta $. 
\end{claim}
\begin{proof}
For $x \sim N(0, \sigma^2 I)$ (even conditioned on $h(x)=1$), we have that with probability at least $0.9$, $|w_1^\top x|, |w_2^\top x| \le O(\sigma)$. Let $\kappa' = \norm{w_1}_{q^*} \ge \norm{w_2}_{q^*}$ without loss of generality. Hence by norm duality, there exists a $z$ with $\norm{z}_q =\delta$ such that $w_1^\top z = \kappa \delta$. Thus if $\kappa \delta > c\sigma$ (for a large enough $c>0$), we have $w_1^\top (x-z) <0$. Hence the adversarial perturbation $\tilde{x}=x-z$ misclassifies the point. 

Let $\Pi^*$ be the projection matrix onto the span of $w_1, w_2$. Let $u$ be any vector in the subspace given by $\Pi^*$. It is easy to see that since $\norm{w_1 - w_2}_2 \ge \gamma$, $u$ can be expressed as a linear combination $u = \alpha_1 w_1 + \alpha_2 w_2$ where $\alpha_1^2 + \alpha_2^2 = O(1/\gamma^2)$.  This is because the minimum singular value of the matrix with columns $w_1, w_2$ is at least $\Omega(\gamma)$. Hence $\norm{u}_{q^*} \le |\alpha_1| \norm{w_1}_{q^*}+ |\alpha_2| \norm{w_2}_{q^*} \le O(\kappa/ \gamma)$. This proves the lemma.   

\end{proof}

\section{Trading off Natural Accuracy for Adversarial Robustness in Classification via Robust Projections}
\label{sec:gaussian-model}
\anote{Changed section heading from "Adversarial Robust Classification via Robust Projections"}
%\pnote{Use N instead of caligraphic N for normal distribution}
\anote{Organized the section into subsections, paragraphs to give more structure }
\anote{I think the intro paragraph can provide more motivation/context? Bring in sentences from the next page on this.}
\xnote{Oct 31st: Move two sentence from the intro in Section 2 to here.}
In many other natural scenarios it might be desirable to trade off natural accuracy for significant robustness to test-time perturbations. In this section we demonstrate how our techniques can be used for this purpose. 

We study the simple binary classification setting under a natural Gaussian model~\citep{Anderson_Multi_stat} for data generation. that was studied in recent works~\citep{tsipras2018robustness, schmidt2018adversarially}. In this model positive examples are drawn from a Gaussian distribution with mean $\mu_1$ and covariance matrix $\Sigma$, whereas the negative examples are drawn from a Gaussian with mean $\mu_2$ and the same covariance $\Sigma$. It is easy to see that if the means are well separated, e.g., if $\|\mu_1 - \mu_2\|_2 \geq \Omega(\sqrt{\log 1/\epsilon}) \sqrt{\norm{\Sigma}}$,\anote{did you want the more general condition?} then the Bayes optimal classifier will have error at most $\epsilon$. The works of~\citet{tsipras2018robustness, schmidt2018adversarially} used this simple model to study adversarial robustness and demonstrate that in many settings there is a natural tradeoff between the error and the robust error of any classifier (at test time); in particular there are settings where no robust classifier can achieve high natural accuracy. We continue this line of investigation and demonstrate that in many natural settings there do exists adversarially robust classifiers that also have small (natural) error. Furthermore, our algorithmic techniques can be used to learn such classifiers whereas standard approaches will fail. \anote{Moved it here from later. Feel free to move it back.} 

\anote{Edited the paragraph that follows.}
To demonstrate this, we will consider settings where the means $\mu_1, \mu_2$ remain reasonable well separated when projected onto a robust subspace $\Pi^*$, albeit by a smaller amount. Furthermore, the component of vector $\mu_1-\mu_2$ orthogonal to $\Pi^*$ will be spread out, as measured by the analytic sparsity $\ell_{q^*}/\ell_2$ (where $q=\infty$, this is $\ell_1/\ell_2$ sparsity). In such a setting, an adversary can make a small $\ell_q$ perturbation to a fresh test example and make the error of the Bayes optimal classifier close to half. On the other hand, projecting the data onto the robust subspace $\Pi^*$ first and performing classification in the subspace will let us tradeoff a small amount of natural accuracy for significant robustness gains.
%lead to classifier thatensure that both natural and adversarial accuracy remain small. 
We now formally define the Gaussian model.

\noindent \textbf{Gaussian Data Model $\calM(\mu_1, \mu_2, \Sigma)$.} In the Gaussian model an example label pair $(x,y) \in \R^n \times \{-1,+1\}$ is generated as follows. Pick $y=1$ with probability $1/2$ and $y=-1$ with probability $1/2$. Conditioned on $y$ generate $x$ as
\[
x|y \sim \left\{ \begin{array}{lr}
            {N}(\mu_1, \Sigma), \text{ if } y=1\\
            {N}(\mu_2, \Sigma), \text{ if } y=-1\\
             \end{array}\right\}
\]
We will use $(x,y) \sim \calM(\mu_1, \mu_2, \Sigma)$ to denote a labeled example generated via the above process. Given a positive-definite matrix $A$, we will use $A^+$ to denote the Penrose-Moore pseudoinverse of $A$. Given a classifier $f: \R^n \rightarrow \{-1,+1\}$ we define the error of $f$ to be the standard classification error given by
\begin{align}
    err(f) = \Psymb_{(x,y) \sim \calM(\mu_1, \mu_2, \Sigma)} [f(x) \neq y].
\end{align}
We will measure perturbations in $\ell_q$ norm for $q > 2$ and for a given perturbation radius $\delta > 0$, we define the robust error of $f$ to be
\begin{align}
    err_{rob}(f) = \Psymb_{(x,y) \sim \calM(\mu_1, \mu_2, \Sigma)} \Big[\exists z \in B_{\delta}(0): f(x+z) \neq y \Big], ~ \text{ where } B_{\delta}(0)=\{z: \|z\|_q \leq \delta\}.
\end{align}
%The works of~\cite{tsipras2018robustness, schmidt2018adversarially} used the above Gaussian model to study adversarial robustness and demonstrate that in many settings there is a natural tradeoff between the error and the robust error of any classifier, and in particular there might be settings where no robust classifier can achieve high natural accuracy. We continue this line of investigation and demonstrate that in many natural settings there do exists adversarially robust classifiers that also have small (natural) error. Furthermore, our algorithmic techniques can be used to learn such classifiers whereas standard approaches will fail. 
\anote{4/14: Very nice description here. Moved it to the front.} 
To demonstrate this, we make the following natural assumptions on the structure of the above Gaussian model. \anote{4/14 Added:} We suspect that these results will hold in more general settings as well. For notational convenience, $\norm{\cdot}$ will be used by default to denote the $\ell_2$ norm for vectors, and the spectral norm for matrices, in the rest of the section.  \anote{4/14: Added last line, just in case.}

\noindent \textbf{Assumptions I.}
\begin{enumerate}
    \item {\em Mean Separation.} For a fixed $\epsilon, \epsilon_1 \in (0,1)$ and a constant $c \geq 1$, there exists a rank-$r$ projection $\Pi^*$ with $\norm{\Pi^*}_{\infty \to 2} \leq \kappa$, $\delta > 0$, and $\alpha, \beta \in (0,1)$, such that
    \begin{align}
            \norm{(\mu_1 - \mu_2)} &\geq c \sqrt{\log \big(\frac 1 \epsilon \big)} \sqrt{\norm{\Sigma}}  \label{eq:mean-separation}\\
        \norm{\Pi^*(\mu_1 - \mu_2)} &\geq c \sqrt{\log \big(\frac{1}{\epsilon + \epsilon_1} \big)} \sqrt{\norm{\Sigma}} + \frac{\kappa \delta \beta}{\alpha} \label{eq:projected-mean-separation}\\
        \alpha^2 \norm{\Sigma} &\leq \sigma_r(\Pi^* \Sigma \Pi^*)  \leq \beta^2 \norm{\Sigma}. \label{eq:projected-covariance-bound}
    \end{align}
    \item {\em Spread Condition.} Let $v$ be the vector defined as $v =  \Sigma^+ (\mu_1 - \mu_2)$. Then we have that
    \begin{align}
        \norm{v}_{q^*} \geq n^{0.1 (\frac 1 2 - \frac{1}{q})} \norm{v}_2. \label{eq:spread-out-condition}
    \end{align}
    \item $\delta$ satisfies the following bound
    \begin{align}
        \delta \geq \frac{\sqrt{\norm{\Sigma}}\norm{(\Sigma^+)^{\frac 1 2}(\mu_1 - \mu_2})}{n^{0.1 (\frac 1 2 - \frac{1}{q})}}. \label{eq:bound-on-delta}
    \end{align}
\end{enumerate}
\anote{4/14 Added headings to make the flow clearer.}
The constant $0.1$ in \eqref{eq:spread-out-condition} above is chosen for ease of exposition and in general one can define a similar condition in terms of $n^c$ for a small constant $c > 0$.

\noindent {\bf Bayes Optimal Classifier: } For the Gaussian model above, the Bayes optimal classification for a point $x$ is obtained by comparing the density functions $p(x|y=1)$ and $p(x|y=-1)$. In particular we have that
\begin{align*}
    \frac{p(x|y=1)}{p(x|y=-1)} &= \frac{e^{-(x-\mu_1)^\top \Sigma^+ (x-\mu_1)}}{e^{-(x-\mu_2)^\top \Sigma^+ (x-\mu_2)}}.
\end{align*}
The above corresponds to the classifier
\begin{align}
    f^*(x) &= sgn \Big((x-\mu_2)^\top \Sigma^+ (x-\mu_2) - (x-\mu_1)^\top \Sigma^+ (x-\mu_1) \Big) \nonumber \\
    &= sgn \Big(\norm{(\Sigma^+)^{\frac 1 2} (x-\mu_2)}^2 - \norm{(\Sigma^+)^{\frac 1 2} (x-\mu_1)}^2 \Big). \label{eq:bayes-optimal}
\end{align}

\noindent {\bf Robust Projection-based Classifier:} In a similar manner we define the robust classifier that performs classification after projection onto $\Pi^*$. When projected onto $\Pi^*$, the conditional distribution of $x$ is again Gaussian with means either $\Pi^* \mu_1$ or $\Pi^* \mu_2$ depending on $y$, and the covariance matrix being $\Gamma = \Pi^* \Sigma \Pi^*$. Then we define the robust classifier as
\begin{align}
    f_{\Pi^*}(x) &=     
    sgn \Big(\norm{(\Gamma^+)^{\frac 1 2} \Pi^* (x-\mu_2)}^2 - \norm{(\Gamma^+)^{\frac 1 2} \Pi^* (x-\mu_1)}^2 \Big). \label{eq:bayes-optimal-robust}
\end{align}
The assumptions in \eqref{eq:mean-separation} and \eqref{eq:projected-mean-separation} ensure that the means are separated in the ambient space as well as when projected onto the robust subspace $\Pi^*$. This will ensure that $f^*(x)$ has error at most $\epsilon$ and at the same time $f_{\Pi^*}(x)$ is not too much worse and has error at most $\epsilon + \epsilon_1$. The assumption in \eqref{eq:projected-covariance-bound} will be used to argue that adversarial perturbations do not hurt the robust classifier $f_{\Pi^*}(x)$ and its robust error also remains at most $\epsilon + \epsilon_1$. Finally, the assumptions in \eqref{eq:spread-out-condition} and \eqref{eq:bound-on-delta} will be used to show that orthogonal to $\Pi^*$ the vector $\mu_1 - \mu_2$ is spread out and an adversary can take advantage of this fact to design test-time perturbations that make the robust error of the Bayes optimal classifier close to half. Notice that the assumptions allow for a fairly large range of the perturbation radius $\delta$ as stated in \eqref{eq:bound-on-delta}. We next formalize these arguments.

\subsection{Trading off Natural Accuracy for More Robustness Statistically}
\anote{Made it a proposition from a lemma, and added a name. }

\begin{proposition}[Non-robustness of the Bayes-optimal classifier]
\label{lem:bayes-classifier-gaussian-model}
If the Gaussian model $\calM(\mu_1, \mu_2, \Sigma)$ satisfies Assumptions I then for the Bayes optimal classifier $f^*(x)$ defined in \eqref{eq:bayes-optimal} it holds that
\begin{align*}
    err(f^*) &\leq \epsilon\\
    err_{rob}(f^*) &\geq \frac 1 2.
\end{align*}
\end{proposition}
\begin{proof}
We first analyze the error of $f^*$. From \eqref{eq:bayes-optimal} we have 
\begin{align*}
    err(f^*) &= \frac 1 2 \Psymb_{x \sim {N}(\mu_1, \Sigma)} \Big[\norm{(\Sigma^+)^{\frac 1 2} (x-\mu_2)}^2  \leq \norm{(\Sigma^+)^{\frac 1 2} (x-\mu_1)}^2 \Big]\\
    &+ \frac 1 2 \Psymb_{x \sim {N}(\mu_2, \Sigma)} \Big[\norm{(\Sigma^+)^{\frac 1 2} (x-\mu_1)}^2  \leq \norm{(\Sigma^+)^{\frac 1 2} (x-\mu_2)}^2 \Big] 
\end{align*}
Using the fact that when $x \sim {N}(\mu, \Sigma)$ then we have that $(\Sigma^+)^{\frac 1 2}(x-\mu)$ is distributed as $\Pi g$ where  $g \sim {N}(0,I)$ and $\Pi$ is the orthogonal projection onto the subspace spanned by the singular vectors of $\Sigma$. Then we get have
\begin{align*}
    err(f^*) &= \frac 1 2 \Psymb_{g \sim {N}(0, I)} \Big[\norm{\Pi g + (\Sigma^+)^{\frac 1 2} (\mu_1-\mu_2)}^2  \leq \norm{\Pi g}^2 \Big] \\
    &\quad + \frac 1 2 \Psymb_{g \sim {N}(0, I)} \Big[\norm{\Pi g + (\Sigma^+)^{\frac 1 2} (\mu_2-\mu_1)}^2  \leq \norm{ \Pi g}^2 \Big] \\
    &= \Psymb_{g \sim {N}(0, I)} \Big[\norm{\Pi g + (\Sigma^+)^{\frac 1 2} (\mu_1-\mu_2)}^2\leq \norm{\Pi g}^2 \Big]\\
    &= \Psymb_{g \sim {N}(0, I)} \Big[ \iprod{\Pi g, (\Sigma^+)^{\frac 1 2} (\mu_2 - \mu_1)} \geq \frac 1 2 \norm{(\Sigma^+)^{\frac 1 2} (\mu_2 - \mu_1)}^2 \Big].
\end{align*}
Since $\iprod{\Pi g, (\Sigma^+)^{\frac 1 2} (\mu_2 - \mu_1)}$ is a Gaussian with mean zero and variance $\norm{(\Sigma^+)^{\frac 1 2} (\mu_2 - \mu_1)}^2$, we know from standard facts about Gaussian distribution that with probability at least $1-\epsilon$, it holds that
\begin{align*}
|\iprod{\Pi g, (\Sigma^+)^{\frac 1 2} (\mu_2 - \mu_1)}| \leq c'\sqrt{\log (\frac 1 \epsilon)} \norm{(\Sigma^+)^{\frac 1 2} (\mu_2 - \mu_1)},
\end{align*}
for a universal constant $c' > 0$. Furthermore, \eqref{eq:mean-separation} implies that
\begin{align*}
\norm{(\Sigma^+)^{\frac 1 2} (\mu_2 - \mu_1)} &\geq \frac{1}{\sqrt{\norm{\Sigma}}} \norm{\mu_2 - \mu_1} > c \sqrt{\log (\frac 1 \epsilon)}.
\end{align*}
Setting $c > 2c'$ we get that
\begin{align*}
    err(f^*) &= \Psymb_{g \sim {N}(0, I)} \Big[ \iprod{\Pi g, (\Sigma^+)^{\frac 1 2} (\mu_2 - \mu_1)} \geq \frac 1 2 \norm{(\Sigma^+)^{\frac 1 2} (\mu_2 - \mu_1)}^2 \Big]\\
    &\leq \epsilon.
\end{align*}
Next we analyze the robust error of $f^*$. We have
\begin{align*}
    err_{rob}(f^*) &= \frac 1 2 \Psymb_{x \sim {N}(\mu_1, \Sigma)} \Big[\sup_{z \in B_\delta(0)}\norm{(\Sigma^+)^{\frac 1 2} (x+z-\mu_1)}^2  - \norm{(\Sigma^+)^{\frac 1 2} (x+z-\mu_2)}^2 \geq 0\Big]\\
    &+ \frac 1 2 \Psymb_{x \sim {N}(\mu_2, \Sigma)} \Big[\sup_{z \in B_\delta(0)} \norm{(\Sigma^+)^{\frac 1 2} (x+z-\mu_2)}^2  - \norm{(\Sigma^+)^{\frac 1 2} (x+z-\mu_1)}^2 \geq 0\Big] 
\end{align*}
Again from symmetry we can rewrite this as
\begin{align*}
    err_{rob}(f^*) &= \Psymb_{g \sim {N}(0, I)} \Big[\sup_{z \in B_\delta(0)} \norm{\Pi g + (\Sigma^+)^{\frac 1 2} z}^2 - \norm{\Pi g + (\Sigma^+)^{\frac 1 2} (\mu_1-\mu_2) + \Sigma^{-\frac 1 2}z }^2  \geq 0 \Big]\\
    &= \Psymb_{g \sim {N}(0, I)} \Big[ \iprod{\Pi g, (\Sigma^+)^{\frac 1 2} (\mu_2 - \mu_1)} + \sup_{z \in B_\delta(0)} \iprod{(\Sigma^+)^{\frac 1 2}z, (\Sigma^+)^{\frac 1 2} (\mu_2 - \mu_1)} \geq \frac 1 2 \norm{(\Sigma^+)^{\frac 1 2} (\mu_2 - \mu_1)}^2 \Big]\\
    &= \Psymb_{g \sim {N}(0, I)} \Big[ \iprod{\Pi g, (\Sigma^+)^{\frac 1 2} (\mu_2 - \mu_1)} + \delta \norm{\Sigma^+(\mu_2 - \mu_1)}_{q^*} \geq \frac 1 2 \norm{(\Sigma^+)^{\frac 1 2} (\mu_2 - \mu_1)}^2 \Big].
\end{align*}
We will next show that $\delta \norm{\Sigma^+(\mu_2 - \mu_1)}_{q^*} \geq \frac 1 2 \norm{(\Sigma^+)^{\frac 1 2} (\mu_2 - \mu_1)}^2$, thereby establishing that $err_{rob}(f^*) \geq 1/2$. From \eqref{eq:spread-out-condition} we have that
\begin{align*}
\delta \norm{\Sigma^+(\mu_2 - \mu_1)}_{q^*} &\geq \delta n^{0.1(\frac 1 2 - \frac{1}{q})} \norm{\Sigma^+(\mu_2 - \mu_1)}_2\\
&\geq \sqrt{\norm{\Sigma}}\norm{(\Sigma^+)^{\frac 1 2}(\mu_1-\mu_2)}\norm{\Sigma^+(\mu_2 - \mu_1)}_2, \,\, \text{[ using the lower bound on $\delta$ in \eqref{eq:bound-on-delta}]}\\
&\geq \norm{(\Sigma^+)^{\frac 1 2} (\mu_2-\mu_1)}^2.
\end{align*}
\end{proof}

\anote{Changed it to proposition}
\begin{proposition}[Guarantees for the Robust Projection-based classifier]
\label{lem:robust-classifier-gaussian-model}
If the Gaussian model $\calM(\mu_1, \mu_2, \Sigma)$ satisfies Assumptions I then for the robust classifier $f_{\Pi^*}(x)$ defined in \eqref{eq:bayes-optimal-robust} it holds that
\begin{align*}
    err(f_{\Pi^*}) &\leq \epsilon + \epsilon_1\\
    err_{rob}(f_{\Pi^*}) &\leq  \epsilon + \epsilon_1.
\end{align*}
\end{proposition}
\begin{proof}
Notice that $\Pi^* x$ is distributed as ${N}(\Pi^* \mu, \Gamma)$ when $x \sim {N}(\mu, \Sigma)$, where $\Gamma = \Pi^* \Sigma \Pi^*$. Hence similar to the calculation in Proposition~\ref{lem:bayes-classifier-gaussian-model} we have that $f_{\Pi^*}(x)$ has error at most $\epsilon + \epsilon_1$ provided 
$$
\norm{\Pi^*(\mu_2 - \mu_1)} \geq c \sqrt{\log \big(\frac{1} {\epsilon + \epsilon_1} \big)} \sqrt{\norm{\Gamma}}.
$$
From \eqref{eq:projected-mean-separation} and noticing that $\norm{\Gamma} \leq \norm{\Sigma}$ we get that $err(f_{\Pi^*}) \leq \epsilon + \epsilon_1$. Next we analyze the robust error of the classifier. Again from the calculations in the previous lemma we have $err_{rob}(f_{\Pi^*})$:
\begin{align*}
     &= \Psymb_{g \sim {N}(0, I)} \Big[ \iprod{\Pi^* g, (\Gamma^+)^{\frac 1 2} \Pi^*(\mu_2 - \mu_1)} + \sup_{z \in B_\delta(0)} \iprod{\Pi^* z, \Gamma^+ \Pi^*(\mu_2 - \mu_1)} \geq \frac 1 2 \norm{(\Gamma^+)^{\frac 1 2} \Pi^* (\mu_2 - \mu_1)}^2 \Big].
\end{align*}
Next notice that
\begin{align*}
    \sup_{z \in B_\delta(0)} \iprod{\Pi^* z, \Gamma^+ \Pi^*(\mu_2 - \mu_1)} &\leq \sup_{z \in B_\delta(0)}\norm{\Pi^* z} \norm{\Gamma^+ \Pi^* (\mu_2 - \mu_1)}\\
    &\leq \kappa \delta \frac{\norm{(\Gamma^+)^{\frac 1 2}\Pi^* (\mu_2 - \mu_1)}}{\alpha \sqrt{\norm{\Sigma}}}. \,\, \text{[from \eqref{eq:projected-covariance-bound}]}
\end{align*}
Notice that $\iprod{\Pi^* g, (\Gamma^+)^{\frac 1 2} \Pi^*(\mu_2 - \mu_1)}$ is a Gaussian that with probability at least $1-\epsilon - \epsilon_1$ takes a value at most $c'\norm{(\Gamma^+)^{\frac 1 2} \Pi^*(\mu_2 - \mu_1)}\sqrt{\log \big(\frac{1}{\epsilon+\epsilon_1} \big)}$. Hence to ensure that $err_{rob}(f_{\Pi^*}) \leq \epsilon + \epsilon_1$ it is enough to have
\begin{align*}
    \frac{1}{2} \norm{u}^2 \geq c' \sqrt{\log \big(\frac{1}{\epsilon + \epsilon_1} \big)}\norm{u} + \frac{\kappa \delta}{\alpha \sqrt{\norm{\Sigma}}}\norm{u},
\end{align*}
where $u = (\Gamma^+)^{\frac 1 2} \Pi^*(\mu_2 - \mu_1)$. The bound follows from \eqref{eq:projected-mean-separation} and noticing that \eqref{eq:projected-covariance-bound} implies that
\begin{align*}
\norm{u} &= \norm{(\Gamma^+)^{\frac 1 2} \Pi^*(\mu_2 - \mu_1)}\geq \frac{1}{\beta \sqrt{\norm{\Sigma}}} \norm{\Pi^* (\mu_2 - \mu_1)}.
\end{align*}
\end{proof}

\anote{Made it into a subsection}
%\noindent \textbf{Finding the robust classifier.} 
\subsection{Efficient Algorithms for Finding a Robust Classifier}
Next we discuss how our techniques from Section~\ref{sec:recovery} can be used to find a robust classifier in the Gaussian model discussed above. Given labeled examples, the first step of the learning algorithm is to use %existing methods for learning a mixture of two Gaussians 
\anote{4/14: Removed mention of mixtures of Gaussians' algorithms}
standard estimators for mean, covariance of single Gaussians on both the positive and negative examples separately, to find approximate parameters $\hat{\mu}_1, \hat{\mu}_2, \hat{\Sigma}$. Our recovery guarantee will depend on $\gamma_1, \gamma_2$ such that $\norm{(I-\Pi^*)(\Sigma^+)^{\frac 1 2}}_2 \leq \gamma_1 \norm{(\Sigma^+)^{\frac 1 2}}$ and $\sigma_r(\Pi^* (\Sigma^+)^{\frac 1 2} \Pi^*) \geq \gamma_2 \norm{(\Sigma^+)^{\frac 1 2}}$.
In this case the SDP from \eqref{eq:spectralsdp} when run on $\hat{\Sigma}^+$ implies that we will obtain a rank-$r$ subspace $\hat{\Pi}$ such that $\norm{\hat{\Pi}^{\perp}\Pi^*} \leq O(\frac{\gamma_1}{\gamma_2})$. Using the same analysis as above we can then show that the classifier 
$$
f_{\hat{\Pi}}(x) = sgn \Big(\norm{ ({\hat{\Gamma}}^+)^{\frac 1 2} \hat{\Pi}(x- \hat{\mu}_2)}^2 - \norm{ ({\hat{\Gamma}}^+)^{\frac 1 2} \hat{\Pi}(x- \hat{\mu}_1)}^2 \Big),
$$
where $\hat{\Gamma} = \hat{\Pi} \hat{\Sigma} \hat{\Pi}$. The learning algorithm is described below. To analyze the robust classifier we mildly strengthen Assumptions I below, to account for the error in the estimate of $\Pi^*$. 

\noindent \textbf{Assumptions II.}
\begin{itemize}
    \item For a fixed $\epsilon \in (0,1)$ and a constant $c \geq 1$, there exists a rank-$r$ projection $\Pi^*$ with $\norm{\Pi^*}_{\infty \to 2} \leq \kappa$, $\delta > 0$, and $\alpha, \beta, \gamma_1, \gamma_2 \in (0,1)$, such that
    \begin{align}
        \norm{\Pi^*(\mu_1 - \mu_2)} -\frac{10\gamma_1}{\gamma_2} \norm{\mu_1 - \mu_2} &\geq c \sqrt{\log \big(\frac 1 \epsilon \big)} \sqrt{\norm{\Sigma}} + \frac{\kappa \delta \beta}{\alpha'} \label{eq:projected-mean-separation-2}\\
        \alpha^2 \norm{\Sigma} &\leq \sigma_r(\Pi^* \Sigma \Pi^*) \leq \beta^2 \norm{\Sigma} \label{eq:projected-covariance-bound-2}\\
        \norm{(I-\Pi^*)(\Sigma^+)^{\frac 1 2}} &\leq \gamma_1 \norm{(\Sigma^+)^{\frac 1 2}} \label{eq:projection-of-Sigma-on-Pi}\\
        \sigma_r(\Pi^* (\Sigma^+)^{\frac 1 2} \Pi^*) &\geq \gamma_2 \norm{(\Sigma^+)^{\frac 1 2}}. \label{eq:projection-of-Sigma-on-Pi-lb}       
    \end{align}
    % \item Let $v$ be the vector defined as $v =  \Sigma^{-1} (\mu_1 - \mu_2)$. Then we have that
    % \begin{align}
    %     \norm{v}_1 \geq d^{0.1} \norm{v}_2. \label{eq:spread-out-condition}
    % \end{align}
    % \item $\delta$ satisfies the following bound
    % \begin{align}
    %     \frac{\sqrt{\norm{\Sigma}}\norm{\Sigma^{-\frac 1 2}(\mu_1 - \mu_2})}{d^{0.1}} \leq \delta \leq \frac{\alpha \norm{\Pi^*(\mu_1 - \mu_2)}}{\kappa}. \label{eq:bound-on-delta}
    % \end{align}
\end{itemize}
Here $\alpha' = \sqrt{\alpha^2 - \frac{8\gamma_1}{\gamma_2}}$.

\begin{figure}[htbp]
\begin{center}
\fbox{\parbox{0.98\textwidth}{
{\bf RobustClassification($(x_1, y_1), \dots, (x_m, y_m)$)}

{\bf Input:} Labeled examples $(x_1, y_1), \dots, (x_m, y_m)$.
\begin{enumerate}
%\item Let $\tilde{A}$ be the corrupted $n \times m$ data matrix with columns $A_i$ for $i \in [m]$. Let $\kappa$ be the given input value representing the bound on $\|\Pi^*\|_{q \to 2}$.
\item Use the first $\frac m 2$ examples to get estimates of $\hat{\mu}_1, \hat{\mu}_2, \hat{\Sigma}$ from standard estimators applied to positive and negative examples separately. \anote{Removed KMV mention.}
\item Use the algorithm from Corollary~\ref{corr:recovery:fr} with $A = ({\hat{\Sigma}}^+)^{\frac 1 2}$ to get $\hat{\Pi}$. \anote{4/14: Should this SDP be run on $A=\hat{\Sigma}^{1/2}$ or actually $\hat{\Sigma}^{-1/2}$? Also, maybe you just want to say Algorithm from Theorem/Corollary? because there's rounding involved as well?}
\item Define $\hat{\Gamma} = \hat{\Pi} \hat{\Sigma} \hat{\Pi}$. 
\item Return the classifier
$$
f_{\hat{\Pi}}(x) = sgn \Big(\norm{ ({\hat{\Gamma}}^+)^{\frac 1 2} \hat{\Pi}(x- \hat{\mu}_2)}^2 - \norm{ ({\hat{\Gamma}}^+)^{\frac 1 2} \hat{\Pi}(x- \hat{\mu}_1)}^2 \Big).
$$
\end{enumerate}
}}
\end{center}
\caption{\label{ALG:Robust-Classifier} Adversarially Robust Classification.}
\end{figure}

A natural illustrative example to keep in mind is the case when the two Gaussians are in the {\em parallel pancake} orientation~\citep{brubaker2008isotropic}, where the variance is small along one  direction $u$ e.g., the x-axis (more generally, the variance is small along a robust subspace), and the variance in all the orthogonal directions are very large. If $u$ is sparse (robust), and in addition the vector between the means has a reasonable projection on to $u$, then our algorithmic techniques can approximate $u$ well with another sparse vector $\hat{u}$. Performing classification after projecting onto $\hat{u}$ will lead to a significantly more robust classifier without hurting the natural accuracy a lot.
%, then our algorithmic techniques can be used to find a robust classifier.

% A natural setting where the above assumptions are expected to hold is the case when the two Gaussians are in the {\em parallel pancake} orientation~\cite{brubaker2008isotropic}. Here the two means are separated along a direction of small variance. If in addition, this direction has a reasonable projection on to a robust subspace then our algorithmic techniques can be used to find a robust classifier.
Next we state and prove our main theorem regarding the Algorithm in Figure~\ref{ALG:Robust-Classifier}.
\begin{theorem}
\label{thm:robust-classifier}
Given $\epsilon \in (0,1)$, perturbation radius $\delta > 0$, and $m = poly(n, 1/\epsilon, 1/\norm{\Sigma^+}, 1/\norm{\mu_1 - \mu_2})$ labeled examples from the Gaussian data model $\calM(\mu_1, \mu_2, \Sigma)$ satisfying Assumptions II, the algorithm in Figure~\ref{ALG:Robust-Classifier} runs in polynomial time and with high probability outputs a classifier $f_{\hat{\Pi}}(x)$ such that $err_{rob}(f_{\hat{\Pi}}) \leq \epsilon$.
\end{theorem}
\begin{proof}
\anote{4/14: Removed mention of KMV. But need to modify the bound for $\epsilon' = \epsilon/poly(n, 1/\norm{\mu_1 - \mu_2}_2, \norm{\Sigma})$. Removed dependence on mean sep. Not sure about covariance..} 
By taking the empirical mean and covariance of the positive and negative examples with $\epsilon'$ set to be  $\epsilon' = \min \big(\epsilon/poly(n, \norm{\Sigma^+}), (\gamma_1/\gamma_2) \norm{\mu_1-\mu_2} \big)$ we get estimates $\hat{\mu}_1, \hat{\mu}_2, \hat{\Sigma}$ such that
\begin{align*}
    \norm{\mu_1 - \hat{\mu}_1} \leq \epsilon', ~~~    \norm{\mu_2 - \hat{\mu}_2} \leq \epsilon', ~ \text{ and }   \norm{\Sigma - \hat{\Sigma}}_2 \leq \epsilon'.
\end{align*}
The above combined with \eqref{eq:projected-covariance-bound-2} and \eqref{eq:projection-of-Sigma-on-Pi} implies that
\begin{align*}
    \sigma_r(\Pi^* ({\hat{\Sigma}}^+)^{\frac 1 2} \Pi^*) \geq (1-o(1)) \gamma_2 \norm{(\Sigma^+)^{\frac 1 2}}\\
    \norm{(I-\Pi^*)({\hat{\Sigma}}^+)^{\frac 1 2}} \leq (1+o(1))\gamma_1 \norm{(\Sigma^+)^{\frac 1 2}}.
\end{align*}
Hence, Corollary~\ref{corr:recovery:fr} implies that step 2 of Algorithm~\ref{ALG:Robust-Classifier} will output $\hat{\Pi}$ such that $\norm{\hat{\Pi}^{\perp} \Pi^*} \leq \frac{4 \gamma_1}{\gamma_2}$. Next, to establish that $f_{\hat{\Pi}}$ has robust error at most $\epsilon$, we need to verify that \eqref{eq:projected-mean-separation}, \eqref{eq:projected-covariance-bound} from Assumptions I hold with $\hat{mu}_1, \hat{\mu}_2, \hat{\Sigma}$ and $\hat{\Pi}$. We have
\begin{align*}
    \norm{\hat{\Pi} \hat{\Sigma} \hat{\Pi}} &\geq \norm{{\Pi^*} \hat{\Sigma} \hat{\Pi}} - \norm{(\hat{\Pi}-\Pi^*) \hat{\Sigma} \hat{\Pi}}\\
    &\geq \norm{{\Pi^*} \hat{\Sigma} {\Pi^*}} - \norm{{\Pi^*} \hat{\Sigma} (\hat{\Pi}-\Pi^*)} - 4\frac{\gamma_1}{\gamma_2} \norm{\Sigma}\\
    &\geq \norm{{\Pi^*} {\Sigma} {\Pi^*}} - \norm{{\Pi^*} (\hat{\Sigma}-\Sigma) {\Pi^*}} - 8\frac{\gamma^2}{\alpha^2} \norm{\Sigma}\\
    &\geq \alpha^2 \norm{\Sigma} - \epsilon' \norm{\Sigma} - 8\frac{\gamma_1}{\gamma_2}\\
    &\geq (1-o(1)) \Big(\alpha^2-8\frac{\gamma_1}{\gamma_2} \Big) \norm{\Sigma} \,\, \text{[from \eqref{eq:projected-covariance-bound}]}.
\end{align*}
Finally, we have
\begin{align*}
    \norm{\hat{\Pi}(\hat{\mu}_1 - \hat{\mu_2})} &\geq \norm{\hat{\Pi}({\mu}_1 - {\mu_2})} - 2\epsilon'\\
    &\geq \norm{{\Pi^*}({\mu}_1 - {\mu_2})} - 2\epsilon' - \norm{(\hat{\Pi} - \Pi^*)({\mu}_1 - {\mu_2})}\\
    &\geq \norm{{\Pi^*}({\mu}_1 - {\mu_2})} - 2\epsilon' - 4\frac{\gamma_1}{\gamma_2}\norm{\mu_1 - \mu_2}\\
    &\geq c \sqrt{\log \big(\frac 1 \epsilon \big)} \sqrt{\norm{\Sigma}} + \frac{\kappa \delta}{\alpha'} \,\, \text{[from \eqref{eq:projected-mean-separation-2}]}.
\end{align*}
\end{proof}
To demonstrate the applicability of Theorem~\ref{thm:robust-classifier} we instantiate it for a special case when $\Sigma = I - \theta \Pi^*$ for a constant $\theta < 1$ and $\Pi^*$ being a rank-$r$ $\kappa$-robust projection matrix. In this case the assumptions simplify to get the following corollary that we state for the case of $q = \infty$. We remark that this particular instantiation is only used to demonstrate the flavour of the condition; in this specialized setting one can of course use our knowledge of the form of covariance matrix to just find $\Pi^*$ (by subtracting off the identity matrix), instead of our algorithm in step 2.  
\anote{4/14: Added the above line as a disclaimer.} \anote{Is $\theta>1$ or $\theta<1$?}
\begin{corollary}
Let $s < 1$ be a fixed constant and $\calM(\mu_1, \mu_2, \Sigma)$ be the Gaussian data model with $\Sigma = I-\theta \Pi^*$ with $\theta \in (0,s)$ such that
\begin{align*}
    \norm{\Pi^* (\mu_1 - \mu_2)} - {10}\sqrt{1-\theta} \norm{\mu_1 - \mu_2} &\geq c\sqrt{\log (\frac 1 \epsilon)} + c'\kappa \delta\\
    \norm{(I-\Pi^*)(\mu_1 - \mu_2)}_1 &\geq (n^{0.1} + \frac{\kappa}{1-\theta}) \norm{\mu_1 - \mu_2},
\end{align*}
where $c'$ is constant that depends on $s$. If $\delta$ satisfies
\begin{align*}
    \frac{\norm{\mu_1 - \mu_2}}{n^{0.1}} \leq \delta \leq \frac{\norm{\Pi^* (\mu_1 - \mu_2)}}{\kappa},
\end{align*}
then we have that the robust error of the Bayes optimal classifier satisfies $err_{rob}(f^*) \geq \frac 1 2$. On the other hand given $m = poly(n, 1/\epsilon, 1/\theta)$ labeled examples from the Gaussian data model $\calM(\mu_1, \mu_2, \Sigma)$ satisfying Assumptions II, the algorithm in Figure~\ref{ALG:Robust-Classifier} runs in polynomial time and with high probability outputs a classifier $f_{\hat{\Pi}}(x)$ such that $err_{rob}(f_{\hat{\Pi}}) \leq \epsilon$.
\end{corollary}
\section{Related Work}
\label{sec:related-work}
\anote{4/13: Do we need to add more related work?}

\paragraph{Adversarial Robustness.} Existing theoretical work on adversarial robustness has almost exclusively focused on supervised learning, and in particular on binary classification. These works include the study of adversarial counterparts of notions such as VC dimension and Rademacher complexity~\citep{cullina2018pac, khim2018adversarial, yin2018rademacher}, evidence of computational barriers~\citep{bubeck2018adversarial, bubeck2018adversarial2, nakkiran2019adversarial, degwekar2019computational} and statistical barriers towards ensuring both low test error and low adversarial test error~\citep{tsipras2018robustness}, and computationally efficient algorithms for adversarially robust learning of restricted classes such as degree-$1$ and degree-$2$ polynomial threshold functions~\citep{awasthi2019robustness}. Furthermore, recent works also provide evidence that adversarially robust supervised learning requires more training data than its non-robust counterpart~\citep{schmidt2018adversarially, montasser2019vc, min2020curious}. 
The closest to our work is the result of~\citet{garg2018spectral} that studies a particular formulation of adversarially robust features. The authors consider computing, given i.i.d. samples from a distribution, a map $f$ such that, with high probability over a new example $x$ drawn from the same distribution, points close to $x$ get a nearby mapping (in $\ell_2$ distance) under $f$. While similar in motivation to our work, the results in~\citet{garg2018spectral} do not aim to minimize the projection error and simply require the projection $f$ to be mean zero and variance one to avoid trivial solutions. Furthermore, the authors look at a specific type of spectral embedding given by the top eigenvectors of the Laplacian of an appropriate graph constructed on the training data. The bounds presented for this embedding depend on the eigenvalue gap present in the Laplacian matrix. Finally, it is not clear how to efficiently use the embedding on new test points, as it involves recomputing the Laplacian by incorporating the new point into the training set. 
\pnote{Need to mention Woodruff stuff and column subset selection.}\pnote{Done.} 

\paragraph{Low Rank Approximations.} There is a large body of work in randomized numerical linear algebra on methods such as column subset selection and CUR decompositions~\citep{kannan2017randomized, boutsidis2009improved, deshpande2010efficient, boutsidis2014near, drineas2008relative, boutsidis2017optimal, song2017low} that aim to approximate a given matrix via a low dimensional subspace spanned by a small number of rows/columns of the matrix. %These formulations, although study low dimensional representations, are not motivated from a robustness perspective. %In contrast, we aim to find a robust low dimensional subspace that approximates the data. 
%In fact, it might be the case that no subspace spanned by a small number of rows/columns of a matrix is robust in our sense~($\|\|_{q \to 2}$ norm). 
However these algorithms do not necessarily yield robust representations; in particular the subspace that is spanned may not be robust in our sense (${q \to 2}$ operator norm).
%\anote{Removed the line:
%There is also a large body of work on fast algorithms for computing low rank approximations~\cite{clarkson2017low, price2017fast, musco2017sublinear, song2017low, ban2019ptas}. Some of these works study the problem when the approximation error is measured in a more robust metric such as the $\ell_1$ norm as opposed to the Frobenius norm~\cite{song2017low}. Again these results are not directly related to the notion of subspaces robustness that we study in this paper.} %Finally many other different variants of PCA have been 

\vspace{5pt}
%we focus on deriving the bayes optimal robust classifiers, which is a different problem. note that we do not know how to do this in general beyond ell_2 perturbations.
%More technically distant
%analogs, More distantly related are previous works that treat 
\noindent \textbf{Sparse PCA.} The problem of sparse PCA has been studied both in average-case and worst-case settings. \\%\xnote{July/11: may wanna revise this paragraph about spike-covariance model.} \anote{July/13:done} 
\emph{Average-case setting:} In the high dimensional regime where the number of samples is much less than the dimensionality, several works have pointed out inconsistent behavior of PCA~\citep{paul2007asymptotics, nadler2008finite, johnstone2009consistency}. As a result this led to the study of the sparse PCA problem where it is assumed that the leading eigenvector is sparse. This problem is typically studied under an average case model known as the {\em spiked covariance model}~\citep{johnstone2001distribution}\footnote{Such ``spiked'' models (signal plus noise) have a long history in statistics~\citep{Anderson_Multi_stat} (first edition in 1962), viewed as certain types of \textit{factor} models.}. In this model the data is assumed to be generated from a Gaussian with covariance matrix $I + \theta vv^\top$, where the leading eigenvector $v$ is assumed to be a sparse vector and $\theta$ is a parameter characterizing the signal strength. There have been several works that study minimax rates of estimating the leading eigenvector (and eigenspaces) under the spiked covariance model and for various notions of sparsity~\citep{aminiwainwright, ma2013sparse, cai2013sparse, shen2013consistency, VuLei12, VuLei13}. More distantly related works include  linearly transformed spiked
models~\citep{dobriban2020optimal}, where the focus is on deriving the Bayes optimal robust classifiers and recovering the unobserved signals of interest under noisy linear transforms.

\noindent \emph{Worst-case setting:} There has also been work on the worst-case version of the problem in the special case when the rank $r=1$, but for the maximization variant of the sparse PCA objective~\citep{chan2016approximability}. For the maximization objective, the $\ell_0$ and $\ell_1$ versions (for capturing sparsity) are within a factor of $2$ from each other (see Section 10.3.3 of ~\citet{Vershynin}). 
%This is in stark contrast 
%We remark that these constant factor (bicriteria) approximations are in stark contrast 
%to the approximability of the maximization version of the problem. 
Even when $r=1$ the best known polynomial time algorithm gives a $O(n^{1/3})$ factor approximation in the worst-case (for both the $\ell_1$ and $\ell_0$ versions) ~\citep{chan2016approximability}.  %\footnote{The $\ell_0$ and $\ell_1$ versions of the sparse PCA maximization problem are within a factor of $2$ from each other~\citep[see Section 10.3.3 of ][]{Vershynin}})
% This is true even when the sparsity can be relaxed by a $O(1)$ factor; 
Moreover no constant factor approximation is possible assuming the SSE conjecture~\citep{chan2016approximability}.  Appendix~\ref{sec:computational-lower-bound} shows how this also immediately implies computational hardness of our minimization version \eqref{intro:obj}. The recent work of~\citep{woodruff2020spca} also studies the maximization version of the $\ell_0$-sparse PCA problem and presents bicriteria algorithms for the sparsity and the quality of the returned solution. However, none of those results translate to multiplicative factor approximation algorithms for the minimization variant of the problem that we study. More importantly, these works are restricted to the rank $r=1$ setting, while we study the more general version of the problem. %Another important difference with our work is that they do not aim for robust representations.

\paragraph{Robustness to Corruptions in the Training Data.} There is large body of work, spanning both the theoretical computer science and the statistical communities, that formulates and studies robustness to training data corruptions in the context of both supervised and unsupervised learning~\citep{valiant1985learning, kearns1993learning, huber2011robust, diakonikolas2019survey}. 
% Here, we survey works that are most relevant in the context of our results. For classification problems, the commonly studied models are the random classification noise model~\cite{angluin1988learning} and the agnostic learning model~\cite{kearns1994toward}, that capture corruptions to the labels of training points. More generally, there are also works studying the malicious noise~\cite{valiant1985learning, kearns1993learning} and the nasty noise~\cite{bshouty2002pac} models that allow for corruptions in both the training points and the corresponding labels. These models led to exciting developments in the design of robust algorithms for classification problems (see e.g.,~\cite{kalai2008agnostically, klivans2009learning, kalai2012reliable, feldman2009distribution, awasthi2014power, diakonikolas2018learning, daniely2015ptas}). 
However they do not study the notion of adversarial perturbations to the data, to the best of our knowledge. 

There is also a large body of work on {Robust Optimization}~\citep{ben2009robust}, where the input is uncertain and is assumed to belong to a structured {\em uncertainty} set. In robust optimization one looks for a single solution that is simultaneously good for all inputs in the uncertainty set, leading to a max-min formulation of the problem. In our model of corruption, we are interested in instance wise guarantees - for every input $A$ and its corruption $\tilde{A}$, the algorithm is required to output a solution that is good for $A$ (the solution is not required to be simultaneously good for all possible corruptions $\tilde{A}$). Moreover the resilience of a solution for $A$ to corruptions implies structural properties that can be leveraged algorithmically.  
%\anote{Added an extra line here.}
Moreover we are not aware of any results related to PCA in this context.

\paragraph{\bf Robust variants of PCA. } The problem of robust PCA has received significant attention in recent years~\citep{de2003framework, candes2011robust, chandrasekaran2011rank}. Here one assumes that a given corrupted matrix $\tilde{A}$ is a sum of two matrices, the true matrix $A$ that is low-rank and a sparse corruption matrix $S$ with sparsity pattern being essentially random. The corruptions, although sparse, can be unbounded in magnitude. This necessitates an incoherence type assumption that the ``mass'' or the principal components of $A$ is spread out -- recovery of $A$ is impossible under unbounded sparse corruptions when the signal is localized or sparse. 
On the other hand, the corruptions may not be sparse in our case. In particular, {\em every} data point (in fact every entry of $A$) could be corrupted up to some specified magnitude $\delta$. Here as our results show (particularly Theorem~\ref{ithm:robusttraining:frob}), localization (or sparsity) of the signal is crucial in tolerating adversarial perturbations in the training data (e.g., a spread out signal can be completely overwhelmed by the corruption in each entry of $A$). The very recent works of~\citet{ACV20coltsub, kothari2020pca} study a well-studied average-case model for sparse PCA under the same notion of training-time corruptions that we study in this work. The work of \citet{ACV20coltsub} builds on the current paper and characterizes the recovery error in terms of the $q \to 2$ operator norm, while the authors of \citet{kothari2020pca} focus on $r=1$ and characterize some computational vs statistical tradeoffs for this average-case model. On the other hand, our results apply to higher rank settings and for worst case data.

\paragraph{Huber's Contamination Model.} In statistics, Huber's $\epsilon$-contamination model~\citep{huber2011robust} is the most widely studied. In this model the dataset is assumed to be generated i.i.d. from a mixture namely, $(1-\epsilon)P + \epsilon Q$. Here $P$ is the true distribution and is assumed to be well behaved, for example the Gaussian distribution, and $Q$ is an arbitrary distribution modeling the noise. The study of this model has led to insightful results for a variety of problems. %The work of Yatracos~\cite{yatracos1985rates} and more recently of~\cite{chen2016general} studies general estimation in Huber's model. 
%More relevant to us are works on {\em mean estimation} and {\em clustering} in Huber's model. 
% The classical work of Tukey~\cite{tukey1975mathematics} proposed a robust estimator, now known as Tukey's median, for robust mean estimation of Gaussian data. The more recent work of~\cite{chen2018robust} showed that Tukey's estimator is minimax optimal and also proposed a minimax optimal estimator for robust covariance estimation.
 Recently, there have been many exciting developments in designing robust estimators of mean and covariance that are also computationally efficient~\citep{diakonikolas2019robust, lai2016agnostic, charikar2017learning, diakonikolas2019survey}. 
 %We discuss a few here. The works of~\cite{diakonikolas2019robust, lai2016agnostic} were the first to propose polynomial time algorithms for robust mean and covariance estimation of Gaussian data in Huber's model, with dimension independent error bounds. This was later extended to more general distributions and the list-decodable setting~\cite{charikar2017learning, steinhardt2017resilience}, optimal bounds for Gaussian data~\cite{diakonikolas2018robustly} and the study of computational/statistical tradeoffs~\cite{diakonikolas2017statistical,hopkins2019hard} and robust method-of-moments~\cite{kothari2017outlier}. There have also been works providing better sample complexity bounds in the Huber model if the mean vector is sparse~\cite{balakrishnan2017computationally, klivans2018efficient}. These works also study estimation in the spiked covariance model under corruptions. More recent developments include a linear time estimator for robust mean estimation~\cite{cheng2019high} and fast algorithms for robust covariance estimation~\cite{cheng2019faster}. There are also recent works studying computationally efficient robust optimization of more general objectives~\cite{diakonikolas2018sever, prasad2018robust}. 
We would like to point out that in these works (and several other recent works), the model of corruption is different than ours. In particular, rather than assuming that the data contains a few outliers~(Huber's model), in our model an adversary can potentially corrupt {\em every} data point up to magnitude $\delta$ (measured in $\ell_q$ norm for $q \ge 2$). 
%Hence the guarantees from these works do not translate into our setting. %\pnote{Aravindan: read the two lines above.} 

\paragraph{Clustering.} From the computational point of view, the work of~\citet{dasgupta1999learning} formulated the goal of clustering data generated from a mixture of well-separated Gaussians. There is a large body of work on designing efficient algorithms for clustering in this setting, both for Gaussians and more general distributions~\citep{arora2005learning, vempala2004spectral, achlioptas2005spectral, moitra2010settling, belkin2010polynomial}. See recent works~\citep{regev2017learning,hopkinsli2018,diakonikolas2018list,kothari2017better} for a detailed discussion. The work of~\citet{kumar2010clustering} abstracted out a common property of datasets~(spectral stability as defined in (\ref{eq:spectral-stability})) that captures mixtures of well separated Gaussians, the planted partitioning model, and other well clustered instances. They showed that a single algorithm, namely the popular Lloyd's algorithm, with the right initialization, provably computes optimal solution for such stable instances. The separation factor needed for Lloyd's to work in~\citet{kumar2010clustering} was later improved by~\cite{awasthi2012improved}. %The recent work of~\cite{cohen2017local} shows that local search obtains a polynomial time approximation scheme~(PTAS) on such spectrally stable instances. However, this in general does not guarantee closeness of the clustering obtained to the optimal one. 
\anote{Removed Cohen et al.}
\citet{Duttaetal} study Euclidean $k$-means clustering on instances that satisfy a notion of additive perturbation stability or resilience, where the optimal solution is stable even when each point is moved by a small amount. Analogously, in our problem the $\infty \to 2$ norm and sparsity captures the stability of the solution to small perturbations in $\ell_\infty$ norm. However the perturbations in Dutta et al. are measured in $\ell_2$ norm, and the problem flavor and algorithms are very different. 
%The works of~\cite{kalai2010efficiently, moitra2010settling, belkin2010polynomial} provided algorithms for clustering Gaussian mixtures with no separation requirement. These algorithms, inherently have an exponential dependence on the number of clusters, $k$, in the running time. 

Building on robust algorithms for mean estimation, there have also been works to perform robust clustering of well separated instances under Huber's contamination model and its variants~\citep{brubaker2009robust, diakonikolas2018list, kothari2017better, kothari2017outlier,hopkinsli2018}. There have also been works in analyzing the EM algorithm for Gaussian mixtures ( see e.g., ~\citep{balakrishnan2017statistical}). While motivated by the study of the phenomenon of robustness, the above results do not provide guarantees in our model of corruption. As in the case of mean estimation, these results are designed to be robust to a small number of outliers~(e.g., a small constant fraction) in the training set. In our corruption model on the other hand, every data point could be potentially corrupted up to magnitude $\delta$ (measured in $\ell_q$ norm for $q \ge 2$).

\section{Auxillary and additional claims}

% !TEX root=main.tex

\subsection{Counterexamples}
\begin{claim}\label{claim:counterexample:monotone}
The matrix norms $\norm{\cdot}_q$ (entry-wise $\ell_q$ norm) and $\norm{\cdot}_{\infty \to \infty}$, $\norm{\cdot}_{1 \to 1}$ are not monotone. 
\end{claim}
\begin{proof}
Let $v=\frac{1}{\sqrt{n}}\vec{1}$, where $\vec{1}=(1,1,\dots,1)$. 
Consider the matrix $M = I - \tfrac{1}{n}\vec{1} \vec{1}^\top$. Note that $\norm{v}_2=1$ and $M \succeq 0$. Clearly $\norm{I}_{1}=\sum_{i,j} |I_{i,j}|=n$. On the other hand, when $q \in [1,2)$, 
$$\norm{M}_q^q = \sum_{i,j} M_{i,j}^q = \sum_{i=1}^n (1-v(i)^2)^q + \sum_{i \ne j} |v(i)|^q |v(j)|^q = \norm{v}_q^{2q} + n(1-\tfrac{1}{n})^q - \frac{n}{n^q} = (n-1)n^{1-q}+ n(1-\tfrac{1}{n})^q > n.$$

Note that this particular instance is symmetric, and each column contributes equally; hence 
$\norm{M}_{1 \to q}^q = \tfrac{1}{n} \norm{M}_q^q$, and similarly for $I$. Hence, the same counterexample also works for $1 \to q$ and $\infty \to q^*$ operator norms where $q \in [1,2)$. %Finally, for $\infty \to 2$ operator norm, we have $\norm{I}_{\infty \to 2} \le \sqrt{n}$.
\end{proof}

\begin{claim}\label{clm:counterexample_infinity_2}
The matrix norms $\norm{\cdot}_{\infty \to 2}$ and $\norm{\cdot}_{2 \to 1}$ are not monotone. 
\end{claim}

\begin{proof}
We consider a similar pair of matrices as the above claim: 
\[ 
A=\text{diag}(u) \text{ where } u=(\underbrace{2,\ldots,2}_{n/3},\underbrace{1,\ldots,1}_{2n/3}) \text{ and } M=A-\frac{1}{n} \vec{1} \vec{1}^{\top}.
\]
It will be easier to reason about $\norm{M}_{\infty \to 2}$ and $\norm{A}_{\infty \to 2}$. Recall $\norm{A}_{\infty \to 2} = \max_{y : \norm{y}_\infty \le 1} \norm{Ay}_2$; since $A$ is a diagonal matrix, the maximum value is $\tau:=\sqrt{4 \cdot (1/3)+2/3} \cdot \sqrt{n}$, and it is attained by  every vector in $\sset{\pm 1}^n$. To establish the claim, we now show that for a specific vector $y \in \sset{\pm 1}^n$, $\norm{My}_2 > \norm{Ay}_2$. 
\begin{align*}
 \text{Consider }   y&=(\underbrace{1,\ldots,1}_{n/3},\underbrace{-1,\ldots,-1}_{2n/3})\\
    \iprod{\vec{1},y}&= -\tfrac{1}{3}n, \text{ and }My= (\underbrace{2+\tfrac{1}{3},\ldots,2+\tfrac{1}{3}}_{n/3},\underbrace{-1+\tfrac{1}{3},\ldots,-1+\tfrac{1}{3}}_{2n/3})\\
   \noindent \text{Hence } \norm{My}_2 &= \sqrt{ \big(\tfrac{1}{3} \big)n \cdot (\tfrac{7}{3} )^2 + \big(\tfrac{2}{3} \big)n \cdot (\tfrac{2}{3} )^2 } \\
    &=\sqrt{\tau^2+ n \big(\frac{1}{9}\big)} > \tau = \norm{Ay}_2,
\end{align*}
as required. Hence $\norm{A}_{\infty \to 2} < \norm{M}_{\infty \to 2}$, which violates the monotonicity property. 
\end{proof}

\subsection{What do robust projection matrices look like?} \label{app:looklike}
Our robustness parameter $\norm{\Pi}_{\infty\to 2}$ ($\norm{\Pi}_{q \to 2}$ for general $q \ge 2$) generalizes analytic notions of sparsity for the subspace associated with the orthogonal projector $\Pi$ (see Lemma~\ref{lem:robustproperties}). For the purposes of this discussion let us restrict our attention to $q=\infty$.  As mentioned earlier, for a $r=1$-dimensional subspace this exactly corresponds to the $\ell_1$ sparsity of the unit vector $v$ in that subspace. 
The $\norm{\Pi}_{\infty \to 2}$ of a projector is the largest $\ell_1$ norm among unit vectors (in $\ell_2$ norm) that belong to the subspace. %This is a basis-independent quantity that only depends on the subspace. 
We remark that for higher-dimensional subspaces, there are several other notions of sparsity that have been explored~\citep{VuLei13, Liuetal}, typically measured for a fixed orthonormal basis $V \in \R^{n \times r}$ of the subspace (so $\Pi=VV^\top$). Some of the notions that have been considered include the entry-wise norm $\norm{V}_1$ (the sum of the $\ell_1$ norms of the basis vectors), the maximum $\ell_1$ norm among the columns of $V$, the sparsity of the diagonal of $\Pi$ and the sum of the row $\ell_2$ norms of $V$, among other quantities. Many of these quantities are the same for $r=1$ but may vary by factors of $\sqrt{r}$ or more depending on the quantity. On the other hand, the quantity $\norm{\Pi}_{q \to 2}$ is a basis-independent quantity that only depends on the subspace.

Consider three different subspaces (or projectors) given by the orthonormal basis $V_1, V_2, V_3 \in \R^{n \times r}$ of the following form (think of $\kappa = \sqrt{k}$, $r \ll \kappa$); assume that the signs of the entries are chosen randomly in a way that also satisfies the necessary orthogonality properties (e.g., random Fourier characters over $\sset{\pm 1}^k$).
\begin{equation*}
V_1= 
\begin{pmatrix}
 \tfrac{\pm 1}{\sqrt{k}} &  \tfrac{\pm 1}{\sqrt{k}} & \cdots &  \tfrac{\pm 1}{\sqrt{k}} \\
 \tfrac{\pm 1}{\sqrt{k}} &  \tfrac{\pm 1}{\sqrt{k}} & \cdots &  \tfrac{\pm 1}{\sqrt{k}} \\
\vdots  & \vdots  & \ddots & \vdots  \\
 \frac{\pm 1}{\sqrt{k}} &  \frac{\pm 1}{\sqrt{k}} & \cdots &  \frac{\pm 1}{\sqrt{k}} \\
 0 & 0 & \cdots & 0 \\
\vdots & \vdots & \vdots & \vdots \\
0 & 0 & \cdots & 0 
\end{pmatrix},
~~~~
V_2= 
\begin{pmatrix}
\tfrac{\pm \sqrt{r}}{\sqrt{k}} & 0 & \cdots & 0 \\
\cdot & \cdot & \cdots & \cdot \\
\tfrac{\pm \sqrt{r}}{\sqrt{k}} & 0 & \cdots & 0 \\
0 &  \tfrac{\pm \sqrt{r}}{\sqrt{k}} & \cdots & 0 \\
\cdot & \cdot & \cdots & \cdot \\
0 & \tfrac{\pm \sqrt{r}}{\sqrt{k}} & \cdots & 0 \\
0 &  0 & \cdots & 0 \\
\vdots  & \vdots  & \ddots & \vdots  
\end{pmatrix},
~~~V_3= 
\begin{pmatrix}
\tfrac{\pm 1}{\sqrt{r}} & \tfrac{\pm 1}{\sqrt{r}} & \cdots & \tfrac{\pm 1}{\sqrt{r}}& \tfrac{\pm 1}{\sqrt{k}} \\
\cdot & \cdot & \cdots &\cdot & \cdot \\
\tfrac{\pm 1}{\sqrt{r}} & \tfrac{\pm 1}{\sqrt{r}} & \cdots & \tfrac{\pm 1}{\sqrt{r}}& \tfrac{\pm 1}{\sqrt{k}} \\
0 & 0 & \cdots&0 & \tfrac{\pm 1}{\sqrt{k}} \\
\cdot & \cdot & \cdots&\cdot & \cdot \\
0 & 0 & \cdots & 0& \tfrac{\pm 1}{\sqrt{k}} \\
0 &0 &\cdots & 0&0\\
\vdots  & \vdots  & \ddots & \vdots & \vdots  
\end{pmatrix}
\end{equation*} 
%In these three examples, $V_1, V_2, V_3$ represent orthonormal basis;  
The main difference between $V_1, V_2$ is that in $V_2$ the sparse basis vectors have disjoint support, whereas in $V_1$ they are commonly supported. However, there is an alternate basis for the subspace $V_2$ which looks like $V_1$, but basis dependent quantities get very different values for $V_1, V_2$. In the third example, the first $r-1$ basis vectors are extremely sparse with $\ell_1$ norm $O(\sqrt{r})$, whereas only one of the basis vectors has $\ell_1$ sparsity $\sqrt{k}$. Many aggregate notions of sparsity like $\norm{V}_1$  take very different values for $V_1$ and $V_3$ that differ by a $\sqrt{r}$ factor. Our robustness parameter $\norm{\Pi}_{\infty \to 2} \approx \sqrt{k}$ for all of them; this is because each of these subspaces are supported on at most $k$ co-ordinates (and a spread out vector of this form exists), so the maximum $\ell_1$ length among unit $\ell_2$ norm vector is $\sqrt{k}$. Finally we remark that as mentioned in Figure~\ref{fig:cifar-10} (and shown in \citet{awasthi2020adversarial}), many natural data sets have such robust low-dimensional projections with small error.

\subsection{Computational Lower Bound}
\label{sec:computational-lower-bound}
The main result in this section is to show that it is NP-hard to solve Program~\eqref{int:obj} in Section~\ref{sec:worstcase} exactly for $q=\infty$ under the small set expansion(SSE) hypothesis.
%\begin{align}
%    \min_{\Pi} &\norm{\Pi^{\perp} A }_F^2\\
%    \text{s.t.\ }& \Pi \text{ is a projection matrix of rank } r, \text{ and } \norm{\Pi}_{\infty \to 2} \le \kappa \label{int:operatornorm}.
%\end{align}
\begin{conjecture}[SSE hypothesis~\citep{raghavendra2010graph}] For any $\eta>0$, there is $\delta>0$ such that it is NP-hard to distinguish between the following two cases given a graph $G=(v,E)$:
\begin{itemize}
\item Yes: Some subset $S\subseteq V$ with $|S| = \delta n$ has expansion $\tfrac{|E(S,V\setminus S|}{|S|}\le \eta$
\item No: Any set $S\subseteq V$ with $|S| \le 2\delta n$ has $\tfrac{|E(S,V\setminus S|}{|S|}\ge 1-\eta$
\end{itemize}
\end{conjecture}
Based on SSE hypothesis, we state the hardness of solving Program~\eqref{int:obj} as follows.
\begin{theorem}\label{thm:hardness_rank1_max}
It is SSE-hard to solve problem~\eqref{int:obj} given $q=\infty,r=1$.
\end{theorem}
We note the objective in problem~\eqref{int:obj} is the same as $\norm{A}_F^2-\max_{\Pi} \iprod{AA^\top,\Pi}$ and use the hardness of $\max_{\Pi} \iprod{AA^\top,\Pi}$ to finish the proof.

\begin{theorem}[Theorem 4 in \citet{chan2016approximability}]\label{thm:hardness_SSE}
It is SSE-hard to solve the following program given $k$ and a matrix $A \in \mathbb{R}^{n \times n}$ with any constant approximation ratio: $\underset{\norm{x}_2=1,\norm{x}_0\le k}{\max}  \|A x\|_2^2$.
\end{theorem}

Next we state the relation between the $\ell_0$-sparse and $\ell_1$-sparse programs.
\begin{theorem}[\citep{Vershynin}]\label{thm:relation_ell0_ell1}
Given any matrix $A \in \mathbb{R}^{n \times n}$ and any $k \ge 1$, let $OPT_{\ell_0}=\underset{\norm{x}_2=1,\norm{x}_0\le k}{\max}  \|Ax\|_2^2$ and $OPT_{\ell_1}=\underset{\norm{x}_2=1,\norm{x}_1\le \sqrt{k}}{\max}  \|Ax\|_2^2$. Then we have $
OPT_{\ell_0}\le OPT_{\ell_1}\le 2 \cdot OPT_{\ell_0}.$
\end{theorem}
Finally we finish the proof of Theorem~\ref{thm:hardness_rank1_max}.

\begin{proof}{of Theorem~\ref{thm:hardness_rank1_max}}
For contradiction, suppose there is an algorithm that solves problem~\eqref{int:obj} for $q=\infty$ and any $\kappa$. Given a matrix $A$ and $k$, since \[\underset{\Pi:\|\Pi\|_{\infty \to 2} \le \kappa}{\max} \langle  A A^{\top}, \Pi \rangle = \underset{\norm{x}_2=1,\norm{x}_1\le \sqrt{k}}{\max} \langle  A A^{\top}, x x^{\top} \rangle\] for rank 1 projection matrices, the algorithm for problem~\eqref{int:obj} also solves $\underset{\norm{x}_2=1, \norm{x}_1\le \sqrt{k}}{\max} \|Ax\|_2^2$ by the reformulation~\eqref{eq:obj_rank1}. Because of Theorem~\ref{thm:relation_ell0_ell1}, this gives a $0.5$ approximation of $\underset{\norm{x}_2=1,\norm{x}_0\le k}{\max}  \|A x\|_2^2$, which refutes the SSE hypothesis based on Theorem~\ref{thm:hardness_SSE}.
\end{proof}
Notice that the above proof only establishes computational hardness for exact minimization of \eqref{int:obj} under the small set expansion conjecture. It would be interesting to establish hardness of approximation results for this problem.
\pnote{Check above.}

\subsection{Proof of Lemma~\ref{thm:robust-mean}}
\label{app:robust-mean}
\begin{proof}
%For the sake of exposition we will prove the statement assuming $\kappa \delta = O(\sigma)$. 
%It is a straightforward exercise to then get the general guarantee stated above for any $\kappa, \delta$. 
By assumption we know that $\|A-C\| \leq \sigma \sqrt{m}$. This implies that
\begin{align*}
    \|A-\Pi^*A\| &\leq \|A-C\| + \|C-\Pi^*A\|\\
    &= \|A-C\| + \|\Pi^*(C-A)\|\\
    &\leq 2\|A-C\| \leq 2\sigma \sqrt{m}.
\end{align*}
Since we set $\tau = 2\sigma\sqrt{m}$, from the guarantee of Theorem~\ref{thm:training:spectralnorm}, we know that if the algorithm outputs \Bad, the data must be poisoned, i.e., $\|A-\tilde{A}\| > 2\sigma \sqrt{m}$. Next suppose that the algorithm outputs a projection matrix $\Pi$. Setting $\hat{\mu} := \Mean(\Pi \tilde{A})$, and $\mathbf{1}$ to be the all-ones vector $(1,1,\dots,1)^\top$ we have that
\begin{align}
    \|\hat{\mu} - {\mu}\|_2 &= \frac{1}{m} \Bignorm{\sum_{j=1}^m (A_j - \Pi \tilde{A}_j) }_2 \nonumber \\    
    &\leq \frac{1}{m} \Bignorm{\sum_{j=1}^m (A_j - \Pi {A}_j) }_2 + \frac{1}{m} \Bignorm{\sum_{j=1}^m (\Pi A_j - \Pi \tilde{A}_j)}_2  \nonumber \\
    &\leq \frac{1}{m} \norm{\mathbf{1}^\top (A-\Pi A)}_2 + \frac{1}{m} \sum_{i=1}^m \norm{\Pi(A_j - \tilde{A}_j)}_2 \nonumber \\
    &\leq \frac{1}{\sqrt{m}} \|A-\Pi A\| + c_q\kappa \delta. \label{eq:proof-robust-mean-1}
\end{align}
% If the algorithm did not output \Bad then from the value of $\tau$ we know that
% $$
% \|\tilde{A} - \Pi \tilde{A}\| \leq \tau = 2\sigma \sqrt{m}.
% $$
% Substituting above we get that
% \begin{align}
%     \label{eq:proof-robust-mean-2}
%     \|\hat{\mu} - {\mu}\|_2 &\leq \frac{1}{\sqrt{m}} \|A-\Pi A\| + c_q\kappa \delta + 2\sigma.
% \end{align}
%where the second inequality follows by . 
Next we make a crucial observation that if $\Pi$ is good for $\tilde{A}$ then it is also good for $A$ and hence $\|A- \Pi A\|$ is small. This is formally established in Lemma~\ref{lem:spectral-norm-closeness-of-subspaces}. Applying the lemma on $A$ and $\tilde{A}$ with $\Pi_1 = \Pi^*$, $\Pi_2 = \Pi$, $\kappa_1 = \kappa$, $\kappa_2 = c_q \kappa$, and $\epsilon = 4\sigma \sqrt{m}/\|A\|$ we get that
\begin{align*}
    \|A - \Pi A\| &\leq (\epsilon + \sqrt{\epsilon})\|A\| + 8\frac{\kappa \delta \sqrt{m}}{\sqrt{\epsilon}}.
\end{align*}
Substituting into (\ref{eq:proof-robust-mean-1}) we get that
\begin{align}
    \|\hat{\mu} - {\mu}\|_2 &\leq (\epsilon + \sqrt{\epsilon})\frac{\|A\|}{\sqrt{m}} + 8\frac{\kappa \delta}{\sqrt{\epsilon}}+ c_q \kappa \delta \nonumber \\
    &\leq (\epsilon + \sqrt{\epsilon})\frac{4\sigma}{\epsilon} + 8\frac{\kappa \delta}{\sqrt{\epsilon}}+ c_q \kappa \delta \,\, (\text{by writing $\|A\|$ in terms of $\epsilon$})\nonumber \\
    &\leq O(c_q)(\sigma + \kappa \delta)\Big(1+\sqrt{\frac{1}{\epsilon}} \Big) \label{eq:proof-robust-mean-3}.
\end{align}
% Using the fact that $\kappa \delta = O(\sigma)$, we get
% \begin{align}
% \label{eq:proof-robust-mean-3}
%     \|\hat{\mu} - {\mu}\|_2 &\leq O(c_q \sigma) \Big(1+\sqrt{\frac{1}{\epsilon}} \Big).
% \end{align}
Next, notice that by triangle inequality,
\begin{align*}
    \|A\| &\leq \|A-C\| + \|C\|\\
    &\leq (\sigma + \|\mu\|)\sqrt{m}.
\end{align*}
Hence we get that
\begin{align*}
    \sqrt{\frac{1}{\epsilon}} = \sqrt{\frac{\|A\|}{4\sigma \sqrt{m}}}
    \leq \frac{1}{2} \Big( 1 + \frac{\|\mu\|}{\sigma}\Big)^{\frac{1}{2}}.
\end{align*}
Substituting into (\ref{eq:proof-robust-mean-3}) we get that
\begin{align}
\label{eq:proof-robust-mean-4}
    \|\hat{\mu} - {\mu}\|_2 &\leq O(c_q) (\sigma + \kappa \delta) \Big(1+\Big( 1 + \frac{\|\mu\|}{\sigma}\Big)^{\frac{1}{2}} \Big)\\
    &\leq O(c_q)(1+\frac{\kappa \delta}{\sigma}) \max \Big(\sigma, \sqrt{\sigma \|\mu\|} \Big).
\end{align}
From the above we get the relative error guarantee of
\begin{align}
\label{eq:proof-robust-mean-5}
    \frac{\|\hat{\mu} - {\mu}\|_2}{\|\mu\|} &\leq O(c_q)(1+\frac{\kappa \delta}{\sigma}) \max \Big(\frac{\sigma}{\|\mu\|}, \sqrt{\frac{\sigma}{\|\mu\|}} \Big).
\end{align}

% Furthermore, from the guarantee of Theorem~\ref{thm:training:spectralnorm} we have that for any $\eta \in (0,1)$,
% \begin{align*}
%     \|A-\Pi A\| &\leq O(1+\frac{1}{\eta}) \Big(2\sigma \sqrt{m} + \|A-\Pi^* A\| + \kappa \delta \sqrt{m} \Big) + \sqrt{2\eta}\|A\|\\
%     &\leq O(1+\frac{1}{\eta}) \Big(4\sigma + \kappa \delta \Big)\sqrt{m} + \sqrt{2\eta}\|A\|.
% \end{align*}
% Substituting into (\ref{eq:proof-robust-mean-1}) we get that
% \begin{align*}
%     \|\hat{\mu} - {\mu}\| \leq O(1+\frac{1}{\eta})(4\sigma + \kappa \delta) + c_q \kappa \delta + \frac{\sqrt{2\eta}\|A\|}{\sqrt{m}}.
% \end{align*}
% If $\|A\| \leq 5\sigma \sqrt{m}$, then we set $\eta = 1/2$ above and use the assumption that $\kappa \delta = o(\sigma)$ to get
% \begin{align*}
%     \|\hat{\mu} - \mu\| = O(c_q \sigma).
% \end{align*}
% Otherwise setting $\eta = (5\sigma \sqrt{m}/\|A\|)^{2/3}$ we get that
% \begin{align}
% \label{eq:proof-robust-mean-1}
%     \|\hat{\mu} - {\mu}\| \leq O(c_q \sigma)\Big(1+ \Big(\frac{\|A\|}{\sigma \sqrt{m}} \Big)^{2/3} \Big).
% \end{align}
% The relative error guarantee follows from the fact that 
% \begin{align*}
% \|A\| \leq \|A-C\| + \|C\|
% \leq \sigma \sqrt{m} + \|\mu\|\sqrt{m}.
% \end{align*}
% \noindent Substituting into (\ref{eq:proof-robust-mean-1}) we get that
% \begin{align*}
%     \frac{\|\hat{\mu} - {\mu}\|}{\|\mu\|} &\leq O(c_q \sigma) \Big(1+ (\sigma + \|\mu\|)^{2/3} \Big)\\
%     &\leq \min \Big(O\Big (\frac{c_q \sigma}{\|\mu\|} \Big)^{1/3}, O\Big(\frac{c_q \sigma}{\|\mu\|}\Big) \Big).
% \end{align*}
\end{proof}
\noindent \textbf{Note:} We would like to point out that for robust mean estimation, our analysis also shows that in step $2$ of the algorithm above, we can replace $\Mean(\Pi \tilde{A})$ with $\Mean(\tilde{A})$. This is because if the algorithm did not output \Bad then $\|\tilde{A} - \Pi \tilde{A}\|_2/\sqrt{m} \leq 2\sigma$ and hence mean of $\tilde{A}$ and that of $\Pi \tilde{A}$ will be close. However, in this case, the subspace spanned by the output vector, i.e., \Mean($A$) might not be robust and hence susceptible to test-time perturbations.
\begin{lemma}
\label{lem:spectral-norm-closeness-of-subspaces}
Fix $q \geq 2, \delta > 0, \kappa \geq 1$. Let $A$ and $\tilde{A}$ be two $n \times m$ matrices, each representing $m$ data points in $n$ dimensions such that for every $j \in [m]$, columns $A_j$ and $\tilde{A}_j$ are close, i.e., $\|A_j - \tilde{A}_j\|_q \leq \delta$. Furthermore, assume that there exist projection matrices, $\Pi_1 = vv^\top$ and $\Pi_2 = uu^\top$ such that $\|\Pi_1\|_{q \to 2} \leq \kappa_1$ and $\|\Pi_2\|_{q \to 2} \leq \kappa_2$ and that $\|A - \Pi_1 A\| \leq \epsilon_1 \|A\|$ and $\|\tilde{A} - \Pi_2 \tilde{A}\| \leq \epsilon_2 \|A\|$. Then, letting $\epsilon = \epsilon_1 + \epsilon_2$ and $\kappa = \kappa_1 + \kappa_2$, it also holds that
\begin{align}
    \|A - \Pi_2 A\| &\leq O(\epsilon + \sqrt{\epsilon})\|A\| + \frac{8\kappa \delta \sqrt{m}}{\sqrt{\epsilon}} \label{eq:proof-spectral-norm-closeness-of-subspaces-1}\\
        \|\tilde{A} - \Pi_1 \tilde{A}\| &\leq O(\epsilon + \sqrt{\epsilon})\|A\| + \frac{8\kappa \delta \sqrt{m}}{\sqrt{\epsilon}}
\end{align}
\end{lemma}
\begin{proof}
We will show the desired bound on $\|A - \Pi_2 A\|$ and by symmetry the same bound will also apply to $\|\tilde{A} - \Pi_1 \tilde{A}\|$. Notice that both $\Pi_1$ and $\Pi_2$ are projections onto one dimensional subspaces and a bound on $\|\|_{q \to 2}$ norm of the projection matrices implies that $\|v\|_{q^*} \leq \kappa_1$ and $\|u\|_{q^*} \leq \kappa_2$, where $q^*$ is such that $1/q + 1/q^*=1$. Next, let $\Pi$ be the projection matrix onto the subspace spanned by $v$ and $u$. By triangle inequality we have that
\begin{align}
    \|A - \Pi_2 A\| &\leq \|A - \Pi A\| + \| \Pi A - \Pi_2 A\| \nonumber \\
    &\leq \|A - \Pi A \| + \| \Pi A - \Pi_2 \tilde{A}\| + \| \Pi_2 \tilde{A} - \Pi_2 {A}\| \nonumber \\
    &\leq \|A - \Pi A\| + \|\Pi A - \Pi \tilde{A}\| + \|\Pi \tilde{A} - \Pi_2 \tilde{A}\| + \| \Pi_2 \tilde{A} - \Pi_2 {A}\| \label{eq:proof-spectral-norm-closeness-of-subspaces-2}.
\end{align}
Recall the standard fact that if $P_1$ and $P_2$ are projection matrices on to subspaces $S_1$ and $S_2$ such that $S_1 \subseteq S_2$, then for any matrix $B$, $\|P_1 B\| \leq \|P_2 B\|$. Using this we to get that
\begin{align}
    \|A - \Pi A\| &\leq \|A - \Pi_1 A\| \leq \epsilon_1 \|A\| \label{eq:proof-spectral-norm-closeness-of-subspaces-3}
\end{align}
and,
\begin{align}
    \|\Pi_2 \tilde{A} - \Pi \tilde{A}\| &\leq \|\Pi_2 \tilde{A} - \tilde{A}\| + \|\tilde{A} - \Pi \tilde{A}\| \nonumber \\
    &\leq 2\|\Pi_2 \tilde{A} - \tilde{A}\| \leq 2\epsilon_2 \|A\| \label{eq:proof-spectral-norm-closeness-of-subspaces-4}.
\end{align}
From the closeness of $A$ and $\tilde{A}$ and the robustness of $\Pi_2$ we also know that
\begin{align}
    \label{eq:proof-spectral-norm-closeness-of-subspaces-5}
    \|\Pi_2 \tilde{A} - \Pi_2 A\| &\leq \|\Pi_2\|_{q \to 2} \delta \sqrt{m} \leq \kappa_2 \delta \sqrt{m}.
\end{align}
\pnote{Equation \ref{eq:proof-spectral-norm-closeness-of-subspaces-5} is used a lot and should be included in prelims.}
Substituting (\ref{eq:proof-spectral-norm-closeness-of-subspaces-3}), (\ref{eq:proof-spectral-norm-closeness-of-subspaces-4}), and (\ref{eq:proof-spectral-norm-closeness-of-subspaces-5}) into (\ref{eq:proof-spectral-norm-closeness-of-subspaces-2}) we get that
\begin{align}
    \|A - \Pi_2 A\| \leq \epsilon_1 \|A\| + 2\epsilon_2 \|A\| + \kappa_2 \delta \sqrt{m} + \|\Pi {A} - \Pi \tilde{A} \| \label{eq:proof-spectral-norm-closeness-of-subspaces-6}.
\end{align}
Note that if $\|\Pi A - \Pi \tilde{A}\| \leq \kappa \delta \sqrt{m}/\sqrt{\epsilon}$, then we have the desired bound on $\|A- \Pi_2 A\|$. We now look at the case when $\|\Pi A - \Pi \tilde{A}\| > \kappa \delta \sqrt{m}/\sqrt{\epsilon}$. Notice that $\Pi$ is the union of robust subspaces and $A-\tilde{A}$ has columns bounded in $q$ norm. Hence, the only way $\|\Pi A - \Pi \tilde{A}\|$ can be very large is if the $\|\|_{q \to 2}$ norm of the projection matrix of a union of two subspaces~($\Pi$) is much larger than the $\|\|_{q \to 2}$ norm of the projection matrices of individual subspaces~($\Pi_1$ and $\Pi_2$). For this to happen the two subspaces must be very close to each other and then we can bound $\|A - \Pi_2  A\|$ in a different way. Formally, we have that
\begin{align}
    \|\Pi A - \Pi \tilde{A}\| &= \max_{z: \|z\|=1} \|\Pi (A-\tilde{A}) \cdot z\| \nonumber \\
    &= \max_{z:\|z\|=1} \|\sum_{j=1}^m z_j \Pi(A_j - \tilde{A}_j)\| \nonumber \\
    &\leq \max_{z:\|z\|=1} \sum_{j=1}^m |z_j| \|\Pi(A_j - \tilde{A}_j)\| \nonumber \\
    &\leq \|\Pi\|_{q \to 2} \delta \sqrt{m} \label{eq:proof-spectral-norm-closeness-of-subspaces-7}.
\end{align}
Next we establish an upper bound on $\|\Pi\|_{q \to 2}$ in terms of the distance between subspaces $\Pi_1 = vv^\top$ and $\Pi_2 = uu^\top$. Suppose $\|u-v\| = \gamma$ and that $u \cdot v \geq 0$ (otherwise we work with $-u$). We also know that $\|v\|_{q^*} \leq \kappa_1$ and $\|u\|_{q^*} \leq \kappa_2$. Now, $\|\Pi\|_{q \to 2}$ is the maximum $q^*$ norm of any unit vector in the span of $v$ and $u$. We can write any such vector $z$ as
\begin{align*}
    z = \alpha_1 v + \alpha_2 v^{\perp}
\end{align*}
where $\alpha^2_1 + \alpha^2_2=1$ and $v^{\perp} = \frac{u- (u \cdot v)v}{\|u- (u \cdot v)v\|}$. Next we have that 
\begin{align*}
    \|u - (u \cdot v)v\|^2 = 1- (u \cdot v)^2\\
    \geq 1 - u \cdot v = \frac{\gamma^2}{2}.\\
\end{align*}
Hence we get that for any $z$ in the span of $v$ and $u$
\begin{align*}
    \|z\|_{q^*} &\leq \|v\|^{q^*} + \|v^{\perp}\|_{q^*}\\
    &\leq \kappa_1 + {\frac{\sqrt{2}}{\gamma}} \Big(\|u\|_{q^*} + \|v\|_{q^*} \Big)\\
    &\leq {\frac{2\sqrt{2}}{\gamma}}\kappa.
\end{align*}
The above also establishes that 
$$
\|\Pi\|_{q \to 2} \leq 2 \frac{\sqrt{2}}{\gamma}\kappa.
$$
Substituting into (\ref{eq:proof-spectral-norm-closeness-of-subspaces-7}) we get that
\begin{align}
    \|\Pi A - \Pi \tilde{A}\| \leq {\frac{2\sqrt{2}}{\gamma}}\kappa \delta \sqrt{m}.
\end{align}
Hence, if $\|\Pi A - \Pi \tilde{A}\| > \kappa \delta \sqrt{m}/\sqrt{\epsilon}$ we must have that
\begin{align}
\label{eq:proof-spectral-norm-closeness-of-subspaces-8}
    \|v-u\| = \gamma \leq 2\sqrt{2} \sqrt{\epsilon}.
\end{align}
In this case we can bound $\|A - \Pi_2 A\|$ as 
\begin{align*}
    \|A - \Pi_2 A\| &\leq \|A - \Pi_1 A\| + \|(\Pi_1 - \Pi_2)A\|\\ 
    &\leq \epsilon_1 \|A\| + \|\Pi_1 - \Pi_2\|\|A\|\\
    &\leq \epsilon \|A\| + \|vv^\top - uu^\top\|\|A\|\\
    &= \epsilon\|A\| + \|\frac{1}{2} (v+u) (v-u)^\top + \frac{1}{2} (v-u) (v+u)^\top\| \|A\|\\
    &\leq \epsilon \|A\| + 2\|v-u\|\|A\| \,\, (\text{by triangle inequality and the fact that $\|v+u\| \leq 2$})\\
    &\leq \epsilon \|A\| + 2\gamma \|A\|\\
    &\leq \epsilon \|A\| + 4\sqrt{2} \sqrt{\epsilon}\|A\|.
\end{align*}
\end{proof}
\noindent \textbf{Tightness of the Guarantee in Theorem~\ref{thm:robust-mean}}. 
%Notice that in the regime when $\|\mu\|=O(1)$, the guarantee from the theorem has an $\sqrt{\sigma}$ error in mean estimation \anote{Maybe say $\sqrt{\sigma/ \norm{\mu}}$? Because $\sqrt{\sigma}$ has dimension root of length, which is weird}. 
We close out this section by showing that the dependence on $\sqrt{\sigma \|\mu\|}$ in our bound on mean estimation is necessary even information theoretically. In what follows $A$ will be an $n \times m$ matrix with $\mu = \Mean(A)$ such that $\Pi^* = \mu \mu^\top/\|\mu\|^2$ has small norm, i.e., $\|\Pi^*\|_{\infty \to 2} = \kappa$. Furthermore, let $C = \mu \mathbf{1}^\top$ and define $\sigma = \|A - C\|/\sqrt{m}$. We will prove the following.
\begin{theorem}
\label{thm:robust-mean-lower-bound}
Fix $q = \infty$. Let $\mathcal{M}$ be the set of $n \times m$ matrices $A$ with mean $\mu$ that satisfies $\|\mu\| \in [1,2]$, variance $\sigma^2$ around the mean that satisfies $\sigma \in (0,1/6]$, and the subspace spanned by $\mu$ being $\kappa$-robust. Call a perturbation $\tilde{A}$ of $A \in \mathcal{M}$ of be valid if $\|A - \tilde{A}\|_{\infty} \leq \delta$.
%and $\kappa \delta \leq 6\sigma$. 
Then, any algorithm that takes as input a valid perturbation $\tilde{A}$ of a matrix $A \in \mathcal{M}$ and either certifies that the data is poisoned, i.e., $\|A-\tilde{A}\| > 8\sigma \sqrt{m}$ or outputs an estimate $\hat{\mu}$ of $\mu$ must incur an error of
\begin{align*}
    {\|\mu - \hat{\mu}\|} \geq \Omega\Big((1+\frac{\kappa \delta}{\sigma})\max(\sigma, \sqrt{\sigma \|\mu\|})\Big),
\end{align*}
where $\mu = \Mean(A)$.
%where $\mu_{\max} = \max(\|\mu\|, \|\hat{\mu}\|)$.
\end{theorem}
\begin{proof}
We will establish the lower bound by constructing two matrices $A$ and $\tilde{A}$, both of which lie in the set $\mathcal{M}$, satisfy $\|A - \tilde{A}\| = \delta$ for $\kappa \delta = O(\sigma)$, but have means separated by $\Omega(\max(\sigma, \sqrt{\sigma \mu_{\max}}))$, where $\mu_{\max}$ is the maximum $\ell_2$ norm among \Mean($A$) and $\Mean(\tilde{A})$. In this case, given either $A$ or $\tilde{A}$ as input, any provably robust certification procedure cannot output \Bad and must output an estimate $\hat{\mu}$, thereby making $\Omega(\max(\sigma, \sqrt{\sigma \mu_{\max}}))$ error on either $A$ or $\tilde{A}$. We next describe our construction.

For a $k$ to be determined later, let $\mu_1 = (1/\sqrt{k}, 1/\sqrt{k}, \dots, 1/\sqrt{k}, 0, 0, \dots, 0)$. Hence, $\mu_1$ is a unit length sparse vector with $\|\mu_1\|_1 = \sqrt{k}$. We define the set of $m$ points in $A$ by generating i.i.d. points of the form $\mu_1 + g$, where $g$ is a mean zero Gaussian with variance $0$ in the first $k$ coordinates and variance $\sigma^2$ in the other coordinates. Then it is a standard fact that with high probability $\|A - \mu_1 \mathbf{1}^\top\| \leq 2\sigma \sqrt{m}$ and that \Mean(A) will be $\sigma \sqrt{d/m} = o(\sigma)$-close to $\mu_1$.\anote{What do you mean by $o(1)$? Do you mean $o(\sigma)?$ or maybe even $\sigma \cdot O(\sqrt{d/m})$?} \pnote{Done.} Next we define the set of points in $\tilde{A}$ to be $\tilde{A}_j = A_j + \delta sgn(\mu_1)$, where $sgn(\mu_1)$ is a $\pm 1$ vector representing the sign of the corresponding coordinate of $\mu_1$. Here we can arbitrarily set $sgn(0)$ to be $+1$. \anote{What is the sign of the $0$ co-ordinates of $\mu$? Is it zero (I think you want it to be arbitrary? } \pnote{Done.} It is easy to see that \Mean($\tilde{A}$) will be $o(\sigma)$-close to $\mu_2 = \mu_1 + \delta sgn(\mu_1)$, and that $\|\tilde{A} - \mu_2 \mathbf{1}^\top\| \leq 2\sigma \sqrt{m}$. Next notice that 
\begin{align*}
    \|\mu_2\|^2 &= 1 + \delta^2 n + \delta \sqrt{k}\\
    \|\mu_2\|_1 &= \sqrt{k} + \delta \sqrt{k}.
\end{align*}
By setting $\delta \sqrt{k} = 3\sigma$ and $\delta n = \sqrt{k}$ we ensure that $\|\mu_2\| \in [1,2]$ and $\|\mu_2\|_1 \leq 2\sqrt{k}$. Hence we get that for the matrix $\Pi = \mu_2 \mu^\top_2/\|\mu_2\|^2$, $\|\Pi\|_{\infty \to 2} \leq 2\sqrt{k}$.
Hence, both $A$ and $\tilde{A}$ lie in the set $\mathcal{M}$ with sparsity bound $\kappa = 2\sqrt{k}$. Furthermore, the fact that $\delta \sqrt{k} = 3\sigma$, ensures that $\kappa \delta \leq 6\sigma$.
%and hence both are valid perturbations of each other. 
Finally, notice that the difference between two means is 
\begin{align*}
    \|\mu_1 - \mu_2\| &= \delta \sqrt{n} = \sqrt{\delta} {k}^{1/4}= \sqrt{3 \sigma}
    = \Omega\Big((1+\frac{\kappa \delta}{\sigma})\max(\sigma,\sqrt{\sigma \mu_{\max}})\Big).
\end{align*}
\end{proof}
We end this section by showing that via an (inefficient) algorithm one can get the same guarantee for mean estimation as in Theorem~\ref{thm:robust-mean} without the need for certification.
\begin{theorem}[Information Theoretic Upper Bound]
\label{thm:robust-mean-stat}
Let $A$ be an $n \times m$ matrix representing $m$ data points in $n$ dimensions and let $\mu$ be the mean of the data points in the matrix $A$ with $C$ representing the $n \times m$ matrix with each column being $\mu$. Let $\Pi^* = \mu \mu^\top/\|\mu\|^2$ be the one dimensional subspace denoting the projection onto $\mu$ and assume that $\|\Pi^*\|_{q \to 2} \leq \kappa$, for some $q \geq 2$. Let $\tilde{A}$ be the given input such that for every column $j \in [m]$ we have $\|A_j-\tilde{A}_j\|_q \leq \delta$. Furthermore, let $\sigma^2 >0$ be a given upper bound on the variance of the data around the mean, i.e., $\|A-C\| \leq \sigma \sqrt{m}$. Then there is an (inefficient) exponential time algorithm that takes $\tilde{A}$ as input
%if $\kappa \delta = O(\sigma)$, 
% the algorithm from Figure~\ref{ALG:Robust-Mean} when run on $\tilde{A}$, w.h.p. runs in polynomial time, and either certifies that the data has been poisoned, i.e., $\|\tilde{A} - A\| = \Omega(\sigma \sqrt{m})$, or 
and outputs an estimate $\hat{\mu}$ of the true mean $\mu$ such that
% \begin{align*}
%     \|\hat{\mu} - {\mu}\|_2 &\leq O(c_q ) \max \Big(\sigma, \sqrt{\sigma \|\mu\|} \Big),
% \end{align*}
\begin{align*}
    \|\hat{\mu} - {\mu}\|_2 &\leq O(c_q)(1+\frac{\kappa \delta}{\sigma}) \max \Big(\sigma, \sqrt{\sigma \|\mu\|} \Big),
\end{align*}
where $c_q$ is a constant that depends on $q$. In particular, the above implies a relative error guarantee of
\begin{align*}
    \frac{\|\hat{\mu} - {\mu}\|_2}{\|\mu\|} &\leq O(c_q)(1+\frac{\kappa \delta}{\sigma}) \max \Big(\frac{\sigma}{\|\mu\|}, \sqrt{\frac{\sigma}{\|\mu\|}} \Big).
\end{align*}
% In general, for any $\kappa, \delta$ we get a guarantee of
% \begin{align*}
%     \frac{\|\hat{\mu} - {\mu}\|_2}{\|\mu\|} &\leq O(c_q) \max \left(\frac{\sigma + \kappa \delta}{\|\mu\|}, \sqrt{\frac{\sigma + \kappa \delta}{\|\mu\|}} \right).
% \end{align*}
\end{theorem}
\begin{proof}
In order to establish the theorem above we first optimize (\ref{eq:problem:spectral}) exactly. Let $A'$ be the matrix and $\Pi$ be the projection that achieve the minimum of (\ref{eq:problem:spectral}). Then we have that
\begin{align*}
    \|A'-\Pi A'\| &\leq \|A - \Pi^* A\|
\end{align*}
Furthermore, we also have that 
\begin{align*}
    \|A- \Pi^* A\| &\leq \|A-C\| + \|C - \Pi^* A\|\\
    &\leq 2\|A-C\|\\
    &\leq 2\sigma \sqrt{m}.
\end{align*}
Hence both $A$ and $A'$ have good projections onto rank-$1$ subspaces. Plugging into Lemma~\ref{lem:spectral-norm-closeness-of-subspaces} we get that
\begin{align*}
    \|A - \Pi A\| \leq O(\epsilon + \sqrt{\epsilon})\|A\| + \frac{8 \kappa \delta \sqrt{m}}{\sqrt{\epsilon}},
\end{align*}
where $\epsilon = 2\sigma\sqrt{m}/\|A\|$.
Hence, letting $\hat{\mu} = \Mean(\Pi A')$ we get from (\ref{eq:proof-robust-mean-1}) that
\begin{align*}
    \|\mu - \hat{\mu}\| &\leq \frac{1}{\sqrt{m}} \|A-\Pi A\| + c_q\kappa \delta\\
    &\leq O(\epsilon + \sqrt{\epsilon})\frac{\|A\|}{\sqrt{m}} + \frac{8\kappa \delta \sqrt{m}}{\sqrt{\epsilon}} + c_q \kappa \delta.
\end{align*}
The rest of the argument proceeds exactly as in the Proof of Theorem~\ref{thm:robust-mean} by writing $\|A\|$ in terms of $\epsilon$, as done in (\ref{eq:proof-robust-mean-3}).
\end{proof}

\end{document}